\documentclass[12pt,a4paper,oneside]{book}
\usepackage[top=2.5cm,bottom=2.5cm,left=3cm,right=2.5cm]{geometry}
\usepackage[dvipsnames,hyperref]{xcolor}
\usepackage{tikz,latexsym,amsfonts,amsmath,amssymb,amsthm,amscd,%
titlesec,fancyhdr,varioref,contour,color,setspace}
\usepackage[numbers]{natbib}
\defcitealias{FG14}{P1}
\defcitealias{FG15}{P2}
\defcitealias{FG16}{P3}
\defcitealias{FG16-2}{P4}
\defcitealias{GF15}{P5}
\defcitealias{Go16}{P6}
\defcitealias{PG16}{P7}
\defcitealias{Go14}{P8}
\defcitealias{AFG12}{P9}
\usepackage[magyar,english]{babel}
\usepackage[T1]{fontenc}
\usepackage[toc,page]{appendix}
\usepackage{pgfplots,chngcntr}
\counterwithout{figure}{chapter}

\contourlength{1pt}
\contournumber{1}

\definecolor{left}{HTML}{000080}
\definecolor{alizarin}{rgb}{0.82,0.1,0.26}
\definecolor{midnightblue}{rgb}{0,0.2,0.4}

\newcounter{summ}

\newcounter{publ}

\newcounter{bibl}

\newcounter{appe}

\usepackage[%
bookmarksopen=false,
bookmarksnumbered,
ocgcolorlinks,%
colorlinks=true,%
linktoc=page,%
linkcolor=midnightblue,%
urlcolor=midnightblue,%
citecolor=alizarin,%
pdfauthor={Tamás~F.~Görbe},%
pdftitle={Integrable many-body systems of Calogero-Ruijsenaars type},%
pdfsubject={Ph.D. thesis},%
pdfkeywords={thesis,dissertation,university of szeged}]{hyperref}

\usepackage[figure]{hypcap}

\usepackage{doi}
\usepackage{bookmark}
\newcommand{\arXiv}[2]{arXiv:\href{http://arxiv.org/abs/#1}{#1 #2}}

\newcommand{\be}{\begin{equation}}
\newcommand{\ee}{\end{equation}}

\usetikzlibrary{3d}
\usetikzlibrary{shapes.multipart}
\usetikzlibrary{arrows}
\usetikzlibrary{positioning}
\usetikzlibrary{decorations.text}
\usetikzlibrary{patterns}
\usetikzlibrary{arrows}
\usetikzlibrary{calc}
\usetikzlibrary{positioning}
\usetikzlibrary{decorations.pathreplacing}
\usetikzlibrary{decorations.pathmorphing}
\usetikzlibrary{decorations.text}
\usetikzlibrary{backgrounds}

\titleformat{\chapter}[hang]{\LARGE\bfseries}{\thechapter\hspace{20pt}}{0pt}{\LARGE\bfseries}

\numberwithin{equation}{chapter}
\newcounter{thmcounter}
\numberwithin{thmcounter}{chapter}
\theoremstyle{definition}

\newtheorem{definition}[thmcounter]{Definition}
\newtheorem*{definition*}{Definition}
\theoremstyle{remark}
\newtheorem{remark}[thmcounter]{Remark}
\newtheorem*{remark*}{Remark}
\theoremstyle{plain}
\newtheorem{corollary}[thmcounter]{Corollary}
\newtheorem{lemma}[thmcounter]{Lemma}
\newtheorem{proposition}[thmcounter]{Proposition}
\newtheorem{theorem}[thmcounter]{Theorem}
\newtheorem*{lathm*}{Liouville-Arnold theorem}

\renewcommand{\bibname}{References}

\def\BC{\mathrm{BC}}                        %
\def\1{{\boldsymbol 1}}                     %
\def\0{{\boldsymbol 0}}                     %
\def\blambda{{\boldsymbol\lambda}}          %
\def\bJ{{\boldsymbol J}}      %
\def\bP{{\boldsymbol P}}      %
\def\bQ{{\boldsymbol Q}}      %
\def\bR{{\boldsymbol R}}      %
\def\bU{{\boldsymbol U}}      %
\def\bV{{\boldsymbol V}}      %
\def\bW{{\boldsymbol W}}      %
\def\bX{{\boldsymbol X}}      %
\def\cA{{\mathcal A}}                       %
\def\cC{{\mathcal C}}                       %
\def\cD{{\mathcal D}}                       %
\def\cE{\mathcal{E}}
\def\cF{{\mathcal F}}                       %
\def\cG{{\mathcal G}}                       %
\def\cH{{\mathcal H}}                       %
\def\cI{\mathcal{I}}      %
\def\cJ{\mathcal{J}}      %
\def\cL{{\mathcal L}}                       %
\def\cO{{\mathcal O}}                       %
\def\cP{{\mathcal P}}                       %
\def\cZ{{\mathcal Z}}                       %
\def\fc{\mathfrak{c}}     %
\def\fL{{\mathfrak L}}                  %
\def\fQ{{\mathfrak Q}}                  %
\def\tr{\mathrm{tr}}                        %
\def\diag{\mathrm{diag}}                    %
\def\ri{{\rm i}}                            %
\def\rA{{\rm A}}                            %
\def\rB{{\rm B}}                            %
\def\rC{{\rm C}}                            %
\def\rD{{\rm D}}                            %
\DeclareMathOperator{\sgn}{sgn}
\DeclareMathOperator{\ws}{s}
\def\adj{\mathrm{adj}}

\def\id{{\rm id}}                           %
\def\vD{\mathrm{vD}}                         %
\def\Pu{\mathrm{P}}                         %
\def\red{{\rm red}}                         %
\def\vreg{{\rm vreg}}                         %
\def\b{\mathrm{b}}
\def\loc{\mathrm{loc}}
\def\C{\mathbb{C}}                          %
\def\CP{\mathbb{CP}}
\def\D{\mathbb{D}}                          %
\def\N{\mathbb{N}}                          %
\def\R{\mathbb{R}}                          %
\def\T{\mathbb{T}}                          %
\def\Z{\mathbb{Z}}                          %
\def\SL{{\rm SL}}                           %
\def\SB{{\rm SB}}                           %
\def\GL{{\rm GL}}                           %
\def\UN{{\rm U}}                            %
\def\SO{{\rm SO}}                           %
\def\SU{{\rm SU}}                           %
\def\gl{\mathfrak{gl}}                      %
\def\sl{\mathfrak{sl}}                      %
\def\un{\mathfrak{u}}                       %
\def\su{\mathfrak{su}}                      %
\def\fH{\mathfrak{H}}                       %
\def\sV{{\mathsf V}}                        %
\def\sv{{\mathsf v}}                        %
\def\sw{{\mathsf w}}                        %
\newcommand{\bC}{\mathbb{C}}
\newcommand{\bN}{\mathbb{N}}

\newcommand{\cB}{\mathcal{B}}
\newcommand{\cK}{\mathcal{K}}
\newcommand{\cM}{\mathcal{M}}
\newcommand{\cN}{\mathcal{N}}
\newcommand{\cR}{\mathcal{R}}
\newcommand{\cS}{\mathcal{S}}
\newcommand{\cU}{\mathcal{U}}
\newcommand{\cW}{\mathcal{W}}
\newcommand{\mfa}{\mathfrak{a}}
\newcommand{\mfc}{\mathfrak{c}}
\newcommand{\mfg}{\mathfrak{g}}
\newcommand{\mfgl}{\mathfrak{gl}}
\newcommand{\mfk}{\mathfrak{k}}
\newcommand{\mfm}{\mathfrak{m}}
\newcommand{\mfp}{\mathfrak{p}}
\newcommand{\mfu}{\mathfrak{u}}
\newcommand{\mfX}{\mathfrak{X}}
\newcommand{\bsone}{\boldsymbol{1}}
\newcommand{\bsLambda}{\boldsymbol{\Lambda}}
\newcommand{\bsTheta}{\boldsymbol{\Theta}}
\newcommand{\bsX}{\boldsymbol{X}}

\newcommand{\dd}{\mathrm{d}}
\newcommand{\reg}{\mathrm{reg}}
\newcommand{\Id}{\mathrm{Id}}
\newcommand{\ad}{\mathrm{ad}}
\newcommand{\wad}{\widetilde{\mathrm{ad}}}
\renewcommand{\Re}{\mathrm{Re}}
\renewcommand{\Im}{\mathrm{Im}}
\newcommand{\Real}{\mathrm{Re}}
\newcommand{\Imag}{\mathrm{Im}}
\newcommand{\End}{\mathrm{End}}

\newcommand{\htheta}{{\hat{\theta}}}
\newcommand{\hTheta}{{\hat{\Theta}}}
\newcommand{\hbsTheta}{{\hat{\bsTheta}}}
\newcommand{\eps}{\varepsilon}
\newcommand{\half}{\frac{1}{2}}
\newcommand{\PD}[2]{\frac{\partial #1}{\partial #2}}

\makeatletter
\ifpdf
\else
\def\pgf@sys@pdf@mark@pos@pgfpageorigin{\pgfqpoint{0sp}{-3281837sp}}
\fi
\makeatother

\begin{document}
\thispagestyle{empty}
\setcounter{page}{1}
\renewcommand\thepage{\Alph{page}}
\pdfbookmark{Cover}{Cover}
\begin{tikzpicture}[remember picture,overlay,shift={(current page.center)}]
\node[shading=axis,rectangle,left color=left,right color=left!40!white,shading angle=135, anchor=north,minimum width=\paperwidth,minimum height=\paperheight] (box) at (current page.north){};

\begin{scope}[white,scale=1.8,shift={(-3,-2)}]
\draw(3.12,6.2)node{\Large\bf\contour{black}{INTEGRABLE MANY-BODY SYSTEMS}};
\draw(3.12,5.7)node{\Large\bf\contour{black}{OF CALOGERO-RUIJSENAARS TYPE}};
\draw(3.12,4.35)node{\large\sc\bfseries\contour{black}{Tam\'as~F.~G\"orbe}};

\draw(0,0)--(4,0)--(6,2)--(2,2)--cycle (1,0)--(3,2) (2,0)--(4,2) (3,0)--(5,2) (1,1)--(5,1);
\draw(4.5,.5)node[below,rotate=45]{\small Classical}
(5.5,1.5)node[below,rotate=45]{\small Quantum}
(2.5,2)node[above]{I} (3.5,2)node[above]{II}  (4.5,2)node[above]{III}
(5.5,2)node[above]{IV};
\draw[->,thick](3.4,2.5)--(2.6,2.5) node[midway,yshift=1em]{$\alpha\to 0$};
\draw[->,thick](4.4,2.5)--(3.6,2.5) node[midway,yshift=1em]{$\alpha\to\mathrm{i}\alpha$};
\draw(5.4,2.5)--(4.6,2.5) node[midway,yshift=2em]{$\omega\to\pi/2\alpha$};
\draw[->,thick](5.4,2.5)--(4.6,2.5) node[midway,yshift=1em]{$\omega'\to\mathrm{i}\infty$};
\draw(1,1)node[above,xshift=-.7em,yshift=.7em,rotate=45]{Relativistic};
\draw[->,thick](0,-.25)--(0,-1.5) node[midway,right,yshift=1em]{$\beta\to 0$};
\draw[thick,<-](5,0)--(6,1.) node[midway,below,xshift=.7em,yshift=-.7em]{$\hbar\to 0$};
\draw[thick,yellow] (1,-2.5) to [out=45,in=135,looseness=300] (1.01,-2.5);
\draw[thick,yellow] (2,.5) to [out=45,in=135,looseness=300] (2.01,.5);
\draw[thick,yellow] (3,-2.5) to [out=90,in=240,looseness=1] (1,.5);
\draw[thick,yellow] (2,-2.5) to [out=90,in=240,looseness=1] (1,.5);
\draw[thick,yellow] (3,-2.5) to [out=90,in=250,looseness=1] (3,.5);
\draw[thick,yellow] (3,.5) to [out=90,in=170,looseness=.8] (4,1.5);
\draw[thick,yellow] (3,.5) to [out=60,in=120,looseness=1] (4,.5);
\draw[black,fill=white] (1,-2.5)circle(.05) (2,-2.5)circle(.05) (3,-2.5)circle(.05)
(1,.5)circle(.05) (2,.5)circle(.05) (3,.5)circle(.05) (4,.5)circle(.05) (4,1.5)circle(.05);

\begin{scope}[shift={(0,-3)}]
\draw(0,0)--(4,0)--(6,2)--(2,2)--cycle (1,0)--(3,2) (2,0)--(4,2) (3,0)--(5,2) (1,1)--(5,1);
\draw(4.5,.5)node[below,rotate=45]{\small Classical}
(5.5,1.5)node[below,rotate=45]{\small Quantum}
(.5,0)node[below]{I} (1.5,0)node[below]{II}  (2.5,0)node[below]{III}
(3.5,0)node[below]{IV};
\draw(1,1)node[above,xshift=-.7em,yshift=.7em,rotate=45]{Nonrelativistic};
\end{scope}
\end{scope}
\end{tikzpicture}

\titlepage
\setcounter{page}{2}
\renewcommand\thepage{\Alph{page}}
\pdfbookmark{Titlepage}{Title page}
\null\vspace{\stretch{1}}
\begin{center}
{\Large\bf
INTEGRABLE MANY-BODY SYSTEMS\\[.5em]OF CALOGERO-RUIJSENAARS TYPE}\\[1em]
Ph.D. thesis\\[2em]
{\large\sc\bfseries Tam\'{a}s~F.~G\"{o}rbe}\\[1.5em]
\rm Department of Theoretical Physics, University of Szeged\\
Tisza Lajos krt 84-86, H-6720 Szeged, Hungary\\
website: \href{http://www.staff.u-szeged.hu/~tfgorbe/}{www.staff.u-szeged.hu/$\sim$tfgorbe}\\
e-mail: \href{mailto:tfgorbe@physx.u-szeged.hu}{tfgorbe@physx.u-szeged.hu}\\[2em]
Supervisor: \textsc{Prof. L\'{a}szl\'{o}~Feh\'{e}r}\\
Department of Theoretical Physics, University of Szeged

\vspace{\stretch{3}}

\includegraphics[width=.4\textwidth]{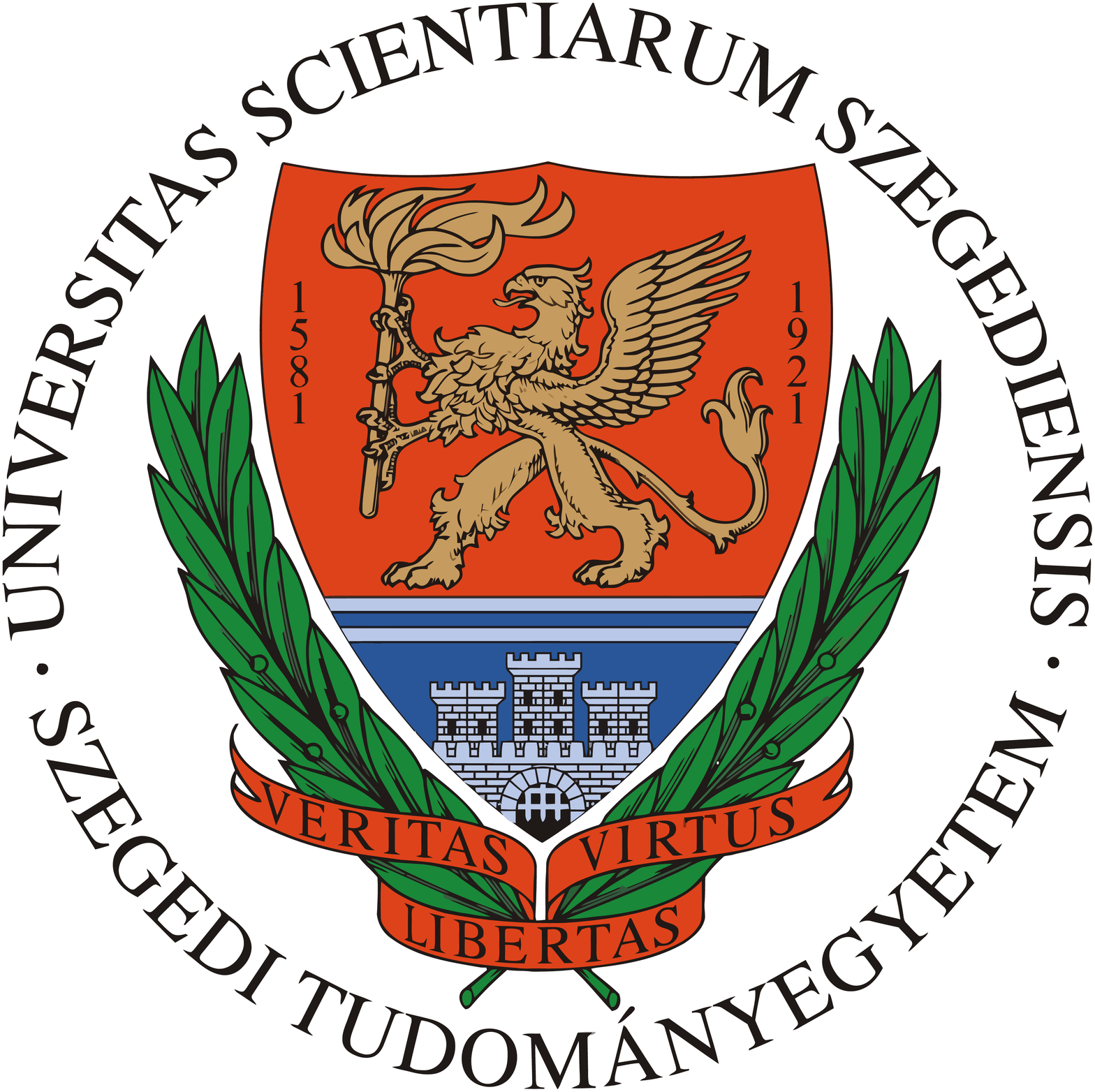}\vspace{\stretch{5}}

Doctoral School of Physics\\
Department of Theoretical Physics\\
Faculty of Science and Informatics\\
University of Szeged\\
Szeged, Hungary\\
2017
\null
\end{center}

\newpage

\thispagestyle{empty}

\noindent
\textbf{Author's declaration}

\medskip\noindent
This thesis is submitted in accordance with the regulations for the Doctor of Philosophy
degree at the University of Szeged. The results presented in the thesis are the author's
original work (see \ref{chap:publ}) with the exceptions of Section \ref{sec:1.1}, which
reviews some pre-existing material, and Subsections \ref{subsec:4.2.2}, \ref{subsec:4.3.1},
\ref{subsec:4.3.3}, \ref{subsec:4.3.4}, \ref{subsec:4.3.5} that contain results obtained
by B.G.~Pusztai. These are included to make the exposition self-contained.

The research was carried out within the Ph.D. programme
``Geometric and field-theoretic aspects of integrable systems''
at the Department of Theoretical Physics, University of Szeged
between September 2013 and August 2016.

\bigskip\bigskip\noindent
\textbf{On the cover}

\medskip\noindent
A schematic diagram of the various versions of Calogero-Ruijsenaars type integrable systems
with dots and lines indicating the ones studied in the thesis.

\bigskip\bigskip\noindent
\textbf{Eprint}

\medskip\noindent
An eprint of the thesis is freely available in the SZTE Repository of Dissertations:\\
\url{http://doktori.bibl.u-szeged.hu/3595/}

\bigskip\bigskip\noindent
\textbf{Keywords}

\medskip\noindent
\emph{integrable systems},
\emph{many-body systems},
\emph{Hamiltonian reduction},
\emph{action-angle duality},
\emph{action-angle variables},
\emph{Calogero-Moser-Sutherland},
\emph{Ruijsenaars-Schneider-van Diejen},
\emph{Poisson-Lie group},
\emph{Heisenberg double},
\emph{Lax matrix},
\emph{spectral coordinates},
\emph{compact phase space},
\emph{root system}

\bigskip\bigskip\noindent
\textbf{2010 Mathematics Subject Classification (MSC2010)}

\medskip\noindent
14H70, 37J15, 37J35, 53D20, 70E40, 70G65, 70H06

\bigskip\bigskip\noindent
\textbf{2010 Physics and Astronomy Classification Scheme (PACS2010)}

\medskip\noindent
02.30.Ik, 05.45.-a, 45.20.Jj, 47.10.Df

\bigskip\bigskip\noindent
\textbf{Please cite this thesis as}

\medskip\noindent
T.F.~G\"{o}rbe,
\emph{Integrable many-body systems of Calogero-Ruijsenaars type},
PhD thesis (2017); \doi{10.14232/phd.3595}

\vfill

\noindent
Copyright \copyright\ 2017 Tam\'{a}s~F.~G\"{o}rbe

\newpage

\thispagestyle{empty}
\null\vspace{\stretch{1}}
\begin{center}
\Large\emph{Orsinak}
\end{center}
\vspace{\stretch{3}}\null

\newpage

\chapter*{Acknowledgements}
\thispagestyle{empty}

I consider myself very lucky for I have so many people to thank.

First of all, it is with great pleasure that I express my deepest gratitude to my Ph.D. supervisor L\'{a}szl\'{o} Feh\'{e}r for the persistent support, guidance, and inspiration he provided me throughout my studies. He introduced me to the fascinating world of Integrable Systems and taught me the importance of honesty, modesty, and having high standards when it comes to research. I feel privileged to have worked with him.

I am grateful to my academic brother G\'{a}bor Pusztai for the work we did together. His superb lectures on Functional Analysis gave me a great appreciation of not only the subject, but also what didactic skill can achieve.

Thanks are due to Martin Halln\"{a}s for making my research visit to Loughborough University possible and for being such a great host. Our joint project helped me to delve into the related area of Multivariate Orthogonal Polynomials.

I am very thankful to Simon Ruijsenaars for hosting me at the University of Leeds. His insightful comments improved my work substantially.

I also want to thank everyone at the Department of Theoretical Physics for creating such a stimulating work environment. I especially enjoyed working alongside my fellow inhabitants of office 231, Gerg\H{o} Ro\'{o}sz and L\'{o}r\'{a}nt Szab\'{o}.

I am indebted to my high school maths teacher J\'{a}nos Mike for the unforgettable classes, which shifted my interest towards mathematics and physics.

The work was supported in part by the Hungarian Scientific Research Fund (OTKA) under the grant no. K-111697, and was sponsored by the EU and the State of Hungary, co-financed by the European Social Fund in the framework of T\'{A}MOP-4.2.4.A/2-11/1-2012-0001 National Excellence Program. The support by the \'{U}NKP-16-3 New National Excellence Program of the Ministry of Human Capacities is also acknowledged. Several short term study programs, made possible by Campus Hungary Scholarships and the Hungarian Templeton Program, contributed to my work.

Last but most important, I am grateful to my wonderful parents and brothers for their encouragement and support. I want to thank all my family and friends. Finally, this thesis is dedicated to my fianc\'{e}e Orsi, without whom it would not have been possible.

\frontmatter

\pdfbookmark{\contentsname}{Contents}
\tableofcontents

\mainmatter

\chapter*{Introduction}
\markboth{Introduction}{}
\addcontentsline{toc}{chapter}{Introduction}

Integrable Systems is a broad area of research that joins seemingly unrelated problems of natural sciences amenable to exact mathematical treatment\footnote{For those who are unfamiliar with Integrable Systems, we recommend reading the survey \cite{Ru15}.}. It serves as a busy crossroad of many subjects ranging from pure mathematics to experimental physics. As a result, the notion of `integrability' is hard to pinpoint as, depending on context, it can refer to different phenomena, and ``where you have two scientists you have (at least) three different definitions of integrability''\footnote{A quote from another good read, the article \emph{Integrability -- and how to detect it} \cite[pp. 31-94]{KGT04}.}. Fortunately, the systems of our interest are integrable in the Liouville sense, which has a precise definition (see below). Loosely speaking, in such systems an abundance of conservation laws restricts the motion and allows the solutions to be exactly expressed with integrals, hence the name.

\section{The golden age of integrable systems}

Studying integrable systems is by no means a new activity as its origins can be traced back to the early days of modern science, when Newton solved the gravitational two-body problem and derived Kepler's laws of planetary motion (for more, see \cite{Si12}). With hindsight, one might say that the solution of the Kepler problem was possible due to the existence of many conserved quantities, such as energy, angular momentum, and the Laplace-Runge-Lenz vector. In fact, the Kepler problem is a prime example of a (super)integrable system (also to be defined). As the mathematical foundations of Newtonian mechanics were established through work of Euler, Lagrange, and Hamilton, more and more examples of integrable/solvable mechanical problems were discovered. Just to name a few, these systems include the harmonic oscillator, the ``spinning tops''/ rigid bodies \cite{Au96} of Euler (1758), Lagrange (1788), and Kovalevskaya (1888), the geodesic motion on the ellipsoid solved by Jacobi (1839), and Neumann's oscillator model (1859). This golden age of integrable systems was ended abruptly in the late 1800s, when Poincar\'{e}, while trying to correct his flawed work on the three-body problem, realized that integrability is a fragile property, that even small perturbations can destroy \cite{Ch07}. This subsided scientific interest and the subject went into a dormant state for more than half a century.

\section{Definition of Liouville integrability}

In the Hamiltonian formulation of Classical Mechanics the state of a physical system, which has $n$ degrees of freedom, is encoded by $2n$ real numbers. These numbers consist of (generalised) positions $q=(q_1,\dots,q_n)$ and (generalised) momenta $p=(p_1,\dots,p_n)$ and are collectively called canonical coordinates of the space of states, the phase space. The time evolution of an initial state $(q_0,p_0)\in\R^{2n}$ is governed by Hamilton's equations of motion, a first-order system of ordinary differential equations that can be written as
\begin{equation*}
\dot{q}_j=\frac{\partial H}{\partial p_j},\quad
\dot{p}_j=-\frac{\partial H}{\partial q_j},\quad j=1,\dots,n,
\end{equation*}
where $H$ is the Hamiltonian, i.e. the total energy of the system. In modern terminology, a Hamiltonian system is a triple $(M,\omega,H)$, where the phase space $(M,\omega)$ is a $2n$-dimensional symplectic manifold\footnote{A symplectic manifold $(M,\omega)$ is a manifold $M$ equipped with a non-degenerate, closed $2$-form $\omega$.} and $H$ is a sufficiently smooth real-valued function on $M$. An initial state $x_0\in M$ evolves along integral curves of the Hamiltonian vector field $X_H$ of $H$ defined via $\omega(X_H,\cdot)=dH$. Darboux's theorem \cite[3.2.2 Theorem]{AM78} guarantees the existence of canonical coordinates\footnote{Notice the slight and customary abuse of notation as we use the symbols $q_j,p_j$ for representing real numbers as well as coordinate functions on $M$. Hopefully, this does not cause any confusion.} $(q,p)$ locally, in which by definition the symplectic form $\omega$ can be written as
\begin{equation*}
\omega=\sum_{j=1}^ndq_j\wedge dp_j,
\end{equation*}
and the equations of motion take the canonical form displayed above. The symplectic form $\omega$ gives rise to a Poisson structure on $M$, which is a handy device that takes two observables $f,g\colon M\to\R$ and turns them into a third one $\{f,g\}$, the Poisson bracket of $f$ and $g$ given by $\{f,g\}=\omega(X_f,X_g)$. In canonical coordinates, we have
$$\{f,g\}=\sum_{j=1}^n\bigg(
\frac{\partial f}{\partial q_j}\frac{\partial g}{\partial p_j}-
\frac{\partial f}{\partial p_j}\frac{\partial g}{\partial q_j}\bigg).$$
It is bilinear, skew-symmetric, satisfies the Jacobi identity and the Leibniz rule. The equations of motion, for any $f\colon M\to\R$, can be rephrased using the Poisson bracket
\begin{equation*}
\dot{f}=\{f,H\}.
\end{equation*}
Consequently, if $\{f,H\}=0$, that is $f$ Poisson commutes with the Hamiltonian $H$, then $f$ is a constant of motion. In fact, this relation is symmetric, since $\{f,H\}=0$ ensures that $H$ is constant along the integral curves of the Hamiltonian vector field $X_f$.

Having conserved quantities can simplify things, since it restricts the motion to the intersection of their level surfaces, selected by the initial conditions. Thus one should aim at finding as many independent Poisson commuting functions as possible. By independence we mean that at generic points (on a dense open subset) of the phase space the functions have linearly independent derivatives. Of course, the non-degeneracy of the Poisson bracket limits the maximum number of independent functions in involution to $n$. If this maximum is reached, we found a Liouville integrable system.

\begin{definition*}
A Hamiltonian system $(M,\omega,H)$, with $n$ degrees of freedom, is called \emph{Liouville integrable}, if there exists a family of independent functions $H_1,\dots,H_n$ in involution, i.e. $\{H_j,H_k\}=0$ for all $j,k$, and $H$ is a function of $H_1,\dots,H_n$.
\end{definition*}

The most prominent feature of Liouville integrable systems is the existence of action-angle variables. This is a system of canonical coordinates $I=(I_1,\dots,I_n)$, $\varphi=(\varphi_1,\dots,\varphi_n)$, in which the (transformed) Hamiltonians $H_1,\dots,H_n$ depend only on the action variables $I$, which are themselves first integrals, while the angle variables $\varphi$ evolve linearly in time. An important result is the following

\begin{lathm*}{\rm \cite[5.2.24 Theorem]{AM78}}
Consider $(M,\omega,H)$ to be a Liouville integrable system with the Poisson commuting functions $H_1,\dots,H_n$. Then the level set
\begin{equation*}
M_c=\{x\in M\mid H_j(x)=c_j,\ j=1,\dots,n\}
\label{}
\end{equation*}
is a smooth $n$-dimensional submanifold of $M$, which is invariant under the Hamiltonian flow of the system. Moreover, if $M_c$ is compact and connected, then it is diffeomorphic to an $n$-torus $\T^n=\{(\varphi_1,\dots,\varphi_n)\mod 2\pi\}$, and the Hamiltonian flow is linear on $M_c$, i.e. the angle variables $\varphi$ on $M_c$ satisfy $\dot{\varphi_j}=\nu_j$, for some constants $\nu_j$, $j=1,\dots,n$. 
\end{lathm*}

The action variables $I$ are also encoded in the level set $M_c$. Roughly speaking, they determine the size of $M_c$, since $I_j$ is obtained by integrating the canonical $1$-form the phase space over the $j$-th cycle of the torus $M_c$.

Another relevant notion is superintegrability, which requires the existence of extra constants of motion.

\begin{definition*}
A Liouville integrable system is called \emph{superintegrable}, if in addition to the Hamiltonians $H_1,\dots,H_n$ there exist independent first integrals $f_1,\dots,f_k$ ($1\leq k<n$). If $k=n-1$, then the system is \emph{maximally superintegrable}.
\end{definition*}

Examples of maximally superintegrable systems include the Kepler problem, the harmonic oscillator with rational frequencies, and the rational Calogero-Moser system considered in Chapter \ref{chap:1}. For more on the theory of integrable systems, see \cite{BBT03}.

\begin{remark*}
It should be noted that, although there is no generally accepted notion of integrability at the quantum level, there are quantum mechanical systems that are called \emph{integrable}.
\end{remark*}

\section{Solitary splendor: The renascence of integrability}

About fifty years ago a revival has taken place in the field of Integrable Systems, when Zabusky and Kruskal \cite{ZK65} conducted a numerical study of the Korteweg-de Vries (KdV) equation\footnote{The motivation for Zabusky and Kruskal's work was to understand the recurrent behaviour in the Fermi-Pasta-Ulam-Tsingou problem \cite{FPUT55}, which turns into the KdV equation in the continuum limit.}, that is the nonlinear $(1+1)$-dimensional partial differential equation
\begin{equation*}
u_t+6uu_x+u_{xxx}=0,
\end{equation*}
and re-discovered its stable solitary wave solutions\footnote{Korteweg and de Vries \cite{KdV1895} devised their equation to reproduce such stable travelling waves, that were first observed by Russell \cite{R1845} in the canals of Edinburgh.}, whose interaction resembled that of colliding particles, hence they gave them the name \emph{solitons}. Subsequently, Kruskal et al. \cite{GGKM67} started a detailed investigation of the KdV equation and found an infinite number of conservation laws associated to it. More explicitly, they showed that the eigenvalues of the Schr\"{o}dinger operator
\begin{equation*}
L=\partial_x^2+u
\end{equation*}
are invariant in time if the `potential' $u$ is a solution of the KdV equation. Moreover, they used the Inverse Scattering Method to reconstruct the potential from scattering data. Lax showed \cite{La68} that the KdV equation is equivalent to an equation involving a pair of operators, now called \emph{Lax pair}, of the form
\begin{equation*}
\dot{L}=[B,L],
\end{equation*}
where $L$ is the Schr\"{o}dinger operator above, and $B$ is a skew-symmetric operator. The commutator form of the Lax equation explains the isospectral nature of the operator $L$. The connection to integrable systems was made by Faddeev and Zakharov \cite{ZF71}, who showed that the KdV equation can be viewed as a completely integrable Hamiltonian system with infinitely many degrees of freedom. These initial findings renewed interest in integrable systems and their applications. For example, Lax pairs associated to other integrable systems were found and used to generate conserved quantities.

The ideas and developments presented so far were all about the KdV equation. However, there are other physically relevant nonlinear PDEs with soliton solutions, which have been solved using the Inverse Scattering Method. For example, the sine-Gordon equation \cite{AKNS73}
\begin{equation*}
\varphi_{tt}-\varphi_{xx}+\sin\varphi=0,
\end{equation*}
which can be interpreted as the equation that describes the twisting of a continuous chain of needles attached to an elastic band. It has different kinds of soliton solutions, called \emph{kink}, \emph{antikink}, and \emph{breather}, that can interact with one another. It is a relativistic equation, since its solutions are invariant under the action of the Poincar\'{e} group of $(1+1)$-dimensional space-time.

The nonlinear Schr\"{o}dinger equation \cite{ZS72} is another famous example. It reads
\begin{equation*}
\ri\psi_t+\frac{1}{2}\psi_{xx}-\kappa|\psi|^2\psi=0,
\end{equation*}
where $\psi$ is a complex-valued wave function and $\kappa$ is constant. It is also an exactly solvable Hamiltonian system \cite{ZM74}.
The equation is nonrelativistic (Galilei invariant).

Now let us list some applications of these soliton equations. The Korteweg-de Vries equation can be applied to describe shallow-water waves with weakly non-linear restoring forces and  long internal waves in a density-stratified ocean. It is also useful in modelling ion acoustic waves in a plasma and acoustic waves on a crystal lattice. The kinks and breathers of the sine-Gordon equation can used as models of nonlinear excitations in complex systems in physics and even in cellular structures. The nonlinear Schr\"{o}dinger equation is of central importance in fluid dynamics, plasma physics, and nonlinear optics as it appears in the Manakov system, a model of wave propagation in fibre optics.

Parallel to soliton theory, various exactly solvable quantum many-body systems appeared, that describe the interaction of quantum particles in one spatial dimension. These models proved to be a fruitful source of ideas and a great influence on the development of mathematical physics. Earlier important milestones include Bethe's solution of the one-dimensional Heisenberg model (Bethe Ansatz, 1931), Pauling's work on the $6$-vertex model (1935), Onsager's solution of the planar Ising model (1944), and the delta Bose gas of Lieb-Liniger (1963). At the level of classical mechanics, a crucial step was Toda's discovery of a nonlinear, one-dimensional lattice model \cite{To67} with soliton solutions. The Toda lattice is an infinite chain of particles interacting via exponential nearest neighbour potential. The nonperiodic and periodic Toda chains are $n$ particles with such interaction put on a line and a ring, and have the Hamiltonians
\begin{equation*}
H_{\mathrm{np}}=\frac{1}{2}\sum_{j=1}^np_j^2+\sum_{j=1}^{n-1}e^{2(q_{j+1}-q_j)},\quad\text{and}\quad
H_{\mathrm{per}}=\frac{1}{2}\sum_{j=1}^np_j^2+\sum_{j=1}^{n-1}e^{2(q_{j+1}-q_j)}+g^2e^{2(q_1-q_n)},
\end{equation*}
respectively. H\'{e}non \cite{He74} found $n$ conserved quantities for both of these systems, and Flashka \cite{Fl74,Fl74-2} and Manakov \cite{Ma75} found Lax pairs giving rise to these first integrals and proved them to be in involution. Therefore the Toda lattices are completely integrable. The scattering theory of the nonperiodic Toda lattice was examined by Moser \cite{Mo75-2}. Bogoyavlensky \cite{Bo76} generalised the Toda lattice to root systems of simple Lie algebras. Olshanetsky, Perelomov \cite{OP79,OP80} and Kostant \cite{Ko79} initiated group-theoretic treatments.

\section{Calogero-Ruijsenaars type systems}

In the early 1970s further exactly solvable quantum many-body systems were found by Calogero \cite{Ca69,Ca71} and Sutherland \cite{Su71,Su72}. Calogero considered particles on a line in harmonic confinement with a pairwise interaction inversely proportional to the square of their relative distances (rational case). Sutherland solved the corresponding problem of particles on a ring, i.e. interacting via a periodic pair-potential (trigonometric case). The classical versions were examined by Moser \cite{Mo75}, who provided Lax pairs, analysed the particle scattering in the rational case, which he proved to be Liouville integrable\footnote{The rational three-body system was treated by Marchioro \cite{Ma70} and to some extent by Jacobi \cite{J1886}.}. Models with short-range interaction (hyperbolic case) \cite{CRM75} and with elliptic potentials (elliptic case) \cite{Ca75} were also formulated (see Figures \ref{fig:1} and \ref{fig:2}).

We give a short description of the classical systems. Let the number of particles $n$ be fixed, $q=(q_1,\dots,q_n)\in\R^n$ collect the particle-positions and $p=(p_1,\dots,p_n)\in\R^n$ the conjugate momenta. The configuration space is usually some open domain $C\subseteq\R^n$, and the phase space $M$ is its cotangent bundle
\begin{equation*}
M=T^\ast C=\{(q,p)\mid q\in C,\ p\in\R^n\},
\end{equation*}
equipped with the canonical symplectic form
\begin{equation*}
\omega=\sum_{j=1}^ndq_j\wedge dp_j.
\end{equation*}
The Hamiltonian of the models can be written in the general form
\begin{equation*}
H_{\mathrm{nr}}=\frac{1}{2m}\sum_{j=1}^np_j^2+\frac{g^2}{m}\sum_{j<k}V(q_j-q_k),
\end{equation*}
where $m>0$ denotes the mass of particles, $g$ is a positive coupling constant regulating the strength of particle repulsion\footnote{The interaction is attractive, if $g^2<0$. Setting $g=0$ yields free particles.}, and the pair-potential $V$ can be one of four types:
\begin{equation*}
V(q)=\begin{cases}
1/q^2,&\text{rational (I)},\\
\alpha^2/\sinh^2(\alpha q),&\text{hyperbolic (II)},\\
\alpha^2/\sin^2(\alpha q),&\text{trigonometric (III)},\\
\wp(q;\omega,\omega'),&\text{elliptic (IV)}.\\
\end{cases}
\label{}
\end{equation*}
Here $\wp$ stands for Weierstrass's elliptic function with half-periods $(\omega,\omega')\in\R_+\times\ri\R_+$. By taking the parameter $\alpha\to\ri\alpha$, II and III are exchanged, while $\alpha\to 0$ produces I.

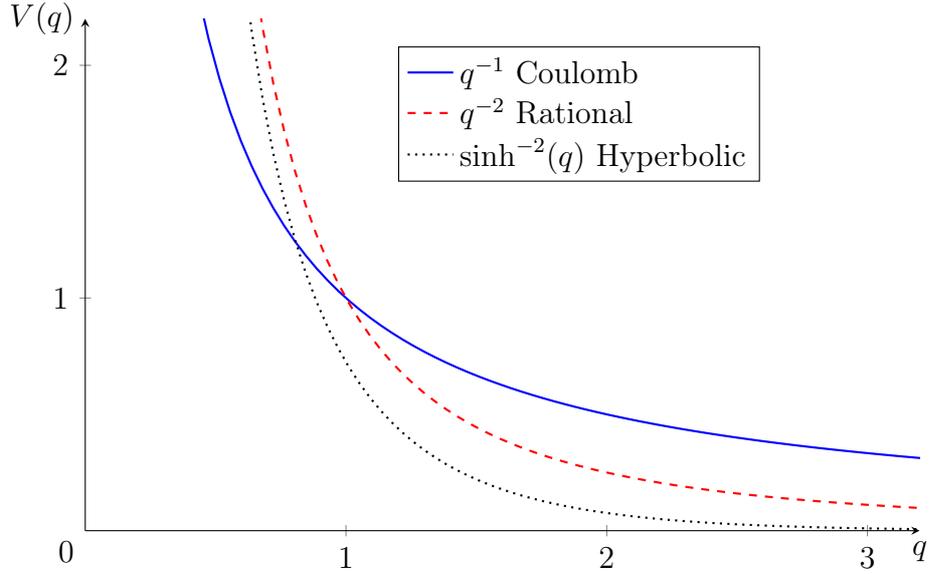
\begin{figure}[h!]
\centering
\begin{tikzpicture}
\begin{axis}[width=.81\textwidth,height=.54\textwidth,axis lines=center,
xlabel=$q$,ylabel = {$V(q)$},xlabel style={below},ylabel style={left},
legend cell align={left},xmin=0,xmax=3.2,ymin=0,ymax=2.2,xtick={1,2,3},ytick={1,2},
legend style={at={(axis cs:1.2,1.5)},anchor=south west},
after end axis/.code={\path (axis cs:0,0) node [anchor=north east] {0};}
]

\addplot[domain=.1:4.2,samples=100,color=blue,thick]{1/x};
\addlegendentry{$q^{-1}$ Coulomb}

\addplot[domain=.1:4.2,samples=100,color=red,dashed,thick]{1/x^2};
\addlegendentry{$q^{-2}$ Rational}

\addplot[domain=.1:4.2,samples=100,color=black,dotted,thick]{4/(exp(x)-exp(-x))^2};
\addlegendentry{$\sinh^{-2}(q)$ Hyperbolic}

\end{axis}
\end{tikzpicture}
\caption{Three repulsive potential functions. The Coulomb potential $V(q)=q^{-1}$ (solid blue) and rational potential $V(q)=q^{-2}$ (dashed red) express long-range interaction in comparison to the hyperbolic potential $V(q)=\sinh^{-2}(q)$ (dotted black).}
\label{fig:1}
\end{figure}

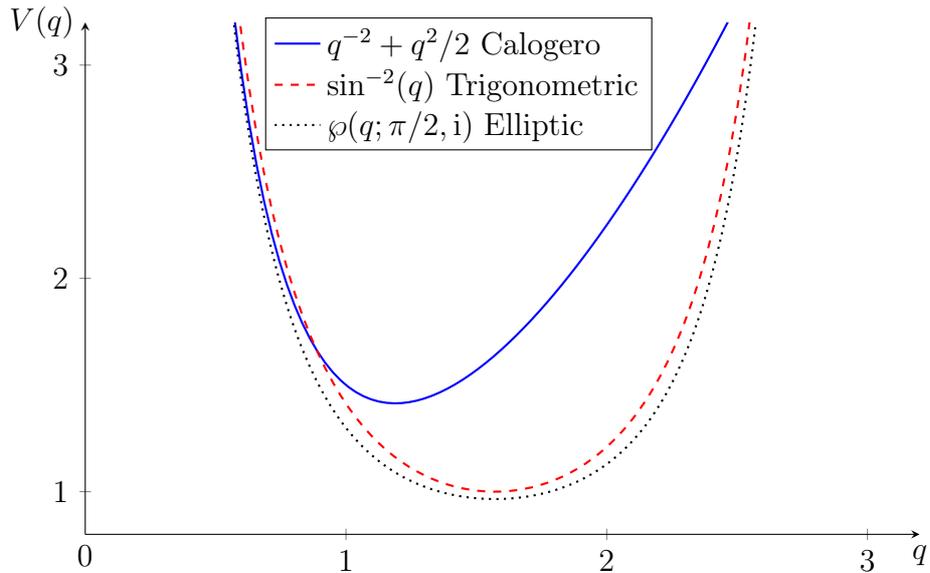
\begin{figure}[h!]
\centering
\begin{tikzpicture}
\begin{axis}[width=.81\textwidth,height=.54\textwidth,axis lines=center,
xlabel=$q$,ylabel={$V(q)$},xlabel style={below},ylabel style={left},
legend cell align={left},xmin=0,xmax=3.2,ymin=.8,ymax=3.2,xtick={1,2,3},ytick={1,2,3},
legend style={at={(axis cs:.69,2.6)},anchor=south west},
after end axis/.code={\path (axis cs:0,.8) node [anchor=north] {0};}
]

\addplot[domain=.1:3.,samples=100,color=blue,thick]{1/x^2+.5*x^2};
\addlegendentry{$q^{-2}+q^2/2$ Calogero}

\addplot[domain=.1:3.,samples=100,color=red,dashed,thick]{1/sin(deg(x))^2};
\addlegendentry{$\sin^{-2}(q)$ Trigonometric}

\addplot[domain=.1:3.,color=black,dotted,thick]
coordinates{
(0.1,100.004)(0.11,82.6496)(0.12,69.4503)(0.13,59.1785)(0.14,51.0284)(0.15,44.4536)(0.16,39.0729)(0.17,34.6138)(0.18,30.8774)(0.19,27.7155)(0.2,25.0162)(0.21,22.6936)(0.22,20.6807)(0.23,18.9249)(0.24,17.3843)(0.25,16.0251)(0.26,14.82)(0.27,13.7466)(0.28,12.7864)(0.29,11.9241)(0.3,11.1469)(0.31,10.4439)(0.32,9.80614)(0.33,9.22572)(0.34,8.69604)(0.35,8.21138)(0.36,7.76682)(0.37,7.35809)(0.38,6.98148)(0.39,6.63373)(0.4,6.312)(0.41,6.01379)(0.42,5.73689)(0.43,5.47935)(0.44,5.23942)(0.45,5.01557)(0.46,4.80642)(0.47,4.61072)(0.48,4.42738)(0.49,4.25539)(0.5,4.09387)(0.51,3.942)(0.52,3.79904)(0.53,3.66434)(0.54,3.53729)(0.55,3.41735)(0.56,3.304)(0.57,3.19679)(0.58,3.09531)(0.59,2.99916)(0.6,2.90801)(0.61,2.82152)(0.62,2.7394)(0.63,2.66137)(0.64,2.58718)(0.65,2.5166)(0.66,2.44941)(0.67,2.3854)(0.68,2.32441)(0.69,2.26624)(0.7,2.21075)(0.71,2.15778)(0.72,2.10719)(0.73,2.05886)(0.74,2.01266)(0.75,1.96847)(0.76,1.92621)(0.77,1.88576)(0.78,1.84703)(0.79,1.80994)(0.8,1.77441)(0.81,1.74035)(0.82,1.7077)(0.83,1.67639)(0.84,1.64635)(0.85,1.61753)(0.86,1.58987)(0.87,1.56332)(0.88,1.53782)(0.89,1.51333)(0.9,1.4898)(0.91,1.46719)(0.92,1.44547)(0.93,1.42459)(0.94,1.40451)(0.95,1.38521)(0.96,1.36665)(0.97,1.34879)(0.98,1.33162)(0.99,1.3151)(1.,1.29921)(1.01,1.28392)(1.02,1.26921)(1.03,1.25505)(1.04,1.24143)(1.05,1.22832)(1.06,1.21571)(1.07,1.20357)(1.08,1.1919)(1.09,1.18066)(1.1,1.16985)(1.11,1.15944)(1.12,1.14944)(1.13,1.13981)(1.14,1.13056)(1.15,1.12166)(1.16,1.1131)(1.17,1.10488)(1.18,1.09698)(1.19,1.08939)(1.2,1.0821)(1.21,1.0751)(1.22,1.06839)(1.23,1.06195)(1.24,1.05577)(1.25,1.04986)(1.26,1.04419)(1.27,1.03877)(1.28,1.03358)(1.29,1.02863)(1.3,1.02389)(1.31,1.01938)(1.32,1.01508)(1.33,1.01098)(1.34,1.00709)(1.35,1.0034)(1.36,0.999895)(1.37,0.996582)(1.38,0.993455)(1.39,0.990509)(1.4,0.987741)(1.41,0.985147)(1.42,0.982725)(1.43,0.980471)(1.44,0.978383)(1.45,0.976459)(1.46,0.974696)(1.47,0.973091)(1.48,0.971645)(1.49,0.970353)(1.5,0.969216)(1.51,0.968232)(1.52,0.967399)(1.53,0.966716)(1.54,0.966184)(1.55,0.965801)(1.56,0.965566)(1.57,0.96548)(1.58,0.965542)(1.59,0.965753)(1.6,0.966113)(1.61,0.966621)(1.62,0.96728)(1.63,0.968089)(1.64,0.969049)(1.65,0.970162)(1.66,0.971429)(1.67,0.972851)(1.68,0.974429)(1.69,0.976167)(1.7,0.978066)(1.71,0.980127)(1.72,0.982355)(1.73,0.98475)(1.74,0.987316)(1.75,0.990056)(1.76,0.992974)(1.77,0.996072)(1.78,0.999355)(1.79,1.00283)(1.8,1.00649)(1.81,1.01035)(1.82,1.01441)(1.83,1.01868)(1.84,1.02316)(1.85,1.02786)(1.86,1.03278)(1.87,1.03793)(1.88,1.04331)(1.89,1.04894)(1.9,1.05482)(1.91,1.06095)(1.92,1.06735)(1.93,1.07402)(1.94,1.08097)(1.95,1.08821)(1.96,1.09575)(1.97,1.1036)(1.98,1.11177)(1.99,1.12027)(2.,1.12912)(2.01,1.13832)(2.02,1.14788)(2.03,1.15783)(2.04,1.16816)(2.05,1.17891)(2.06,1.19008)(2.07,1.20168)(2.08,1.21375)(2.09,1.22628)(2.1,1.23931)(2.11,1.25285)(2.12,1.26692)(2.13,1.28154)(2.14,1.29674)(2.15,1.31253)(2.16,1.32895)(2.17,1.34601)(2.18,1.36376)(2.19,1.3822)(2.2,1.40139)(2.21,1.42134)(2.22,1.44209)(2.23,1.46368)(2.24,1.48614)(2.25,1.50952)(2.26,1.53385)(2.27,1.55919)(2.28,1.58557)(2.29,1.61305)(2.3,1.64168)(2.31,1.67152)(2.32,1.70262)(2.33,1.73506)(2.34,1.76888)(2.35,1.80418)(2.36,1.84102)(2.37,1.87948)(2.38,1.91965)(2.39,1.96162)(2.4,2.00549)(2.41,2.05136)(2.42,2.09934)(2.43,2.14956)(2.44,2.20215)(2.45,2.25723)(2.46,2.31496)(2.47,2.37549)(2.48,2.439)(2.49,2.50567)(2.5,2.5757)(2.51,2.6493)(2.52,2.7267)(2.53,2.80815)(2.54,2.89393)(2.55,2.98432)(2.56,3.07964)(2.57,3.18025)(2.58,3.28652)(2.59,3.39886)(2.6,3.51773)(2.61,3.64361)(2.62,3.77705)(2.63,3.91865)(2.64,4.06905)(2.65,4.22899)(2.66,4.39925)(2.67,4.58072)(2.68,4.77438)(2.69,4.98131)(2.7,5.20273)(2.71,5.44)(2.72,5.69463)(2.73,5.96832)(2.74,6.26299)(2.75,6.58082)(2.76,6.92424)(2.77,7.29605)(2.78,7.69943)(2.79,8.13801)(2.8,8.61597)(2.81,9.13812)(2.82,9.71004)(2.83,10.3382)(2.84,11.0302)(2.85,11.7949)(2.86,12.6429)(2.87,13.5865)(2.88,14.6407)(2.89,15.8235)(2.9,17.1564)(2.91,18.6661)(2.92,20.385)(2.93,22.3538)(2.94,24.623)(2.95,27.2571)(2.96,30.3386)(2.97,33.9747)(2.98,38.3069)(2.99,43.5248)(3.,49.8873)};
\addlegendentry{$\wp(q;\pi/2,\ri)$ Elliptic}
\end{axis}
\end{tikzpicture}
\caption{Three confining potential functions. Calogero potential $V(q)=q^{-2}+q^2/2$ (solid blue), trigonometric potential $V(q)=\sin^{-2}(q)$ (dashed red), and elliptic potential $V(q)=\wp(q;\omega,\omega')$ (dotted black) with half-periods $\omega=\pi/2$, $\omega'=\ri$.}
\label{fig:2}
\end{figure}

The elliptic potential degenerates to the other ones in various limits\footnote{It is worth mentioning that the Toda lattices (both periodic and nonperiodic) can be also obtained from the elliptic model. For details, see \cite{In89,Ru90-2,Ru99}.}
\begin{equation*}
\wp(q;\omega,\omega')\to\begin{cases}
1/q^2,&\text{if}\ \omega\to\infty,\ \omega'\to\ri\infty,\\
\alpha^2/3+\alpha^2/\sinh^2(\alpha q),&\text{if}\ \omega\to\infty,\ \omega'\to\ri\pi/2\alpha,\\
-\alpha^2/3+\alpha^2/\sin^2(\alpha q),&\text{if}\ \omega\to\pi/2\alpha,\ \omega'\to\ri\infty.
\end{cases}
\end{equation*}
These models are nonrelativistic, that is invariant under the Galilei group of $(1+1)$-dimensional space-time. Relativistic (i.e. Poincar\'{e}-invariant) integrable deformations were constructed\footnote{With the motivation to reproduce the scattering of sine-Gordon solitons using interacting particles.} by Ruijsenaars and Schneider \cite{RS86}, and Ruijsenaars \cite{Ru87}. The Hamiltonians of the relativistic systems read
\begin{equation*}
H_{\mathrm{rel}}=\frac{1}{\beta^2m}\sum_{j=1}^n\cosh(\beta p_j)\prod_{k\neq j}f(q_j-q_k),
\end{equation*}
where $\beta=1/mc>0$ is the deformation parameter ($c$ can be interpreted as the speed of light), and the function $f$ can be one of the following
\begin{equation*}
f(q)=\begin{cases}
(1+\beta^2g^2/q^2)^{1/2},&\text{rational (I)},\\
(1+\sin^2(\alpha\beta g)/\sinh^2(\alpha q))^{1/2},&\text{hyperbolic (II)},\\
(1+\sinh^2(\alpha\beta g)/\sin^2(\alpha q))^{1/2},&\text{trigonometric (III)},\\
(\sigma^2(\ri\beta g;\omega,\omega')[\wp(\ri\beta g;\omega,\omega')-\wp(q;\omega,\omega')])^{1/2},&\text{elliptic (IV)}.\\
\end{cases}
\end{equation*}
Here $\sigma$ is the Weierstrass sigma function. In the nonrelativistic limit $\beta\to 0$ we get
\begin{equation*}
\lim_{\beta\to 0}(H_{\mathrm{rel}}-\frac{n}{\beta^2m})=H_{\mathrm{nr}}.
\end{equation*}
The quantum Hamiltonians at the nonrelativistic level consist of commuting partial differential operators, obtained from classical Hamiltonians by canonical quantization. For example, the Hamiltonian operator can be written as
\begin{equation*}
\hat{H}_{\mathrm{nr}}=-\frac{\hbar^2}{2m}\sum_{j=1}^n\frac{\partial^2}{\partial q_j^2}+\frac{g(g-\hbar)}{m}\sum_{j<k}V(\hat{q}_j-\hat{q}_k).
\end{equation*}
The corresponding Hilbert space, on which these operators act, is the space $L^2(C,dq)$ of square integrable complex-valued functions over the classical configuration space $C$.
In contrast, the relativistic quantum Hamiltonians have an exponential dependence on the momentum operators, resulting in analytic differential operators, such as
\begin{equation*}
\hat{H}_{\mathrm{rel}}=\frac{1}{2\beta^2m}(S_1+S_{-1}),\quad\text{with}\quad
S_{\pm1}=\sum_{j=1}^n\bigg[\prod_{k\neq j}f_\mp(\hat{q}_j-\hat{q}_k)\bigg]
e^{\mp\ri\hbar\beta\partial_j}
\bigg[\prod_{k\neq j}f_\pm(\hat{q}_j-\hat{q}_k)\bigg].
\end{equation*}
In the elliptic case $f_\pm(q)=\sigma(\ri\beta g+q)/\sigma(q)$ and the other cases are obtained as limits. Therefore these operators act on functions that have an analytic continuation to an at least $2\hbar\beta$ wide strip in the complex plane. For more details on these models, the reader is referred to the articles \cite{Ca08,Ru06,Ru09} or the exhaustive surveys \cite{Ru90-2,Ru99}.

A scheme of the Calogero-Ruijsenaars type systems is depicted in Figure \ref{fig:3}.

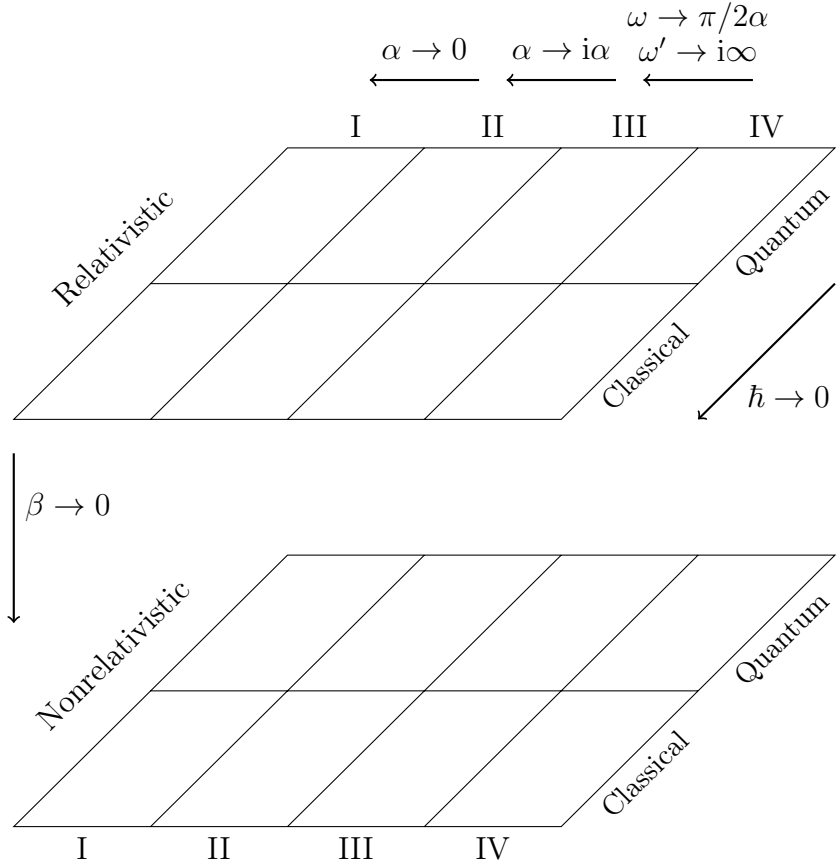
\begin{figure}[h!]
\centering
\begin{tikzpicture}[scale=1.8]
\draw(0,0)--(4,0)--(6,2)--(2,2)--cycle (1,0)--(3,2) (2,0)--(4,2) (3,0)--(5,2) (1,1)--(5,1);
\draw(4.5,.5)node[below,rotate=45]{\small Classical}
(5.5,1.5)node[below,rotate=45]{\small Quantum}
(2.5,2)node[above]{I} (3.5,2)node[above]{II}  (4.5,2)node[above]{III}
(5.5,2)node[above]{IV};
\draw[->,thick](3.4,2.5)--(2.6,2.5) node[midway,yshift=1em]{$\alpha\to 0$};
\draw[->,thick](4.4,2.5)--(3.6,2.5) node[midway,yshift=1em]{$\alpha\to\mathrm{i}\alpha$};
\draw(5.4,2.5)--(4.6,2.5) node[midway,yshift=2em]{$\omega\to\pi/2\alpha$};
\draw[->,thick](5.4,2.5)--(4.6,2.5) node[midway,yshift=1em]{$\omega'\to\mathrm{i}\infty$};
\draw(1,1)node[above,xshift=-.7em,yshift=.7em,rotate=45]{Relativistic};
\draw[->,thick](0,-.25)--(0,-1.5) node[midway,right,yshift=1em]{$\beta\to 0$};
\draw[thick,<-](5,0)--(6,1.) node[midway,below,xshift=.7em,yshift=-.7em]{$\hbar\to 0$};
\begin{scope}[shift={(0,-3)}]
\draw(0,0)--(4,0)--(6,2)--(2,2)--cycle (1,0)--(3,2) (2,0)--(4,2) (3,0)--(5,2) (1,1)--(5,1);
\draw(4.5,.5)node[below,rotate=45]{\small Classical}
(5.5,1.5)node[below,rotate=45]{\small Quantum}
(.5,0)node[below]{I} (1.5,0)node[below]{II}  (2.5,0)node[below]{III}
(3.5,0)node[below]{IV};
\draw(1,1)node[above,xshift=-.7em,yshift=.7em,rotate=45]{Nonrelativistic};
\end{scope}
\end{tikzpicture}
\caption{Schematics of Calogero-Ruijsenaars type systems.}
\label{fig:3}
\end{figure}

The above-mentioned models have generalisations formulated using root systems\footnote{A short summary of facts about root systems can be found in \cite{Sa06}. For more details, see \cite{Hu72}.}. To this end, notice that in the Hamiltonians presented above $q_j-q_k=a\cdot q$ are the inner product of $q$ and the root vectors $a\in\rA_{n-1}$ of the simple Lie algebra $\sl(n,\C)$. It turns out that if $\rA_{n-1}$ is replaced with any root system the resulting system is still integrable. Such root system generalisations were introduced by Olshanetsky and Perelomov \cite{OP81,OP83}, who found Lax pairs and proved integrability for models attached to the classical root systems $\rB_n,\rC_n,\rD_n$ (and $\BC_n$). For arbitrary root systems, the integrability of non-elliptic quantum systems was showed by Heckman and Opdam \cite{HO87}, and Sasaki et al. \cite{KPS00}, and for classical systems (including the elliptic case) by Khastgir and Sasaki \cite{KS01}. Integrable Ruijsenaars-Schneider models attached to non-A type root systems were found by van Diejen \cite{vD94,vD94-2,vD95,vD95-2}. It is a remarkable fact that the eigenfunctions of these generalised Calogero-Ruijsenaars type operators are multivariate orthogonal polynomials, and the equilibrium positions of the classical systems are given by the zeros of classical orthogonal polynomials \cite{Ca77,OS06}.

There are other ways to generalise the Calogero-Ruijsenaars type systems, e.g. by allowing internal degrees of freedom (spins) \cite{GH84,Re17} or supersymmetry \cite{SS93,BTW98,BDM15}.

\section{Basic idea of Hamiltonian reduction}

In their pioneering work, Kazhdan, Kostant, and Sternberg \cite{KKS78} offered a key insight into the origin of the Poisson commuting first integrals of Calogero-Moser-Sutherland models. In a nutshell, they derived the complicated motion of these many-body systems by applying Marsden-Weinstein reduction \cite{MW74} to a higher dimensional free particle. The reduction framework and its application to Hamiltonian systems have undergone considerable development since then \cite{OR04,MMOPR07}. Here we only present a description of the reduction machinery that is tailored to our purposes. Part \ref{part:1} of the thesis contains specific implementations of this approach.

The reduction procedure starts with choosing a `big phase space' of group-theoretic origin. This might be, say, the cotangent bundle $P=T^\ast X$ of a matrix Lie group or algebra $X$. The natural symplectic structure $\Omega$ of the cotangent bundle $P$ permits one to define a Hamiltonian system $(P,\Omega,\cH)$ by specifying a Hamiltonian $\cH\colon P\to\R$. If $\cH$ is simple enough, then the equations of motion can be solved, or even better, a family of Poisson commuting functions $\{\cH_j\}$ be found, which $\cH$ is a member/function of. Then by choosing an appropriate group action (of some group $G$) on $X$ (hence $P$), under which $\cH_j$ are invariant\footnote{It can go the other way around, that is have a group action first, then find invariant functions.}, one can construct the momentum map $\Phi\colon P\to\mathfrak{g}^\ast$ corresponding to this action. Fixing the value $\mu$ of the momentum map $\Phi$ produces a level surface $\Phi^{-1}(\mu)$ in the `big phase space'. This constraint surface is foliated by the orbits of the isotropy/gauge group $G_\mu\subset G$ of the momentum value. The reduced phase space $(P_\red,\omega_\red)$ consists of these orbits. The point is that the flows of the commuting `free' Hamiltonians $\{\cH_j\}$ preserve the momentum surface and are constant along obits. Therefore they admit reduced versions $H_j\colon P_\red\to\R$, which still Poisson commute\footnote{With respect to the Poisson bracket induced by the reduced symplectic form $\omega_\red$.} and the resulting Hamiltonian system $(P_\red,\omega_\red,H)$ is Liouville integrable. In practice, we model the reduced phase space by a smooth slice $S$ of the gauge orbits (see Figure \ref{fig:4}). This slice $S$ is obtained by solving the momentum equation $\Phi=\mu$. Systems in action-angle duality (see below) can emerge in this picture if one has two sets of invariant functions and two models $S,\tilde{S}$ of the reduced phase space.

\begin{figure}[h!]
\centering
\begin{tikzpicture}[scale=1.3,vector/.style={thick, ->, >=stealth'}]
\draw[thick] (0,.5) to[out=40, in=160]
     (8,1) to[out=190, in=30]
     (5,0) to[out=160, in=5] (0,.5);
\shade[left color=lightgray,right color=white] (0,.5) to[out=40, in=160]
     (8,1) to[out=190, in=30]
     (5,0) to[out=160, in=5] (0,.5);
\draw (7.4,1.2) to[out=190, in=30] (4.5,.17);
\draw (6.8,1.4) to[out=190, in=35] (4.1,.3);
\draw (6.2,1.57) to[out=185, in=40] (3.7,.4);
\draw (5.6,1.71) to[out=185, in=40] (3.2,.48);
\draw (4.95,1.83) to[out=185, in=45] (2.7,.55);
\draw (4.13,1.89) to[out=195, in=47] (2.2,.58);
\draw (3.4,1.88) to[out=200, in=50] (1.7,.6);
\draw (2.7,1.8) to[out=205, in=55] (1.2,.57);
\draw (1.7,1.53) to[out=210, in=60] (.7,.55);
\draw (.8,1.07) to[out=215, in=60] (.3,.52);
\def\myshift1#1{\raisebox{-2.5ex}}
\draw[postaction={decorate,decoration={text along path,text align=center,text={|\small\myshift1|momentum level surface}}}] (5,0) to[out=30, in=190] (8,1);
\def\myshift2#1{\raisebox{1.5ex}}
\draw[postaction={decorate,decoration={text along path,text align=center,text={|\small\myshift2|
orbits of isotropy group = points of reduced phase space}}}]
(0,.5) to[out=40, in=160] (8,1);
\draw[thick,red] (.31,.7) to[out=30, in=165] (6.8,.95)
node[xshift=-2ex,yshift=-1ex]{$S$};;
\draw[thick,blue] (.29,.6) to[out=20, in=150] (5.4,.4) node[xshift=-2ex,yshift=-1ex]{$\tilde S$};
\end{tikzpicture}
\caption{The geometry of reduction and action-angle duality.}
\label{fig:4}
\end{figure}
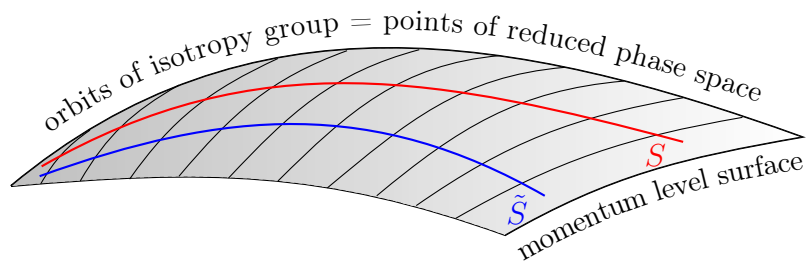

\section{Action-angle dualities}

Action-angle duality is a relation between two Liouville integrable systems, say $(M,\omega,H)$ and $(\tilde{M},\tilde{\omega},\tilde{H})$, requiring the existence of canonical coordinates $(q,p)$ on $M$ and $(\tilde{q},\tilde{p})$ on $\tilde M$ (or on dense open submanifolds of $M$ and $\tilde{M}$) and a global symplectomorphism $\cR\colon M\to\tilde{M}$, the \emph{action-angle map}, such that $(\tilde{q},\tilde{p})\circ\cR$ are action-angle variables for the Hamiltonian $H$ and $(q,p)\circ\cR^{-1}$ are action-angle variables for the Hamiltonian $\tilde H$. This means that $H\circ\cR^{-1}$ depends only on $\tilde{q}$ and $\tilde H \circ \cR$ only on $q$. Then one says that the systems $(M,\omega,H)$ and $(\tilde{M},\tilde{\omega},\tilde{H})$ are in \emph{action-angle duality}. In addition, for the systems of our interest it also happens that the Hamiltonian $H$, when expressed in the coordinates $(q,p)$, admits interpretation in terms of interacting `particles' with position variables $q$, and similarly, $\tilde{H}$ expressed in $(\tilde{q},\tilde{p})$ describes the interacting points with positions $\tilde{q}$. Thus $q$ are particle positions for $H$ and action variables for $\tilde{H}$, and the $\tilde{q}$ are positions for $\tilde{H}$ and actions for $H$. The significance of this curious property is clear for instance from the fact that it persists at the quantum mechanical level as the bispectral character of the wave functions \cite{DG86,Ru90}, which are important special functions.\vspace{-1em}

\begin{figure}[h!]
\centering
\begin{tikzpicture}[%
xscale=2,yscale=1.8,%
vector/.style={very thick, ->, >=stealth'},%
axis/.style={->, >=stealth'},%
angle/.style={thick, ->, >=stealth'},%
length/.style={thick, <->, >=stealth'}]
\draw 
(-1,0) node[rectangle,draw]{Rational Calogero-Moser} (3.2,0) node[rectangle,draw]{Rational Calogero-Moser}
(-1,1) node[rectangle,draw]{Hyperbolic Calogero-Moser} (3.2,1) node[rectangle,draw]{Rational Ruijsenaars-Schneider}
(-1,2) node[rectangle,draw]{Hyperbolic Ruijsenaars-Schneider} (3.2,2) node[rectangle,draw]{Hyperbolic Ruijsenaars-Schneider};
\draw[<-,thick] (-1,1.3)--(-1,1.7) node[midway,left]{$\beta\to 0$};
\draw[<->,thick] (.7,2)--(1.5,2) node[midway,above]{$\cR$};
\draw[<-,thick] (3.2,1.3)--(3.2,1.7) node[midway,right]{$\alpha\to 0$};
\draw[<->,thick] (.7,1)--(1.5,1)  node[midway,above]{$\cR$};
\draw[<-,thick] (-1,.3)--(-1,.7) node[midway,left]{$\alpha\to 0$};
\draw[<->,thick] (.7,0)--(1.5,0)  node[midway,above]{$\cR$};
\draw[<-,thick] (3.2,.3)--(3.2,.7) node[midway,right]{$\beta\to 0$};
\end{tikzpicture}
\caption{Action-angle dualities among Calogero-Ruijsenaars type systems.}
\label{fig:5}
\end{figure}
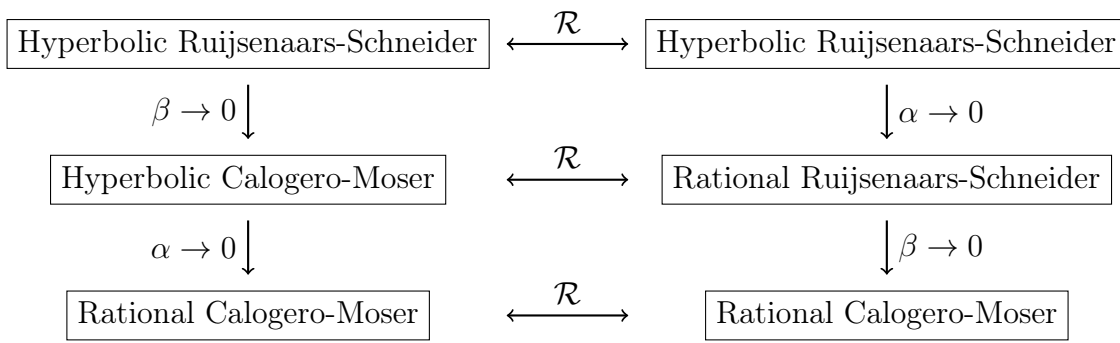

\noindent
Dual pairs of many-body systems were exhibited by Ruijsenaars (see Figure \ref{fig:5}) in the course of his direct construction \cite{Ru88,Ru90-2,Ru95,Ru99} of action-angle variables for the many-body systems (of non-elliptic Calogero-Ruijsenaars type and non-periodic Toda type) associated with the root system $\rA_{n-1}$. The idea that dualities can be interpreted in terms of Hamiltonian reduction can be distilled from \cite{KKS78} and was put forward explicitly in several papers in the 1990s, e.g.~\cite{FGNR00,GN95}. These papers contain a wealth of interesting ideas and results, but often stated without full proofs. In the last decade or so, Feh\'{e}r and collaborators undertook the systematic study of these dualities within the framework of reduction \cite{FK09,FA10,AF10,FK11,FK12,Fe13,FKl14}. It seems natural to expect that action-angle dualities exist for many-body systems associated with other root systems. Substantial evidence in favour of this expectation was given by Pusztai \cite{Pu11,Pu11-2,Pu12,Pu13,Pu15}. This thesis presents results (see \ref{chap:publ}) that were obtained in connection to these earlier developments.

\section{Outline of the thesis}

The main content of the thesis is divided into two parts with a total of five chapters.

Part \ref{part:1} takes the reduction approach to Calogero-Ruijsenaars type systems. In each of its chapters the basic idea of reduction that we just sketched is put into practice, only at an increasing level of complexity. In particular, Chapter \ref{chap:1} presents a streamlined derivation of the rational Calogero-Moser system using reduction. Section \ref{sec:1.2} exhibits the utility of the reduction perspective, as we give a simple proof of a formula providing action-angle coordinates. Chapter \ref{chap:2} is a study of the trigonometric $\BC_n$ Sutherland system. We provide a physical interpretation of the model in Section \ref{sec:2.1} and prepare the ingredients of reduction in Section \ref{sec:2.2}. In Section \ref{sec:2.3}, we solve the momentum equations and obtain the action-angle dual of the $\BC_n$ Sutherland system. In Section \ref{sec:2.4}, we apply our duality map to various problems, such as equilibrium configurations, proving superintegrability, and showing the equivalence of two sets of Hamiltonians. Chapter \ref{chap:3} generalises certain results of the previous chapter as it derives a $1$-parameter deformation of the trigonometric $\BC_n$ Sutherland system using Hamiltonian reduction of the Heisenberg double of $\SU(2n)$. We define the pertinent reduction in Section \ref{sec:3.1}, solve the momentum constraints in Section \ref{sec:3.2}, and characterize the reduced system in Section \ref{sec:3.3}. In Section \ref{sec:3.4}, we complete a recent derivation of the hyperbolic analogue.

Part \ref{part:2} is a collection of work motivated by, but not involving reduction techniques. Chapter \ref{chap:4} reports our discovery of a Lax pair for the hyperbolic van Diejen system with two independent coupling parameters. The preparatory Section \ref{sec:4.1} is followed by the explicit formulation of our Lax matrix in Section \ref{sec:4.2}. In Section \ref{sec:4.3}, we show that the dynamics can be solved by a projection method, which in turn allows us to initiate the study of the scattering
properties. We prove the equivalence between the first integrals provided by the eigenvalues of the Lax matrix and the family of van Diejen's commuting Hamiltonians in Section \ref{sec:4.4}. Chapter \ref{chap:5} is concerned with the explicit construction of compactified versions of trigonometric and elliptic Ruijsenaars-Schneider systems. In Section \ref{sec:5.1}, we embed the local phase space of the model into the complex projective space $\CP^{n-1}$. Section \ref{sec:5.2} contains our proof of the global extension of the trigonometric Lax matrix to $\CP^{n-1}$. We use our direct construction to introduce new compactified elliptic systems in Section \ref{sec:5.3}.

The chapters are complemented by \ref{part:app} collecting supplementary material (alternative proofs, detailed derivations, etc.). A \ref{chap:sum} presents the most important results in a concise form. A list of \ref{chap:publ}, on which this thesis is based, and a \ref{chap:bibl} are also included.

\cleardoublepage
\part{Reduction approach, action-angle duality, applications}
\label{part:1}

\chapter{A pivotal example}
\label{chap:1}

We start this chapter by describing the rational Calogero-Moser system and
recalling how it originates from Hamiltonian reduction \cite{KKS78}.
Then we use reduction treatment to simplify Falqui and Mencattini's recent
proof \cite{FM16} of Sklyanin's expression \cite{Sk09} providing spectral
Darboux coordinates of the rational Calogero-Moser system.

\section{Rational Calogero-Moser system}
\label{sec:1.1}

The Hamiltonian $H$ \eqref{1.1} with rational potential models equally massive interacting particles
moving along a line with a pair potential inversely proportional to the square of
the distance. The model was introduced and solved at
the quantum level by Calogero \cite{Ca71}. The complete integrability of its classical
version was established by Moser \cite{Mo75}, who employed the Lax formalism \cite{La68} to identify
a complete set of commuting integrals as coefficients of the characteristic polynomial
of a certain Hermitian matrix function, called the Lax matrix.

These developments might prompt one to consider the Poisson commuting eigenvalues
of the Lax matrix and be interested in searching for an expression of conjugate variables.
Such an expression was indeed formulated by Sklyanin \cite{Sk09} in his work on
bispectrality, and worked out in detail for the open Toda chain \cite{Sk13}.
Sklyanin's formula for the rational Calogero-Moser model was recently confirmed
within the framework of bi-Hamiltonian geometry by Falqui and Mencattini \cite{FM16}
in a somewhat circuitous way, although a short-cut was pointed out in the form of a
conjecture. The purpose of this chapter is to prove this conjecture and offer
an alternative simple proof of Sklyanin's formula using results of Hamiltonian
reduction.

\subsection{Description of the model}
\label{subsec:1.1.1}

For $n$ particles, let the $n$-tuples $q=(q_1,\dots,q_n)$ and $p=(p_1,\dots,p_n)$ collect
their coordinates and momenta, respectively. Then the Hamiltonian of the model reads
\begin{equation}
H(q,p)=\frac{1}{2}\sum_{j=1}^np_j^2
+\sum_{\substack{j,k=1\\(j<k)}}^n\frac{g^2}{(q_j-q_k)^2},
\label{1.1}
\end{equation}
where $g$ is a real coupling constant tuning the strength of particle interaction.
The pair potential is singular at $q_j=q_k$ $(j\neq k)$, hence any initial ordering
of the particles remains unchanged during time-evolution. The configuration space is
chosen to be the domain $\cC=\{q\in\R^n\mid q_1>\dots>q_n\}$, and the phase space is
its cotangent bundle
\begin{equation}
T^\ast\cC=\{(q,p)\mid q\in\cC,\ p\in\R^n\},
\label{1.2}
\end{equation}
endowed with the standard symplectic form
\begin{equation}
\omega=\sum_{j=1}^n dq_j\wedge dp_j.
\label{1.3}
\end{equation}

\subsection{Calogero particles from free matrix dynamics}
\label{subsec:1.1.2}

The Hamiltonian system $(T^\ast\cC,\omega,H)$, called the rational Calogero-Moser
system, can be obtained as an appropriate Marsden-Weinstein reduction of the free
particle moving in the space of $n\times n$ Hermitian matrices as follows.

Consider the manifold of pairs of $n\times n$ Hermitian matrices
\begin{equation}
M=\{(X,P)\mid X,P\in\mathfrak{gl}(n,\mathbb{C}),\ X^\dag=X,\ P^\dag=P\},
\label{1.4}
\end{equation}
equipped with the symplectic form
\begin{equation}
\Omega=\tr(dX\wedge dP).
\label{1.5}
\end{equation}
The Hamiltonian of the analogue of a free particle reads
\begin{equation}
\cH(X,P)=\frac{1}{2}\tr(P^2).
\label{1.6}
\end{equation}
The equations of motion can be solved explicitly for this Hamiltonian system
$(M,\Omega,\cH)$, and the general solution is given by $X(t)=tP_0+X_0$,
$P(t)=P_0$. Moreover, the functions $\cH_k(X,P)=\frac{1}{k}\tr(P^k)$, $k=1,\dots,n$
form an independent set of commuting first integrals.

The group of $n\times n$ unitary matrices $\UN(n)$ acts on $M$ \eqref{1.4} by conjugation
\begin{equation}
(X,P)\to (UXU^\dag,UPU^\dag),\quad U\in\UN(n),
\label{1.7}
\end{equation}
leaves both the symplectic form $\Omega$ \eqref{1.5} and the Hamiltonians $\cH_k$
invariant, and the matrix commutator $(X,P)\to[X,P]$ is a momentum map for this
$\UN(n)$-action. Consider the Hamiltonian reduction performed by factorizing the
momentum constraint surface
\begin{equation}
[X,P]=\ri g(vv^\dag-\1_n)=\mu, \quad v=(1\dots 1)^\dag\in\R^n,\quad g\in\R,
\label{1.8}
\end{equation}
with the stabilizer subgroup $G_\mu\subset\UN(n)$ of $\mu$, e.g. by diagonalization of
the $X$ component. This yields the gauge slice $S=\{(Q(q,p),L(q,p))\mid q\in\cC,\ p\in\R^n\}$,
where
\begin{equation}
Q_{jk}=(UXU^\dag)_{jk}=q_j\delta_{jk},\quad
L_{jk}=(UPU^\dag)_{jk}=p_j\delta_{jk}+\ri g\frac{1-\delta_{jk}}{q_j-q_k},\quad
j,k=1,\dots,n.
\label{1.9}
\end{equation}
This $S$ is symplectomorphic to the reduced phase space and to $T^\ast\cC$
\eqref{1.2} since it inherits the reduced symplectic form $\omega$ \eqref{1.3}.
The unreduced Hamiltonians project to a commuting set of independent integrals
$H_k=\frac{1}{k}\tr(L^k)$, $k=1,\dots,n$, such that $H_2=H$ \eqref{1.1} and what's more,
the completeness of Hamiltonian flows follows automatically from the reduction.
Therefore the rational Calogero-Moser system is completely integrable.

The similar role of matrices $X$ and $P$ in the derivation above can be exploited to
construct action-angle variables for the rational Calogero-Moser system. This is done by
switching to the gauge, where the $P$ component is diagonalized by some matrix
$\tilde{U}\in G_\mu$, and it boils down to the gauge slice
$\tilde{S}=\{(\tilde{Q}(\phi,\lambda),\tilde{L}(\phi,\lambda))\mid \phi\in\R^n,\ \lambda\in\cC\}$, where
\begin{equation}
\tilde{Q}_{jk}=(\tilde{U}X\tilde{U}^\dag)_{jk}
=\phi_j\delta_{jk}-\ri g\frac{1-\delta_{jk}}{\lambda_j-\lambda_k},\quad
\tilde{L}_{jk}=(\tilde{U}P\tilde{U}^\dag)_{jk}=\lambda_j\delta_{jk},\quad
j,k=1,\dots,n.
\label{1.10}
\end{equation}
By construction, $\tilde{S}$ with the symplectic form
$\tilde{\omega}=\sum_{j=1}^nd\phi_j\wedge d\lambda_j$
is also symplectomorphic to the reduced phase space, thus a canonical transformation
$(q,p)\to(\phi,\lambda)$ is obtained, where the reduced Hamiltonians depend only on
$\lambda$, viz. $H_k=\frac{1}{k}(\lambda_1^k+\dots+\lambda_n^k)$, $k=1,\dots,n$.

\section{Application: Canonical spectral coordinates}
\label{sec:1.2}

Now, we turn to the question of variables conjugate to the Poisson commuting eigenvalues
$\lambda_1,\dots,\lambda_n$ of $L$ \eqref{1.9}, i.e. such functions
$\theta_1,\dots,\theta_n$ in involution that
\begin{equation}
\{\theta_j,\lambda_k\}=\delta_{jk},\quad j,k=1,\dots,n.
\label{1.11}
\end{equation}
At the end of Subsection \ref{subsec:1.1.2} we saw that the variables $\phi_1,\dots,\phi_n$ are
such functions. These action-angle variables $\lambda,\phi$ were already obtained by
Moser \cite{Mo75} using scattering theory, and also appear in Ruijsenaars's proof of
the self-duality of the rational Calogero-Moser system \cite{Ru88}.

Let us define the following functions over the phase space $T^\ast\cC$
\eqref{1.2} with dependence on an additional variable $z$:
\begin{equation}
A(z)=\det(z\1_n-L),\quad
C(z)=\tr(Q\,\adj(z\1_n-L)vv^\dag),\quad
D(z)=\tr(Q\,\adj(z\1_n-L)),
\label{1.12}
\end{equation}
where $Q$ and $L$ are given by \eqref{1.9}, $v=(1\dots1)^\dag\in\R^n$ and $\adj$ denotes the
adjugate matrix, i.e. the transpose of the cofactor matrix.
Sklyanin's formula \cite{Sk09} for $\theta_1,\dots,\theta_n$ then reads
\begin{equation}
\theta_k=\frac{C(\lambda_k)}{A'(\lambda_k)},\quad k=1,\dots,n.
\label{1.13}
\end{equation}
In \cite{FM16} Falqui and Mencattini have shown that
\begin{equation}
\mu_k=\frac{D(\lambda_k)}{A'(\lambda_k)},\quad k=1,\dots,n
\label{1.14}
\end{equation}
are conjugate variables to $\lambda_1,\dots,\lambda_n$, and
\begin{equation}
\theta_k=\mu_k+f_k(\lambda_1,\dots,\lambda_n),\quad k=1,\dots,n,
\label{1.15}
\end{equation}
with such $\lambda$-dependent functions $f_1,\dots,f_n$ that
\begin{equation}
\frac{\partial f_j}{\partial\lambda_k}=\frac{\partial f_k}{\partial\lambda_j},\quad
j,k=1,\dots,n
\label{1.16}
\end{equation}
thus $\theta_1,\dots,\theta_n$ given by
Sklyanin's formula \eqref{1.13} are conjugate to $\lambda_1,\dots,\lambda_n$.
This was done in a roundabout way, although the explicit form of
relation \eqref{1.15} was conjectured.

Here we take a different route by making use of the reduction viewpoint of Subsection
\ref{subsec:1.1.2}. From this perspective, the problem becomes transparent and can be solved
effortlessly. First, we show that $\mu_1,\dots,\mu_n$ \eqref{1.14} are nothing else
than the angle variables $\phi_1,\dots,\phi_n$.

\begin{lemma}
\label{lem:1.1}
The variables $\mu_1,\dots,\mu_n$ defined in \eqref{1.14} are the angle variables
$\phi_1,\dots,\phi_n$ of the rational Calogero-Moser system.
\end{lemma}

\begin{proof}
Notice that, by definition, $\mu_1,\dots,\mu_n$ are gauge invariant, thus by working in
the gauge, where the $P$ component is diagonal, that is with the matrices $\tilde{Q}$,
$\tilde{L}$ \eqref{1.10}, we get
\begin{equation}
\frac{D(z)}{A'(z)}
=\frac{\sum_{j=1}^n\phi_j\prod_{\substack{\ell=1\\(\ell\neq j)}}^n(z-\lambda_\ell)}
{\sum_{j=1}^n\prod_{\substack{\ell=1\\(\ell\neq j)}}^n(z-\lambda_\ell)}.
\label{1.17}
\end{equation}
Substituting $z=\lambda_k$ into \eqref{1.17} yields $\mu_k=\phi_k$, for each $k=1,\dots,n$.
\end{proof}

Next, we prove the relation of functions $A$, $C$, $D$ \eqref{1.12}, that was
conjectured in \cite{FM16}.

\begin{theorem}
\label{thm:1.2}
For any $n\in\N$, $(q,p)\in T^\ast\cC$ \eqref{1.2}, and $z\in\C$ we have
\begin{equation}
C(z)=D(z)+\frac{\ri g}{2}A''(z).
\label{1.18}
\end{equation}
\end{theorem}

\begin{proof}
Pick any point $(q,p)$ in the phase space $T^\ast\cC$ and consider the
corresponding point $(\lambda,\phi)$ in the space of action-angle variables.
Since $A(z)=(z-\lambda_1)\dots(z-\lambda_n)$ we have
\begin{equation}
\frac{\ri g}{2}A''(z)
=\ri g\sum_{\substack{j,k=1\\(j<k)}}^n
\prod_{\substack{\ell=1\\(\ell\neq j,k)}}^n(z-\lambda_\ell).
\label{1.19}
\end{equation}
The difference of functions $C$ and $D$ \eqref{1.12} reads
\begin{equation}
C(z)-D(z)=\tr\big(Q\,\adj(z\1_n-L)(vv^\dag-\1_n)\big).
\label{1.20}
\end{equation}
Since this is a gauge invariant function, we are allowed to work with $\tilde{Q},\tilde{L}$ \eqref{1.10}
instead of $Q,L$ \eqref{1.9}. Therefore \eqref{1.20} can be written as the sum of all
off-diagonal components of $\tilde{Q}\,\adj(z\1_n-\tilde{L})$, that is
\begin{equation}
C(z)-D(z)=\ri g\sum_{\substack{j,k=1\\(j\neq k)}}^n\frac{-1}{\lambda_j-\lambda_k}
\prod_{\substack{\ell=1\\(\ell\neq k)}}^n(z-\lambda_\ell)=
\ri g\sum_{\substack{j,k=1\\(j<k)}}^n
\prod_{\substack{\ell=1\\(\ell\neq j,k)}}^n(z-\lambda_\ell).
\label{1.21}
\end{equation}
This concludes the proof.
\end{proof}

Our theorem confirms that indeed relation \eqref{1.15} is valid with
\begin{equation}
f_k(\lambda_1,\dots,\lambda_n)
=\frac{\ri g}{2}\frac{A''(\lambda_k)}{A'(\lambda_k)}
=\ri g\sum_{\substack{\ell=1\\(\ell\neq k)}}^n\frac{1}{\lambda_k-\lambda_\ell}
,\quad
k=1,\dots,n,
\label{1.22}
\end{equation}
for which \eqref{1.16} clearly holds. An immediate consequence, as we indicated before,
is that $\theta_1,\dots,\theta_n$ \eqref{1.13} are conjugate variables to
$\lambda_1,\dots,\lambda_n$, thus Sklyanin's formula is verified.

\begin{corollary}[Sklyanin's formula]
\label{cor:1.3}
The variables $\theta_1,\dots,\theta_n$ defined by
\begin{equation}
\theta_k=\frac{C(\lambda_k)}{A'(\lambda_k)},\quad k=1,\dots,n
\label{1.23}
\end{equation}
are conjugate to the eigenvalues $\lambda_1,\dots,\lambda_n$ of the Lax
matrix $L$.
\end{corollary}

\section{Discussion}
\label{sec:1.3}

There seem to be several ways for generalisation. For example, one might consider
rational Calogero-Moser models associated to root systems other than type A.
The hyperbolic Calogero-Moser systems as well as, the `relativistic' Calogero-Moser
systems, also known as Ruijsenaars-Schneider systems, are also of considerable interest.

In Appendix \ref{sec:A.1}, we give another proof for Theorem \ref{thm:1.2} based on the scattering theory of particles in the rational Calogero-Moser system.

\chapter{Trigonometric BC${}_{\textit{n}}$ Sutherland system}
\label{chap:2}

In this chapter, we present a new case of action-angle duality
between integrable many-body systems of Calogero-Ruijsenaars type.
This chapter contains our results reported in \citepalias{FG14,Go14,GF15}.

The two systems live on the action-angle phase spaces of each other in such a way that
the action variables of each system serve as the particle positions of the other one.
Our investigation utilizes an idea that was exploited previously to provide group-theoretic
interpretation for several dualities discovered originally by Ruijsenaars.
In the group-theoretic framework one applies Hamiltonian reduction to two Abelian
Poisson algebras of invariants on a higher dimensional phase space and identifies
their reductions as action and position variables of two integrable systems living
on two different models of the single reduced phase space. Taking the cotangent bundle
of $\UN(2n)$ as the upstairs space, we demonstrate how this mechanism leads to a new
dual pair involving the $\BC_n$ trigonometric Sutherland system. Thereby we generalise
earlier results pertaining to the $\rA_{n-1}$ trigonometric Sutherland system \cite{FA10}
as well as a recent work by Pusztai \cite{Pu12} on the hyperbolic $\BC_n$ Sutherland
system.

The specific goal in this chapter is to find out how this result can be generalised
if one replaces the hyperbolic $\BC_n$ system with its trigonometric analogue. A similar
problem has been studied previously in the $\rA_{n-1}$ case, where it was found that the
dual of the trigonometric Sutherland system possesses intricate global structure
\cite{FA10,Ru95}. The global description of the duality necessitates a separate
investigation also in the $\BC_n$ case, since it cannot be derived by naive analytic
continuation between trigonometric and hyperbolic functions. This problem turns out to
be considerably more complicated than those studied in \cite{FA10,Pu12}.

The trigonometric $\BC_n$ Sutherland system is defined by the Hamiltonian
\begin{equation}
H=\frac{1}{2}\sum_{j=1}^np_j^2
+\sum_{1\leq j<k\leq n}\bigg[\frac{\gamma}{\sin^2(q_j-q_k)}
+\frac{\gamma}{\sin^2(q_j+q_k)}\bigg]
+\sum_{j=1}^n \frac{\gamma_1}{\sin^2(q_j)}
+\sum_{j=1}^n \frac{\gamma_2}{\sin^2(2q_j)}.
\label{2.1}
\end{equation}
Here $(q,p)$ varies in the cotangent bundle $M=T^\ast C_1=C_1\times\R^n$ of the domain
\begin{equation}
C_1=\bigg\{q\in\R^n\bigg|\frac{\pi}{2}>q_1>\dots>q_n>0\bigg\},
\label{2.2}
\end{equation}
and the three independent real coupling constants $\gamma,\gamma_1,\gamma_2$
are supposed to satisfy
\begin{equation}
\gamma>0,\quad
\gamma_2>0,\quad
4\gamma_1+\gamma_2>0.
\label{2.3}
\end{equation}
The inequalities in \eqref{2.3} guarantee that the $n$ particles with coordinates $q_j$
cannot leave the open interval $(0,\pi/2)$ and they cannot collide. At a `semi-global'
level, the dual system will be shown to have the Hamiltonian
\begin{align}
\tilde{H}^0=&\sum_{j=1}^n\cos(\vartheta_j)
\bigg[1-\frac{\nu^2}{\lambda_j^2}\bigg]^{\tfrac{1}{2}}
\bigg[1-\frac{\kappa^2}{\lambda_j^2}\bigg]^{\tfrac{1}{2}}
\prod_{\substack{k=1\\(k\neq j)}}^n
\bigg[1-\frac{4\mu^2}{(\lambda_j-\lambda_k)^2}\bigg]^{\tfrac{1}{2}}
\bigg[1-\frac{4\mu^2}{(\lambda_j+\lambda_k)^2}\bigg]^{\tfrac{1}{2}}\nonumber\\
&-\frac{\nu\kappa}{4\mu^2}\prod_{j=1}^n
\bigg[1-\frac{4\mu^2}{\lambda_j^2}\bigg]
+\frac{\nu\kappa}{4\mu^2}.
\label{2.4}
\end{align}
Here $\mu>0,\nu,\kappa$ are real constants, $\vartheta_1,\dots,\vartheta_n$
are angular variables, and $\lambda$ varies in the Weyl chamber with thick walls:
\begin{equation}
C_2=\bigg\{\lambda\in\R^n\bigg|
\begin{matrix}\lambda_a-\lambda_{a+1}>2\mu,\\
(a=1,\dots,n-1)\end{matrix}
\quad\text{and}\quad
\lambda_n>\max\{|\nu|,|\kappa|\}\bigg\}.
\label{2.5}
\end{equation}
The inequalities defining $C_2$ ensure the reality and the smoothness of $\tilde{H}^0$ on
the phase space $\tilde{M}^0= C_2\times\T^n$, which is equipped with the symplectic form
\begin{equation}
\tilde\omega^0=\sum_{k=1}^nd\lambda_k\wedge d\vartheta_k.
\label{2.6}
\end{equation}
Duality will be established under the following relation between the couplings,
\begin{equation}
\gamma=\mu^2,\quad
\gamma_1=\frac{\nu\kappa}{2},\quad
\gamma_2=\frac{(\nu-\kappa)^2}{2},
\label{2.7}
\end{equation}
where in addition to $\mu>0$ we also adopt the condition
\begin{equation}
\nu>|\kappa|\geq 0.
\label{2.8}
\end{equation}
This entails that equation \eqref{2.7} gives a one-to-one correspondence of
the parameters $(\gamma,\gamma_1,\gamma_2)$ subject to \eqref{2.3} and $(\mu,\nu,\kappa)$,
and also simplifies our analysis. In the above, the qualification `semi-global'
indicates that $\tilde{M}^0$ represents a dense open submanifold of the full dual phase
space $\tilde{M}$. The completion of $\tilde{M}^0$ into $\tilde{M}$ guarantees both the
completeness of the Hamiltonian flows of the dual system and the global nature of the
symplectomorphism between $M$ and $\tilde{M}$. The structure of $\tilde{M}$ will also be
clarified. For example, we shall see that the action variables of the
Sutherland system fill the closure of the domain $C_2$, with the boundary
points corresponding to degenerate Liouville tori.

The integrable systems $(M,\omega,H)$ and $(\tilde{M},\tilde{\omega},\tilde{H})$
as well as their duality relation will emerge from an appropriate Hamiltonian reduction.
Specifically, we will reduce the cotangent bundle $T^\ast\UN(2n)$ with respect to the
symmetry group $G_+\times G_+$, where $G_+\cong\UN(n)\times\UN(n)$ is the fix-point
subgroup of an involution of $\UN(2n)$. This enlarges the range of the reduction approach
to action-angle dualities \cite{FGNR00,Go95,Ne99}.

\section{Physical interpretation}
\label{sec:2.1}

The trigonometric $\BC_n$ Sutherland model has the following physical interpretation. Consider a circle of radius $1/2$ with centre $O$. First, put one particle on the circle to an arbitrary point $Q_0$, hence creating reference direction $\overrightarrow{OQ_0}$, which coordinates a point $Q$ on the circle with the angle $\phi(Q)=\angle QOQ_0\in(-\pi,\pi]$, i.e. $\phi(Q_0)=0$. Next, place $n$ particles on the circle at some points $Q_1,\dots,Q_n$, such that their angles $\phi_j=\phi(Q_j)$ ($j=1,\dots,n$) satisfy $\pi>\phi_1>\dots>\phi_n>0$. Put $n$ additional particles on the circle at `mirror images' $Q_{-j}$ of $Q_j$ with respect to the point $Q_0$, that is $\phi(Q_j)=-\phi(Q_{-j})$.

\begin{figure}[h!]
\centering
\begin{tikzpicture}[%
scale=4.7,%
vector/.style={very thick, ->, >=stealth'},%
axis/.style={->, >=stealth'},%
angle/.style={thick, ->, >=stealth'},%
length/.style={thick, <->, >=stealth'}]

\draw (0,0)node[left]{$O$};
\draw (1,0) arc(0:360:1cm);
\draw[decorate,decoration={brace,mirror,amplitude=7pt}]
 (0,0) -- 
 ({cos(deg(-2*pi/3))},{sin(deg(-2*pi/3))}) 
 node[midway,anchor=south,yshift=0.2cm,rotate={deg(pi/3)}] {$R=1/2$};
\draw[axis,thin] 
 (0,0) -- 
 (1.2,0);
\filldraw[black] 
(1,0) node[anchor=north west]{$Q_0$} circle [radius=0.015]
({cos(deg(pi/9))},{sin(deg(pi/9))}) circle [radius=0.015]
({cos(deg(-pi/9))},{sin(deg(-pi/9))}) circle [radius=0.015]
({cos(deg(2*pi/3))},{sin(deg(2*pi/3))}) circle [radius=0.015]
({cos(deg(-2*pi/3))},{sin(deg(-2*pi/3))}) circle [radius=0.015];
\draw[dashed] 
 (0,0) -- 
 ({cos(deg(pi/9))},{sin(deg(pi/9))}) 
 node[anchor=south west] {$Q_k$};
\draw[dashed] 
 (0,0) -- 
 ({cos(deg(2*pi/3))},{sin(deg(2*pi/3))}) 
 node[anchor=south east] {$Q_j$};
\draw 
 ({cos(deg(-pi/9))},{sin(deg(-pi/9))}) 
 node[anchor=north west] {$Q_{-k}$};
\draw[dashed] 
 (0,0) -- 
 ({cos(deg(-2*pi/3))},{sin(deg(-2*pi/3))}) 
 node[anchor=north] {$Q_{-j}$};
\draw[angle] (.5,0) node[anchor=south east]{$\phi_k$} arc(0:20:.5cm);
\draw[angle] (.3,0) node[anchor=south east,xshift=-.7cm,yshift=.3cm]{$\phi_j$} arc(0:120:.3cm);
\draw[angle] (.3,0) node[anchor=north east,xshift=-.5cm,yshift=-.4cm]{$\phi_{-j}$} arc(0:-120:.3cm);
\draw[length] (1.05,0) arc(0:120:1.05cm);
\draw ({1.05*cos(deg(pi/3))},{1.05*sin(deg(pi/3))}) node[fill=white,circle,inner sep=0pt]{$q_j$};
\draw[length] (1,0) arc(0:20:1cm);
\draw ({cos(deg(pi/18))},{sin(deg(pi/18))}) node[fill=white,circle,inner sep=0pt]{$q_k$};
\draw[thick] 
 ({cos(deg(pi/9))},{sin(deg(pi/9))}) -- 
 ({cos(deg(2*pi/3))},{sin(deg(2*pi/3))})
 node[midway,above,rotate={-deg(pi/9)}]{$\sin(q_j-q_k)$};
\draw[thick] 
 ({cos(deg(-2*pi/3))},{sin(deg(-2*pi/3))}) -- 
 ({cos(deg(pi/9))},{sin(deg(pi/9))}) 
 node[pos=.45,below,rotate={deg(2*pi/9)}]{$\sin(q_j+q_k)$};
\draw[thick] 
 ({cos(deg(2*pi/3))},{sin(deg(2*pi/3))}) -- 
 ({cos(deg(-2*pi/3))},{sin(deg(-2*pi/3))}) 
 node[midway,above,rotate={deg(pi/2)}]{$\sin(2q_j)$};
 \draw[thick] 
  ({cos(deg(2*pi/3))},{sin(deg(2*pi/3))}) -- 
  (1,0) 
  node[midway,below,rotate={deg(-pi/6)}]{$\sin(q_j)$};
\end{tikzpicture}
\caption{The schematics of trigonometric BC$_n$ Sutherland model.}
\label{fig:6}
\end{figure}
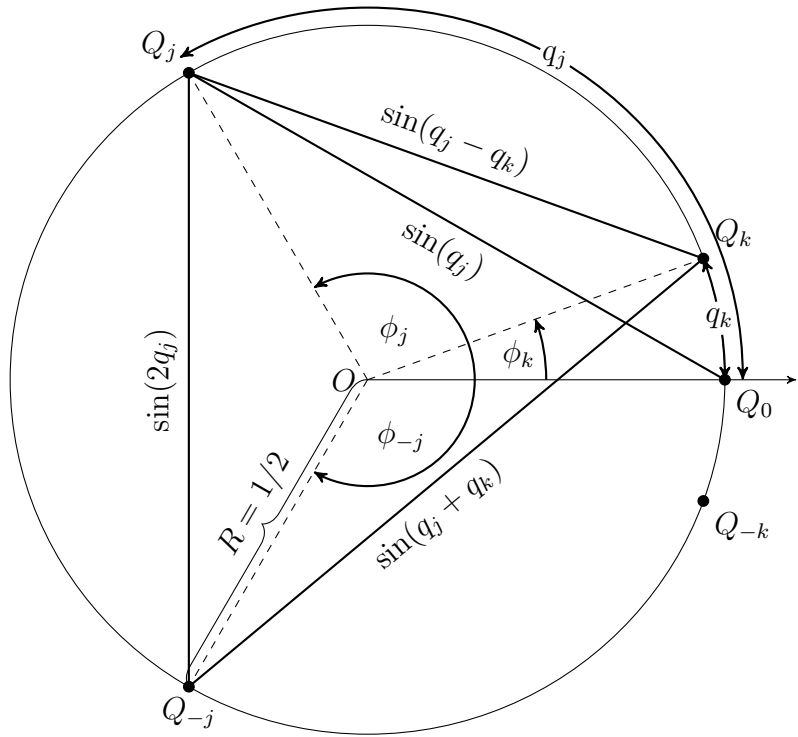

\noindent
Now, let these particles interact via a pair-potential that is inversely proportional to the square of the chord-distance. This interaction clearly preserves the initial symmetric configuration. Therefore $Q_0$ is fixed and acts as a boundary. Let us use the arc lengths $q_j=\phi_j/2$ instead of the angles. Due to the symmetry, the configuration is specified by $q_1,\dots,q_n$, which satisfy the inequalities in $C_1$ \eqref{2.2}. Let $\gamma_1,\gamma_2,\gamma$ be particle-boundary, particle-mirror particle, and bulk interaction couplings, respectively.

One can distinguish four types of chord-distances corresponding to these couplings (see Figure \ref{fig:6}), namely
\begin{equation}
\gamma_1:\ \sin(q_j),\quad
\gamma_2:\ \sin(2q_j),\quad
\gamma:\ \sin(q_j-q_k),\quad
\sin(q_j+q_k).
\label{2.9}
\end{equation}
Let $p_1,\dots,p_n$ stand for the generalised momenta of the particles at $q_1,\dots,q_n$. Then the total energy of the system is given by the Hamiltonian $H$ \eqref{2.1}, which exhibits symmetry under the Weyl group of the $\BC_n$ root system.

\section{Definition of the Hamiltonian reduction}
\label{sec:2.2}

Next we describe the starting data which will lead to integrable many-body systems
in duality by means of the mechanism outlined in the Introduction. We also collect
some group-theoretic facts that will be used in the demonstration of this claim.

Our investigation requires the unitary group of degree $2n$, i.e.
\begin{equation}
G=\UN(2n)=\{y\in\GL(2n,\C)\mid y^\dag y=\1_{2n}\},
\label{2.10}
\end{equation}
and its Lie algebra
\begin{equation}
\cG=\un(2n)=\{Y\in\gl(2n,\C)\mid Y^\dag+Y=\0_{2n}\},
\label{2.11}
\end{equation}
where $\1_{2n}$ and $\0_{2n}$ denote the identity and null matrices of size $2n$,
respectively. We endow the Lie algebra $\cG$
with the Ad-invariant bilinear form
\begin{equation}
\langle\cdot,\cdot\rangle\colon\cG\times\cG\to\R,\quad
(Y_1,Y_2)\mapsto\langle Y_1,Y_2\rangle=\tr(Y_1Y_2),
\label{2.12}
\end{equation}
and identify $\cG$ with the dual space $\cG^\ast$ in the usual manner. By using
left-translations to trivialize the cotangent bundle $T^\ast G$, we also adopt
the identification
\begin{equation}
T^\ast G\cong G\times\cG^\ast\cong G\times\cG
=\{(y,Y)\mid y\in G,\ Y\in\cG\}.
\label{2.13}
\end{equation}
Then the canonical symplectic form of $T^\ast G$ can be written as
\begin{equation}
\Omega^{T^\ast G}=-d\langle y^{-1}dy,Y\rangle.
\label{2.14}
\end{equation}
It can be evaluated according to the formula
\begin{equation}
\Omega^{T^\ast G}_{(y,Y)}(\Delta y\oplus\Delta Y,\Delta'y\oplus\Delta'Y)
=\langle y^{-1}\Delta y,\Delta'Y\rangle
-\langle y^{-1}\Delta'y,\Delta Y\rangle
+\langle[y^{-1}\Delta y,y^{-1}\Delta'y],Y\rangle,
\label{2.15}
\end{equation}
where $\Delta y\oplus\Delta Y,\Delta'y\oplus\Delta'Y\in T_{(y,Y)} T^\ast G$
are tangent vectors at a point $(y,Y)\in T^\ast G$.

\noindent
We introduce the $2n\times 2n$ Hermitian, unitary matrix partitioned into
four $n\times n$ blocks
\begin{equation}
C=\begin{bmatrix}
\0_n&\1_n\\\1_n&\0_n
\end{bmatrix}\in G,
\label{2.16}
\end{equation}
and the involutive automorphism of $G$ defined as conjugation with $C$
\begin{equation}
\Gamma\colon G\to G,\quad
y\mapsto\Gamma(y)=CyC^{-1}.
\label{2.17}
\end{equation}
The set of fix-points of $\Gamma$ forms the subgroup of $G$ consisting of $2n\times 2n$
unitary matrices with centro-symmetric block structure,
\begin{equation}
G_+=\{y\in G\mid\Gamma(y)=y\}=\bigg\{
\begin{bmatrix}
a&b\\b&a
\end{bmatrix}\in G\bigg\}\cong\UN(n)\times\UN(n).
\label{2.18}
\end{equation}
We also introduce the closed submanifold $G_-$ of $G$ by the definition
\begin{equation}
G_-=\{y\in G\mid\Gamma(y)=y^{-1}\}=\bigg\{
\begin{bmatrix}
a&b\\c&a^\dag
\end{bmatrix}\in G\bigg\vert\ b,c\in\ri\un(n)\bigg\}.
\label{2.19}
\end{equation}
By slight abuse of notation, we let $\Gamma$ stand for the induced
involution of the Lie algebra $\cG$, too. We can decompose $\cG$ as
\begin{equation}
\cG=\cG_+\oplus\cG_-,\quad Y=Y_++Y_-,
\label{2.20}
\end{equation}
where $\cG_\pm$ are the eigenspaces of $\Gamma$
corresponding to the eigenvalues $\pm 1$, respectively, i.e.
\begin{equation}
\begin{split}
\cG_+&=\ker(\Gamma-\id)=\bigg\{
\begin{bmatrix}
A&B\\
B&A
\end{bmatrix}
\bigg\vert\ A,B\in\un(n)\bigg\},\\
\cG_-&=\ker(\Gamma+\id)=\bigg\{
\begin{bmatrix}
A&B\\
-B&-A
\end{bmatrix}
\bigg\vert\ A\in\un(n),\ B\in\ri\un(n)\bigg\}.
\end{split}
\label{2.21}
\end{equation}
We are interested in a reduction of $T^\ast G$ based on the symmetry group
$G_+ \times G_+$. We shall use the shifting trick of symplectic reduction \cite{OR04},
and thus we first prepare a coadjoint orbit of the symmetry group.
To do this, we take any vector $V\in \C^{2n}$ that satisfies $CV+V=0$, and associate
to it the element $\upsilon_{\mu,\nu}^\ell(V)$ of $\cG_+$ by the definition
\begin{equation}
\upsilon_{\mu,\nu}^\ell(V)=\ri\mu\big(VV^\dag-\1_{2n}\big)+\ri(\mu-\nu)C,
\label{2.22}
\end{equation}
where $\mu,\nu\in\R$ are real parameters.
The set
\begin{equation}
\cO^\ell=\big\{\upsilon^\ell\in\cG_+\mid
\exists\ V\in\C^{2n},\ V^\dag V=2n,\ CV+V=0,\
\upsilon^\ell=\upsilon_{\mu,\nu}^\ell(V)
\big\}
\label{2.23}
\end{equation}
represents a coadjoint orbit of $G_+$ of dimension $2(n-1)$.
We let $\cO^r=\{\upsilon^r\}$ denote the one-point coadjoint orbit of $G_+$
containing the element
\begin{equation}
\upsilon^r=-\ri\kappa C\quad\text{with some constant}\;\kappa\in\R,
\label{2.24}
\end{equation}
and consider
\begin{equation}
\cO=\cO^\ell\oplus\cO^r\subset\cG_+\oplus\cG_+\cong(\cG_+\oplus\cG_+)^\ast,
\label{2.25}
\end{equation}
which is a coadjoint orbit\footnote{The same coadjoint orbit was used in \cite{Pu12}.}
of $G_+ \times G_+$. Our starting point for symplectic reduction will be the phase space
$(P,\Omega)$ with
\begin{equation}
P =T^\ast G\times\cO
\quad\mbox{and}\quad
\Omega=\Omega^{T^\ast G}+\Omega^{\cO},
\label{2.26}
\end{equation}
where $\Omega^{\cO}$ denotes the Kirillov-Kostant-Souriau symplectic form on $\cO$.
The natural symplectic action of $G_+\times G_+$ on $P$ is defined by
\begin{equation}
\Phi_{(g_L,g_R)}(y,Y,\upsilon^\ell\oplus\upsilon^r)=
\big(g_L^{\phantom{1}}yg_R^{-1},g_R^{\phantom{1}}Yg_R^{-1},
g_L^{\phantom{1}}\upsilon^\ell g_L^{-1}\oplus\upsilon^r \big).
\label{2.27}
\end{equation}
The corresponding momentum map $J\colon P\to\cG_+\oplus\cG_+$ is given by the formula
\begin{equation}
J(y,Y,\upsilon^\ell\oplus\upsilon^r)=
\big((yYy^{-1})_++\upsilon^\ell\big)\oplus\big(-Y_++\upsilon^r\big).
\label{2.28}
\end{equation}
We shall see that the reduced phase space
\begin{equation}
P_\red=P_0/(G_+\times G_+),\quad P_0=J^{-1}(0),
\label{2.29}
\end{equation}
is a smooth symplectic manifold, which inherits two Abelian Poisson algebras from $P$.

Using the identification $\cG^\ast\cong\cG$, the invariant functions $C^\infty(\cG)^G$
form the center of the Lie-Poisson bracket. Denote by $C^\infty(G)^{G_+\times G_+}$
the set of smooth functions on $G$ that are invariant under the $(G_+\times G_+)$-action
on $G$ that appears in the first component of \eqref{2.27}. Let us also introduce the maps
\begin{equation}
\pi_1\colon P\to G,\quad (y,Y, \upsilon^\ell, \upsilon^r)\mapsto y,
\label{2.30}
\end{equation}
and
\begin{equation}
\pi_2\colon P\to\cG,\quad (y,Y, \upsilon^\ell, \upsilon^r)\mapsto Y.
\label{2.31}
\end{equation}
It is clear that
\begin{equation}
\fQ^1= \pi_1^\ast(C^\infty(G)^{G_+\times G_+})
\quad\text{and}\quad
\fQ^2= \pi_2^\ast(C^\infty(\cG)^G)
\label{2.32}
\end{equation}
are two Abelian subalgebras in the Poisson algebra of smooth functions on $(P,\Omega)$
and these Abelian Poisson algebras descend to the reduced phase space $P_\red$.

Later we shall construct two models of $P_\red$ by exhibiting two global cross-sections
for the action of $G_+\times G_+$ on $P_0$. For this, we shall apply two different
methods for solving the constraint equations that, according to \eqref{2.28}, define
the level surface $P_0\subset P$:
\begin{equation}
(yYy^{-1})_++\upsilon^\ell=\0_{2n}
\quad\text{and}\quad
-Y_++\upsilon^r=\0_{2n},
\label{2.33}
\end{equation}
where $\upsilon^\ell=\upsilon^\ell_{\mu,\lambda}(V)$ \eqref{2.22} for some vector
$V\in\C^{2n}$ subject to $CV+V=0$, $V^\dag V=2n$ and $\upsilon^r=-\ri\kappa C$.
We below collect the group-theoretic results needed for our constructions.
To start, let us associate the diagonal $2n\times 2n$ matrix
\begin{equation}
Q(q)=\diag(q,-q)
\label{2.34}
\end{equation}
with any $q\in\R^n$. Notice that the set
\begin{equation}
\cA=\{\ri Q(q)\mid q\in\R^n\}\subset\cG_-
\label{2.35}
\end{equation}
is a maximal Abelian subalgebra in $\cG_-$. The corresponding subgroup of $G$ has the form
\begin{equation}
\exp(\cA)=\big\{
e^{\ri Q(q)}=
\diag\big(e^{\ri q_1},\dots,e^{\ri q_n},e^{-\ri q_1},\dots,e^{-\ri q_n}\big)
\mid q\in\R^n\big\}.
\label{2.36}
\end{equation}
The centralizer of $\cA$ inside $G_+$ \eqref{2.18} (with respect to conjugation)
is the Abelian subgroup
\begin{equation}
Z=Z_{G_+}(\cA)=\big\{
e^{\ri \xi}=
\diag\big(e^{\ri x_1},\dots,e^{\ri x_n},e^{\ri x_1},\dots,e^{\ri x_n}\big)
\mid x\in\R^n\big\}<G_+.
\label{2.37}
\end{equation}
The Lie algebra of $Z$ is
\begin{equation}
\cZ=\{\ri \xi=\ri\, \diag(x,x)\mid x\in\R^n\} <\cG_+.
\label{2.38}
\end{equation}

The results that we now recall (see e.g. \cite{He78,Ma97,Sc94}) will be used later.
First, for any $y\in G$ there exist elements $y_L$, $y_R$ from $G_+$ and
unique $q\in\R^n$ satisfying
\begin{equation}
\frac{\pi}{2}\geq q_1\geq\dots\geq q_n\geq 0
\label{2.39}
\end{equation}
such that
\begin{equation}
y=y_L^{\phantom{1}}e^{\ri Q(q)}y_R^{-1}.
\label{2.40}
\end{equation}
If all components of $q$ satisfy strict inequalities, then the pair $y_L, y_R$ is
unique precisely up to the replacements $(y_L,y_R)\to(y_L\zeta,y_R\zeta)$ with arbitrary
$\zeta\in Z$. The decomposition \eqref{2.40} is referred to as the generalised Cartan
decomposition corresponding to the involution $\Gamma$.

Second, every element $g\in G_-$ can be written in the form
\begin{equation}
g=\eta e^{2\ri Q(q)}\eta^{-1}
\label{2.41}
\end{equation}
with some $\eta\in G_+$ and uniquely determined $q\in\R^n$ subject to \eqref{2.39}.
In the case of strict inequalities for $q$, the freedom in $\eta$ is given precisely
by the replacements $\eta\to\eta\zeta$, $\forall\,\zeta\in Z$.

Third, every element $Y_-\in\cG_-$ can be written in the form
\begin{equation}
Y_-=g_R\ri D g_R^{-1},\quad
D=\diag(d_1,\dots,d_n,-d_1,\dots,-d_n),
\label{2.42}
\end{equation}
with $g_R \in G_+$ and uniquely determined real $d_i$ satisfying
\begin{equation}
d_1\geq\dots\geq d_n\geq 0.
\label{2.43}
\end{equation}
If the $d_j$ ($j=1,\dots,n$) satisfy strict inequalities, then the freedom in $g_R$ is exhausted
by the replacements $g_R\to g_R\zeta$, $\forall\,\zeta\in Z$.

The first and the second statements are essentially equivalent since the map
\begin{equation}
G\to G_-,\quad y\mapsto y^{-1}CyC
\label{2.44}
\end{equation}
descends to a diffeomorphism from
\begin{equation}
G/G_+ = \{ G_+ g \mid g\in G\}
\label{2.45}
\end{equation}
onto $G_-$ \cite{He78}.

\section{Action-angle duality}
\label{sec:2.3}

\subsection{The Sutherland gauge}
\label{subsec:2.3.1}

We here exhibit a symplectomorphism between the reduced phase space $(P_\red,\Omega_\red)$
and the Sutherland phase space
\begin{equation}
M=T^\ast C_1=C_1\times\R^n
\label{2.46}
\end{equation}
equipped with its canonical symplectic form, where $C_1$ was defined in \eqref{2.2}.
As preparation, we associate with any $(q,p)\in M$ the $\cG$-element
\begin{equation}
Y(q,p)=K(q,p)-\ri\kappa C,
\label{2.47}
\end{equation}
where $K(q,p)$ is the $2n\times 2n$ matrix
\begin{equation}
\begin{gathered}
K_{j,k}=-K_{n+j,n+k}=\ri p_j\delta_{j,k}-\mu(1-\delta_{j,k})/\sin(q_j-q_k),\\
K_{j,n+k}=-K_{n+j,k}=(\nu/\sin(2q_j)+\kappa\cot(2q_j))\delta_{j,k}
+\mu(1-\delta_{j,k})/\sin(q_j+q_k),\\
\end{gathered}
\label{2.48}
\end{equation}
with $j,k=1,\dots,n$. We also introduce the $2n$-component vector
\begin{equation}
V_\R=(\underbrace{1,\dots,1}_{n\ {\rm times}},
\underbrace{-1,\dots,-1}_{n\ {\rm times}})^\top.
\label{2.49}
\end{equation}
Notice from \eqref{2.21} that $K(q,p) \in \cG_-$.

Throughout the chapter we adopt the conditions \eqref{2.8} and take $\mu>0$, although
the next result requires only that the real parameters $\mu,\nu,\kappa$ satisfy
\begin{equation}
\mu\neq 0\quad\text{and}\quad|\nu|\neq|\kappa|.
\label{2.50}
\end{equation}

\begin{theorem}
\label{thm:2.1}
Using the notations introduced in \eqref{2.22}, \eqref{2.34} and \eqref{2.47},
the subset $S$ of the phase space $P$ \eqref{2.26} given by
\begin{equation}
S=\left\{(e^{\ri Q(q)},Y(q,p),\upsilon_{\mu,\nu}^\ell(V_\R),\upsilon^r)
\mid (q,p)\in M \right\},
\label{2.51}
\end{equation}
is a global cross-section for the action of $G_+\times G_+$ on $P_0=J^{-1}(0)$.
Identifying $P_\red$ with $S$, the reduced symplectic form is equal to the Darboux
form $\omega=\sum_{k=1}^ndq_k \wedge dp_k$. Thus the obvious identification between
$S$ and $M$ provides a symplectomorphism
\begin{equation}
(P_\red,\Omega_\red)\simeq(M,\omega).
\label{2.52}
\end{equation}
\end{theorem}

\begin{proof}
We saw in Section \ref{sec:2.2} that the points of the level surface $P_0$
satisfy the equations
\begin{equation}
(yYy^{-1})_++\upsilon^\ell_{\mu,\nu}(V)=\0_{2n}
\quad\text{and}\quad
-Y_+ -\ri \kappa C=\0_{2n},
\label{2.53}
\end{equation}
for some vector $V\in\C^{2n}$ subject to $CV+V=0$, $V^\dag V=2n$.
Remember that the block-form of any Lie algebra element $Y\in\cG$ is
\begin{equation}
Y=\begin{bmatrix}A&B\\-B^\dag&D\end{bmatrix}
\quad\text{with}\quad
A+A^\dag=\0_n=D+D^\dag,\quad
B\in\C^{n\times n}.
\label{2.54}
\end{equation}
Now the second constraint equation in \eqref{2.53} can be written as
\begin{equation}
2Y_+=\begin{bmatrix}A+D&B-B^\dag\\B-B^\dag&A+D\end{bmatrix}
=\begin{bmatrix}\0_n&-2\ri\kappa\1_n\\-2\ri\kappa\1_n&\0_n\end{bmatrix}
=-2\ri\kappa C,
\label{2.55}
\end{equation}
which implies that
\begin{equation}
D=-A\quad\text{and}\quad B^\dag=B+2\ri\kappa\1_n.
\label{2.56}
\end{equation}
Thus every point of $P_0$ has $\cG$-component $Y$ of the form
\begin{equation}
Y=\begin{bmatrix}A&B\\-B-2\ri\kappa\1_n&-A\end{bmatrix}
\quad\text{with}\quad
A+A^\dag=\0_n,
\quad
B\in\C^{n\times n}.
\label{2.57}
\end{equation}
By using the generalised Cartan decomposition \eqref{2.40} and applying a gauge
transformation (the action of $G_+ \times G_+$ on $P_0$), we may assume that
$y=e^{\ri Q(q)}$ with some $q$ satisfying \eqref{2.38}. Then the first equation
of the momentum map constraint \eqref{2.53} yields the matrix equation
\begin{equation}
\frac{1}{2\ri}\big(e^{\ri Q(q)}Ye^{-\ri Q(q)}+e^{-\ri Q(q)}CYCe^{\ri Q(q)}\big)
+\mu(VV^\dag-\1_{2n})+(\mu-\nu)C=\0_{2n}.
\label{2.58}
\end{equation}
If we introduce the notation $V=(u,-u)^\top$, $u\in \C^n$, and assume that $Y$
has the form \eqref{2.57} then \eqref{2.58} turns into the following equations
for $A$ and $B$
\begin{equation}
\frac{1}{2\ri}\big(e^{\ri q}Ae^{-\ri q}-e^{-\ri q}Ae^{\ri q}\big)
+\mu(uu^\dag-\1_n)=\0_n,
\label{2.59}
\end{equation}
and
\begin{equation}
\frac{1}{2\ri}\big(e^{\ri q}Be^{\ri q}-e^{-\ri q}Be^{-\ri q}\big)
-\kappa e^{-2\ri q}-\mu uu^\dag+(\mu-\nu)\1_n=\0_n.
\label{2.60}
\end{equation}
Since $\mu \neq 0$, equation \eqref{2.59} implies that $|u_j|^2=1$ for all $j=1,\dots,n$.
Therefore we can apply a `residual' gauge transformation by an element
$(g_L,g_R)=(e^{\ri\xi(x)},e^{\ri\xi(x)})$, with suitable $e^{\ri\xi(x)}\in Z$ \eqref{2.37}
 to transform $\upsilon_{\mu,\nu}^\ell(V)$ into $\upsilon_{\mu,\nu}^\ell(V_\R)$.
This amounts to setting $u_j=1$ for all $j=1,\dots,n$. After having done this, we return
to equations \eqref{2.59} and \eqref{2.60}. By writing out the equations entry-wise,
we obtain that the diagonal components of $A$ are arbitrary imaginary numbers (which
we denote by $\ri p_1,\dots,\ri p_n$) and we also obtain the following system of equations
\begin{equation}
\begin{split}
A_{j,k}\sin(q_j-q_k)=-\mu=-B_{j,k}\sin(q_j+q_k),&\quad j\neq k,\\
B_{j,j}\sin(2q_j)=\nu+\kappa\cos(2q_j)-\ri\kappa\sin(2q_j),&\quad j,k=1,\dots,n.
\end{split}
\label{2.61}
\end{equation}
So far we only knew that $q$ satisfies $\pi/2\geq q_1\geq\dots\geq q_n\geq 0$.
By virtue of the conditions \eqref{2.50}, the system \eqref{2.61} can be solved
if and only if $\pi/2>q_1>\dots>q_n>0$. Substituting the unique solution for $A$
and $B$ back into \eqref{2.57} gives the formula $Y=Y(q,p)$ as displayed in \eqref{2.47}.

The above arguments show that every gauge orbit in $P_0$ contains a point of $S$
\eqref{2.51}, and it is immediate by turning the equations backwards that every
point of $S$ belongs to $P_0$. By using that $q$ satisfies strict inequalities
and that all components of $V_\R$ are non-zero, it is also readily seen that no
two different points of $S$ are gauge equivalent. Moreover, the effectively acting
symmetry group, which is given by
\begin{equation}
(G_+ \times G_+)/\UN(1)_{\diag}
\label{2.62}
\end{equation}
where $\UN(1)$ contains the scalar unitary matrices, acts \emph{freely} on $P_0$.

It follows from the above that $P_\red$ is a smooth manifold diffeomorphic to $M$.
Now the proof is finished by direct computation of the pull-back of the symplectic
form $\Omega$ of $P$ \eqref{2.26} onto the global cross-section $S$.
\end{proof}

Let us recall that the Abelian Poisson algebras $\fQ^1$ and $\fQ^2$ \eqref{2.32}
consist of $(G_+ \times G_+)$-invariant functions on $P$, and thus descend to
Abelian Poisson algebras on the reduced phase space $P_\red$. In terms of the
model $M\simeq S\simeq P_\red$, the Poisson algebra $\fQ^2_\red$ is obviously
generated by the functions $(q,p)\mapsto\tr((-\ri Y(q,p)))^m$ for $m=1,\dots,2n$.
It will be shown in the following section\footnote{In fact, we shall see that
$Y(q,p)$ is conjugate to a diagonal matrix $\ri\Lambda$ of the form in equation
\eqref{2.71}.} that these functions vanish identically for the odd integers,
and functionally independent generators of $\fQ^2_\red$ are provided by the functions
\begin{equation}
H_k(q,p)=\frac{1}{4k}\tr(-\ri Y(q,p))^{2k},\quad k=1,\dots,n.
\label{2.63}
\end{equation}

The first of these functions reads
\begin{equation}
\begin{split}
H_1(q,p)=\frac{1}{4}\tr(-\ri Y(q,p))^2=&
\frac{1}{2}\sum_{j=1}^np_j^2
+\sum_{1\leq j<k\leq n}\bigg(\frac{\mu^2}{\sin^2(q_j-q_k)}
+\frac{\mu^2}{\sin^2(q_j+q_k)}\bigg)\\
&+\frac{1}{2}\sum_{j=1}^n\frac{\nu\kappa}{\sin^2(q_j)}
+\frac{1}{2}\sum_{j=1}^n\frac{(\nu-\kappa)^2}{\sin^2(2q_j)}.
\end{split}
\label{2.64}
\end{equation}
That is, upon the identification \eqref{2.7} it coincides with the Sutherland Hamiltonian
\eqref{2.1}. This implies the Liouville integrability of the Hamiltonian \eqref{2.1}.
Since its spectral invariants yield a commuting family of $n$ independent functions in
involution that include the Sutherland Hamiltonian, the Hermitian matrix function
$-\ri Y(q,p)$ \eqref{2.47} serves as a Lax matrix for the Sutherland system $(M,\omega,H)$.

As for the reduced Abelian Poisson algebra $\fQ^1_\red$, we notice that the cross-section
$S$ permits to identify it with the Abelian Poisson algebra of the smooth functions of the
variables $q_1,\dots,q_n$. This is so since the level set $P_0$ lies completely in the
`regular part' of the phase space $P$, where the $G$-component $y$ of
$(y,Y, \upsilon^\ell,\upsilon^r)$ is such that $Q(q)$ in its decomposition \eqref{2.40}
satisfies strict inequalities $\pi/2>q_1>\dots>q_n>0$. It is a well-known fact
that in the regular part the components of $q$ are smooth (actually real-analytic)
functions of $y$ (while globally they are only continuous functions). To see that every
smooth function depending on $q\in C_1$ is contained in $\fQ^1_\red$, one may further use
that every $(G_+\times G_+)$-invariant smooth function on $P_0$ can be extended to an
invariant smooth function on $P$. Indeed, this holds since $G_+ \times G_+$ is compact and
$P_0 \subset P$ is a regular submanifold, which itself follows from the free action
property established in the course of the proof of Theorem \ref{thm:2.1}.

We can summarize the outcome of the foregoing discussion as follows. Below, the generators
of Poisson algebras are understood in the functional sense, i.e. if some $f_1,\dots, f_n$
are generators then all smooth functions of them belong to the Poisson algebra.

\begin{corollary}
\label{cor:2.2}
By using the model $(M,\omega)$ of the reduced phase space $(P_\red,\Omega_\red)$ provided
by Theorem \ref{thm:2.1}, the Abelian Poisson algebra $\fQ^2_\red$ \eqref{2.31} can be
identified with the Poisson algebra generated by the spectral invariants \eqref{2.62} of
the `Sutherland Lax matrix' $-\ri Y(q,p)$ \eqref{2.47}, which according to \eqref{2.64}
include the many-body Hamiltonian $H(q,p)$ \eqref{2.1}, and $\fQ^1_\red$ can be identified
with the algebra generated by the corresponding position variables $q_j$ $(j=1,\dots,n)$.
\end{corollary}

\subsection{The Ruijsenaars gauge}
\label{subsec:2.3.2}

It follows from the group-theoretic results quoted in Section \ref{sec:2.2} that the Abelian
Poisson algebra $\fQ^1$ is generated by
the functions
\begin{equation}
\tilde\cH_k(y,Y,\upsilon^\ell,\upsilon^r)=\frac{(-1)^k}{2k}\tr\big(y^{-1}CyC\big)^k,
\quad k=1,\dots,n,
\label{2.65}
\end{equation}
and thus the unitary and Hermitian matrix
\begin{equation}
L=-y^{-1}CyC
\label{2.66}
\end{equation}
serves as an `unreduced Lax matrix'. It is readily seen in the Sutherland gauge \eqref{2.51}
that these $n$ functions remain functionally independent after reduction. Here, we shall
prove that the evaluation of the invariant function $\tilde\cH_1$ in another gauge
reproduces the dual Hamiltonian \eqref{2.4}. The reduction of the matrix function $L$ will
provide a Lax matrix for the corresponding integrable system. Before turning to details,
we advance the group-theoretic interpretation of the dual position variable $\lambda$
that features in the Hamiltonian \eqref{2.4}, and sketch the plan of this section.

To begin, recall that on the constraint surface $Y= Y_--\ri\kappa C$, and
for any $Y_-\in\cG_-$ there is an element $g_R\in G_+$ such that
\begin{equation}
g_R^{-1}Y_- g_R^{\phantom{1}}
=\diag(\ri d_1,\dots,\ri d_n,-\ri d_1,\dots,-\ri d_n)=\ri D\in\cA
\quad\text{with}\quad d_1\geq\dots\geq d_n\geq 0.
\label{2.67}
\end{equation}

Then introduce the real matrix $\blambda=\diag(\lambda_1,\dots,\lambda_n)$ whose diagonal
components are\footnote{From now on we frequently use the notations $\N_n=\{1,\dots,n\}$
and $\N_{2n}=\{1,\dots,2n\}$.}
\begin{equation}
\lambda_j=\sqrt{d_j^2+\kappa^2},\quad j\in\N_n.
\label{2.68}
\end{equation}
One can diagonalize the matrix $D-\kappa C$ by conjugation with the unitary matrix
\begin{equation}
h(\lambda)=\begin{bmatrix}
\alpha(\blambda)&\beta(\blambda)\\
-\beta(\blambda)&\alpha(\blambda)
\end{bmatrix},
\label{2.69}
\end{equation}
where the real functions $\alpha(x),\beta(x)$ are defined on the interval
$[|\kappa|,\infty)\subset\R$ by the formulae
\begin{equation}
\alpha(x)=\frac{\sqrt{x+\sqrt{x^2-\kappa^2}}}{\sqrt{2x}},\quad
\beta(x)=\kappa\frac{1}{\sqrt{2x}}\frac{1}{\sqrt{x+\sqrt{x^2-\kappa^2}}},
\label{2.70}
\end{equation}
at least if $\kappa\neq 0$. If $\kappa=0$, then we set $\alpha(x)=1$ and $\beta(x)=0$.
Indeed, it is easy to check that
\begin{equation}
h(\lambda)\Lambda h(\lambda)^{-1}=D-\kappa C
\quad\text{with}\quad
\Lambda=\diag(\lambda_1,\dots,\lambda_n,-\lambda_1,\dots,-\lambda_n).
\label{2.71}
\end{equation}
Note that $h(\lambda)$ belongs to the subset $G_-$ of $G$ \eqref{2.19}.

The above diagonalization procedure can be used to define the map
\begin{equation}
\fL\colon P_0\to \R^n,\quad
(y,Y,\upsilon^\ell,\upsilon^r)\mapsto\lambda.
\label{2.72}
\end{equation}
This is clearly a continuous map, which descends to a continuous map
$\fL_\red\colon P_\red\to\R^n$. One readily sees also that these maps are smooth
(even real-analytic) on the open submanifolds $P_0^\reg\subset P_0$ and
$P_\red^\reg\subset P_\red$, where the $2n$ eigenvalues of $Y_-$ are pairwise different.

The image of the constraint surface $P_0$ under the map $\fL$ will turn out to be the
closure of the domain
\begin{equation}
C_2=\bigg\{\lambda\in\R^n\bigg|
\begin{matrix}\lambda_a-\lambda_{a+1}>2\mu,\\
(a=1,\dots,n-1)\end{matrix}
\quad\text{and}\quad
\lambda_n>\nu\bigg\}.
\label{2.73}
\end{equation}

By solving the constrains through the diagonalization of $Y$, we shall construct a
model of the open submanifold of $P_\red$ corresponding to the open submanifold
$\fL^{-1}(C_2) \subset P_0$. This model will be symplectomorphic to the semi-global
phase-space $C_2 \times \T^n$ of the dual Hamiltonian \eqref{2.4}.

In the rest of this section, we present the construction of the aforementioned model
of $\fL_\red^{-1}(C_2)\subset P_\red$. We demonstrate that $\fL_\red^{-1}(C_2)$ is a
dense subset of $P_\red$ and present the global characterization of the dual model of
$P_\red$.

Many of the local formulae that appear in this section have analogues in
\cite{Pu11,Pu11-2,Pu12}, which inspired our considerations. However, the
global structure is different.

\subsubsection{The dual model of the open subset $\fL^{-1}_\red(C_2)\subset P_\red$}

We first prepare some functions on $C_2\times \T^n$. Denoting
the elements of this domain as pairs
\begin{equation}
(\lambda,e^{\ri \vartheta})
\quad\text{with}\quad
\lambda=(\lambda_1,\dots,\lambda_n)\in C_2,\quad
e^{\ri\vartheta}=(e^{\ri\vartheta_1},\dots,e^{\ri\vartheta_n})\in\T^n,
\label{2.74}
\end{equation}
we let
\begin{eqnarray}
&&\phantom{+c} f_{c} =\bigg[1-\frac{\nu}{\lambda_c}\bigg]^{\frac{1}{2}}
\prod_{\substack{a=1\\(a\neq c)}}^n
\bigg[1-\frac{2\mu}{\lambda_c-\lambda_a}\bigg]^{\frac{1}{2}}
\bigg[1-\frac{2\mu}{\lambda_c+\lambda_a}\bigg]^{\frac{1}{2}}, \quad \forall c\in \N_n,
\nonumber\\
&&f_{n+c}=e^{\ri \vartheta_c} \bigg[1+\frac{\nu}{\lambda_c}\bigg]^{\frac{1}{2}}
\prod_{\substack{a=1\\(a\neq c)}}^n
\bigg[1+\frac{2\mu}{\lambda_c-\lambda_a}\bigg]^{\frac{1}{2}}
\bigg[1+\frac{2\mu}{\lambda_c+\lambda_a}\bigg]^{\frac{1}{2}}.
\label{2.75}
\end{eqnarray}
For $\lambda\in C_2$ \eqref{2.73}, all factors under the square roots are positive.
Using the column vector $f= (f_1,\dots, f_{2n})^\top$ together with $\Lambda_c=\lambda_c$
and $\Lambda_{c+n}=-\lambda_c$ for $c\in \N_n$, we define the $2n\times 2n$ matrices
$\check{A}(\lambda,\vartheta)$ and $B(\lambda,\vartheta)$ by
\begin{equation}
\check A_{j,k}=\frac{2\mu f_j\overline{(Cf)}_k-
2(\mu-\nu)C_{j,k}}{2\mu+\Lambda_k-\Lambda_j},\quad
j,k\in\N_{2n},
\label{2.76}
\end{equation}
and
\begin{equation}
B(\lambda,\vartheta)=-\big(h(\lambda)
\check A(\lambda,\vartheta)h(\lambda)\big)^\dagger.
\label{2.77}
\end{equation}
We shall see that these are unitary matrices from $G_-\subset G$ \eqref{2.19}.
Then we write $B$ in the form
\begin{equation}
B=\eta e^{2\ri Q(q)}\eta^{-1}
\label{2.78}
\end{equation}
with some $\eta\in G_+$ and unique $q= q(\lambda,\vartheta)$ subject to
\eqref{2.39}. (It turns out that $q(\lambda,\vartheta)\in C_1$ \eqref{2.2}
and thus $\eta$ is unique up to replacements $\eta\to\eta\zeta$ with
arbitrary $\zeta \in Z$ \eqref{2.37}.) Relying on \eqref{2.78}, we set
\begin{equation}
y(\lambda,\vartheta)=\eta e^{\ri Q(q(\lambda,\vartheta))}\eta^{-1}
\label{2.79}
\end{equation}
and introduce the vector $V(\lambda,\vartheta)\in\C^{2n}$ by
\begin{equation}
V(\lambda,\vartheta)=y(\lambda,\vartheta)h(\lambda)f(\lambda,\vartheta).
\label{2.80}
\end{equation}
It will be shown that $V+CV=0$ and $|V|^2=2n$, which ensures that
$\upsilon^\ell_{\mu,\nu}(V)\in\cO^\ell$ \eqref{2.23}.

Note that $\check A$, $y$ and $V$ given above depend on $\vartheta$ only through
$e^{\ri \vartheta}$ and are $C^\infty$ functions on $C_2\times\T^n$. It should be
remarked that although the matrix element $\check A_{n,2n}$ \eqref{2.76} has an
apparent singularity at $\lambda_n=\mu$, the zero of the denominator cancels.
Thus $\check A$ extends by continuity to $\lambda_n=\mu$ and remains smooth there,
which then also implies the smoothness of $y$ and $V$.

\begin{theorem}
\label{thm:2.3}
By using the above notations, consider the set
\begin{equation}
\tilde{S}^0=\{(y(\lambda,\vartheta),\ri h(\lambda)\Lambda(\lambda)h(\lambda)^{-1},
\upsilon^\ell_{\mu,\nu}(V(\lambda,\vartheta)),
\upsilon^r)\mid(\lambda,e^{\ri \vartheta})\in
C_2\times\T^n\}.
\label{2.81}
\end{equation}

This set is contained in the constraint surface $P_0=J^{-1}(0)$ and it provides a
cross-section for the $G_+\times G_+$-action restricted to $\fL^{-1}(C_2)\subset P_0$.
In particular, $C_2\subset\fL(P_0)$ and $\tilde S^0$ intersects every gauge orbit in
$\fL^{-1}(C_2)$ precisely in one point. Since the elements of $\tilde S^0$ are
parametrized by $C_2\times\T^n$ in a smooth and bijective manner, we obtain the
identifications
\begin{equation}
\fL^{-1}_\red(C_2)\simeq\tilde S^0\simeq C_2\times\T^n.
\label{2.82}
\end{equation}
Letting $\tilde\sigma_0\colon\tilde S^0\to P$ denote the tautological injection,
the pull-backs of the symplectic form $\Omega$ \eqref{2.26} and the function
$\tilde\cH_1$ \eqref{2.65} obey
\begin{equation}
\tilde\sigma_0^\ast(\Omega)=\sum_{c=1}^n d\lambda_c\wedge d\vartheta_c,
\quad
(\tilde \cH_1 \circ \tilde \sigma_0)(\lambda,\vartheta)
=\frac{1}{2}\tr\big(h(\lambda)\check A(\lambda,\vartheta)h(\lambda)\big)
=\tilde H^0(\lambda,\vartheta)
\label{2.83}
\end{equation}
with the RSvD type Hamiltonian $\tilde H^0$ in \eqref{2.4}. Consequently, the Hamiltonian
reduction of the system $(P,\Omega,\tilde\cH_1)$ followed by restriction to the open
submanifold $\fL_\red^{-1}(C_2)\subset P_\red$ reproduces the system
$(\tilde M^0,\tilde\omega^0,\tilde H^0)$ defined in \eqref{2.4}-\eqref{2.5}.
\end{theorem}

\begin{remark}
\label{rem:2.4}
Referring to \eqref{2.66}, we have the Lax matrix
\begin{equation}
L(y(\lambda, \vartheta))=h(\lambda)\check A(\lambda,\vartheta)h(\lambda).
\label{2.84}
\end{equation}
Later we shall also prove that $\fL_{\red}^{-1}(C_2)$ is a dense subset of $P_\red$,
whereby the reduction of $(P,\Omega,\tilde\cH_1)$ may be viewed as a completion of
$(\tilde M^0,\tilde\omega^0,\tilde H^0)$.
\end{remark}

\subsubsection{Proof of Theorem \ref{thm:2.3}}

The proof will emerge from a series of lemmas. Our immediate aim is to construct gauge
invariant functions that will be used for parametrizing the orbits of $G_+\times G_+$
in (an open submanifold of) $P_0$. For introducing gauge invariants we can restrict
ourselves to the submanifold $P_1\subset P_0$ where $Y$ in $(y,Y,\upsilon^\ell,\upsilon^r)$
has the form
\begin{equation}
Y=h(\lambda)\ri\Lambda(\lambda)h(\lambda)^{-1}
\label{2.85}
\end{equation}
with some $\lambda\in\R^n$ for which
\begin{equation}
\lambda_1\geq\dots\geq\lambda_n\geq|\kappa|.
\label{2.86}
\end{equation}
Indeed, every element of $P_0$ can be gauge transformed into $P_1$. It will be
advantageous to further restrict attention to $P_1^\reg\subset P_1$ where we have
\begin{equation}
\lambda_1>\dots>\lambda_n>|\kappa|.
\label{2.87}
\end{equation}
The residual gauge transformations that map $P_1^\reg$ to itself belong to the group
$G_+\times Z<G_+\times G_+$ with $Z$ defined in \eqref{2.37}. Since $\upsilon^r$ is
constant and $\upsilon^\ell=\upsilon^\ell_{\mu,\nu}(V)$, we may label the elements
of $P_1$ by triples $(y,Y,V)$, with the understanding that $V$ matters up to phase.
Then the gauge action of $(g_L,\zeta)\in G_+ \times Z$ operates by
\begin{equation}
(y,V)\mapsto(g_Ly\zeta^{-1},g_LV),
\label{2.88}
\end{equation}
while $Y$ is already invariant. Now we can factor out the residual $G_+$-action by
introducing the $G_-$-valued function
\begin{equation}
\check A(y,Y,V)=h(\lambda)^{-1}L(y)h(\lambda)^{-1}
\label{2.89}
\end{equation}
and the $\C^{2n}$-valued function
\begin{equation}
F(y,Y,V)=h(\lambda)^{-1}y^{-1}V.
\label{2.90}
\end{equation}
Here $\lambda=\fL(y,Y,V)$, which means that \eqref{2.85} holds, and we used $L(y)$ in
\eqref{2.66}. Like $V$, $F$ is defined only up to a $\UN(1)$ phase. We obtain the
transformation rules
\begin{equation}
\check A(g_Ly\zeta^{-1},Y,g_L V)=\zeta\check A(y,Y,V)\zeta^{-1},
\label{2.91}
\end{equation}
\begin{equation}
F(g_Ly\zeta^{-1},Y,g_L V)=\zeta F(y,Y,V),
\label{2.92}
\end{equation}
and therefore the functions
\begin{equation}
\cF_k(y,Y,V)=|F_k(y,Y,V)|^2,\quad k=1,\dots,2n
\label{2.93}
\end{equation}
are well-defined, gauge invariant, smooth functions on $P_1^\reg$. They represent
$(G_+\times G_+)$-invariant smooth functions on $P_0^\reg$. We shall see shortly
that the functions $\cF_k$ depend only on $\lambda=\fL(y,Y,V)$ and shall derive
explicit formulae for this dependence. Then the non-negativity of $\cF_k$ will be
used to gain information about the set $\fL(P_0)$ of $\lambda$ values that actually occurs.

Before turning to the inspection of the functions $\cF_k$, we present a crucial lemma.

\begin{lemma}
\label{lem:2.5}
Fix $\lambda\in\R^n$ subject to \eqref{2.86} and set
$\Lambda=\diag(\lambda,-\lambda)$ and $Y=h\ri\Lambda h^{-1}$.
If $y\in G$ and $\upsilon^\ell_{\mu,\nu}(V)\in\cO^\ell$ solve the momentum map
constraint given according to the first equation in \eqref{2.53} by
\begin{equation}
yYy^{-1}+CyYy^{-1}C+2\upsilon^\ell_{\mu,\nu}(V)=0,
\label{2.94}
\end{equation}
then $\check A\in G_-$ and $F\in\C^{2n}$ defined by \eqref{2.89} and \eqref{2.90}
solve the following equation:
\begin{equation}
2\mu\check{A}+\check{A}\Lambda-\Lambda\check{A}
=2\mu F(CF)^\dag-2(\mu-\nu)C.
\label{2.95}
\end{equation}
Conversely, for any $\check A\in G_-$, $F\in\C^{2n}$ that satisfy $|F|^2=2n$ and
equation \eqref{2.95}, pick $y\in G$ such that $L(y)=h(\lambda)\check Ah(\lambda)$
and define $V=yh(\lambda)F$. Then $CV+V=0$ and $(y,Y,\upsilon^\ell_{\mu,\nu}(V))$
solve the momentum map constraint \eqref{2.94}.
\end{lemma}

\begin{proof}
If eq.~\eqref{2.94} holds, then we multiply it by $h(\lambda)^{-1}y^{-1}$ on
the left and by $CyCh(\lambda)^{-1}$ on the right.
Using \eqref{2.58}, with $CV+V=0$ and $|V|^2=2n$, and the notations \eqref{2.89} and \eqref{2.90},
this immediately gives \eqref{2.95}.
Conversely, suppose that \eqref{2.95} holds for some $\check A\in G_-$ and $F\in\C^{2n}$ with
$|F|^2=2n$. Since $h(\lambda)\check Ah(\lambda)$ belongs to $G_-$, there exists $y\in G$
such that
\begin{equation}
h(\lambda)\check Ah(\lambda)=L(y).
\label{2.96}
\end{equation}
Such $y$ is unique up to left-multiplication by an arbitrary element of $G_+$
(whereby one may bring $y$ into $G_-$ if one wishes to do so).
Picking $y$ according to \eqref{2.96}, and then setting
\begin{equation}
V=yh(\lambda)F,
\label{2.97}
\end{equation}
it is an elementary matter to show that \eqref{2.95} implies the following equation:
\begin{equation}
yYy^{-1}+CyYy^{-1}C+2\ri\mu(-V(CV)^\dagger-\1_{2n})+2\ri(\mu-\nu)C=0.
\label{2.98}
\end{equation}
It is a consequence of this equation that
\begin{equation}
(V(CV)^\dagger)^\dagger=(CV)V^\dagger=V(CV)^\dagger.
\label{2.99}
\end{equation}
This entails that $CV=\alpha V$ for some $\alpha\in\UN(1)$. Then
$V^\dagger=\alpha(CV)^\dagger$ also holds, and thus we must have $\alpha^2=1$.
Hence $\alpha$ is either $+1$ or $-1$. Taking the trace of the equality \eqref{2.98},
and using that $|V|^2=2n$ on account of $|F|^2=2n$, we obtain that $\alpha =-1$, i.e.
$CV+V=0$. This means that equation \eqref{2.98} reproduces \eqref{2.94}.
\end{proof}

To make progress, now we restrict our attention to the subset of $P_1^\reg$, where
the eigenvalue-parameter $\lambda$ of $Y$ verifies in addition to \eqref{2.87} also
the conditions
\begin{equation}
|\lambda_a\pm\lambda_b|\neq 2\mu
\quad\text{and}\quad
(\lambda_a-\nu)(\lambda_a-\vert 2\mu-\nu\vert)\neq 0,
\quad
\forall a,b\in\N_n.
\label{2.100}
\end{equation}

We call such $\lambda$ values `strongly regular', and let
$P_1^{\vreg}\subset P_1$ and $P_0^{\vreg}\subset P_0$ denote the corresponding
open subsets. Later we shall prove that $P_0^\vreg$ is \emph{dense} in $P_0$.
The above conditions will enable us to perform calculations that will lead to
a description of a dense subset of the reduced phase space. They ensure that we
never divide by zero in relevant steps of our arguments. The first such step is
the derivation of the following consequence of equation \eqref{2.95}.

\begin{lemma}
\label{lem:2.6}
The restriction of the matrix function $\check A$ \eqref{2.89} to $P_1^{\vreg}$
has the form
\begin{equation}
\check A_{j,k}=\frac{2\mu F_j\overline{(C F)}_k-
2(\mu-\nu)C_{j,k}}{2\mu+\Lambda_k-\Lambda_j},\quad
j,k\in\N_2n,
\label{2.101}
\end{equation}
where $F\in\C^{2n}$ satisfies $|F|^2=2n$ and $\Lambda=\diag(\lambda,-\lambda)$
varies on $P_1^\vreg$ according to \eqref{2.85}.
\end{lemma}

\begin{lemma}
\label{lem:2.7}
For any strongly regular $\lambda$ and $a\in\N_n$ define
\begin{equation}
w_a=\prod_{\substack{b=1\\(b\neq a)}}^n
\frac{(\lambda_a-\lambda_b)(\lambda_a+\lambda_b)}
{(2\mu-(\lambda_a-\lambda_b))(2\mu-(\lambda_a+\lambda_b))},\quad
w_{a+n}=\prod_{\substack{b=1\\(b\neq a)}}^n
\frac{(\lambda_a-\lambda_b)(\lambda_a+\lambda_b)}
{(2\mu+\lambda_a-\lambda_b)(2\mu+\lambda_a+\lambda_b)},
\label{2.102}
\end{equation}
and set $W_k=w_k\cF_k$ with $\cF_k=|F_k|^2$. Then the unitarity of the matrix
$\check A$ as given by \eqref{2.101} implies the following system of equations for
the pairs of functions $W_c$ and $W_{c+n}$ for any $c\in\N_n$:
\begin{equation}
(\mu+\lambda_c)W_c+(\mu-\lambda_c)W_{n+c}-2(\mu-\nu)=0,
\label{2.103}
\end{equation}
\begin{equation}
\lambda_c^2W_cW_{n+c}
-\mu(\mu-\nu)(W_c+W_{n+c})
+(\mu-\nu)^2+\mu^2-\lambda_c^2=0.
\label{2.104}
\end{equation}
For fixed $c\in\N_n$ and strongly regular $\lambda$, this system of equations admits two solutions, which are given by
\begin{equation}
(W_c,W_{n+c})=(W_c^+,W_{n+c}^+)=(w_c\cF_c^+,w_{c+n}\cF_{c+n}^+)
=(1-\frac{\nu}{\lambda_c},1+\frac{\nu}{\lambda_c}),
\label{2.105}
\end{equation}
and by
\begin{equation}
(W_c,W_{n+c})=(W_c^-,W_{n+c}^-)=(w_c\cF_c^-,w_{c+n}\cF_{c+n}^-)
=(-1+\frac{2\mu-\nu}{\lambda_c},-1-\frac{2\mu-\nu}{\lambda_c}).
\label{2.106}
\end{equation}
The functions $\cF_k^\pm$ satisfy the identities
\begin{equation}
\sum_{k=1}^{2n}\cF_k^+(\lambda)=2n
\quad\text{and}\quad
\sum_{k=1}^{2n}\cF_k^-(\lambda)=-2n.
\label{2.107}
\end{equation}
\end{lemma}

\begin{proof}
The derivation of equations \eqref{2.103}, \eqref{2.104} follows a similar derivation
due to Pusztai \cite{Pu11}, and is summarized in the appendix. We then solve the linear
equation \eqref{2.103} say for $W_{c+n}$ and substitute it into \eqref{2.104}.
This gives a quadratic equation for $W_c$ whose two solutions we can write down.
We note that the derivation of the equations \eqref{2.103} and \eqref{2.104}
presented in the appendix utilizes the full set of the conditions \eqref{2.100}.

To verify the identities \eqref{2.107}, we first extend $\lambda$ to vary in the open
subset of $\C^n$ subject to the conditions $\lambda_a^2\neq\lambda_b^2$ and
$\lambda_c \neq 0$, and then consider the sums that appear in \eqref{2.107} as
functions of a chosen component of $\lambda$ with the other components fixed.
These explicitly given sums are meromorphic functions having only first order poles,
and one may check that all residues at the apparent poles vanish. Hence the sums are
constant over $\C^n$, and the values of the constants can be established by looking at
a suitable asymptotic limit in the domain $C_2$ \eqref{2.73}, whereby all $w_k$ tend to $1$
and the pre-factors in \eqref{2.105} and \eqref{2.106} tend to $1$ and $-1$, respectively.
\end{proof}

Observe that neither any $w_k$ nor any $\cF_k^\pm$ ($k\in\N_{2n}$) can vanish if
$\lambda$ is strongly regular.
We know that the value of $\cF_k$ \eqref{2.93} is uniquely defined at every point of $P_1^\reg$.
Therefore only one of the solutions $(\cF_c^\pm,\cF_{c+n}^\pm)$ can be acceptable at any
$\lambda\in\fL(P_1^\vreg)$.
The identities in \eqref{2.107} and analyticity arguments strongly suggest that the acceptable solutions
are provided by $\cF_k^+$.
The first statement of the following lemma confirms that this is the case
for $\lambda\in C_2$ \eqref{2.73}.

\begin{lemma}
\label{lem:2.8}
The formulae \eqref{2.105} and \eqref{2.106} can be used to define $\cF_k^\pm$ as smooth real functions
on the domain $C_2$, and none of these functions vanishes at any $\lambda\in C_2$.
Then for any $\lambda\in C_2$ and $c\in\N_n$ at least one out of $\cF_c^-$ and $\cF^-_{c+n}$ is negative,
while $\cF_k^+>0$ for all $k\in\N_{2n}$.
Hence for $\lambda\in C_2\cap\fL(P_0)$
only $\cF_k^+(\lambda)$ can give the value of the function $\cF_k$ as defined in
\eqref{2.93}.
Taking any
$\lambda\in C_2$ and any $F\in\C^{2n}$ satisfying $|F_k|^2=\cF_k^+(\lambda)$,
the formula \eqref{2.101}
yields a unitary matrix that belongs to $G_-$ \eqref{2.19}.
This matrix $\check A$ and vector $F\in\C^{2n}$ solve equation \eqref{2.95}.
\end{lemma}

\begin{proof}
It is easily seen that $w_k(\lambda)>0$ for all $\lambda\in C_2$ and $k\in\N_{2n}$.
The statement about the negativity of either $\cF_c^-$ or $\cF^-_{c+n}$ thus follows from
the identity $W_c^-+W_{n+c}^-=-2$. The positivity of
$\cF_k^+$ is easily checked.
It is also readily verified that $\check A^\dagger=C\check AC$,
which entails that $\check A \in G_-$ once we know that $\check A$ is unitary.
For $\lambda \in C_2$ and $|F_k|^2=\cF_k^+(\lambda)$, the unitarity of $\check A$ \eqref{2.101}
can be shown by almost verbatim adaptation of the arguments proving Proposition 6 in \cite{Pu11-2}.

If $\lambda\in C_2$ is such that the denominators in \eqref{2.101} do not vanish,
then the formula \eqref{2.101} is plainly
equivalent to \eqref{2.95}.
Observe that only those elements $\lambda\in C_2$ for which $\lambda_n=\mu$ fail to satisfy
this condition.
At such $\lambda$ the matrix element
$\check A_{n,2n}$ has an apparent `first order pole', but one can check
by inspection of the formula \eqref{2.76} that $\check A_{n,2n}$ actually
remains finite and smooth even at such exceptional points,
and thus solves also \eqref{2.95} because of continuity.
\end{proof}

Before presenting the proof of Theorem \ref{thm:2.3}, note that at the point of
$\tilde S^0$ labelled by $(\lambda,e^{\ri\vartheta})$ the value of the function
$F$ \eqref{2.90} is equal to $f(\lambda,e^{\ri\vartheta})$ given in \eqref{2.75}.

\begin{proof}[Proof of Theorem \ref{thm:2.3}]
It follows from Lemma \ref{lem:2.5} and Lemma \ref{lem:2.8} that $\tilde S^0$ is
a subset of $P_1^\reg$ and $\fL(\tilde S^0)= C_2$. Taking into account Theorem
\ref{thm:2.1}, this implies that $y(\lambda,\vartheta)$ \eqref{2.79} and
$V(\lambda, \vartheta)$ \eqref{2.80} are well-defined smooth functions on
$C_2\times\T^n$. We next show that $\tilde S^0$ is a cross-section for the
residual gauge action on $\fL^{-1}(C_2) \cap P_1$. To do this, pick an arbitrary element
\begin{equation}
(\tilde y, h(\lambda) \ri\Lambda h(\lambda)^{-1},
\upsilon_{\mu,\nu}^\ell(\tilde V), \upsilon^r)
\in \fL^{-1}(C_2) \cap P_1.
\label{2.108}
\end{equation}
Because $\cF_k(\lambda)\neq 0$, we can find a unique element $e^{\ri \vartheta}\in \T^n$
and an element $\zeta \in Z$ \eqref{2.37} (which is unique up to scalar multiple) such that
\begin{equation}
F_k(\tilde y\zeta^{-1},h(\lambda)\ri\Lambda h(\lambda)^{-1},\tilde{V})
=f_k(\lambda,e^{\ri \vartheta}),\quad\forall k\in\N_{2n},
\label{2.109}
\end{equation}
up to a $k$-independent phase. We then see from \eqref{2.95} that
$L(\tilde y \zeta^{-1})=L(y(\lambda,\vartheta))$, which in turn
implies the existence of some (unique after $\zeta$ was chosen)
$\eta_+ \in G_+$ for which
\begin{equation}
\eta_+ \tilde y \zeta^{-1} = y(\lambda, \vartheta).
\label{2.110}
\end{equation}
Using also that $\zeta^{-1}h(\lambda)\zeta=h(\lambda)$,
we conclude from the last two equations that
\begin{equation}
\eta_+\tilde{V}=\eta_+\tilde y h(\lambda)F(\tilde y,h(\lambda)\ri\Lambda h(\lambda)^{-1},
\tilde{V})=y(\lambda,\vartheta)h(\lambda)f(\lambda,\vartheta)
=V(\lambda,e^{\ri\vartheta}).
\label{2.111}
\end{equation}
Thus we have shown that the element \eqref{2.108} can be gauge transformed into a point
of $\tilde S^0$, and this point is uniquely determined since \eqref{2.109} fixes
$e^{\ri\vartheta}$ uniquely. In other words, $\tilde S^0$ intersects every orbit of
the residual gauge action on $\fL^{-1}(C_2) \cap P_1$ in precisely one point.

The map from $C_2$ into $P$, given by the parametrization of $\tilde S^0$,
is obviously smooth, and hence we obtain the identifications
\begin{equation}
C_2 \simeq\tilde S^0\simeq(\fL^{-1}(C_2)\cap P_1)/(G_+\times Z)
\simeq\fL^{-1}(C_2)/(G_+\times G_+)=\fL_\red^{-1}(C_2).
\label{2.112}
\end{equation}
To establish the formula \eqref{2.83} of the reduced symplectic structure, we proceed
as follows. We define $G_+ \times G_+$ invariant real functions on $P$ by
\begin{equation}
\varphi_m(y,Y,V)=\dfrac{1}{m}\Re\big(\tr(Y^m)\big),\quad m\in\N,
\label{2.113}
\end{equation}
and
\begin{equation}
\chi_k(y,Y,\upsilon)=\Re\big(\tr(Y^ky^{-1} V V^\dagger yC)\big),
\quad k\in\N \cup\{0\}.
\label{2.114}
\end{equation}
The restrictions of these functions to $\tilde S^0$ are the respective
functions $\varphi_m^\red$ and $\chi_k^\red$:
\begin{equation}
\varphi_m^\red(\lambda,\vartheta)=
\begin{cases}
0,&\text{if}\ m\ \text{is odd},\\
\displaystyle(-1)^{\tfrac{m}{2}}\frac{2}{m}\sum_{j=1}^n\lambda_j^m,
&\text{if}\ m\ \text{is even},
\end{cases}
\label{2.115}
\end{equation}
and
\begin{equation}
\chi_k^\red(\lambda,\vartheta)=
\begin{cases}
\displaystyle - 2 (-1)^{\tfrac{k-1}{2}} \sum_{j=1}^n\lambda_j^k
\bigg[1-\frac{\kappa^2}{\lambda_j^2}\bigg]^{\frac{1}{2}}
X_j\sin(\vartheta_j),&\text{if }k\text{ is odd},\\
\displaystyle 2(-1)^{\tfrac{k}{2}} \sum_{j=1}^n\lambda_j^k
\bigg[1-\frac{\kappa^2}{\lambda_j^2}\bigg]^{\frac{1}{2}}
X_j\cos(\vartheta_j)
-\kappa\lambda_j^{k-1}\big(\cF_j-\cF_{n+j}\big),&\text{if }k\text{ is even},
\end{cases}
\label{2.116}
\end{equation}
where
\begin{equation}
X_j=\sqrt{\cF_j \cF_{n+j}}
=e^{-\ri\vartheta_j}\bigg[1-\frac{\nu^2}{\lambda_j^2}\bigg]^{\tfrac{1}{2}}
\prod_{\substack{k=1\\(k\neq j)}}^n
\bigg[1-\frac{4\mu^2}{(\lambda_j-\lambda_k)^2}\bigg]^{\tfrac{1}{2}}
\bigg[1-\frac{4\mu^2}{(\lambda_j+\lambda_k)^2}\bigg]^{\tfrac{1}{2}}.
\label{2.117}
\end{equation}
Then we calculate the pairwise Poisson brackets
of the set of functions $\varphi_m$, $\chi_k$ on $P$ and restrict the results to
$\tilde S^0$. This must coincide with the results of the direct calculation of the
Poisson brackets of the reduced functions $\varphi_m^\red$, $\chi_k^\red$ based on
the pull-back of the symplectic form $\Omega$ onto $\tilde S^0 \subset P$. Inspection
shows that the required equalities hold if and only if we have the formula in \eqref{2.83}
for the pull-back in question. This reasoning is very similar to that used in \cite{Pu11-2}
to find the corresponding reduced symplectic form. Since the underlying calculations are
rather laborious, we break them up into smaller pieces and only detail them following this
proof. As for the formula for the restriction of $\tilde \cH_1$ to $\tilde S^0$ displayed
in \eqref{2.83}, this is a matter of direct verification.
\end{proof}

The following line of thought is an appropriate adaptation of an argument presented by
Pusztai in \cite{Pu11-2} which since has been applied in the simpler case of A${}_n$
root system in \citepalias{AFG12}. Differences between these earlier results
and the calculations below are highlighted in the Discussion.

Consider the families of real-valued smooth functions $\varphi_m$ \eqref{2.113},
$\chi_k$ \eqref{2.114} on the phase space $P$ \eqref{2.26}, and the corresponding reduced
functions $\varphi_m^\red$ \eqref{2.115}, $\chi_k^\red$ \eqref{2.116} on $\tilde S^0$
\eqref{2.81}. Now let us take an arbitrary point $x=(y,Y,\upsilon^\ell,\upsilon^r)\in P$
and an arbitrary tangent vector
$\delta x=\delta y\oplus\delta Y\oplus\delta\upsilon^\ell\oplus 0\in T_x P$.
The derivative of $\varphi_m$ can be easily obtained and has the form
\begin{equation}
(d\varphi_m)_x(\delta x)=
\begin{cases}
0,&\text{if}\ m\ \text{is odd},\\
\langle Y^{m-1},\delta Y\rangle,&\text{if}\ m\ \text{is even}.
\end{cases}
\label{2.118}
\end{equation}
The derivative of $\chi_k$ can be written as
\begin{equation}
\begin{split}
(d\chi_k)_x(\delta x)=&
\bigg\langle\dfrac{\big[[Y^k,C]_\pm,y^{-1}Z(\upsilon^\ell)y\big]}{2},y^{-1}
\delta y\bigg\rangle\\
&+\bigg\langle\sum_{j=0}^{k-1}\dfrac{Y^{k-j-1}[y^{-1}Z(\upsilon^\ell)y,C]_\pm
Y^j}{2},\delta Y\bigg\rangle\\
&+\bigg\langle\dfrac{y[C,Y^k]_\pm y^{-1}+Cy[C,Y^k]_\pm
y^{-1}C}{4\ri\mu},\delta\upsilon^\ell\bigg\rangle,
\end{split}
\label{2.119}
\end{equation}
where $[A,B]_\pm=AB\pm BA$ with the sign of $(-1)^k$.
The Hamiltonian vector field of $\varphi_m$ is
\begin{equation}
(\bX_{\varphi_m})_x
=\Delta y\oplus\Delta Y\oplus\Delta\upsilon^\ell\oplus 0
=yY^{m-1}\oplus 0\oplus 0\oplus 0,
\label{2.120}
\end{equation}
while the Hamiltonian vector field corresponding to $\chi_k$ is
\begin{equation}
(\bX_{\chi_k})_x
=\Delta' y\oplus\Delta' Y\oplus\Delta' \upsilon^\ell\oplus 0,
\label{2.121}
\end{equation}
where
\begin{alignat}{3}
\Delta'y&=\dfrac{y}{2}\sum_{j=0}^{k-1}Y^{k-j-1}[y^{-1}Z(\upsilon^\ell)y,C]_\pm Y^j,
\label{2.122}\\
\Delta'Y&=\dfrac{1}{2}\big[[Y^k,y^{-1}Z(\upsilon^\ell)y]_\pm,C\big],
\label{2.123}\\
\Delta'\upsilon^\ell&=\dfrac{1}{4\ri\mu}\big[\big(y[C,Y^k]_\pm y^{-1}+Cy[C,Y^k]_\pm
y^{-1}C\big),\upsilon^\ell\big].
\label{2.124}
\end{alignat}

\begin{lemma}
\label{lem:2.9}
$\{\lambda_a,\lambda_b\}=0$ for any $a,b\in\{1,\dots,n\}$.
\end{lemma}

\begin{proof}
Using \eqref{2.118} one has $\{\varphi_m,\varphi_l\}\equiv 0$ for any $m,l\in\N$
which implies that $\{\varphi_m^\red,\varphi_l^\red\}\equiv 0$. Let $m,l\in\N$ be
arbitrary even numbers. Direct calculation of the Poisson bracket
$\{\varphi_m^\red,\varphi_l^\red\}$ using \eqref{2.115} and the Leibniz rule results in the
formula
\begin{equation}
\{\varphi_m^\red,\varphi_l^\red\}=(-1)^{\tfrac{m+l}{2}}4\sum_{a,b=1}^n
\lambda_a^{m-1}\{\lambda_a,\lambda_b\}\lambda_b^{l-1}.
\label{2.125}
\end{equation}
By introducing the $n\times n$ matrices
\begin{equation}
\bP_{a,b}=\{\lambda_a,\lambda_b\}\quad
\text{and}\quad
\bU_{a,b}=\lambda_a^{2b-1},\qquad
a,b\in\{1,\dots,n\}
\label{2.126}
\end{equation}
and choosing $m$ and $l$ from the set $\{1,\dots,2n\}$, the equation
$\{\varphi_m^\red,\varphi_l^\red\}\equiv 0$ can be cast into the matrix equation
\begin{equation}
(-1)^{\tfrac{m+l}{2}}\bU^\dag\bP\bU=\0_n.
\label{2.127}
\end{equation}
Since $\bU$ is an invertible Vandermonde-type matrix it follows from \eqref{2.127}
that $\bP=\0_n$ which reads as $\{\lambda_a,\lambda_b\}=0$ for all $a,b\in\{1,\dots,n\}$.
\end{proof}

\begin{lemma}
$\{\lambda_a,\vartheta_b\}=\delta_{a,b}$ for any $a,b\in\{1,\dots,n\}$.
\label{lem:2.10}
\end{lemma}

\begin{proof}
By choosing two even numbers, $k$ and $m$, and calculating the Poisson bracket
$\{\chi_k,\varphi_m\}$ at an arbitrary point $x=(y,Y,\upsilon^\ell,\upsilon^r)\in P$
the results \eqref{2.120}-\eqref{2.124} imply that
\begin{equation}
\{\chi_k,\varphi_m\}(x)=\chi_{k+m-1}(x)
+\frac{1}{2}\tr\big((Y^kCY^{m-1}-Y^{m-1}CY^k)y^{-1}Z(\upsilon^\ell)y\big).
\label{2.128}
\end{equation}
The computation of the reduced form of \eqref{2.128} shows that
\begin{equation}
\{\chi_k^\red,\varphi_m^\red\}=2\chi_{k+m-1}^\red.
\label{2.129}
\end{equation}
By utilizing \eqref{2.115}, \eqref{2.116} and the result of the previous lemma one can write
the l.h.s. of \eqref{2.129} as
\begin{equation}
\{\chi_k^\red,\varphi_m^\red\}=(-1)^{\tfrac{k+m}{2}}
4\sum_{b=1}^n\lambda_b^k\bigg[1-\frac{\kappa^2}{\lambda_b^2}\bigg]^{\tfrac{1}{2}}
|X_b(\lambda)|\sin(\vartheta_b)\sum_{a=1}^n\{\lambda_a,\vartheta_b\}\lambda_a^{m-1}.
\label{2.130}
\end{equation}
Now, returning to equation \eqref{2.129} together with \eqref{2.130}
one can obtain the following equivalent form
\begin{equation}
\sum_{b=1}^n
\lambda_b^k\bigg[1-\frac{\kappa^2}{\lambda_b^2}\bigg]^{\tfrac{1}{2}}
|X_b(\lambda)|\sin(\vartheta_b)
\bigg(\sum_{a=1}^n\{\lambda_a,\vartheta_b\}\lambda_a^{m-1}-\lambda_b^{m-1}\bigg)
=0.
\label{2.131}
\end{equation}
By introducing the $n\times n$ matrices
\begin{equation}
\bV_{b,d}=\bigg[1-\frac{\kappa^2}{\lambda_b^2}\bigg]^{\tfrac{1}{2}}
|X_b(\lambda)|\sin(\vartheta_b)\bigg(\sum_{a=1}^n\{\lambda_a,\vartheta_b\}
\lambda_a^{2d-1}-\lambda_b^{2d-1}\bigg),
\quad b,d\in\{1,\dots,n\}
\label{2.132}
\end{equation}
and using the Vandermonde-type matrix $\bU$ defined in \eqref{2.126}
one is able to write \eqref{2.131} into the matrix equation $\bU^\dag\bV=\0_n$.
Since $\bU$ is invertible $\bV=\0_n$ and therefore in the dense subset of
$C_2\times\T^n$ where $\sin(\vartheta_b)\neq 0$ the following holds
\begin{equation}
\sum_{a=1}^n\{\lambda_a,\vartheta_b\}\lambda_a^{m-1}-\lambda_b^{m-1}=0,
\quad\forall\,b\in\{1,\dots,n\}.
\label{2.133}
\end{equation}
With the matrices $\bU$ and
\begin{equation}
\bQ_{b,a}=\{\lambda_a,\vartheta_b\},\quad
a,b\in\{1,\dots,n\}
\label{2.134}
\end{equation}
equation \eqref{2.133} can be written equivalently as $\bQ\bU-\bU=\0_n$,
which immediately implies that $\bQ=\1_n$. Due to the continuity of Poisson bracket
$\bQ=\1_n$ must hold for every point in $C_2\times\T^n$, therefore one has
$\{\lambda_a,\vartheta_b\}=\delta_{a,b}$ for all $a,b\in\{1,\dots,n\}$.
\end{proof}

\begin{lemma}
\label{lem:2.11}
$\{\vartheta_a,\vartheta_b\}=0$ for any $a,b\in\{1,\dots,n\}$.
\end{lemma}

\begin{proof}
Let $k$ and $l$ be two arbitrary odd integers, and calculate the Poisson bracket
$\{\chi_k^\red,\chi_l^\red\}$ indirectly, that is, work out the Poisson bracket
$\{\chi_k,\chi_l\}=\Omega(\bX_{\chi_l},\bX_{\chi_k})$
explicitly and restrict it to the gauge \eqref{2.81}. The first term $\langle y^{-1}\Delta y,\Delta'Y\rangle$ in \eqref{2.15} reads
\begin{equation}
\begin{split}
\langle y^{-1}\Delta y,\Delta'Y\rangle=&
(-1)^{\tfrac{k+l+2}{2}} 2\,l\sum_{a=1}^n
\lambda_a^{k+l-1}\bigg[1-\frac{\kappa}{\lambda_a^2}\bigg]
|X_a(\lambda)|^2\sin(2\vartheta_a)\\
&(-1)^{\tfrac{k+l+2}{2}} 2\sum_{\substack{a,b=1\\(a\neq b)}}^n
\lambda_a^k\lambda_b^l
\bigg[1-\frac{\kappa^2}{\lambda_a^2}\bigg]^{\tfrac{1}{2}}
\bigg[1-\frac{\kappa^2}{\lambda_b^2}\bigg]^{\tfrac{1}{2}}
|X_a||X_b|
\frac{\sin(\vartheta_a-\vartheta_b)}
{\lambda_a+\lambda_b}\\
&(-1)^{\tfrac{k-l+2}{2}} 2\sum_{\substack{a,b=1\\(a\neq b)}}^n
\lambda_a^k\lambda_b^l
\bigg[1-\frac{\kappa^2}{\lambda_a^2}\bigg]^{\tfrac{1}{2}}
\bigg[1-\frac{\kappa^2}{\lambda_b^2}\bigg]^{\tfrac{1}{2}}
|X_a||X_b|
\frac{\sin(\vartheta_a+\vartheta_b)}
{\lambda_a-\lambda_b}.
\end{split}
\label{2.135}
\end{equation}
Due to antisymmetry in the indices $\langle y^{-1}\Delta' y,\Delta Y\rangle$ in
\eqref{2.15} is obtained by interchanging $k$ and $l$
\begin{equation}
\begin{split}
\langle y^{-1}\Delta' y,\Delta Y\rangle=&
(-1)^{\tfrac{k+l+2}{2}} 2\,k\sum_{a=1}^n
\lambda_a^{k+l-1}\bigg[1-\frac{\kappa}{\lambda_a^2}\bigg]
|X_a(\lambda)|^2\sin(2\vartheta_a)\\
&(-1)^{\tfrac{k-l+2}{2}} 2\sum_{\substack{a,b=1\\(a\neq b)}}^n
\lambda_a^k\lambda_b^l
\bigg[1-\frac{\kappa^2}{\lambda_a^2}\bigg]^{\tfrac{1}{2}}
\bigg[1-\frac{\kappa^2}{\lambda_b^2}\bigg]^{\tfrac{1}{2}}
|X_a||X_b|
\frac{\sin(\vartheta_a-\vartheta_b)}
{\lambda_a+\lambda_b}\\
&(-1)^{\tfrac{k+l+2}{2}} 2\sum_{\substack{a,b=1\\(a\neq b)}}^n
\lambda_a^k\lambda_b^l
\bigg[1-\frac{\kappa^2}{\lambda_a^2}\bigg]^{\tfrac{1}{2}}
\bigg[1-\frac{\kappa^2}{\lambda_b^2}\bigg]^{\tfrac{1}{2}}
|X_a||X_b|
\frac{\sin(\vartheta_a+\vartheta_b)}
{\lambda_a-\lambda_b}.
\end{split}
\label{2.136}
\end{equation}
One can easily check that the third term in \eqref{2.15} vanishes.
The last term takes the form
\begin{equation}
\begin{split}
\langle[D_\upsilon,D'_\upsilon],\upsilon\rangle=
&(-1)^{\tfrac{k+l+2}{2}} 4\sum_{\substack{a,b=1\\(a\neq b)}}^n
\lambda_a^k\lambda_b^l
\bigg[1-\frac{\kappa^2}{\lambda_a^2}\bigg]^{\tfrac{1}{2}}
\bigg[1-\frac{\kappa^2}{\lambda_b^2}\bigg]^{\tfrac{1}{2}}
|X_a||X_b|
\frac{\sin(\vartheta_a-\vartheta_b)}
{\big(4\mu^2-(\lambda_a+\lambda_b)^2\big)(\lambda_a+\lambda_b)}\\
&(-1)^{\tfrac{k-l+2}{2}} 4\sum_{\substack{a,b=1\\(a\neq b)}}^n
\lambda_a^k\lambda_b^l
\bigg[1-\frac{\kappa^2}{\lambda_a^2}\bigg]^{\tfrac{1}{2}}
\bigg[1-\frac{\kappa^2}{\lambda_b^2}\bigg]^{\tfrac{1}{2}}
|X_a||X_b|
\frac{\sin(\vartheta_a+\vartheta_b)}
{\big(4\mu^2-(\lambda_a-\lambda_b)^2\big)(\lambda_a-\lambda_b)}.
\end{split}
\label{2.137}
\end{equation}
As a result of this indirect calculation one obtains the following expression for
$\{\chi_k^\red,\chi_l^\red\}$
\begin{equation}
\begin{split}
\{\chi_k^\red,\chi_l^\red\}&=
(-1)^{\tfrac{k-l+2}{2}} 2(k-l)\sum_{a=1}^n\lambda_a^{k+l-1}
\bigg[1-\frac{\kappa^2}{\lambda_a^2}\bigg]|X_a|^2\sin(2\vartheta_a)\\
&(-1)^{\tfrac{k+l+2}{2}} 16\mu^2\sum_{\substack{a,b=1\\(a\neq b)}}^n
\lambda_a^k\lambda_b^l
\bigg[1-\frac{\kappa^2}{\lambda_a^2}\bigg]^{\tfrac{1}{2}}
\bigg[1-\frac{\kappa^2}{\lambda_b^2}\bigg]^{\tfrac{1}{2}}
|X_a||X_b|
\frac{\sin(\vartheta_a-\vartheta_b)}
{\big(4\mu^2-(\lambda_a+\lambda_b)^2\big)(\lambda_a+\lambda_b)}\\
&(-1)^{\tfrac{k-l+2}{2}} 16\mu^2\sum_{\substack{a,b=1\\(a\neq b)}}^n
\lambda_a^k\lambda_b^l
\bigg[1-\frac{\kappa^2}{\lambda_a^2}\bigg]^{\tfrac{1}{2}}
\bigg[1-\frac{\kappa^2}{\lambda_b^2}\bigg]^{\tfrac{1}{2}}
|X_a||X_b|
\frac{\sin(\vartheta_a+\vartheta_b)}
{\big(4\mu^2-(\lambda_a-\lambda_b)^2\big)(\lambda_a-\lambda_b)}.
\end{split}
\label{2.138}
\end{equation}
One can also carry out a direct computation of $\{\chi_k^\red,\chi_l^\red\}$
by using basic properties of the Poisson bracket and the previous two lemmas
\begin{equation}
\begin{split}
\{\chi_k^\red,\chi_l^\red\}&=
(-1)^{\tfrac{k-l+2}{2}} 2(k-l)\sum_{a=1}^n\lambda_a^{k+l-1}
\bigg[1-\frac{\kappa^2}{\lambda_a^2}\bigg]|X_a|^2\sin(2\vartheta_a)\\
&(-1)^{\tfrac{k+l+2}{2}} 16\mu^2\sum_{\substack{a,b=1\\(a\neq b)}}^n
\lambda_a^k\lambda_b^l
\bigg[1-\frac{\kappa^2}{\lambda_a^2}\bigg]^{\tfrac{1}{2}}
\bigg[1-\frac{\kappa^2}{\lambda_b^2}\bigg]^{\tfrac{1}{2}}
|X_a||X_b|
\frac{\sin(\vartheta_a-\vartheta_b)}
{\big(4\mu^2-(\lambda_a+\lambda_b)^2\big)(\lambda_a+\lambda_b)}\\
&(-1)^{\tfrac{k-l+2}{2}} 16\mu^2\sum_{\substack{a,b=1\\(a\neq b)}}^n
\lambda_a^k\lambda_b^l
\bigg[1-\frac{\kappa^2}{\lambda_a^2}\bigg]^{\tfrac{1}{2}}
\bigg[1-\frac{\kappa^2}{\lambda_b^2}\bigg]^{\tfrac{1}{2}}
|X_a||X_b|
\frac{\sin(\vartheta_a+\vartheta_b)}
{\big(4\mu^2-(\lambda_a-\lambda_b)^2\big)(\lambda_a-\lambda_b)}\\
&(-1)^{\tfrac{k-l}{2}} 4\sum_{a,b=1}^n\lambda_a^k\lambda_b^l
\bigg[1-\frac{\kappa^2}{\lambda_a^2}\bigg]^{\tfrac{1}{2}}
\bigg[1-\frac{\kappa^2}{\lambda_b^2}\bigg]^{\tfrac{1}{2}}
|X_a||X_b|\cos(\vartheta_a)\cos(\vartheta_b)\{\vartheta_a,\vartheta_b\}.
\end{split}
\label{2.139}
\end{equation}
Now it is obvious that \eqref{2.138} and \eqref{2.139}
must be equal therefore the extra term must vanish
\begin{equation}
\sum_{a,b=1}^n\lambda_a^k\lambda_b^l
\bigg[1-\frac{\kappa^2}{\lambda_a^2}\bigg]^{\tfrac{1}{2}}
\bigg[1-\frac{\kappa^2}{\lambda_b^2}\bigg]^{\tfrac{1}{2}}
|X_a||X_b|\cos(\vartheta_a)\cos(\vartheta_b)\{\vartheta_a,\vartheta_b\}=0.
\label{2.140}
\end{equation}
By utilizing the $n\times n$ matrices
\begin{equation}
\bW_{a,b}=\lambda_a^b\bigg[1-\frac{\kappa^2}{\lambda_a^2}\bigg]^{\tfrac{1}{2}}
|X_a(\lambda)|\cos(\vartheta_a),\quad
\bR_{a,b}=\{\vartheta_a,\vartheta_b\},\quad
a,b\in\{1,\dots,n\}
\label{2.141}
\end{equation}
one can reformulate \eqref{2.140} as the matrix equation
\begin{equation}
\bW^\dag\bR\,\bW=\0_n.
\label{2.142}
\end{equation}
Since $\bW$ is easily seen to be invertible in a dense subset of the phase space
$C_2\times\T^n$, eq. \eqref{2.142} and the continuity of Poisson bracket imply
$\bR=\0_n$ for the full phase space, i.e. $\{\vartheta_a,\vartheta_b\}=0$ for all
$a,b\in\{1,\dots,n\}$.
\end{proof}

Lemmas \ref{lem:2.9}, \ref{lem:2.10}, and \ref{lem:2.11} together imply the following (claimed in Theorem \ref{thm:2.3})

\begin{theorem}
\label{thm.2.12}
The reduced symplectic structure on $\tilde S^0$ \eqref{2.81}, given by the pull-back of
$\Omega$ \eqref{2.26} by the tautological injection $\tilde\sigma_0\colon\tilde{S}^0\to P$,
has the canonical form
$\tilde\sigma_0^\ast(\Omega)=\sum_{a=1}^nd\lambda_a\wedge d\vartheta_a$.
\end{theorem}

\subsubsection{Density properties}

So far we dealt with the open subset $\fL_\red^{-1}(C_2)$ of the reduced
phase space. Here we show that Theorem \ref{thm:2.3} contains `almost all'
information about the dual system
since $\fL_\red^{-1}(C_2)\subset P_\red$ is a \emph{dense} subset.
This key result will be proved by combining two lemmas.

\begin{lemma}
\label{lem:2.13}
The subset $P_0^\vreg \subset P_0$ of the constraint surface where the
range of the eigenvalue map $\fL$ \eqref{2.72} satisfies the conditions
\eqref{2.87} and \eqref{2.100} is dense.
\end{lemma}

\begin{proof}
Let us first of all note that $P_0$ is a connected regular analytic
submanifold of $P$.
In fact, it is a regular (embedded) analytic submanifold of the analytic
manifold $P$ since the momentum map is analytic and zero is its regular
value (because the effectively acting gauge group \eqref{2.62} acts freely on
$P_0$). The connectedness follows from Theorem \ref{thm:2.1}, which implies that
$P_0$ is diffeomorphic to the product of $S$ \eqref{2.50} and the group \eqref{2.62},
and both are connected.

For any $Y\in \cG$ denote by $\{\ri \Lambda_a\}_{a=1}^{2n}$ the set of its
eigenvalues counted with multiplicities.
Then the following formulae
\begin{equation}
\cR(y,Y,V)=\prod_{\substack{a,b=1\\(a\neq b)}}^{2n}(\Lambda_a-\Lambda_b)
\prod_{a=1}^{2n}(\Lambda_a^2-\kappa^2),
\label{2.143}
\end{equation}
\begin{equation}
\cS(y,Y,V)=\prod_{\substack{a,b=1\\(a\neq b)}}^{2n}[(\Lambda_a-\Lambda_b)^2-4\mu^2]
\prod_{a=1}^{2n}\left[(\Lambda_a^2-\mu^2)(\Lambda_a^2-\nu^2)(\Lambda_a^2-(2\mu-\nu)^2)
\right].
\label{2.144}
\end{equation}
define analytic functions on $P_0$.
Indeed, $\cR$ and $\cS$ are symmetric polynomials in the eigenvalues
of $Y$, and hence can be expressed as polynomials in the coefficients
of the characteristic
polynomial of $Y$, which are polynomials in the matrix elements of $Y$.
The product $\cR \cS$ is also an analytic function on $P_0$, and
the subset $P_0^\vreg$, can be characterized as
\begin{equation}
P_0^\vreg = \{ x\in P_0 \mid \cR(x) \cS(x) \neq 0\}.
\label{2.145}
\end{equation}
It is clear from Theorem \ref{thm:2.3} that $\cR\cS$ does not vanish identically
on $P_0$.
Since the zero set of a non-zero analytic function on a connected analytic
manifold cannot
contain any open set, equation \eqref{2.145} implies that
$P_0^\vreg$ is a dense subset of $P_0$.
\end{proof}

Let $\overline{C}_2$ be the closure of the domain $C_2 \subset \R^n$.
Eventually, it will turn out that $\fL(P_0) = \overline{C}_2$.
For now, we wish to prove the following.

\begin{lemma}
\label{lem:2.14}
For every boundary point $\lambda^0\in\partial\overline{C}_2 $ there
exist an open ball $B(\lambda^0)\subset\R^n$ around $\lambda^0$ that
does not contain any strongly regular $\lambda$ which lies outside
$\overline{C}_2$ and belongs to $\fL(P_0)$.
\end{lemma}

\begin{proof}
We start by noticing that for any boundary point
$\lambda^0\in\partial\overline{C}_2$
there is a ball $B(\lambda^0)$ centred at $\lambda^0$ such
that any strongly regular $\lambda\in B(\lambda^0)\setminus\overline{C}_2$
is subject to either of the following: (i) there is an index
$a\in\{1,\dots,n-1\}$ such that
\begin{equation}
\lambda_a-\lambda_{a+1}<2\mu
\quad\text{and}\quad
\lambda_b-\lambda_{b+1}>2\mu
\quad\forall\;b<a,
\label{2.146}
\end{equation}
or (ii) we have
\begin{equation}
\lambda_a-\lambda_{a+1}>2\mu,
\quad
a=1,\dots,n-1
\quad\text{and}\quad
 \lambda_n<\nu.
\label{2.147}
\end{equation}

Let us consider a strongly regular $\lambda\in B(\lambda^0)$
that falls into case (i) \eqref{2.146} and is so close to $C_2$
that we still have
\begin{equation}
\lambda_k-\lambda_{k+1}>\mu,\quad
\forall\;k\in\{1,\dots,n-1\}.
\label{2.148}
\end{equation}
It then follows that
\begin{equation}
\lambda_a-\lambda_b>2\mu,\quad\forall\;b>a+1,
\label{2.149}
\end{equation}
and
\begin{equation}
\lambda_a+\lambda_b>2\mu,\quad\forall\;b\in\{1,\dots,n\}.
\label{2.150}
\end{equation}
Inspection of the signs of $w_a(\lambda)$ and $w_{a+n}(\lambda)$
in \eqref{2.102} gives
\begin{equation}
w_a(\lambda)<0<w_{a+n}(\lambda).
\label{2.151}
\end{equation}
Since every boundary point $\lambda^0\in\partial\overline{C}_2$ satisfies
$\lambda_a^0>\lambda_n^0\geq \nu$ for all $a\in\{1,\dots,n-1\}$, we may choose
a small enough ball centred at $\lambda^0$ to ensure that for
$\lambda$ inside that ball the above inequalities as well as
$\lambda_a>\nu$ hold.
On account of $\lambda_a>\nu>0$ and $\mu>0$ we then have
\begin{equation}
1 - \frac{\nu}{\lambda_a}>0 \quad\text{and}\quad
-1-\frac{2\mu-\nu}{\lambda_a}<0.
\label{2.152}
\end{equation}
By combining \eqref{2.104} and \eqref{2.105} with \eqref{2.151} and
\eqref{2.152} we conclude that
\begin{equation}
\cF_a^+(\lambda)<0 \quad\text{and}\quad\cF_{a+n}^-(\lambda)<0.
\label{2.153}
\end{equation}
By Lemma \ref{lem:2.7}, these inequalities imply that
$\cF_a(\lambda)$ and $\cF_{a+n}(\lambda)$ cannot be both non-negative,
which contradicts the defining equation \eqref{2.93}.
This proves the claim in the case (i) \eqref{2.146}.

Let us consider a strongly regular $\lambda$ satisfying (ii)
\eqref{2.147}.
In this case we can verify that
\begin{equation}
1-\frac{\nu}{\lambda_n}<0,\quad
w_n(\lambda)>0,\quad w_{n+a}(\lambda)>0.
\label{2.154}
\end{equation}
Thus we see from \eqref{2.105} that $\cF_{2n}^+(\lambda)<0$.
Since the sum of the two components on the right hand side of \eqref{2.106} is negative, we also see
that at least one out of $\cF_{n}^-(\lambda)$ and $\cF_{2n}^-(\lambda)$ is negative.
Therefore equations \eqref{2.103} and \eqref{2.104} exclude the unitarity of $\check A$ \eqref{2.101}
in the case (ii) \eqref{2.147} as well.
\end{proof}

\begin{proposition}
\label{prop:2.15}
The $\lambda$-image of the constraint surface is contained in
$\overline{C}_2$, i.e. we have
\begin{equation}
\fL(P_0) \subseteq \overline{C}_2.
\label{2.155}
\end{equation}
As a consequence, $\fL_\red^{-1}(C_2)$ is dense in $P_\red$.
\end{proposition}

\begin{proof}
Since $P_0^\vreg \subset P_0$ is dense and $\fL\colon P_0\to\R^n$ \eqref{2.72}
is continuous, $\fL(P_0^\vreg) \subset \fL(P_0)$ is dense.
Thus it follows from Lemma \ref{lem:2.14} that for any
$\lambda^0\in\partial C_2$ there exists
a ball around $\lambda^0$ that does not contain \emph{any} element of
$\fL(P_0)$ lying outside $\overline{C}_2$.

Suppose that \eqref{2.155} is not true, which means that there exists some
$\lambda^\ast\in \fL(P_0)\setminus \overline{C}_2$. Taking any element
$\hat \lambda \in \fL(P_0)$ that lies in $C_2$, it is must be possible
to connect $\lambda^\ast$ to $\hat \lambda$
by a continuous curve in $\fL(P_0)$, since $P_0$ is connected.
Starting from the point $\lambda^\ast$, any such continuous curve
must pass through some point of the boundary $\partial C_2$.
However, this is impossible since we know that
$\fL(P_0)\setminus \overline{C}_2$ does not contain any series that
converges to a point of $\partial C_2$.
This contradiction shows that \eqref{2.155} holds.

By \eqref{2.155} we have $P_0^\vreg \subset \fL^{-1}(C_2)$, and we know from Lemma \ref{lem:2.13} that
$P_0^\vreg \subset P_0$ is dense. These together entail that $\fL_\red^{-1}(C_2)\subset P_\red $ is dense.
\end{proof}

\subsubsection{Global characterization of the dual system}

We have seen that
\begin{equation}
P_0^\vreg \subset\fL^{-1}(C_2) \subset P_0
\label{2.156}
\end{equation}
is a chain of dense open submanifolds.
These project onto dense open submanifolds of $P_\red$ and
their images under the map $\fL$ \eqref{2.72} are
dense subsets of $\fL(P_0)=\fL_\red(P_\red)$:
\begin{equation}
\fL(P_0^\vreg) \subset C_2 \subset \fL(P_0).
\label{2.157}
\end{equation}

Now introduce the set
\begin{equation}
\C^n_{\neq }= \{ z\in \C^n \mid \prod_{k=1}^n z_k \neq 0\}.
\label{2.158}
\end{equation}
The parametrization
\begin{equation}
z_j=\sqrt{\lambda_j-\lambda_{j+1}-2\mu}\prod_{a=1}^j e^{\ri\vartheta_a},
\,\,\, j=1,\dots, n-1, \qquad
z_n=\sqrt{\lambda_n-\nu}\prod_{a=1}^n e^{\ri\vartheta_a}
\label{2.159}
\end{equation}
provides a diffeomorphism between $C_2 \times \T^n$ and $\C^n_{\neq }$.
Thus we can view $z\in \C^n_{\neq}$ as a variable parametrizing
$C_2 \times \T^n$ that corresponds to the semi-global cross-section $\tilde S^0$
by Theorem \ref{thm:2.3}. Below, we shall exhibit a \emph{global cross-section} in $P_0$,
which will be diffeomorphic to $\C^n$. In other words, the `semi-global'
model of the dual systems will be completed into a global model by allowing
the zero value for the complex variables $z_k$. This completion results
from the symplectic reduction automatically.

First of all, let us note that the inverse of the parametrization \eqref{2.159} gives
\begin{equation}
\lambda_k(z)=\nu+2(n-k)\mu+\sum_{j=k}^nz_j\bar z_j,\quad k=1,\dots,n,
\label{2.160}
\end{equation}
which extend to smooth functions over $\C^n$.
The range of the extended map $z \mapsto (\lambda_1, \dots, \lambda_n)$
is the closure $\overline{C}_2$ of the polyhedron $C_2$.
The variables $e^{\ri \vartheta_k}$ are well-defined only over
$\C^n_{\neq }$, where the parametrization \eqref{2.159} entails the equality
\begin{equation}
\sum_{k=1}^nd\lambda_k\wedge d\vartheta_k
= \ri\sum_{k=1}^ndz_k \wedge d\bar z_k.
\label{2.161}
\end{equation}
An easy inspection of the formulae \eqref{2.75} shows that the
functions $f_a$ can be recast as
\begin{equation}
f_k(\lambda, e^{\ri \vartheta}) = \vert z_k \vert g_k(z),
\quad
f_{n+k}(\lambda, e^{\ri \vartheta})
= e^{\ri \vartheta_k} \vert z_{k-1}\vert g_{n+k}(z),
\qquad k=1,\dots, n,\,\,\, z_0=1,
\label{2.162}
\end{equation}
with uniquely defined functions $g_1(z),\dots, g_{2n}(z)$
that extend to smooth (actually real-analy\-tic) positive functions on $\C^n$.
Note that these functions depend on $z$ only through $\lambda(z)$, i.e.
one has
\begin{equation}
g_a(z) = \eta_a(\lambda(z)),\quad a=1,\dots,2n,
\label{2.163}
\end{equation}
with suitable functions $\eta_a$ that one could display explicitly.
The absolute values $\vert z_k \vert$ that appear in \eqref{2.162} are not
smooth at $z_k=0$, and the phases
$e^{\ri \vartheta_k}$ are not well-defined there.
The crux is that both of these `troublesome features' can be
removed by applying suitable gauge transformations to the elements
of the cross-section $\tilde S^0$ \eqref{2.81}.
To demonstrate this, we define $m=m(e^{\ri \vartheta})\in Z_{G_+}(\cA)$ by
\begin{equation}
m_k(e^{\ri \vartheta})=\prod_{j=1}^k e^{-\ri \vartheta_j},\quad k=1,\dots,n.
\label{2.164}
\end{equation}
Conforming with \eqref{2.37}, we also set $m_{k+n}=m_k$.
Then the gauge transformation by $(m,m)\in G_+ \times G_+$ operates on the
$\C^{2n}$-valued vector $f(\lambda, e^{\ri \vartheta})$ and on the matrix
$\check A(\lambda, e^{\ri \vartheta})$ according to
\begin{equation}
f(\lambda, e^{\ri \vartheta}) \to
m(e^{\ri \vartheta}) f(\lambda, e^{\ri \vartheta}) \equiv \phi(z),
\qquad
\check A(\lambda, e^{\ri \vartheta)} \to m(e^{\ri \vartheta})
\check A(\lambda, e^{\ri \vartheta}) m(e^{\ri \vartheta})^{-1} \equiv \tilde A(z),
\label{2.165}
\end{equation}
which defines the functions $\phi(z)$ and $\tilde A(z)$ over $\C^n_{\neq}$.
The resulting functions have the form
\begin{equation}
\phi_k(z)= \bar z_k g_k(z), \quad
\phi_{n+k}(z) = \bar z_{k-1} g_{n+k}(z), \quad k=1,\dots, n,
\label{2.166}
\end{equation}
and
\begin{equation}
\tilde{A}_{a,b}(z)
=-\frac{2\mu\bar{z}_az_{b-1}g_a(z)g_{n+b}(z)}
{\lambda_a(z)-\lambda_b(z)-2\mu}, \quad 1\leq a,b \leq n,
\label{2.167}
\end{equation}
\begin{equation}
\tilde{A}_{a,n+b}(z)
=-\frac{2\mu\bar{z}_az_bg_a(z)g_b(z)}
{\lambda_a(z)+\lambda_b(z)-2\mu}
+\delta_{a,b}\frac{\mu -\nu}{\lambda_a(z)-\mu},
\label{2.168}
\end{equation}
\begin{equation}
\tilde{A}_{n+a,b}(z)
=\frac{2\mu\bar{z}_{a-1}z_{b-1}g_{n+a}(z)g_{n+b}(z)}
{\lambda_a(z)+\lambda_b(z)+2\mu}
-\delta_{a,b}\frac{\mu-\nu}{\lambda_a(z)+\mu},
\label{2.169}
\end{equation}
\begin{equation}
\tilde{A}_{n+a,n+b}(z)
=\frac{2\mu\bar{z}_{a-1}z_b g_{n+a}(z)g_b(z)}
{\lambda_a(z)-\lambda_b(z)+2\mu}.
\label{2.170}
\end{equation}
Now the important point is that, as is easily verified, the apparent
singularities coming from vanishing denominators in $\tilde A$ all cancel,
and both $\phi(z)$ and $\tilde A(z)$ extend to smooth (actually real-analytic) functions on
the whole of $\C^n$.
In particular, note the relation
\begin{equation}
\tilde{A}_{k,k+1}(z)= \tilde A_{k+n+1,k+n}(z) = - 2\mu g_k(z) g_{k+n+1}(z),
\quad k=1,\dots, n-1.
\label{2.171}
\end{equation}
Corresponding to \eqref{2.77}, we also have the matrix $\tilde B(z)\equiv-(h(\lambda(z))
\tilde A(z)h(\lambda(z)))^\dagger$.
This is smooth over $\C^n$ since both $\tilde A(z)$ and $h(\lambda(z))$ \eqref{2.69} are smooth.
It follows from their defining equations that the induced gauge transformations
of $y(\lambda, e^{\ri \vartheta})$ \eqref{2.79} and $V(\lambda, e^{\ri \vartheta})$ \eqref{2.80}
are given by
\begin{equation}
y(\lambda, e^{\ri \vartheta}) \to m(e^{\ri \vartheta})
y(\lambda, e^{\ri \vartheta}) m(e^{\ri \vartheta})^{-1} \equiv \tilde y(z),
\label{2.172}
\end{equation}
and
\begin{equation}
V(\lambda, e^{\ri \vartheta}) \to m(e^{\ri \vartheta}) V(\lambda, e^{\ri \vartheta})
= \tilde y(z) h(\lambda(z)) \phi(z)
\equiv \tilde V(z).
\label{2.173}
\end{equation}
Since $\tilde y(z)$ is a uniquely defined smooth function of $\tilde B(z)$, both $\tilde y(z)$ and
$\tilde V(z)$ are smooth functions on the whole of $\C^n$.

After these preparations, we are ready to state the main result of this chapter.

\begin{theorem}
\label{thm:2.16}
By using the above notations, consider the set
\begin{equation}
\tilde{S}=\{(\tilde{y}(z),\ri h(\lambda(z))\Lambda(\lambda(z))h(\lambda(z))^{-1},
\upsilon^\ell_{\mu, \nu}(\tilde V(z)),\upsilon^r)\mid z\in\C^n\,\}.
\label{2.174}
\end{equation}
This set defines a global cross-section for the $G_+\times G_+$-action
on the constraint surface $P_0$.
The parametrization of the elements of $\tilde{S}$ by $z\in\C^n$ gives
rise to a symplectic diffeomorphism between $(P_\red,\Omega_\red)$ and
$\C^n$ equipped with the Darboux form $\ri\sum_{k=1}^n dz_k\wedge d\bar z_k$.
The spectral invariants of the `global RSvD Lax matrix'
\begin{equation}
\tilde L(z)\equiv h(\lambda(z))\tilde A(z)h(\lambda(z))
\label{2.175}
\end{equation}
yield commuting Hamiltonians on $\C^n$ that represent the reductions of
the Hamiltonians spanning the Abelian Poisson algebra $\fQ^1$ \eqref{2.32}.
\end{theorem}

\begin{proof}
Let us denote by
\begin{equation}
z \mapsto \tilde \sigma(z)
\label{2.176}
\end{equation}
the assignment of the element of $\tilde S$ to $z\in \C^n$ as given in \eqref{2.174}.
The map $\tilde\sigma\colon\C^n\to P$ \eqref{2.26} is smooth (even real-analytic)
and we have to verify that it possesses the following properties.
First, $\tilde\sigma$ takes values in the constraint surface $P_0$.
Second, with $\Omega$ in \eqref{2.26},
\begin{equation}
\tilde\sigma^\ast(\Omega)=\ri\sum_{k=1}^ndz_k\wedge d\bar z_k.
\label{2.177}
\end{equation}
Third, $\tilde\sigma$ is injective.
Fourth, the image $\tilde S$ of $\tilde\sigma$ intersects every orbit of $G_+ \times G_+$
in $P_0$ in precisely one point.

Let us start by recalling from Theorem \ref{thm:2.3} the map
$(\lambda, \theta)\mapsto \tilde \sigma_0(\lambda, \theta)$
that denotes the assignment of the general element of $\tilde S^0$ \eqref{2.81} to
$(\lambda, \theta) \in C_2 \times \T^n$, where now we defined
\begin{equation}
\theta= e^{\ri \vartheta}.
\label{2.178}
\end{equation}
Then the first and second properties of $\tilde \sigma$ follow since we have
\begin{equation}
\tilde \sigma(z(\lambda, \theta)) = \Phi_{(m(\theta), m(\theta))}
\left(\tilde \sigma_0(\lambda, \theta)\right),
\quad
\text{for all}\quad (\lambda, \theta) \in C_2 \times \T^n.
\label{2.179}
\end{equation}
We know that $\tilde \sigma_0(\lambda, \theta) \in P_0$ for all $(\lambda, \theta)\in
C_2 \times \T^n$, which implies the first property since $\tilde\sigma$ is continuous and
$P_0$ is a closed subset of $P$. The restriction of the pull-back \eqref{2.177}
to $\C^n_{\neq}$ is easily calculated using the parametrization $(\lambda, \theta) \mapsto
z(\lambda, \theta)$ and using that by Theorem \ref{thm:2.3}
$\tilde \sigma_0^\ast(\Omega)=\sum_{k=1}^nd\lambda_k \wedge d\vartheta_k$.
Indeed, this translates into \eqref{2.177} restricted to $\C^n_{\neq}$,
which implies the claimed
equality because $\tilde \sigma^\ast(\Omega)$ is smooth on $\C^n$.

Before continuing, we remark that the map $(\lambda, \theta) \mapsto z(\lambda, \theta)$
naturally extends to a continuous map on the closed domain $\overline{C}_2 \times \T^n$ and its
`partial inverse' $z\mapsto\lambda(z)$ extends to a smooth map $\C^n\to\overline{C}_2$.
We will use these extended maps without further notice in what follows.
(The extended map $(\lambda, \theta) \mapsto z(\lambda, \theta)$ is not
differentiable at the points for which $\lambda\in\partial C_2$.)

In order to show that $\tilde \sigma$ is injective, consider the equality
\begin{equation}
\tilde \sigma(z) = \tilde \sigma(\zeta)\quad
\text{for some}\quad z, \zeta \in \C^n.
\label{2.180}
\end{equation}
Looking at the `second component' of this equality according to \eqref{2.174} we see
that $\lambda(z) = \lambda(\zeta)$. Then the first component of the equality implies
$\tilde A(z) = \tilde A(\zeta)$. The special case $\tilde A_{a,1}(z) = \tilde A_{a,1}(\zeta)$
of this equality gives
\begin{equation}
\frac{\bar{z}_a \eta_a(\lambda(z))\eta_{n+1}(\lambda(z))}
{\lambda_a(z)-\lambda_1(z)-2\mu}
=\frac{\bar{\zeta}_a \eta_a(\lambda(\zeta))\eta_{n+1}(\lambda(\zeta))}
{\lambda_a(\zeta)-\lambda_1(\zeta)-2\mu}, \qquad 1\leq a \leq n.
\label{2.181}
\end{equation}
We know that the factors multiplying $\bar z_a$ and $\bar \zeta_a$ are
equal and non-zero (actually negative). Thus $z=\zeta$ follows,
establishing the claimed injectivity.

Next we prove that no two different element of $\tilde S$ are gauge equivalent
to each other, i.e. $\tilde S$ can intersect any orbit of $G_+ \times G_+$
at most in one point.
Suppose that
\begin{equation}
\Phi_{(g_L, g_R)}( \tilde \sigma(z)) = \tilde \sigma(\zeta)
\label{2.182}
\end{equation}
for some $(g_L, g_R) \in G_+ \times G_+$ and $z, \zeta \in \C^n$.
We conclude from the second component of this equality that $\lambda(z)= \lambda(\zeta)$.
Because $\lambda(z)\in\overline{C}_2$ holds, $\lambda(z)$ is
regular in the sense that it satisfies \eqref{2.87}. Thus we can also conclude
from the second component of the equality \eqref{2.182} that $g_R$ belongs to
the Abelian subgroup $Z$ of $G_+$ given in \eqref{2.37}.
Then we infer from the first component
\begin{equation}
g_L\tilde y(z)g_R^{-1}=\tilde y(\zeta)
\label{2.183}
\end{equation}
of the equality \eqref{2.182} that $g_L = g_R$.
We here used that $\tilde A(\zeta)$ can be represented in the form \eqref{2.41} with
strict inequalities in \eqref{2.39}, which holds since $S$ \eqref{2.51} is a global cross-section.
Now denote $ g_L = g_R = e^{\ri \xi} \in Z$ referring to \eqref{2.37}.
Then we have $e^{\ri \xi} \tilde A(z) e^{-\ri \xi} = \tilde A(\zeta)$,
and in particular
\begin{equation}
e^{\ri x_a} \tilde A_{a, a+1}(z)e^{-\ri x_{a+1}} = \tilde A_{a,a+1}(\zeta),
\quad
\forall a=1,\dots, n-1.
\label{2.184}
\end{equation}
By using \eqref{2.162} and \eqref{2.171}
\begin{equation}
\tilde{A}_{a,a+1}(z)
=-2\mu \eta_a(\lambda(z))\eta_{n+a+1}(\lambda(z)) \neq 0,
\label{2.185}
\end{equation}
and thus we obtain from $\lambda(z) = \lambda(\zeta)$ that $ e^{\ri \xi}$ must be equal to a multiple
of the identity element of $G_+$.
Hence we have established that $\tilde \sigma(z) = \tilde \sigma (\zeta)$ is implied by
\eqref{2.182}.

It remains to demonstrate that $\tilde S$ intersects every gauge orbit in $P_0$.
We have seen previously that $\fL^{-1}(C_2)$ is dense in $P_0$
and $\tilde S^0$ \eqref{2.81} is a cross-section for the gauge action in $\fL^{-1}(C_2)$.
These facts imply that for any element $x \in P_0$ there exists a series
$x(k) \in \fL^{-1}(C_2)$, $k\in \N$, such that
\begin{equation}
\lim_{k\to \infty}(x(k)) =x,
\label{2.186}
\end{equation}
and there also exist series $(g_L(k), g_R(k))\in G_+ \times G_+$ and $(\lambda(k),\theta(k))
\in C_2 \times \T^n$ such that
\begin{equation}
x(k) = \Phi_{(g_L(k), g_R(k))}\left(\tilde \sigma_0( \lambda(k), \theta(k))\right).
\label{2.187}
\end{equation}
Since $\fL\colon P_0\to\R^n$ is continuous, we have
\begin{equation}
\fL(x) = \lim_{k\to \infty} \fL(x(k)) = \lim_{k\to \infty} \lambda(k).
\label{2.188}
\end{equation}
This limit belongs to $\overline{C}_2$ and we denote it by $\lambda^\infty$.
The non-trivial case to consider is when $\lambda^\infty$ belongs to the boundary
$\partial C_2$.
Now, since $G_+ \times G_+ \times \T^n$ is compact, there exists a
convergent subseries
\begin{equation}
(g_L(k_i), g_R(k_i), \theta(k_i)),
\quad
i \in \N,
\label{2.189}
\end{equation}
of the series $(g_L(k), g_R(k), \theta(k))$.
We pick such a convergent subseries and denote its limit as
\begin{equation}
(g_L^\infty, g_R^\infty, \theta^\infty)=\lim_{i\to \infty} (g_L(k_i), g_R(k_i), \theta(k_i)).
\label{2.190}
\end{equation}
Then we define $z^\infty \in \C^n$ by
\begin{equation}
z^\infty = \lim_{i\to \infty} z(\lambda(k_i), \theta(k_i)) = z(\lambda^\infty, \theta^\infty).
\label{2.191}
\end{equation}
Since $z \mapsto \tilde \sigma(z)$ is continuous, we can write
\begin{equation}
\tilde \sigma(z^\infty) = \lim_{i\to \infty} \tilde \sigma( z(\lambda(k_i), \theta(k_i)))=
\lim_{i\to \infty} \Phi_{(m(\theta(k_i)), m(\theta(k_i)))}
\left(\tilde \sigma_0(\lambda(k_i), \theta(k_i))\right),
\label{2.192}
\end{equation}
where $m(\theta)$ is defined by \eqref{2.164}, with $\theta = e^{\ri \vartheta}$.
By combining these formulae, we finally obtain
\begin{equation}
\begin{split}
x&=\lim_{i\to \infty} \Phi_{(g_L(k_i),g_R(k_i))}
\left(\tilde \sigma_0(\lambda(k_i), \theta(k_i))\right)\\
&=\lim_{i\to \infty} \Phi_{(g_L(k_i) m(\theta(k_i))^{-1},g_R(k_i)m(\theta(k_i))^{-1} )}
\left(\tilde \sigma( z(\lambda(k_i), \theta(k_i)))\right)\\
&=\Phi_{( g_L^\infty m(\theta^\infty)^{-1}, g_R^\infty m(\theta^\infty)^{-1})}
(\tilde \sigma (z^\infty)).
\end{split}
\label{2.193}
\end{equation}
Therefore $\tilde S$ is a global cross-section in $P_0$.

The final statement of Theorem \ref{thm:2.16} about the global RSvD Lax matrix \eqref{2.175} follows since
$\tilde L$ is just the restriction of the `unreduced Lax matrix' $L$ of \eqref{2.66}
to the global cross-section $\tilde S$, which represents a model of the full reduced phase space
$P_\red$.
\end{proof}

\section{Applications}
\label{sec:2.4}

\subsection{On the equilibrium position of the Sutherland system}
\label{subsec:2.4.1}

Since the Sutherland Lax matrix is diagonalizable, that is
\begin{equation}
Y(q,p)\quad\text{and}\quad\ri\Lambda(\lambda)
=\ri\,\diag(\lambda,-\lambda)
\label{2.194}
\end{equation}
are similar matrices, we have the following for the Sutherland
Hamiltonians $H_1,\dots,H_n$
\begin{equation}
H_k(q,p)=\frac{1}{4k}\tr((-\ri Y(q,p))^{2k})
=\frac{1}{4k}\tr(\Lambda(\lambda)^{2k})
=\frac{1}{2k}\sum_{j=1}^n\lambda_j^{2k}=h_k(\lambda),
\quad k=1,\dots,n.
\label{2.195}
\end{equation}
In particular, for the main Hamiltonian $H_1(q,p)$ the above formula reads as
\begin{equation}
H(q,p)=H_1(q,p)=h_1(\lambda)
=\frac{1}{2}(\lambda_1^2+\dots+\lambda_n^2).
\label{2.196}
\end{equation}
It is obvious that $h_1$ has a global minimum on
$\overline{C_2}$ and
\begin{equation}
\min_{(q,p)\in C_1\times\R^n} H(q,p)
=\min_{\lambda\in\overline{C_2}}h_1(\lambda)
=h_1(\lambda^0),
\label{2.197}
\end{equation}
where $\lambda^0$ is the point in $\overline{C_2}$ with the smallest
(Euclidean) norm. Clearly, $\lambda^0$ is the boundary point of $C_2$
at which
\begin{equation}
\lambda^0_a-\lambda^0_{a+1}-2\mu=0,\quad a=1,\dots,n-1\quad
\lambda^0_n-\nu=0,
\label{2.198}
\end{equation}
hold, i.e.
\begin{equation}
\lambda_a^0=(n-a)2\mu+\nu,\quad a=1,\dots,n.
\label{2.199}
\end{equation}
In terms of the ``oscillator variables'' $z_1,\dots,z_n\in\C$ this means that
the equilibrium point $(q,p)=(q^e,0)$ of the Sutherland system corresponds to
$z_1=\dots=z_n=0$. In fact, each function $H_k$ ($k=1,\dots, n$)
possesses a global minimum at $z=0$.

Now, remember that the matrices
\begin{equation}
-(h(\lambda)\check{A}(\lambda,e^{\ri\vartheta})h(\lambda))^\dag\quad\text{and}\quad e^{\ri Q(q)}
\label{2.200}
\end{equation}
are similar. By using this and the fact that $h(\lambda)$ and $m$
commute for any $m\in Z_{G_+}(\cA)$, one concludes that
(with the special choice $m=m(e^{\ri\vartheta})$)
\begin{equation}
\sigma(-(h\hat{A}h)^\dag(z))=\sigma(e^{2\ri Q(q)})
=\{e^{2\ri q_1},\dots,e^{2\ri q_n},
e^{-2\ri q_1},\dots,e^{-2\ri q_n}\},
\label{2.201}
\end{equation}
where $\sigma(M)$ denotes the spectrum of $M$. In particular, for $z=0$ the above
spectral identity provides a useful method to determine the equilibrium coordinates
$q^e$. An interesting characterization of this equilibrium point in terms of the
$(q,p)$ variables can be found in \cite{CS02}.

\subsection{Maximal superintegrability of the dual system}
\label{subsec:2.4.2}

We have seen that the `unreduced RSvD Hamiltonians'
\begin{equation}
\tilde\cH(y,Y,V)=\frac{(-1)^k}{2k}\tr\big(y^{-1}CyC\big)^k,\quad
k=1,\dots,n
\label{2.203}
\end{equation}
take the following form in the `Sutherland gauge'
\begin{equation}
\tilde h_k(q)=\tilde\cH_k|_S=\tilde H_k\circ\cR=\frac{(-1)^k}{k}\sum_{j=1}^n\cos(2kq_j),\quad
k=1,\dots,n.
\label{2.204}
\end{equation}
Note that the RSvD-type Hamiltonians depend \emph{only} on the `action
variables' $q_1,\dots,q_n$.
Our objective is to show that the RSvD-type dual model is maximally
superintegrable. In particular, the construction presented in \citepalias{AFG12}
will be taken out. Let us consider the $n\times n$ matrix
\begin{equation}
X_{a,b}=\frac{\partial \tilde h_a(q)}{\partial q_b},\quad
a,b=1,\dots,n.
\label{2.205}
\end{equation}
As a first step we verify that $X(q)$ is invertible at every $q\in C_1$.

\begin{proposition}
\label{prop:2.17}
For any $q\in C_1=\{q\in\R^n\mid\pi/2>q_1>\dots>q_n>0\}$
we have $\det X(q)\neq 0$.
\end{proposition}

\begin{proof}
Let us first compute the matrix entries $X_{a,b}=X_{a,b}(q)$.
Simple differentiation shows that
\begin{equation}
X_{a,b}=(-1)^{a+1}2\sin(2aq_b),\quad
a,b=1,\dots,n.
\label{2.206}
\end{equation}
Since in each row contains the constants $(-1)^{a+1}2$ ($a=1,\dots,n$)
we have
\begin{equation}
\det(X_{a,b})=(-1)^{(1+1)+\dots+(n+1)}2^n\det(\sin 2aq_b)
=(-1)^{\tfrac{n(n+3)}{2}}2^n\det(\sin 2aq_b).
\label{2.207}
\end{equation}
By introducing the notation
\begin{equation}
\alpha_b=2q_b,\quad b=1,\dots,n,
\label{2.208}
\end{equation}
the above determinant takes a somewhat simpler form
\begin{equation}
\det(X_{a,b})=(-1)^{\tfrac{n(n+3)}{2}}2^n\det(\sin a\alpha_b).
\label{2.209}
\end{equation}
The trigonometric identity
\begin{equation}
2^r\cos^r\alpha\,\sin\alpha
=\sum_{s=0}^{\big\lceil\tfrac{r-1}{2}\big\rceil}
\bigg[{r\choose s}-{r\choose s-1}\bigg]\sin(r-2s+1)\alpha
\label{2.210}
\end{equation}
(which can be easily proven using de Moivre's formula) implies that
applying suitable column-operations on the determinant \eqref{2.209}, it can
be written as
\begin{equation}
\det(X_{a,b})=(-1)^{\tfrac{n(n+3)}{2}}2^n\det(2^{a-1}\cos^{a-1}\alpha_b\,\sin \alpha_b).
\label{2.211}
\end{equation}
Hence the problem can be reduced to the computation of a
Vandermonde-determinant
\begin{equation}
\begin{split}
\det(X_{a,b})
&=(-1)^{\tfrac{n(n+3)}{2}}2^{\tfrac{n(n+1)}{2}}
\prod_{b=1}^n\sin\alpha_b\det(\cos^{a-1}\alpha_b)\\
&=(-1)^{\tfrac{n(n+3)}{2}}2^{\tfrac{n(n+1)}{2}}
\prod_{b=1}^n\sin\alpha_b\prod_{1\leq b<c\leq n}(\cos\alpha_c-\cos\alpha_b).
\end{split}
\label{2.212}
\end{equation}
Now, by substituting back the $q$'s according to \eqref{2.208} we get
\begin{equation}
\det(X_{a,b}(q))
=(-1)^{\tfrac{n(n+3)}{2}}2^{\tfrac{n(n+1)}{2}}
\prod_{b=1}^n\sin2q_b\prod_{1\leq b<c\leq n}(\cos2q_c-\cos2q_b),
\label{2.213}
\end{equation}
which is an obviously non-vanishing function on $C_1$ due to monotonicity.
\end{proof}
Now, by referring to \citepalias{AFG12} and using the previous proposition
for any $\tilde H_k$ one can construct the functions mentioned in the
Discussion of \citepalias{FG14}
\begin{equation}
f_i(q,p)=\sum_{j=1}^np_j(X^{-1}(q))_{j,i},\quad i\in\{1,\dots,n\}\setminus\{k\}.
\label{2.214}
\end{equation}

\subsection{Equivalence of two sets of Hamiltonians}
\label{subsec:2.4.3}

In this subsection, we prove the equivalence of two complete sets of
Poisson commuting Hamiltonians of the (super)integrable rational $\BC_n$
Ruijsenaars-Schneider system. Specifically, the commuting
Hamiltonians constructed by van Diejen are shown to be linear combinations
of the Hamiltonians generated by the characteristic polynomial of the Lax
matrix obtained recently by Pusztai, and the explicit formula of this
invertible linear transformation is found.

\subsubsection{Hamiltonians due to van Diejen}

In \cite{vD94,vD95-2} the following complete set of Poisson commuting
Hamiltonians was given:
\begin{equation}
H_l^\vD(\lambda,\theta)
=\sum_{\substack{J\subset\{1,\dots,n\},\ |J|\leq l\\
\varepsilon_j=\pm 1,\ j\in J}}
\cos(\theta_{\varepsilon J})
V_{\varepsilon J;J^c}^{1/2}
V_{-\varepsilon J;J^c}^{1/2}
U_{J^c,l-|J|},\quad
l=1,\dots,n,
\label{2.215}
\end{equation}
with
\begin{equation}
\begin{split}
\theta_{\varepsilon J}&=\sum_{j\in J}\varepsilon_j\theta_j,\\
V_{\varepsilon J;K}&=\prod_{j\in J}w(\varepsilon_j\lambda_j)
\prod_{\substack{j,j'\in J\\j<j'}}
v^2(\varepsilon_j\lambda_j+\varepsilon_{j'}\lambda_{j'})
\prod_{\substack{j\in J\\k\in K}}v(\varepsilon_j\lambda_j+\lambda_k)
v(\varepsilon_j\lambda_j-\lambda_k),\\
U_{K,p}&=(-1)^p\sum_{\substack{I\subset K,\ |I|=p\\\varepsilon_i=\pm 1,\ i\in I}}
\bigg(\prod_{i\in I}w(\varepsilon_i\lambda_i)\prod_{\substack{i,i'\in I\\i<i'}}
v(\varepsilon_i\lambda_i+\varepsilon_{i'}\lambda_{i'})
v(-\varepsilon_i\lambda_i-\varepsilon_{i'}\lambda_{i'})\\
&\hspace{10em}\times\prod_{\substack{i\in I\\k\in K\setminus I}}
v(\varepsilon_i\lambda_i+\lambda_k)v(\varepsilon_i\lambda_i-\lambda_k)\bigg).
\end{split}
\label{2.216}
\end{equation}
We note that $J^c$ in \eqref{2.215} denotes the complementary set of $J$,
and the contribution to $H_l^\vD$ coming from $J=\emptyset$ is $U_{\emptyset^c,l}$.
The relatively simple form of $U_{K,p}$ above was found in \cite{vD95-2}.
Equation \eqref{2.215} makes sense for $l=0$, as well, giving $H_0^\vD\equiv 1$.
In the rational case the functions $v$ and $w$ take the following
form\footnote{The parameter $\beta$ appearing in \cite{vD94,vD95-2} can be introduced via replacing $\lambda,\theta,\mu,\nu,\kappa$ by $\beta^{-1}\lambda,\beta\theta,\beta\mu,\beta \nu,\beta\kappa$, respectively. In the convention of \cite{Pu13}, our $\mu$, $\theta$ and $q$ correspond to $2\mu$, $2\theta$ and $2q$.}
\begin{equation}
v(x)=\frac{x+\ri\mu}{x},\quad
w(x)=\bigg[\frac{x+\ri\nu}{x}\bigg]
\bigg[\frac{x+\ri\kappa}{x}\bigg].
\label{2.217}
\end{equation}
Up to irrelevant constants, $H_1^\vD$ reproduces the rational $\BC_n$ Ruijsenaars-Schneider
Hamiltonian
\begin{align}
H^\Pu(\lambda,\theta)=&\sum_{j=1}^n\cosh(\theta_j)
\bigg[1+\frac{\nu^2}{\lambda_j^2}\bigg]^{\tfrac{1}{2}}
\bigg[1+\frac{\kappa^2}{\lambda_j^2}\bigg]^{\tfrac{1}{2}}
\prod_{\substack{k=1\\(k\neq j)}}^n
\bigg[1+\frac{\mu^2}{(\lambda_j-\lambda_k)^2}\bigg]^{\tfrac{1}{2}}
\bigg[1+\frac{\mu^2}{(\lambda_j+\lambda_k)^2}\bigg]^{\tfrac{1}{2}}\nonumber\\
&+\frac{\nu\kappa}{\mu^2}\prod_{j=1}^n
\bigg[1+\frac{\mu^2}{\lambda_j^2}\bigg]
-\frac{\nu\kappa}{\mu^2}.
\label{2.218}
\end{align}
Indeed, one can check that $H_1^\vD=2(H^\Pu-n)$. Here $\mu,\nu,\kappa$ are real
parameters for which we impose the conditions $\mu\neq 0$, $\nu\neq 0$ and
$\nu\kappa\geq 0$. The generalised momenta $\theta=(\theta_1,\dots,\theta_n)$
run over $\R^n$ and the `particle positions' $\lambda=(\lambda_1,\dots,\lambda_n)$
vary in the Weyl chamber
\begin{equation}
\fc=\{x\in\R^n\mid x_1>\dots>x_n>0\}.
\label{2.219}
\end{equation}

Now, take any point $(\lambda,\theta)\in\fc\times\R^n$ in the phase space and set
$(q,p)=\cS^{-1}(\lambda,\theta)$ to be the corresponding action-angle
coordinates\footnote{Here $\cS\colon\fc\times\R^n\to\fc\times\R^n$ is the action-angle
map, that was constructed by Pusztai \cite{Pu12}.}.
Consider the $H^\Pu$-trajectory $(\lambda(t),\theta(t))$ with initial condition $(\lambda,\theta)$.
Notice that the Hamiltonian $H_l^\vD$ \eqref{2.215} is constant along the $H^\Pu$-trajectory.
By utilizing the asymptotics (proved in \cite{Pu13})
\begin{equation}
\lambda_k(t)\sim t\sinh(q_k)-p_k
\quad\text{and}\quad
\theta_k(t)\sim q_k,
\qquad
k=1,\dots,n,
\label{2.220}
\end{equation}
one can readily check that
\begin{equation}
(\cS^\ast H_l^\vD)(q,p)
=\lim_{t\to\infty}H_l^\vD(\lambda(t),\theta(t))
=\sum_{\substack{J\subset\{1,\dots,n\},\ |J|\leq l\\
\varepsilon_j=\pm 1,\ j\in J}}
(-2)^{l-|J|}{n-|J|\choose l-|J|}\cosh(q_{\varepsilon J}).
\label{2.221}
\end{equation}
From now on we let $\cH_l^\vD$ stand for the pull-back $\cS^\ast H_l^\vD$ just computed,
and stress that it depends only on the variables $q$.

\subsubsection{Hamiltonians obtained from the Lax matrix}

We recall some relevant objects of \cite{Pu12}.
First, prepare the $2n\times 2n$ Hermitian,
unitary matrix
\begin{equation}
C=\begin{bmatrix}\0_n&\1_n\\\1_n&\0_n\end{bmatrix}
\label{2.222}
\end{equation}
and the $2n\times 2n$ Hermitian matrix
\begin{equation}
h(\lambda)=\begin{bmatrix}
a(\diag(\lambda))&b(\diag(\lambda))\\
-b(\diag(\lambda))&a(\diag(\lambda))
\end{bmatrix}
\label{2.223}
\end{equation}
containing the smooth functions $a(x),b(x)$ given on the interval
$(0,\infty)\subset\R$ by
\begin{equation}
a(x)=\frac{\sqrt{x+\sqrt{x^2+\kappa^2}}}{\sqrt{2x}},\quad
b(x)=\ri\kappa\frac{1}{\sqrt{2x}}\frac{1}{\sqrt{x+\sqrt{x^2+\kappa^2}}}.
\label{2.224}
\end{equation}
Then introduce the vectors $z(\lambda)\in\C^n$, $F(\lambda,\theta)\in\C^{2n}$
by the formulae
\begin{equation}
z_l(\lambda)=-\bigg[1+\frac{\ri\nu}{\lambda_l}\bigg]
\prod_{\substack{m=1\\(m\neq l)}}^n\bigg[1+\frac{\ri\mu}{\lambda_l-\lambda_m}\bigg]
\bigg[1+\frac{\ri\mu}{\lambda_l+\lambda_m}\bigg],
\label{2.225}
\end{equation}
and
\begin{equation}
F_l(\lambda,\theta)=e^{-\tfrac{\theta_l}{2}}|z_l(\lambda)|^{\tfrac{1}{2}},\quad
F_{n+l}(\lambda,\theta)=\overline{z_l(\lambda)}
F_l(\lambda,\theta)^{-1},
\label{2.226}
\end{equation}
$l=1,\dots,n$. With these notations at hand, the $2n\times 2n$ matrix
\begin{equation}
A_{j,k}(\lambda,\theta)=\frac{\ri\mu F_j\overline{F_k}
+\ri(\mu-2\nu)C_{j,k}}{\ri\mu+\Lambda_j-\Lambda_k},\quad
j,k\in\{1,\dots,2n\},
\label{2.227}
\end{equation}
with $\Lambda=\diag(\lambda,-\lambda)$ is used to define the `RSvD Lax matrix' \cite{Pu12}:
\begin{equation}
L(\lambda,\theta)=h(\lambda)^{-1}A(\lambda,\theta)h(\lambda)^{-1}.
\label{2.228}
\end{equation}
The matrices $h$, $A$, and $L$ are invertible and satisfy the relations
\begin{equation}
ChC=h^{-1},\quad
CAC=A^{-1},\quad
CLC=L^{-1}.
\label{2.229}
\end{equation}
Their determinants are
\begin{equation}
\det(h)=\det(A)=\det(L)=1.
\label{2.230}
\end{equation}
Let $K_m$ denote the coefficients of the characteristic polynomial of
$L$ \eqref{2.228},
\begin{equation}
\det(L(\lambda,\theta)-x\1_{2n})
=K_0(\lambda,\theta)x^{2n}+K_1(\lambda,\theta)x^{2n-1}+\dots
+K_{2n-1}(\lambda,\theta)x+K_{2n}(\lambda,\theta).
\label{2.231}
\end{equation}
An immediate consequence of \eqref{2.229},\eqref{2.230} is that
\begin{equation}
K_{2n-m}\equiv K_m,\quad m=0,1,\dots,n,
\label{2.232}
\end{equation}
thus the functions $K_0\equiv 1,K_1,\dots,K_n$ fully determine the characteristic
polynomial \eqref{2.231}. The first non-constant member of this family is
proportional to $H^\Pu$ \eqref{2.218}, that is $K_1=-2H^\Pu$.
The asymptotic form of the Lax matrix $L$ \eqref{2.230} is the diagonal matrix
\begin{equation}
\diag(e^{-q},e^{q}),
\label{2.233}
\end{equation}
hence the action-angle transforms of the functions $K_m$ ($m=0,1,\dots,n$) can be easily
computed to be
\begin{equation}
(\cS^\ast K_m)(q,p)=(-1)^m
\sum_{a=0}^{\big\lfloor\tfrac{m}{2}\big\rfloor}
\sum_{\substack{J\subset\{1,\dots,n\},\ |J|=m-2a\\\varepsilon_j=\pm 1,\ j\in J}}
{n-|J|\choose a}\cosh(q_{\varepsilon J}).
\label{2.234}
\end{equation}
Of course, we used the asymptotics \eqref{2.220}, and that $K_m$ is constant along the
flow of $H^\Pu$. Now we introduce the shorthand $\cK_m=\cS^\ast K_m$, and observe that
it only depends on $q$.

It is worth emphasizing that finding a formula relating the families
$\{H_l^\vD\}_{l=0}^n$ and $\{K_m\}_{m=0}^n$ is equivalent to finding a
relation between their action-angle transforms $\{\cH_l^\vD\}_{l=0}^n$
and $\{\cK_m\}_{m=0}^n$.

\begin{proposition}
\label{prop:2.18}
There exists an invertible linear relation between the two families
$\{\cH_l^\vD\}_{l=0}^n$ and $\{\cK_m\}_{m=0}^n$.
\end{proposition}

\begin{proof}
Let us introduce the auxiliary functions
\begin{equation}
\cM_k(q)=\sum_{\substack{J\subset\{1,\dots,n\},\ |J|=k\\
\varepsilon_j=\pm 1,\ j\in J}}\cosh(q_{\varepsilon J}),
\quad q\in\R^n,\quad k=0,1,\dots,n.
\label{2.235}
\end{equation}
For any $l\in\{0,1,\dots,n\}$ the Hamiltonian $\cH_l^\vD$ \eqref{2.221}
is a linear combination of $\cM_0,\cM_1,\dots,\cM_l$,
\begin{equation}
\cH_l^\vD(q)=\sum_{k=0}^l(-2)^{l-k}{n-k\choose l-k}\cM_k(q).
\label{2.236}
\end{equation}
This shows that the matrix of the linear map transforming $\{\cM_k\}_{k=0}^n$
into $\{\cH_l^\vD\}_{l=0}^n$ is lower triangular with ones on the diagonal,
hence the above relation is invertible.
Similarly, any function $\cK_m$ \eqref{2.234}, $m\in\{0,1,\dots,n\}$ can be
expressed as a linear combination of $\cM_m,\cM_{m-2},\dots,\cM_3,\cM_1$ or
$\cM_m,\cM_{m-2},\dots,\cM_2,\cM_0$ depending on the parity of $m$, that is
\begin{equation}
\cK_m(q)=(-1)^m\sum_{a=0}^{\big\lfloor\tfrac{m}{2}\big\rfloor}
{n-(m-2a)\choose a}\cM_{m-2a}(q).
\label{2.237}
\end{equation}
Hence the linear transformation relating $\{\cM_k\}_{k=0}^n$ to $\{\cK_m\}_{m=0}^n$
has a lower triangular matrix with diagonal components $\pm 1$, implying that it is
invertible. This proves the existence of an invertible linear relation between the
two families $\{\cH_l^\vD\}_{l=0}^n$ and $\{\cK_m\}_{m=0}^n$.
\end{proof}

Now, we prove an explicit formula expressing $\cH_l^\vD$ as linear combination of
$\{\cK_m\}_{m=0}^l$.

\begin{proposition}
\label{prop:2.19}
For any fixed $n\in\N$, $l\in\{1,\dots,n\}$ and $q\in\R^n$ we have
\begin{equation}
(-1)^l\cH_l^\vD(q)=\cK_l(q)+
\sum_{m=0}^{l-1}\frac{2(n-m)}{2(n-m)-(l-m)}{(n-l)+(n-m)\choose l-m}\cK_m(q).
\label{2.238}
\end{equation}
\end{proposition}

\begin{proof}
Substitute $\cK_m$ \eqref{2.234} into the right-hand side of the expression
above to obtain
\begin{multline}
\sum_{k=0}^{l-1}\sum_{a=0}^{\big\lfloor\tfrac{k}{2}\big\rfloor}
\sum_{\substack{J\subset\{1,\dots,n\},\ |J|=k-2a\\\varepsilon_j=\pm 1,\ j\in J}}
(-1)^k\frac{2(n-k)}{2(n-k)-(l-k)}{(n-l)+(n-k)\choose l-k}\times\\
\times{n-(k-2a)\choose a}\cosh(q_{\varepsilon J})
+\sum_{a=0}^{\big\lfloor\tfrac{l}{2}\big\rfloor}
\sum_{\substack{J\subset\{1,\dots,n\},\ |J|=l-2a\\\varepsilon_j=\pm 1,\ j\in J}}
(-1)^l{n-(l-2a)\choose a}\cosh(q_{\varepsilon J}).
\label{2.239}
\end{multline}
Since $k=|J|+2a$ it is obvious that $(-1)^k=(-1)^{-|J|}$. Multiply \eqref{2.239} by
$(-1)^l$ and change the order of summations over $a$ and $J$ to get
\begin{multline}
\sum_{\substack{J\subset\{1,\dots,n\},\ |J|<l\\\varepsilon_j=\pm 1,\ j\in J}}
(-1)^{l-|J|}\sum_{a=0}^{\big\lfloor\tfrac{l-|J|}{2}\big\rfloor}
\frac{2[n-(|J|+2a)]}{2[n-(|J|+2a)]-[l-(|J|+2a)]}\times\\
\times{(n-l)+(n-(|J|+2a))\choose l-(|J|+2a)}
{n-|J|\choose a}\cosh(q_{\varepsilon J})
+\sum_{\substack{J\subset\{1,\dots,n\},\ |J|=l\\\varepsilon_j=\pm 1,\ j\in J}}
\cosh(q_{\varepsilon J}).
\label{2.240}
\end{multline}
Now, comparison of \eqref{2.236} with \eqref{2.240} leads to a relation equivalent to
\eqref{2.238},
\begin{multline}
\sum_{a=0}^{\big\lfloor\tfrac{l-|J|}{2}\big\rfloor}
\frac{2[n-(|J|+2a)]}{2[n-(|J|+2a)]-[l-(|J|+2a)]}\times\\
\times{2n-(l+|J|+2a)\choose l-(|J|+2a)}
{n-|J|\choose a}\bigg/{n-|J|\choose l-|J|}=2^{l-|J|}.
\label{2.241}
\end{multline}
For $n=l$ in \eqref{2.241} one obtains
\begin{equation}
\begin{cases}\displaystyle
2\sum_{a=0}^{\big\lfloor\tfrac{l-|J|}{2}\big\rfloor}{l-|J|\choose a}=2^{l-|J|},&
\text{if}\ l-|J|\ \text{is odd,}\\[12pt]
\displaystyle
2\sum_{a=0}^{\tfrac{l-|J|}{2}-1}{l-|J|\choose a}
+{l-|J|\choose \frac{l-|J|}{2}}=2^{l-|J|},&
\text{if}\ l-|J|\ \text{is even,}\\
\end{cases}
\label{2.242}
\end{equation}
which are well-known identities for the binomial coefficients.
This means that \eqref{2.238} holds for $l=n$ for all $n\in\N$,
which implies that if we consider $n+1$ variables it is sufficient
to check the cases $l<n+1$. With that in mind let us progress by induction
on $n$ and suppose that \eqref{2.238} is verified for all $1\leq l\leq n$
for some $n\in\N$.

First, notice that the Hamiltonians $\cH_l^\vD$ \eqref{2.221} satisfy
the following recursion
\begin{equation}
\cH_l^\vD(q_1,\dots,q_n,q_{n+1})
=\cH_l^\vD(q_1,\dots,q_n)+4\sinh^2(\frac{q_{n+1}}{2})\cH_{l-1}^\vD(q_1,\dots,q_n).
\label{2.243}
\end{equation}
This can be checked either directly or by utilizing that $\cH_l^\vD$ is
the $l$-th elementary symmetric function with variables $\sinh^2(\frac{q_i}{2})$
(see Appendix \ref{sec:B.2}).
Similarly, the functions $\cK_k$ \eqref{2.234} satisfy
\begin{equation}
\cK_k(q_1,\dots,q_n,q_{n+1})
=\cK_k(q_1,\dots,q_n)
-2\cosh(q_{n+1})\cK_{k-1}(q_1,\dots,q_n)
+\cK_{k-2}(q_1,\dots,q_n),
\label{2.244}
\end{equation}
with $\cK_{-1}\equiv 0$. Let us introduce some shorthand notation,
such as the $\R^{l+1}$ vectors
\begin{equation}
\vec\cH^\vD(n)=(\cH_0^\vD,-\cH_1^\vD,\dots,(-1)^l\cH^\vD_l)^\top
\quad\text{and}\quad
\vec\cK(n)=(\cK_0,\cK_1,\dots,\cK_l)^\top
\label{2.245}
\end{equation}
and the $\R^{(l+1)\times(l+1)}$ matrices
\begin{equation}
\cA(n)_{j+1,k+1}=\begin{cases}\displaystyle
\frac{2(n-k)}{2(n-k)-(j-k)}{(n-j)+(n-k)\choose j-k},
&\text{if}\ j\geq k,\\[1em]
0,&\text{if}\ j<k,
\end{cases}
\label{2.246}
\end{equation}
where $j,k\in\{0,\dots,l\}$ and
\begin{equation}
\cH^\vD(n,n+1)=\1_{l+1}-4\sinh^2(\frac{q_{n+1}}{2})\cI_{-1},
\quad
\cK(n,n+1)=\1_{l+1}-2\cosh(q_{n+1})\cI_{-1}+\cI_{-2}
\label{2.247}
\end{equation}
with $(\cI_{-m})_{j+1,k+1}=\delta_{j,k+m}$, $m>0$.
The relations \eqref{2.243} and \eqref{2.244} can be written in the concise form
\begin{equation}
\vec\cH^\vD(n+1)=\cH^\vD(n,n+1)\vec\cH^\vD(n),\quad
\vec\cK(n+1)=\cK(n,n+1)\vec\cK(n)
\label{2.248}
\end{equation}
and our assumption is condensed into
\begin{equation}
\vec\cH^\vD(n)=\cA(n)\vec\cK(n).
\label{2.249}
\end{equation}
Using this notation it is clear that the desired induction step is equivalent to
the matrix equation
\begin{equation}
\cH^\vD(n,n+1)\cA(n)=\cA(n+1)\cK(n,n+1).
\label{2.250}
\end{equation}
Spelling this out at some arbitrary $(j,k)$-th entry gives us
\begin{multline}
\frac{A+B}{A}{A\choose B}
-4\sinh^2\bigg(\frac{\alpha}{2}\bigg)\frac{A+B}{A+1}{A+1\choose B-1}
=\\[.5em]
=\frac{A+B+2}{A+2}{A+2\choose B}
-2\cosh(\alpha)\frac{A+B}{A+1}{A+1\choose B-1}
+\frac{A+B-2}{A}{A\choose B-2},
\label{2.251}
\end{multline}
where
\begin{equation}
A=2n-j-k,\quad B=j-k,\quad\alpha=q_{n+1}.
\label{2.252}
\end{equation}
A simple direct calculation shows that \eqref{2.251} indeed holds implying that
\eqref{2.238} is also true for $n+1$ for any $l\leq n$. The case $l=n+1$ is given
by the argument preceding induction. This completes the proof.
\end{proof}

\begin{remark}
\label{rem:2.20}
We showed in Proposition \ref{prop:2.18} that the relation \eqref{2.238} is invertible.
Without spending space on the proof, we note that the inverse relation can be written explicitly as
\begin{equation}
(-1)^m\cK_m(q)=\sum_{l=0}^{m}{2(n-l)\choose m-l}\cH_l^\vD(q).
\label{2.253}
\end{equation}
\end{remark}

\section{Discussion}
\label{sec:2.5}

In this chapter, we characterized a symplectic reduction of the phase space
$(P,\Omega)$ \eqref{2.26} by exhibiting two models of the reduced phase space
$P_\red$ \eqref{2.29}. These are provided by the global cross-sections $S$ and
$\tilde S$ described in Theorem \ref{thm:2.1} and in Theorem \ref{thm:2.16}.
The two cross-sections naturally give rise to symplectomorphisms
\begin{equation}
(M, \omega) \simeq (P_\red, \Omega_\red) \simeq (\tilde M, \tilde \omega),
\label{2.254}
\end{equation}
where $M=T^\ast C_1$ \eqref{2.2} with the canonical symplectic form
$\omega= \sum_{k=1}^ndq_k \wedge dp_k$ and $\tilde{M}=\C^n$ with
$\tilde{\omega}=\ri\sum_{k=1}^ndz_k\wedge d\bar{z}_k$. The Abelian
Poisson algebras $\fQ^1$ and $\fQ^2$ on $P$ \eqref{2.32} descend
to reduced Abelian Poisson algebras $\fQ^1_\red$ and $\fQ^2_\red$
on $P_\red$. The construction guarantees that any element of the
reduced Abelian Poisson algebras possesses complete Hamiltonian flow.
These flows can be analysed by means of the standard projection
algorithm as well as by utilization of the symplectomorphism \eqref{2.254}.

To further discuss the interpretation of our results, consider the
gauge invariant functions
\begin{equation}
\cH_k(y,Y,V) = \frac{1}{4k} (-\ri Y)^{2k}
\quad\text{and}\quad
{\tilde \cH}_k(y,Y,V) = \frac{(-1)^k}{2k} \tr(y^{-1} C y C)^k,
\quad
k=1,\dots, n.
\label{2.255}
\end{equation}
The restrictions of the functions $\cH_k$ to the global
cross-sections $S$ and $\tilde{S}$ take the form
\begin{equation}
\cH_k|_S=\frac{1}{4k}(-\ri Y(q,p))^{2k}=H_k(q,p)
\quad\text{and}\quad
\cH_k|_{\tilde S}=\frac{1}{2k} \sum_{j=1}^n \lambda_j(z)^{2k}.
\label{2.256}
\end{equation}
According to \eqref{2.64}, the $H_k$ yield the commuting Hamiltonians of the
Sutherland system, while the $\lambda_j$ as functions on $\tilde S\simeq \C^n$
are given by \eqref{2.160}. Since any smooth function on a global cross-section
encodes a smooth function on $P_\red$, we conclude that the Sutherland
Hamiltonians $H_k$ and the `eigenvalue-functions' $\lambda_j$ define two
alternative sets of generators for $\fQ^2_\red$.

The restrictions of the functions $\tilde \cH_k$ read
\begin{equation}
{\tilde \cH}_k \vert_{S} = \frac{(-1)^k}{k}\sum_{j=1}^n \cos( 2k q_j)
\quad\text{and}\quad
{\tilde \cH}_k\vert_{\tilde S}=\frac{1}{2k}\tr(\tilde L(z)^k)
\label{2.257}
\end{equation}
with $\tilde L(z)$ is defined in \eqref{2.175}. On the semi-global cross-section
$\tilde{S}^0$ of Theorem \ref{thm:2.3}, which parametrizes the dense open submanifold
$\fL^{-1}_\red(C_2) \subset P_\red$, we have
\begin{equation}
{\tilde \cH}_1\vert_{\tilde S^0}= \tilde H^0,
\label{2.258}
\end{equation}
where $\tilde H^0$ is the RSvD Hamiltonian displayed in \eqref{2.4}. We see from
\eqref{2.257} that the functions $q_j\in C^\infty(S)$ and the commuting Hamiltonians
${\tilde \cH}_k|_{\tilde S}$ engender two alternative generating sets for $\fQ^1_\red$.
On account of the relations
\begin{equation}
\tilde M^0\simeq\tilde S^0\simeq C_2\times\T^n\simeq\C^n_{\neq}\subset\C^n\simeq\tilde S\simeq\tilde M,
\label{2.259}
\end{equation}
${\tilde \cH}_1\vert_{ \tilde S}$ yields a globally smooth extension of the
many-body Hamiltonian $\tilde H^0$.

It is immediate from our results that both $\fQ^1_\red$ and $\fQ^2_\red$ define Liouville
integrable systems on $P_\red$, since both have $n$ functionally independent generators.
The interpretations of these Abelian Poisson algebras that stem from the models $S$ and
$\tilde{S}$ underlie the action-angle duality between the Sutherland and RSvD systems as
follows. First, the generators $q_k$ of $\fQ^1_\red$ can be viewed alternatively as
particle positions for the Sutherland system or as action variables for the RSvD system.
Their canonical conjugates $p_k$ are of non-compact type. Second, the generators
$\lambda_k$ of $\fQ^2_\red$ can be viewed alternatively as action variables for the
Sutherland systems or as globally well-defined `particle positions' for the completed RSvD
system. In conclusion, the symplectomorphism $\cR\colon M\to\tilde{M}$ naturally induced by
\eqref{2.254} satisfies all properties required by the notion of action-angle duality
outlined in the Introduction.

We end this chapter by pointing out some further consequences. First of all, we note that
the dimension of the Liouville tori of the Sutherland system drops on the locus where
the action variables encoded by $\lambda$ belong to the boundary of the polyhedron
$\overline{C}_2$. This is a consequence of the next statement, which can be proved
by direct calculation.

\begin{proposition}
\label{prop:2.21}
Consider the Sutherland Hamiltonians\footnote{Here $H_k(z)$ denotes the reduction
of the Hamiltonian $\cH_k$ expressed in terms of the model $\tilde{M}$, cf.~\eqref{2.256}.}
$H_k(z) = \frac{1}{2k} \sum_{j=1}^n \lambda_j(z)^{2k}$
and for any $z\in\C^n$ define $\cD(z)= \# \{ z_k\neq 0\mid k=1,\dots, n\}$.
Then one has the equality
\begin{equation}
\operatorname{dim}\left( \operatorname{span} \{d\lambda_k(z) \mid k=1,\dots,n \}\right)
=\operatorname{dim}\left( \operatorname{span}
 \{dH_k(z) \mid k=1,\dots,n \}\right)= \cD(z).
\label{2.260}
\end{equation}
\end{proposition}

Being in control of the action-angle variables for our dual pair of integrable
systems, the following result is readily obtained.

\begin{proposition}
\label{prop:2.22}
Any `Sutherland Hamiltonian' $H_k\in C^\infty(M)$ ($k=1,\dots,n$) given by \eqref{2.63}
defines a non-degenerate Liouville integrable system, i.e. the commutant of $H_k$ in the
Poisson algebra $C^\infty(M)$ is the Abelian algebra generated by the action variables
$\lambda_1,\dots,\lambda_n$. Any `RSvD Hamiltonian' $\tilde H_k \in C^\infty(\tilde M)$,
$k=1,\dots,n$, which by definition coincides with $\tilde \cH_k\vert_{\tilde S}$ in
\eqref{2.254} upon the identification $\tilde{M}\simeq\tilde{S}$, is maximally degenerate
(`superintegrable') since its commutant in the Poisson algebra $C^\infty(\tilde M)$ is
generated by $(2n-1)$ elements.
\end{proposition}

\begin{proof}
The subsequent argument relies on the `action-angle symplectomorphisms'
between $(M,\omega)$ and $(\tilde{M},\tilde{\omega})$ corresponding to
\eqref{2.254}.

Let us first restrict the Sutherland Hamiltonian $H_k$ to the submanifold parametrized
by the action-angle variables varying in $C_2 \times \T^n$. For generic $\lambda$, we
see from \eqref{2.256} that the flow of $H_k$ is dense on the torus $\T^n$. Therefore
any smooth function $f$ that Poisson commutes with $H_k$ must be constant on the
non-degenerate Liouville tori of the Sutherland system. By smoothness, this implies
that $f$ Poisson commutes with all the action variables $\lambda_j$ on the full phase
space. Consequently, it can be expressed as a function of those variables.

Maximal superintegrability for the dual model was proved in Subsection \ref{subsec:2.4.2}.
\end{proof}

In the end, we remark that the matrix functions $-\ri Y(q,p)$ and $\tilde L(z)$,
which naturally arose from the Hamiltonian reduction, serve as Lax matrices for
the pertinent dual pair of integrable systems. We also notice that the $z_j$ can
be viewed as `oscillator variables' for the Sutherland system since the actions
$\lambda_k$ are linear combinations in $\vert z_j \vert^2$ ($j=1,\dots, n$) and
the form $\tilde\omega$ coincides with the symplectic form of $n$ independent harmonic
oscillators. It could be worthwhile to inspect the quantization of the Sutherland
system based on these oscillator variables and to compare the result to the standard
quantization \cite{He87,HO87,Op88}. We plan to return to this issue in the future.

We demonstrated that the commuting Hamiltonians of the rational $\BC_n$
Ruijsenaars-Schneider system, constructed originally by van Diejen, are linear
combinations of the coefficients of the characteristic polynomial of the
Lax matrix found recently by Pusztai, and vice versa. The derivation utilized
the action-angle map and the scattering theory results of \cite{Pu12,Pu13}.
Our Proposition \ref{prop:2.18} gives rise to a determinantal representation of the
somewhat complicated expressions $H_l^\vD$ in \eqref{2.215}. It could be of some interest
to provide a purely algebraic proof of the resulting formula of the characteristic
polynomial of the Lax matrix.

The configuration space $\fc$ \eqref{2.219} is an open Weyl chamber associated with
the Weyl group $W(\BC_n)$, and after extending this domain all Hamiltonians that we
dealt with enjoy $W(\BC_n)$ invariance. In particular, the sets $\{\cH_l^\vD\}_{l=0}^n$,
$\{ \cK_l\}_{l=0}^n$ and $\{\cM_l\}_{l=0}^n$ represent different free generating sets
of the invariant polynomials in the functions $e^{\pm q_k}$ ($k=1,\dots,n$) of
the action variables $q_k$ acted upon by the sign changes and permutations that form
$W(\BC_n)$. In order to verify this, it is useful to point out that the $W(\BC_n)$
invariant polynomials in the variables $e^{\pm q_k}$ are the same as the ordinary
symmetric polynomials in the variables $\cosh(q_k)$. The statement that
$\{\cH_l^\vD\}_{l=0}^n$ is a free generating set for these polynomials then follows,
for example, from the identity presented in Appendix \ref{sec:B.2}.

Analogous statements hold obviously also for our trigonometric version.

An interesting open problem for future work is to extend the considerations reported here
to the hyperbolic RSvD system having five independent coupling parameters.

\chapter{A Poisson-Lie deformation}
\label{chap:3}

In this chapter, which based on our results reported in \citepalias{FG15,FG16}, a deformation of the classical trigonometric $\BC_n$ Sutherland
system is derived via Hamiltonian reduction of the Heisenberg double of $\SU(2n)$.
We apply a natural Poisson-Lie analogue of the Kazhdan-Kostant-Sternberg type
reduction of the free particle on $\SU(2n)$ that led to the $\BC_n$ Sutherland
system in the previous chapter. We prove that this yields a Liouville integrable
Hamiltonian system and construct a globally valid model of the smooth reduced
phase space wherein the commuting flows are complete. We point out that the
reduced system, which contains 3 independent coupling constants besides the
deformation parameter, can be recovered (at least on a dense submanifold) as a
singular limit of the standard 5-coupling deformation due to van Diejen.
Our findings complement and further develop those obtained recently by Marshall \cite{Ma15}
on the hyperbolic case by reduction of the Heisenberg double of $\SU(n,n)$.

Here, we shall deal with a reduction of the Heisenberg double of $\SU(2n)$
and derive a Liouville integrable Hamiltonian system related to Marshall's
one in a way similar to the connection between the original trigonometric
Sutherland system and its hyperbolic variant. Although this is essentially
analytic continuation, it should be noted that the resulting systems are
qualitatively different in their dynamical characteristics and global features. In
addition, what we hope makes our work worthwhile is that our treatment is different
from the one in \cite{Ma15} in several respects and we go considerably further regarding
the global characterization of the reduced phase space and the completeness
of the relevant Hamiltonian flows.

The main Hamiltonian of the system that we obtain can be written as
\begin{equation}
H=\frac{e^{a+b}+e^{a-b}}{2}\sum_{j=1}^ne^{-2\hat{p}_j}
-\sum_{j=1}^n\cos(\hat{q}_j)w(\hat{p}_j;a)^{\tfrac{1}{2}}
\prod_{\substack{k=1\\(k\neq j)}}^n
\bigg[1-\frac{\sinh^2(x)}{\sinh^2(\hat{p}_j-\hat{p}_k)}\bigg]^{\tfrac{1}{2}}
\label{3.1}
\end{equation}
with the (external) Morse potential
\begin{equation}
w(\hat p_j;a)=1-(1+e^{2a})e^{-2\hat p_j}+e^{2a}e^{-4\hat p_j},
\label{3.2}
\end{equation}
and the real coupling constants $a,b,x$ satisfying
\begin{equation}
a>0,\quad b\neq 0,\quad\text{and}\quad x\neq 0.
\label{3.3}
\end{equation}
The components of $\hat{q}$ parametrize the torus $\T_n$ by $e^{\ri\hat{q}}$,
and $\hat{p}$ belongs to the domain
\begin{equation}
\cC_x:=\{\hat{p}\in\R^n\mid 0>\hat{p}_1,\ \hat{p}_k-\hat{p}_{k+1}>|x|/2\ (k=1,\dots,n-1)\}.
\label{3.4}
\end{equation}

The dynamics is then defined via the symplectic form
\begin{equation}
\hat\omega=\sum_{j=1}^nd\hat{q}_j\wedge d\hat{p}_j.
\label{3.5}
\end{equation}
It will be shown that this system results by restricting a reduced free system
on a dense open submanifold of the pertinent reduced phase space. The Hamiltonian
flow is complete on the full reduced phase space, but it can leave the submanifold
parametrized by $\cC_x\times\T_n$. By glancing at the form of the Hamiltonian, one
may say that it represents an RS type system coupled to external fields. Since
differences of the `position variables' $\hat{p}_k$ appear, one feels that this
Hamiltonian somehow corresponds to an A-type root system.

To better understand the nature of this model, let us now introduce new Darboux
variables $q_k$, $p_k$ following essentially \cite{Ma15} as
\begin{equation}
\exp(\hat{p}_k)=\sin(q_k)\quad\text{and}\quad\hat{q}_k=p_k\tan(q_k).
\label{3.6}
\end{equation}
In terms of these variables $H(\hat{p},\hat{q};x,a,b)=\cH_1(q,p;x,a,b)$ with the
`new Hamiltonian'
\begin{multline}
\cH_1=\frac{e^{a+b}+e^{a-b}}{2}\sum_{j=1}^n\frac{1}{\sin^2(q_j)}\\
-\sum_{j=1}^n\cos(p_j\tan(q_j))\bigg[1-\frac{1+e^{2a}}{\sin^2(q_j)}
+\frac{4e^{2a}}{4\sin^2(q_j)-\sin^2(2q_j)}\bigg]^{\tfrac{1}{2}}\\
\times\prod_{\substack{k=1\\(k\neq j)}}^n
\bigg[1-\frac{2\sinh^2\big(\frac{x}{2}\big)\sin^2(q_j)\sin^2(q_k)}
{\sin^2(q_j-q_k)\sin^2(q_j+q_k)}
\bigg]^{\tfrac{1}{2}}.
\label{3.7}
\end{multline}
Remarkably, only such combinations of the new `position variables' $q_k$ appear that
are naturally associated with the $\BC_n$ root system and the Hamiltonian $\cH_1$
enjoys symmetry under the corresponding Weyl group. Thus now one may wish to attach
the Hamiltonian $\cH_1$ to the $\BC_n$ root system. Indeed, this interpretation is
preferable for the following reason. Introduce the scale parameter (corresponding to
the inverse of the velocity of light in the original Ruijsenaars-Schneider system) $\beta>0$ and make
the substitutions
\begin{equation}
a\to\beta a,\quad
b\to\beta b,\quad
x\to\beta x,\quad
p\to\beta p,\quad
\hat\omega\to\beta\hat\omega.
\label{3.8}
\end{equation}
Then consider the deformed Hamiltonian $\cH_\beta$ defined by
\begin{equation}
\cH_\beta(q,p;x,a,b)=\cH_1(q,\beta p;\beta x,\beta a,\beta b).
\label{3.9}
\end{equation}
The point is that one can then verify the following relation:
\begin{equation}
\lim_{\beta\to 0}\frac{\cH_\beta(q,p;x,a,b)-n}{\beta^2}
=H_{\BC_n}^{\text{Suth}}(q,p;\gamma,\gamma_1,\gamma_2),
\label{3.10}
\end{equation}
where $H_{\BC_n}^{\text{Suth}}$ stands for the standard trigonometric $\BC_n$ Sutherland Hamiltonian \eqref{2.1} with coupling constants
\begin{equation}
\gamma=\frac{x^2}{4},\quad
\gamma_1=(b^2-a^2)/2,\quad
\gamma_2=2a^2.
\label{3.11}
\end{equation}
Note that the domain of the variables $\hat{q},\hat{p}$, and correspondingly that of
$q,p$ also depends on $\beta$, and in the $\beta\to 0$ limit it is easily seen that
we recover the usual $\BC_n$ configuration space \eqref{2.2}. In conclusion, we see that
$H$ \eqref{3.1} in its equivalent form $\cH_\beta$ \eqref{3.9} is a 1-parameter
deformation of the trigonometric $\BC_n$ Sutherland Hamiltonian. We remark in passing
that the conditions \eqref{3.3} imply that $\gamma_2>0$ and $4\gamma_1+\gamma_2>0$,
which guarantee that the flows of $H_{\BC_n}^{\text{Suth}}$ are complete on the
domain \eqref{2.2}.

Marshall \cite{Ma15} obtained similar results for an analogous deformation of the
hyperbolic $\BC_n$ Sutherland Hamiltonian. His deformed Hamiltonian differs from
\eqref{3.1} above in some important signs and in the relevant domain of the `position
variables' $\hat p$. Although in our impression the completeness of the reduced
Hamiltonian flows was not treated in a satisfactory way in \cite{Ma15}, the
completeness proof that we shall present can be adapted to Marshall's case as well,
as demonstrated in Section \ref{sec:3.4}.

It is natural to ask how the system studied in this chapter (and its cousin in
\cite{Ma15}) is related to van Diejen's \cite{vD94} $5$-coupling trigonometric $\BC_n$
system? It was shown already in \cite{vD94} that the $5$-coupling trigonometric system
is a deformation of the $\BC_n$ Sutherland system, and later \cite{vD95-2} several other
integrable systems were also derived as its (`Inozemtsev type' \cite{In89}) limits\footnote{It should be mentioned that the so-called `Inozemtsev limit' was discovered by Ruijsenaars \cite{Ru90-2}.}.
Motivated by this, we can show that the Hamiltonian \eqref{3.1} is a
singular limit of van Diejen's general Hamiltonian.
Incidentally, a Hamiltonian of Schneider \cite{Sc87} can be viewed as a subsequent singular
limit of the Hamiltonian \eqref{3.1}. Schneider's system was mentioned in \cite{Ma15},
too, but the relation to van Diejen's system was not described.

The original idea behind the present work and \cite{Ma15} was that a natural Poisson-Lie
analogue of the Hamiltonian reduction treatment \cite{FP07} of the $\BC_n$
Sutherland system should lead to a deformation of this system. It was expected that a
special case of van Diejen's standard $5$-coupling deformation will arise. The
expectation has now been confirmed, although it came as a surprise that a singular
limit is involved in the connection.

The outline of the chapter is as follows. We start in Section \ref{sec:3.1} by defining
the reduction of interest. In Section \ref{sec:3.2} we observe that several technical
results of \cite{FK11} can be applied for analyzing the reduction at hand, and solve
the momentum map constraints by taking advantage of this observation. The heart of
the chapter is Section \ref{sec:3.3}, where we characterize the reduced system.
In Subsection \ref{subsec:3.3.1} we prove that the reduced phase space is smooth, as
formulated in Theorem \ref{thm:3.9}. Then in Subsection \ref{subsec:3.3.2} we focus on a dense
open submanifold on which the Hamiltonian \eqref{3.1} lives. The demonstration of the
Liouville integrability of the reduced free flows is given in Subsection
\ref{subsec:3.3.3}. In particular, we prove the integrability of the completion of the
system \eqref{3.1} carried by the full reduced phase space. Our main result is Theorem
\ref{thm:3.14} (proved in Subsection \ref{subsec:3.3.4}), which establishes a globally valid model
of the reduced phase space. We stress that the global structure of the phase space on
which the flow of \eqref{3.1} is complete was not considered previously at all, and
will be clarified as a result of our group theoretic interpretation. Section
\ref{sec:3.5} contains our conclusions, further comments on the related paper by
Marshall \cite{Ma15} and a discussion of open problems. This chapter is complemented by
four appendices. Appendix \ref{sec:C.1} deals with the connection to van Diejen's system;
the other $3$ appendices contain important details relegated from the main text.

\section{Definition of the Hamiltonian reduction}
\label{sec:3.1}

We below introduce the `free' Hamiltonians and define their reduction.
We restrict the presentation of this background material to a minimum
necessary for understanding our work. The conventions follow \cite{FK11},
which also contains more details. As a general reference, we recommend \cite{CP94}.

\subsection{The unreduced free Hamiltonians}
\label{subsec:3.1.1}

We fix a natural number\footnote{The $n=1$ case would need special treatment
and is excluded in order to simplify the presentation.} $n\geq 2$ and
consider the Lie group $\SU(2n)$ equipped with its standard quadratic Poisson bracket
defined by the compact form of the Drinfeld-Jimbo classical $r$-matrix,
\begin{equation}
r_{\mathrm{DJ}}=\ri\sum_{1\leq\alpha<\beta\leq 2n}E_{\alpha\beta}\wedge E_{\beta\alpha},
\label{3.12}
\end{equation}
where $E_{\alpha\beta}$ is the elementary matrix of size $2n$ having a single
non-zero entry $1$ at the $\alpha\beta$ position. In particular, the Poisson
brackets of the matrix elements of $g\in\SU(2n)$ obey Sklyanin's formula
\begin{equation}
\{g\stackrel{\otimes}{,}g\}_{\SU(2n)}=[g\otimes g, r_{\mathrm{DJ}}].
\label{3.13}
\end{equation}
Thus $\SU(2n)$ becomes a Poisson-Lie group, i.e. the multiplication
$\SU(2n)\times\SU(2n)\to\SU(2n)$ is a Poisson map. The cotangent bundle $T^\ast\SU(2n)$
possesses a natural Poisson-Lie analogue, the so-called Heisenberg double \cite{Se85},
which is provided by the real Lie group $\SL(2n,\C)$ endowed with a certain symplectic
form \cite{AM94}, $\omega$. To describe $\omega$, we use the Iwasawa decomposition and
factorize every element $K\in\SL(2n,\C)$ in two alternative ways
\begin{equation}
K=g_Lb_R^{-1}=b_Lg_R^{-1}
\label{3.14}
\end{equation}
with uniquely determined
\begin{equation}
g_L,g_R\in\SU(2n),\quad
b_L,b_R\in\SB(2n).
\label{3.15}
\end{equation}
Here $\SB(2n)$ stands for the subgroup of $\SL(2n,\C)$ consisting of upper triangular
matrices with positive diagonal entries. The symplectic form $\omega$ reads
\begin{equation}
\omega=\frac{1}{2}\Im\tr(db_Lb_L^{-1}\wedge dg_Lg_L^{-1})+
\frac{1}{2}\Im\tr(db_Rb_R^{-1}\wedge dg_Rg_R^{-1}).
\label{3.16}
\end{equation}
Before specifying free Hamiltonians on the phase space $\SL(2n,\C)$, note that any
smooth function $h$ on $\SB(2n)$ corresponds to a function $\tilde h$ on the space
of positive definite Hermitian matrices of determinant $1$ by the relation
\begin{equation}
\tilde h(bb^\dagger)=h(b),\quad\forall b\in\SB(2n).
\label{3.17}
\end{equation}
Then introduce the invariant functions
\begin{equation}
C^\infty(\SB(2n))^{\SU(2n)}\equiv\{h\in C^\infty(\SB(2n))\mid\tilde h(bb^\dagger)
=\tilde h(gbb^\dagger g^{-1}),\ \forall g\in\SU(2n),b\in\SB(2n)\}.
\label{3.18}
\end{equation}
These in turn give rise to the following ring of functions on $\SL(2n,\C)$:
\begin{equation}
\fH\equiv\{\cH\in C^\infty(\SL(2n,\C))\mid\cH(g_Lb_R^{-1})=h(b_R),\
h\in C^{\infty}(\SB(2n))^{\SU(2n)}\},
\label{3.19}
\end{equation}
where we utilized the decomposition \eqref{3.14}. An important point is that $\fH$
forms an Abelian algebra with respect to the Poisson bracket associated with $\omega$
\eqref{3.16}.

The flows of the `free' Hamiltonians contained in $\fH$ can be obtained effortlessly.
To describe the result, define the derivative $d^Rf\in C^\infty(\SB(2n),\su(2n))$ of
any real function $f\in C^\infty(\SB(2n))$ by requiring
\begin{equation}
\left.\frac{d}{ds}\right\vert_{s=0}f(be^{sX})=\Im\tr\big(Xd^Rf(b)\big),
\quad\forall b\in\SB(2n),\ \forall X\in\mathrm{Lie}(\SB(2n)).
\label{3.20}
\end{equation}
The Hamiltonian flow generated by $\cH\in\fH$ through the initial value
$K(0)=g_L(0)b_R(0)^{-1}$ is in fact given by
\begin{equation}
K(t)=g_L(0)\exp\big[-td^Rh(b_R(0))\big]b_R^{-1}(0),
\label{3.21}
\end{equation}
where $\cH$ and $h$ are related according to \eqref{3.19}. This means that $g_L(t)$
follows the orbit of a one-parameter subgroup, while $b_R(t)$ remains constant.
Actually, $g_R(t)$ also varies along a similar orbit, and $b_L(t)$ is constant.

The constants of motion $b_L$ and $b_R$ generate a Poisson-Lie symmetry, which
allows one to define Marsden-Weinstein type \cite{MW74} reductions.

\subsection{Generalized Marsden-Weinstein reduction}
\label{subsec:3.1.2}

The free Hamiltonians in $\fH$ are invariant with respect to the action of
$\SU(2n)\times\SU(2n)$ on $\SL(2n,\C)$ given by left- and right-multiplications.
This is a Poisson-Lie symmetry, which means that the corresponding action map
\begin{equation}
\SU(2n)\times\SU(2n)\times\SL(2n,\C)\to\SL(2n,C),
\label{3.22}
\end{equation}
operating as
\begin{equation}
(\eta_L,\eta_R,K)\mapsto\eta_L K\eta_R^{-1},
\label{3.23}
\end{equation}
is a Poisson map. In \eqref{3.22} the product Poisson structure is taken using the
Sklyanin bracket on $\SU(2n)$ and the Poisson structure on $\SL(2n,\C)$ associated with
the symplectic form $\omega$ \eqref{3.16}. This Poisson-Lie symmetry admits a momentum
map in the sense of Lu \cite{Lu91}, given explicitly by
\begin{equation}
\Phi\colon\SL(2n,\C)\to\SB(2n)\times\SB(2n),\quad\Phi(K)=(b_L,b_R).
\label{3.24}
\end{equation}
The key property of the momentum map is represented by the identity
\begin{equation}
\left.\frac{d}{ds}\right\vert_{s=0}
f(e^{sX}Ke^{-sY})=\Im\tr\big(X\{f,b_L\}b_L^{-1}+Y\{f,b_R\}b_R^{-1}\big),
\quad\forall X,Y\in\su(2n),
\label{3.25}
\end{equation}
where $f\in C^\infty(\SL(2n,\C))$ is an arbitrary real function and the Poisson bracket
is the one corresponding to $\omega$ \eqref{3.16}. The map $\Phi$ enjoys an
equivariance property and one can \cite{Lu91} perform Marsden-Weinstein type reduction in the
same way as for usual Hamiltonian actions (for which the symmetry group has vanishing
Poisson structure). To put it in a nutshell, any $\cH\in\fH$ gives rise to a reduced
Hamiltonian system by fixing the value of $\Phi$ and subsequently taking quotient
with respect to the
corresponding isotropy group. The reduced flows can be obtained by the standard
restriction-projection algorithm, and under favorable circumstances the reduced phase
space is a smooth symplectic manifold.

Now, consider the block-diagonal subgroup
\begin{equation}
G_+:=\mathrm{S}(\UN(n)\times\UN(n))<\SU(2n).
\label{3.26}
\end{equation}
Since $G_+$ is also a Poisson submanifold of $\SU(2n)$, the restriction of \eqref{3.23}
yields a Poisson-Lie action
\begin{equation}
G_+\times G_+\times\SL(2n,\C)\to\SL(2n,\C)
\label{3.27}
\end{equation}
of $G_+\times G_+$. The momentum map for this action is provided by projecting the
original momentum map $\Phi$ as follows. Let us write every element $b\in\SB(2n)$ in
the block-form
\begin{equation}
b=\begin{bmatrix}b(1)&b(12)\\\0_n&b(2)\end{bmatrix}
\label{3.28}
\end{equation}
and define $G_+^\ast<\SB(2n)$ to be the subgroup for which $b(12)=\0_n$.
If $\pi\colon\SB(2n)\to G_+^\ast$ denotes the projection
\begin{equation}
\pi\colon\begin{bmatrix}b(1)&b(12)\\\0_n&b(2)\end{bmatrix}\mapsto
\begin{bmatrix}b(1)&\0_n\\\0_n&b(2)\end{bmatrix},
\label{3.29}
\end{equation}
then the momentum map $\Phi_+\colon\SL(2n,\C)\to G_+^\ast\times G_+^\ast$ is furnished
by
\begin{equation}
\Phi_+(K)=(\pi(b_L),\pi(b_R)).
\label{3.30}
\end{equation}
Indeed, it is readily checked that the analogue of \eqref{3.25} holds with $X,Y$
taken from the block-diagonal subalgebra of $\su(2n)$ and $b_L,b_R$ replaced by
their projections. The equivariance property of this momentum map means that
in correspondence to
\begin{equation}
K\mapsto\eta_L K\eta_R^{-1}\quad\text{with}\quad(\eta_L,\eta_R)\in G_+\times G_+,
\label{3.31}
\end{equation}
one has
\begin{equation}
\big(\pi(b_L)\pi(b_L)^\dagger,\pi(b_R)\pi(b_R)^\dagger\big)\mapsto
\big(\eta_L\pi(b_L)\pi(b_L)^\dagger\eta_L^{-1},
\eta_R\pi(b_R)\pi(b_R)^\dagger \eta_R^{-1}\big).
\label{3.32}
\end{equation}
We briefly mention here that, as the notation suggests, $G_+^\ast$ is itself a
Poisson-Lie group that can serve as a Poisson dual of $G_+$. The relevant Poisson
structure can be obtained by identifying the block-diagonal subgroup of $\SB(2n)$ with
the factor group $\SB(2n)/L$, where $L$ is the block-upper-triangular normal subgroup.
This factor group inherits a Poisson structure from $\SB(2n)$, since $L$ is a so-called
coisotropic (or `admissible') subgroup of $\SB(2n)$ equipped with its standard Poisson
structure. The projected momentum map $\Phi_+$ is a Poisson map with respect to this
Poisson structure on the two factors $G_+^\ast$ in \eqref{3.30}. The details are not
indispensable for us. The interested reader may find them e.g. in \cite{BCST08}.

Inspired by the papers \cite{FP07,FK11,Ma15}, we wish to study the particular
Marsden-Weinstein reduction defined by imposing the following momentum map constraint:
\begin{equation}
\Phi_+(K)=\mu\equiv(\mu_L,\mu_R),\ \text{where}\ 
\mu_L=\begin{bmatrix}e^u\nu(x)&\0_n\\\0_n&e^{-u}\1_n\end{bmatrix},\ 
\mu_R=\begin{bmatrix}e^v\1_n&\0_n\\\0_n&e^{-v}\1_n\end{bmatrix}
\label{3.33}
\end{equation}
with some real constants $u$, $v$, and $x$. Here, $\nu(x)\in\SB(n)$ is the $n\times n$
upper triangular matrix defined by
\begin{equation}
\nu(x)_{jj}=1,\quad\nu(x)_{jk}=(1-e^{-x})e^{\frac{(k-j)x}{2}},\quad j<k,
\label{3.34}
\end{equation}
whose main property is that $\nu(x)\nu(x)^\dag$ has the largest possible
non-trivial isotropy group under conjugation by the elements of $\SU(n)$.

Our principal task is to characterize the reduced phase space
\begin{equation}
M\equiv\Phi_+^{-1}(\mu)/G_\mu,
\label{3.35}
\end{equation}
where $\Phi_+^{-1}(\mu)=\{K\in\SL(2n,\C)\mid\Phi_+(K)=\mu\}$
and
\begin{equation}
G_\mu=G_+(\mu_L)\times G_+
\label{3.36}
\end{equation}
is the isotropy group of $\mu$ inside $G_+\times G_+$. Concretely, $G_+(\mu_L)$ is
the subgroup of $G_+$ consisting of the special unitary matrices of the form
\begin{equation}
\eta_L=\begin{bmatrix}\eta_L(1)&\0_n\\\0_n&\eta_L(2)\end{bmatrix},
\label{3.37}
\end{equation}
where $\eta_L(2)$ is arbitrary and
\begin{equation}
\eta_L(1)\nu(x)\nu(x)^\dag\eta_L(1)^{-1}=\nu(x)\nu(x)^\dag.
\label{3.38}
\end{equation}
In words, $\eta_L(1)$ belongs to the little group of $\nu(x)\nu(x)^\dag$ in $\UN(n)$.
We shall see that $\Phi_+^{-1}(\mu)$ and $M$ are smooth manifolds for which the
canonical projection
\begin{equation}
\pi_\mu\colon\Phi_+^{-1}(\mu)\to M
\label{3.39}
\end{equation}
is a smooth submersion. Then $M$ \eqref{3.35} inherits a symplectic form
$\omega_M$ from $\omega$ \eqref{3.16}, which satisfies
\begin{equation}
\iota_\mu^\ast(\omega)=\pi_\mu^\ast(\omega_M),
\label{3.40}
\end{equation}
where $\iota_\mu\colon\Phi_+^{-1}(\mu)\to\SL(2n,\C)$ denotes the tautological
embedding.

\section{Solution of the momentum equation}
\label{sec:3.2}

The description of the reduced phase space requires us to solve the momentum map
constraints, i.e. we have to find all elements $K\in\Phi_+^{-1}(\mu)$. Of course,
it is enough to do this up to the gauge transformations provided by the isotropy
group $G_\mu$ \eqref{3.36}. The solution of this problem will rely on the auxiliary
equation \eqref{3.51} below, which is essentially equivalent to the momentum map
constraint, $\Phi_+(K)=\mu$, and coincides with an equation studied previously in
great detail in \cite{FK11}. Thus we start in the next subsection by deriving this
equation.

\subsection{A crucial equation implied by the constraints}
\label{subsec:3.2.1}

We begin by recalling (e.g.~\cite{Ma97}) that any $g\in\SU(2n)$ can be decomposed as
\begin{equation}
g=g_+\begin{bmatrix}\cos q&\ri\sin q\\\ri\sin q&\cos q\end{bmatrix}h_+,
\label{3.41}
\end{equation}
where $g_+,h_+\in G_+$ and $q=\diag(q_1,\dots,q_n)\in\R^n$ satisfies
\begin{equation}
\frac{\pi}{2}\geq q_1\geq\dots\geq q_n\geq 0.
\label{3.42}
\end{equation}
The vector $q$ is uniquely determined by $g$, while $g_+$ and $h_+$ suffer from
controlled ambiguities.

First, apply the above decomposition to $g_L$ in $K=g_Lb_R^{-1}\in\Phi_+^{-1}(\mu)$
and use the right-handed momentum constraint $\pi(b_R)=\mu_R$. It is then easily seen
that up to gauge transformations every element of $\Phi_+^{-1}(\mu)$ can be
represented in the following form:
\begin{equation}
K=\begin{bmatrix}\rho&\0_n\\\0_n&\1_n\end{bmatrix}
\begin{bmatrix}\cos q&\ri\sin q\\\ri\sin q&\cos q\end{bmatrix}
\begin{bmatrix}e^{-v}\1_n&\alpha\\\0_n&e^v\1_n\end{bmatrix}.
\label{3.43}
\end{equation}
Here $\rho\in\SU(n)$ and $\alpha$ is an $n\times n$ complex matrix. By using obvious
block-matrix notation, we introduce $\Omega:=K_{22}$ and record from \eqref{3.43} that
\begin{equation}
\Omega=\ri(\sin q)\alpha+e^v\cos q.
\label{3.44}
\end{equation}
For later purpose we introduce also the polar decomposition of the matrix $\Omega$,
\begin{equation}
\Omega=\Lambda T,
\label{3.45}
\end{equation}
where $T\in\UN(n)$ and the Hermitian, positive semi-definite factor $\Lambda$ is
uniquely determined by the relation $\Omega\Omega^\dag=\Lambda^2$.

Second, by writing $K=b_Lg_R^{-1}$ the left-handed momentum constraint
$\pi(b_L) =\mu_L$ tells us that $b_L$ has the block-form
\begin{equation}
b_L=\begin{bmatrix}e^u\nu(x)& \chi\\\0_n&e^{-u}\1_n\end{bmatrix}
\label{3.46}
\end{equation}
with an $n\times n$ matrix $\chi$. Now we inspect the components of the $2\times 2$
block-matrix identity
\begin{equation}
K K^\dag=b_Lb_L^\dag,
\label{3.47}
\end{equation}
which results by substituting $K$ from \eqref{3.43}. We find that the (22) component of
this identity is equivalent to
\begin{equation}
\Omega\Omega^\dag=\Lambda^2=e^{-2u}\1_n-e^{-2v}(\sin q)^2.
\label{3.48}
\end{equation}
On account of the condition \eqref{3.3}, this uniquely determines $\Lambda$ in terms
of $q$, and shows also that $\Lambda$ is invertible. A further important consequence
is that we must have
\begin{equation}
q_n>0,
\label{3.49}
\end{equation}
and therefore $\sin q$ is an invertible diagonal matrix. Indeed, if $q_n=0$, then
from \eqref{3.44} and \eqref{3.48} we would get
$(\Omega\Omega^\dag)_{nn}=e^{2v}=e^{-2u}$, which is excluded by \eqref{3.3}.

Next, one can check that in the presence of the relations already established, the
(12) and the (21) components of the identity \eqref{3.47} are equivalent to the equation
\begin{equation}
\chi=\rho(\ri\sin q)^{-1}[e^{-u}\cos q-e^{u+v}\Omega^\dag].
\label{3.50}
\end{equation}
Observe that $K$ uniquely determines $q$, $T$ and $\rho$, and conversely $K$ is
uniquely defined by the above relations once $q$, $T$ and $\rho$ are found.

Now one can straightforwardly check by using the above relations that the (11)
component of the identity \eqref{3.47} translates into the following equation:
\begin{equation}
\rho(\sin q)^{-1}T^\dag(\sin q)^2T(\sin q)^{-1}\rho^\dag=\nu(x)\nu(x)^\dag.
\label{3.51}
\end{equation}
This is to be satisfied by $q$ subject to \eqref{3.42}, \eqref{3.49} and $T\in\UN(n)$,
$\rho\in\SU(n)$. What makes our job relatively easy is that this is the same as
equation (5.7) in the paper \cite{FK11} by Feh\'er and Klim\v c\'ik. In fact,
this equation
was analyzed in detail in \cite{FK11}, since it played a crucial role in that work,
too. The correspondence with the symbols used in \cite{FK11} is
\begin{equation}
(\rho,T,\sin q)\Longleftrightarrow(k_L,k_R^\dag,e^{\hat p}).
\label{3.52}
\end{equation}
This motivates to introduce the variable $\hat p\in\R^n$ in our case, by setting
\begin{equation}
\sin q_k=e^{\hat p_k},\quad k=1,\dots,n.
\label{3.53}
\end{equation}
Notice from \eqref{3.42} and \eqref{3.49} that we have
\begin{equation}
0\geq\hat p_1\geq\dots\geq\hat p_n>-\infty.
\label{3.54}
\end{equation}
If the components of $\hat{p}$ are all different, then we can directly rely on
\cite{FK11} to establish both the allowed range of $\hat{p}$ and the explicit
form of $\rho$ and $T$. The statement that $\hat{p}_j\neq\hat{p}_k$ holds for
$j\neq k$ can be proved by adopting arguments given in \cite{FK11,FK12}.
This proof requires combining techniques of \cite{FK11} and \cite{FK12},
whose extraction from \cite{FK11,FK12} is rather involved. We present it
in Appendix \ref{sec:C.2}, otherwise in the next subsection we proceed by
simply stating relevant applications of results from \cite{FK11}.

\begin{remark}
\label{rem:3.1}
In the context of \cite{FK11} the components of $\hat p$ are not restricted to the
half-line and both $k_L$ and $k_R$ vary in $\UN(n)$. These slight differences do
not pose any obstacle to using the results and techniques of \cite{FK11,FK12}. We note that
essentially the same equation \eqref{3.51} surfaced in \cite{Ma15} as well, but the author of
that paper refrained from taking advantage of the previous analyses of this equation.
In fact, some statements of \cite{Ma15} are not fully correct. This will be specified (and
corrected) in Section \ref{sec:3.4}.
\end{remark}

\subsection{Consequences of equation \eqref{3.51}}
\label{subsec:3.2.2}

We start by pointing out the foundation of the whole analysis. For this, we
first display the identity
\begin{equation}
\nu(x)\nu(x)^\dag=e^{-x}\1_n+\sgn(x)\hat v\hat v^\dag,
\label{3.55}
\end{equation}
which holds with a certain $n$-component vector $\hat{v}=\hat{v}(x)$.
By introducing
\begin{equation}
w=\rho^\dag\hat{v}
\label{3.56}
\end{equation}
and setting $\hat{p}\equiv\diag(\hat{p}_1,\dots,\hat{p}_n)$, we rewrite equation
\eqref{3.51} as
\begin{equation}
e^{2\hat{p}-x\1_n}+\sgn(x)e^{\hat{p}}ww^\dag e^{\hat{p}}=T^{-1}e^{2\hat{p}}T.
\label{3.57}
\end{equation}
The equality of the characteristic polynomials of the matrices on the two sides of
\eqref{3.57} gives a polynomial equation that contains $\hat{p}$, the absolute values
$|w_j|^2$ and a complex indeterminate. Utilizing the requirement that $|w_j|^2\geq 0$
must hold, one obtains the following result.

\begin{proposition}
\label{prop:3.2}
If $K$ given by \eqref{3.43} belongs to the constraint surface $\Phi_+^{-1}(\mu)$,
then the vector $\hat p$ \eqref{3.53} is contained in the closed polyhedron
\begin{equation}
\bar{\cC}_x:=\{\hat{p}\in\R^n\mid 0\geq\hat{p}_1,\
\hat{p}_k-\hat{p}_{k+1}\geq|x|/2\ (k=1,\dots,n-1)\}.
\label{3.58}
\end{equation}
\end{proposition}
Proposition \ref{prop:3.2} can be proved by merging the proofs of \cite[Lemma 5.2]{FK11}
and \cite[Theorem 2]{FK12}. This is presented in Appendix \ref{sec:C.2}.

The above-mentioned polynomial equality permits to find the possible vectors $w$ \eqref{3.56}
as well. If $\hat{p}$ and $w$ are given, then $T$ is determined by equation \eqref{3.57}
up to left-multiplication by a diagonal matrix and $\rho$ is determined by
\eqref{3.56} up to left-multiplication by elements from the little group of
$\hat v(x)$. Following this line of reasoning and controlling the ambiguities in the
same way as in \cite{FK11}, one can find the explicit form of the most general $\rho$ and
$T$ at \emph{any} fixed $\hat{p}\in\bar{\cC}_x$. In particular, it turns out that the
range of the vector $\hat{p}$ equals $\bar{\cC}_x$.

Before presenting the result, we need to prepare some notations. First of all, we
pick an arbitrary $\hat p\in\bar\cC_x$ and define the $n\times n$ matrix
$\theta(x,\hat p)$ as follows:
\begin{equation}
\theta(x,\hat p)_{jk}:=\frac{\sinh\big(\frac{x}{2}\big)}{\sinh(\hat p_k-\hat p_j)}
\prod_{\substack{m=1\\(m\neq j,k)}}^n\bigg[\frac{\sinh(\hat p_j-\hat p_m-\frac{x}{2})
\sinh(\hat p_k-\hat p_m+\frac{x}{2})}{\sinh(\hat p_j-\hat p_m)
\sinh(\hat p_k-\hat p_m)}\bigg]^{\tfrac{1}{2}},\quad j\neq k,
\label{3.59}
\end{equation}
and
\begin{equation}
\theta(x,\hat p)_{jj}:=\prod_{\substack{m=1\\(m\neq j)}}^n
\bigg[\frac{\sinh(\hat p_j-\hat p_m-\frac{x}{2})\sinh(\hat p_j-\hat p_m+
\frac{x}{2})}{\sinh^2(\hat p_j-\hat p_m)}\bigg]^{\tfrac{1}{2}}.
\label{3.60}
\end{equation}
All expressions under square root are non-negative and non-negative square roots are
taken. Note that $\theta(x,\hat p)$ is a real orthogonal matrix of determinant 1 for
which $\theta(x,\hat p)^{-1}=\theta(-x,\hat p)$ holds, too.

Next, define the real vector $r(x,\hat{p})\in\R^n$ with non-negative components
\begin{equation}
r(x,\hat p)_j=\sqrt{\frac{1-e^{-x}}{1-e^{-nx}}}\prod_{\substack{k=1\\(k\neq j)}}^n
\sqrt{\frac{1-e^{2\hat p_j-2\hat p_k-x}}{1-e^{2\hat p_j-2\hat p_k}}},
\quad j=1,\dots,n,
\label{3.61}
\end{equation}
and the real $n\times n$ matrix $\zeta(x,\hat{p})$,
\begin{equation}
\begin{split}
&\zeta(x,\hat p)_{aa}=r(x,\hat p)_a,\quad
\zeta(x,\hat p)_{ij}=\delta_{ij}-\frac{r(x,\hat p)_ir(x,\hat p)_j}{1+r(x,\hat p)_a},\\
&\zeta(x,\hat p)_{ia}=-\zeta(x,\hat p)_{ai}=r(x,\hat p)_i,\quad i,j\neq a,
\end{split}
\label{3.62}
\end{equation}
where $a=n$ if $x>0$ and $a=1$ if $x<0$. Introduce also the vector $v=v(x)$:
\begin{equation}
v(x)_j=\sqrt{\frac{n(e^x-1)}{1-e^{-nx}}}e^{-\tfrac{jx}{2}},\quad j=1,\ldots,n,
\label{3.63}
\end{equation}
which is related to $\hat{v}$ in \eqref{3.55} by
\begin{equation}
\hat v(x)=\sqrt{\sgn(x) e^{-x}\frac{e^{nx}-1}{n}}v(x).
\label{3.64}
\end{equation}
Finally, define the $n\times n$ matrix $\kappa(x)$ as
\begin{equation}
\begin{gathered}
\kappa(x)_{aa}=\frac{v(x)_a}{\sqrt{n}},\quad
\kappa(x)_{ij}=\delta_{ij}-\frac{v(x)_iv(x)_j}{n+\sqrt{n}v(x)_a},\\
\kappa(x)_{ia}=-\kappa(x)_{ai}=\frac{v(x)_i}{\sqrt{n}},\quad i,j\neq a,
\end{gathered}
\label{3.65}
\end{equation}
where, again, $a=n$ if $x>0$ and $a=1$ if $x<0$. It can be shown that both
$\kappa(x)$ and $\zeta(x,\hat{p})$ are orthogonal matrices of
determinant 1 for any $\hat{p}\in\bar{\cC}_x$.

Now we can state the main result of this section, whose proof is omitted since it is
a direct application of the analysis of the solutions of \eqref{3.51} presented in
\cite[Section 5]{FK11}.

\begin{proposition}
\label{prop:3.3}
Take any $\hat{p}\in\bar{\cC}_x$ and any diagonal unitary matrix $e^{\ri\hat{q}}\in\T_n$.
By using the preceding notations define $K\in\SL(2n,\C)$ \eqref{3.43} by setting
\begin{equation}
T=e^{\ri\hat{q}}\theta(-x,\hat{p}),\quad\rho=\kappa(x)\zeta(x,\hat{p})^{-1},
\label{3.66}
\end{equation}
and also applying the equations \eqref{3.44}, \eqref{3.45}, \eqref{3.48}, and \eqref{3.53}.
Then the element $K$ belongs to the constraint surface $\Phi_+^{-1}(\mu)$, and every
orbit of the gauge group $G_\mu$ \eqref{3.36} in $\Phi_+^{-1}(\mu)$ intersects the set
of elements $K$ just constructed.
\end{proposition}

\begin{remark}
\label{rem:3.4}
It is worth spelling out the expression of the element $K$ given by
Proposition \ref{prop:3.3}. Indeed, we have
\begin{equation}
K(\hat p, e^{\ri \hat q}) =\begin{bmatrix}\rho &\0_n\\\0_n&\1_n\end{bmatrix}
\begin{bmatrix}\sqrt{\1_n-e^{2\hat p}} &\ri e^{\hat p}\\\ri e^{\hat p}&
\sqrt{\1_n-e^{2\hat p}}\end{bmatrix}
\begin{bmatrix}e^{-v}\1_n&\alpha\\\0_n&e^v\1_n\end{bmatrix}
\label{3.67}
\end{equation}
using the above definitions and
\begin{equation}
\alpha=-\ri\bigg[e^{\ri\hat q}\sqrt{e^{-2u}e^{-2\hat p}-e^{-2v}\1_n}\,\theta(-x,\hat p)
-e^v\sqrt{e^{-2\hat p}-\1_n}\bigg].
\label{3.68}
\end{equation}
\end{remark}

\begin{remark}
\label{rem:3.5}
Let us call $S$ the set of the elements $K(\hat{p},e^{\ri\hat{q}})$ constructed above,
and observe that this set is homeomorphic to
\begin{equation}
\bar\cC_x\times\T_n=\{(\hat p,e^{\ri\hat q})\}
\label{3.69}
\end{equation}
by its very definition. This is not a smooth manifold, because of the presence of the
boundary of $\bar\cC_x$. However, this does not indicate any `trouble' since it is
not true (at the boundary of $\bar{\cC}_x$) that $S$ intersects every gauge orbit in
$\Phi_+^{-1}(\mu)$ in a \emph{single} point. Indeed, it is instructive to verify that
if $\hat{p}$ is the special vertex of $\bar\cC_x$ for which
$\hat{p}_k=(1-k)|x|/2$ for $k=1,\dots,n$, then all points $K(\hat{p},e^{\ri\hat{q}})$
lie on a single gauge orbit. This, and further inspection, can lead to the idea that
the variables $\hat{q}_j$ should be identified with arguments of complex numbers,
which lose their meaning at the origin that should correspond to the boundary of
$\bar{\cC}_x$. Our Theorem \ref{thm:3.14} will show that this idea is correct. It is proper to
stress that we arrived at such idea under the supporting influence of previous works
\cite{Ru95,FK11}.
\end{remark}

\section{Characterization of the reduced system}
\label{sec:3.3}

The smoothness of the reduced phase space and the completeness of the reduced free
flows follows immediately if we can show that the gauge group $G_\mu$ acts in such a
way on $\Phi_+^{-1}(\mu)$ that the isotropy group of every point is just the finite
center of the symmetry group.
In Subsection \ref{subsec:3.3.1}, we prove that the factor of $G_\mu$ by the center
acts freely on $\Phi_+^{-1}(\mu)$. Then in Subsection \ref{subsec:3.3.2} we explain
that $\cC_x\times\T_n$ provides a model of a dense open subset of the reduced phase
space by means of the corresponding subset of $\Phi_+^{-1}(\mu)$ defined by
Proposition \ref{prop:3.3}. Adopting a key calculation from \cite{Ma15}, it turns out that
$(\hat p,e^{\ri\hat q})\in\cC_x\times\T_n$ are Darboux coordinates on this dense open
subset. In Subsection \ref{subsec:3.3.3}, we demonstrate that the reduction
of the Abelian Poisson algebra of free Hamiltonians \eqref{3.19} yields an integrable
system. Finally, in Subsection \ref{subsec:3.3.4}, we present a model of the full reduced phase
space, which is our main result in this chapter.

\subsection{Smoothness of the reduced phase space}
\label{subsec:3.3.1}

It is clear that the normal subgroup of the full symmetry group $G_+\times G_+$
consisting of matrices of the form
\begin{equation}
(\eta,\eta)\quad\text{with}\quad\eta=\diag(z\1_n,z\1_n),\quad z^{2n}=1
\label{3.70}
\end{equation}
acts trivially on the phase space. This subgroup is contained in $G_\mu$ \eqref{3.36}.
The corresponding factor group of $G_\mu$ is called `effective gauge group' and is
denoted by $\bar{G}_\mu$. We wish to show that $\bar{G}_\mu$ acts freely on the
constraint surface $\Phi_+^{-1}(\mu)$.

We need the following elementary lemmas.

\begin{lemma}
\label{lem:3.6}
Suppose that
\begin{equation}
g_+\begin{bmatrix}\cos q&\ri\sin q\\\ri\sin q&\cos q\end{bmatrix}h_+
=g_+'\begin{bmatrix}\cos q&\ri\sin q\\\ri\sin q&\cos q\end{bmatrix}h_+'
\label{3.71}
\end{equation}
with $g_+,h_+,g_+',h_+'\in G_+$ and $q=\diag(q_1,\dots,q_n)$ subject to
\begin{equation}
\frac{\pi}{2}\geq q_1>\dots>q_n>0.
\label{3.72}
\end{equation}
Then there exist diagonal matrices $m_1,m_2\in\T_n$ having the form
\begin{equation}
m_1=\diag(a,\xi),\quad m_2=\diag(b,\xi),
\quad\xi\in\T_{n-1},\ a,b\in\T_1,\quad\det(m_1m_2)=1,
\label{3.73}
\end{equation}
for which
\begin{equation}
(g_+',h_+')=(g_+\diag(m_1,m_2),\diag(m_2^{-1},m_1^{-1})h_+).
\label{3.74}
\end{equation}
If \eqref{3.72} holds with strict inequality $\frac{\pi}{2}>q_1$, then $m_1=m_2$,
i.e. $a=b$.
\end{lemma}

\begin{lemma}
\label{lem:3.7}
Pick any $\hat{p}\in\bar\cC_x$ and consider the matrix $\theta(x,\hat p)$ given
by \eqref{3.59} and \eqref{3.60}. Then the entries $\theta_{n,1}(x,\hat p)$ and
$\theta_{j,j+1}(x,\hat p)$ are all non-zero if $x>0$ and the entries
$\theta_{1,n}(x,\hat p)$ and $\theta_{j+1,j}(x,\hat p)$ are all non-zero if $x<0$.
\end{lemma}

For convenience, we present the proof of Lemma \ref{lem:3.6} in Appendix \ref{sec:C.3}.
The property recorded in Lemma \ref{lem:3.7} is known \cite{Ru95,FK11}, and is
easily checked by inspection.

\begin{proposition}
\label{prop:3.8}
The effective gauge group $\bar G_\mu$ acts freely on $\Phi_+^{-1}(\mu)$.
\end{proposition}

\begin{proof}
Since every gauge orbit intersects the set $S$ specified by Proposition
\ref{prop:3.3}, it is enough to show that if $(\eta_L,\eta_R)\in G_\mu$ maps $K\in S$
\eqref{3.67} to itself, then $(\eta_L,\eta_R)$ equals some element $(\eta,\eta)$ given
in \eqref{3.70}. For $K$ of the form \eqref{3.43}, we can spell out
$K'\equiv\eta_L K\eta_R^{-1}$ as
\begin{equation}
K'=\begin{bmatrix}\eta_L(1)\rho&\0_n\\\0_n&\eta_L(2)\end{bmatrix}
\begin{bmatrix}\cos q&\ri\sin q\\\ri\sin q&\cos q\end{bmatrix}
\begin{bmatrix}\eta_R(1)^{-1}&\0_n\\\0_n&\eta_R(2)^{-1}\end{bmatrix}
\begin{bmatrix}e^{-v}\1_n&\eta_R(1)\alpha\eta_R(2)^{-1}\\\0_n&e^v\1_n\end{bmatrix}.
\label{3.75}
\end{equation}
The equality $K'=K$ implies by the uniqueness of the Iwasawa decomposition and Lemma
\ref{lem:3.6} that we must have
\begin{equation}
\eta_L(2)=\eta_R(1)=m_2,\quad\eta_R(2)=m_1,\quad\eta_L(1)\rho=\rho m_1,
\label{3.76}
\end{equation}
with some diagonal unitary matrices having the form \eqref{3.73}. By using that
$\eta_R(1)=m_2$ and $\eta_R(2)=m_1$, the Iwasawa decomposition of $K'=K$ in
\eqref{3.67} also entails the relation
\begin{equation}
\alpha=m_2\alpha m_1^{-1}.
\label{3.77}
\end{equation}
Because of \eqref{3.68}, the off-diagonal components of the matrix equation \eqref{3.77}
yield
\begin{equation}
\theta(-x,\hat p)_{jk}=\big(m_2\theta(-x,\hat p)m_1^{-1}\big)_{jk},
\quad\forall j\neq k.
\label{3.78}
\end{equation}
This implies by means of Lemma \ref{lem:3.7} and equation \eqref{3.73} that
$m_1=m_2=z\1_n$ is a scalar matrix. But then $\eta_L(1)=m_1$ follows from
$\eta_L(1)\rho=\rho m_1$, and the proof is complete.
\end{proof}

Proposition \ref{prop:3.8} and the general results gathered in Appendix \ref{sec:C.4}
imply the following theorem, which is one of our main results.

\begin{theorem}
\label{thm:3.9}
The constraint surface $\Phi_+^{-1}(\mu)$ is an embedded submanifold of $\SL(2n,\C)$ and
the reduced phase space $M$ \eqref{3.35} is a smooth manifold for which the natural
projection $\pi_\mu\colon\Phi_+^{-1}(\mu)\to M$ is a smooth submersion.
\end{theorem}

\subsection{Model of a dense open subset of the reduced phase space}
\label{subsec:3.3.2}

Let us denote by $S^o\subset S$ the subset of the elements $K$ given by Proposition
\ref{prop:3.3} with $\hat p$ in the interior $\cC_x$ of the polyhedron $\bar\cC_x$ \eqref{3.58}.
Explicitly, we have
\begin{equation}
S^o=\{K(\hat p,e^{\ri\hat q})\mid(\hat p,e^{\ri\hat q})\in\cC_x\times\T_n\},
\label{3.79}
\end{equation}
where $K(\hat p,e^{\ri\hat q})$ stands for the expression \eqref{3.67}.
Note that $S^o$ is in bijection with $\cC_x\times\T_n$. The next lemma
says that no two different point of $S^o$ are gauge equivalent.

\begin{lemma}
\label{lem:3.10}
The intersection of any gauge orbit with $S^o$ consists of at most one point.
\end{lemma}

\begin{proof}
Suppose that
\begin{equation}
K':=K(\hat p',e^{\ri\hat q'})=\eta_L K(\hat p,e^{\ri\hat q})\eta_R^{-1}
\label{3.80}
\end{equation}
with some $(\eta_L,\eta_R)\in G_\mu$. By spelling out the gauge transformation as in
\eqref{3.75}, using the shorthand $\sin q=e^{\hat p}$, we observe that $\hat p'=\hat p$
since $q$ in \eqref{3.41} does not change under the action of $G_+\times G_+$. Since
now we have $\frac{\pi}{2}>q_1$ (which is equivalent to $0>\hat p_1$), the arguments
applied in the proof of Proposition \ref{prop:3.8} permit to translate the equality \eqref{3.80}
into the relations
\begin{equation}
\eta_L(2)=\eta_R(1)=\eta_R(2)=m,\quad\eta_L(1)\rho=\rho m,
\label{3.81}
\end{equation}
complemented with the condition
\begin{equation}
\alpha(\hat p,e^{\ri\hat q'})=m\alpha(\hat p,e^{\ri\hat q})m^{-1},
\label{3.82}
\end{equation}
which is equivalent to
\begin{equation}
e^{\ri\hat q'}\theta(-x,\hat p)=me^{\ri\hat q}\theta(-x,\hat p)m^{-1}.
\label{3.83}
\end{equation}
We stress that $m\in\T_n$ and notice from \eqref{3.60} that for $\hat p\in\cC_x$
all the diagonal entries $\theta(-x,\hat p)_{jj}$ are non-zero. Therefore we conclude
from \eqref{3.83} that $e^{\ri\hat q'}=e^{\ri q}$. This finishes the proof, but of
course we can also confirm that $m=z\1_n$, consistently with Proposition \ref{prop:3.8}.
\end{proof}

Now we introduce the map $\cP\colon\SL(2n,\C)\to\R^n$ by
\begin{equation}
\cP\colon K=g_Lb_R^{-1}\mapsto\hat p,
\label{3.84}
\end{equation}
defined by writing $g_L$ in the form \eqref{3.41} with $\sin q=e^{\hat p}$.
The map $\cP$ gives rise to a map $\bar \cP\colon M\to\R^n$ verifying
\begin{equation}
\bar{\cP}(\pi_\mu(K))=\cP(K),\quad\forall K\in\Phi_+^{-1}(\mu),
\label{3.85}
\end{equation}
where $\pi_\mu$ is the canonical projection \eqref{3.39}. We notice that, since the
`eigenvalue parameters' $\hat p_j$ $(j=1,\dots,n)$ are pairwise different for any
$K\in\Phi_+^{-1}(\mu)$, $\bar \cP$ is a smooth map. The continuity of $\bar \cP$ implies
that
\begin{equation}
M^o:=\bar \cP^{-1}(\cC_x)=\pi_\mu(S^o)\subset M
\label{3.86}
\end{equation}
is an open subset. The second equality is a direct consequence of our foregoing
results about $S$ and $S^o$. Note that $\bar{\cP}^{-1}(\bar{\cC}_x)=\pi_\mu(S)=M$.
Since $\pi_\mu$ is continuous (actually smooth) and any point of $S$ is the limit of
a sequence in $S^o$, $M^o$ is \emph{dense} in the reduced phase space $M$.
The dense open subset $M^o$ can be parametrized by $\cC_x\times\T_n$ according to
\begin{equation}
(\hat p,e^{\ri\hat q})\mapsto\pi_\mu(K(\hat p,e^{\ri\hat q})),
\label{3.87}
\end{equation}
which also allows us to view $S^o\simeq\cC_x\times\T_n$ as a model of $M^o\subset M$.
In principle, the restriction of the reduced symplectic form to $M^o$ can now be
computed by inserting the explicit formula $K(\hat p,e^{\ri\hat q})$ \eqref{3.67}
into the Alekseev-Malkin form \eqref{3.16}. In the analogous reduction of the
Heisenberg double of $\SU(n,n)$, Marshall \cite{Ma15} found a nice way to circumvent such
a tedious calculation. By taking the same route, we have verified that $\hat p$ and
$\hat q$ are Darboux coordinates on $M^o$.

The outcome of the above considerations is summarized by the next theorem.

\begin{theorem}
\label{thm:3.11}
$M^o$ defined by equation \eqref{3.86} is a dense open subset of
the reduced phase space $M$. Parametrizing $M^o$ by $\cC_x\times\T_n$ according to
\eqref{3.87}, the restriction of reduced symplectic form $\omega_M$ \eqref{3.40} to
$M^o$ is equal to $\hat\omega=\sum_{j=1}^nd\hat q_j\wedge d\hat p_j$ \eqref{3.5}.
\end{theorem}

\subsection{Liouville integrability of the reduced free Hamiltonians}
\label{subsec:3.3.3}

The Abelian Poisson algebra $\fH$ \eqref{3.19} consists of $(G_+ \times G_+)$-invariant
functions\footnote{More precisely, $\fH=C^\infty(\SL(2n,\C))^{\SU(2n)\times\SU(2n)}$.}
generating complete flows, given explicitly by \eqref{3.21}, on the unreduced phase
space. Thus each element of $\fH$ descends to a smooth reduced Hamiltonian on $M$
\eqref{3.35}, and generates a complete flow via the reduced symplectic form
$\omega_M$. This flow is the projection of the corresponding unreduced flow,
which preserves the constraint surface $\Phi_+^{-1}(\mu)$. It also follows from the
construction that $\fH$ gives rise to an Abelian Poisson algebra, $\fH_M$, on
$(M,\omega_M)$. Now the question is whether the Hamiltonian vector fields of $\fH_M$
span an $n$-dimensional subspace of the tangent space at the points of a dense open
submanifold of $M$. If yes, then $\fH_M$ yields a Liouville integrable system, since
$\dim(M)=2n$.

Before settling the above question, let us focus on the Hamiltonian $\cH\in \fH$
defined by
\begin{equation}
\cH(K):=\frac{1}{2}\tr\big((K^\dag K)^{-1}\big)=\frac{1}{2}\tr(b_R^\dag b_R).
\label{3.88}
\end{equation}
Using the formula of $K(\hat p,e^{\ri\hat q})$ in Remark \ref{rem:3.4}, it is readily
verified that
\begin{equation}
\cH(K(\hat p,e^{\ri\hat q}))=H(\hat p,\hat q;x,v-u,v+u),
\quad\forall(\hat p,e^{\ri\hat q})\in\cC_x\times\T_n,
\label{3.89}
\end{equation}
with the Hamiltonian $H$ displayed in equation \eqref{3.1}. \emph{Consequently, $H$ in
\eqref{3.1} is identified as the restriction of the reduction of $\cH$ \eqref{3.88}
to the dense open submanifold $M^o$ \eqref{3.86} of the reduced phase space, wherein
the flow of every element of $\fH_M$ is complete.}

Turning to the demonstration of Liouville integrability, consider the $n$ functions
\begin{equation}
\cH_k(K):=\frac{1}{2k}\tr\big((K^\dag K)^{-1}\big)^k
=\frac{1}{2k}\tr(b_R^\dag b_R)^k,\quad k=1,\dots,n.
\label{3.90}
\end{equation}
The restriction of the corresponding elements of $\fH_M$ on $M^o\simeq\cC_x\times\T_n$
gives the functions
\begin{equation}
H_k(\hat p,\hat q)=\frac{1}{2k}\tr
\begin{bmatrix}
e^{2v}\1_n&-e^v\alpha\\
-e^v\alpha^\dag&(e^{-2v}\1_n+\alpha^\dag\alpha)\end{bmatrix}^k,
\label{3.91}
\end{equation}
where $\alpha$ has the form \eqref{3.68}. These are real-analytic functions on
$\cC_x\times\T_n$. It is enough to show that their exterior derivatives are
linearly independent on a dense open subset of $\cC_x\times\T_n$. This follows if we
show that the function
\begin{equation}
f(\hat p,\hat q)=\det\big[d_{\hat q}H_1,d_{\hat q}H_2,\dots,d_{\hat q}H_n\big]
\label{3.92}
\end{equation}
is not identically zero on $\cC_x\times\T_n$. Indeed, since $f$ is an analytic
function and $\cC_x\times\T_n$ is connected, if $f$ is not identically zero then
its zero set cannot contain any accumulation point. This, in turn, implies that $f$ is
non-zero on a dense open subset of $\cC_x\times\T_n\simeq M^o$, which is also dense
and open in the full reduced phase space $M$. In other words, the reductions of
$\cH_k$ $(k=1,\ldots, n)$ possess the property of Liouville integrability.
It is rather obvious that the function $f$ is not identically zero, since $H_k$
involves dependence on $\hat q$ through $e^{\pm\ri k\hat q}$ and lower powers of
$e^{\pm\ri\hat q}$. It is not difficult to inspect the function $f(\hat p,\hat q)$
in the `asymptotic domain' where all
differences $|\hat p_j-\hat p_m|$ $(m\neq j)$ tend to infinity, since in this domain
$\alpha$ becomes close to a diagonal matrix. We omit the details
of this inspection, whereby we checked that $f$ is indeed not identically zero.

The above arguments prove the Liouville integrability of the reduced free
Hamiltonians, i.e. the elements of $\fH_M$. Presumably, there exists a dual set of
integrable many-body Hamiltonians that live on the space of action-angle variables of
the Hamiltonians in $\fH_M$. The construction of such dual Hamiltonians is an
interesting task for the future, which will be further commented upon in Section
\ref{sec:3.5}.

\subsection{The global structure of the reduced phase space}
\label{subsec:3.3.4}

We here construct a global cross-section of the gauge orbits in the constraint
surface $\Phi_+^{-1}(\mu)$. This engenders a symplectic diffeomorphism between
the reduced phase space $(M,\omega_M)$ and the manifold $(\hat M_c,\hat\omega_c)$
below. It is worth noting that $(\hat M_c, \hat \omega_c)$ is symplectomorphic
to $\R^{2n}$ carrying the standard Darboux 2-form, and one can easily find an
explicit symplectomorphism if desired. Our construction was inspired by the
previous papers \cite{Ru95,FK11}, but detailed inspection of the specific
example was also required for finding the final result given by Theorem \ref{thm:3.14}.
After a cursory glance,
the reader is advised to go directly to this theorem and follow the definitions
backwards as becomes necessary. See also Remark \ref{rem:3.15} for the rationale behind the
subsequent definitions.

To begin, consider the product manifold
\begin{equation}
\hat M_c:=\C^{n-1}\times\D,
\label{3.93}
\end{equation}
where $\D$ stands for the open unit disk, i.e. $\D:=\{w\in\C:|w|<1\}$,
and equip it with the symplectic form
\begin{equation}
\hat\omega_c=\ri\sum_{j=1}^{n-1}dz_j\wedge d\bar z_j
+\frac{\ri dz_n\wedge d\bar z_n}{1-z_n\bar z_n}.
\label{3.94}
\end{equation}
The subscript $c$ refers to `complex variables'. Define the surjective map
\begin{equation}
\hat\cZ_x\colon\bar\cC_x\times\T_n\to\hat M_c,\quad
(\hat p,e^{\ri\hat q})\mapsto z(\hat p,e^{\ri\hat q})
\label{3.95}
\end{equation}
by the formulae
\begin{equation}
\begin{gathered}
z_j(\hat p,e^{\ri\hat q})
=(\hat p_j-\hat p_{j+1}-|x|/2)^{\tfrac{1}{2}}\prod_{k=j+1}^ne^{\ri\hat q_k},
\quad j=1,\dots,n-1,\\
z_n(\hat p,e^{\ri\hat q})=(1-e^{\hat p_1})^{\tfrac{1}{2}}\prod_{k=1}^ne^{\ri\hat q_k}.
\end{gathered}
\label{3.96}
\end{equation}
Notice that the restriction $\cZ_x$ of $\hat \cZ_x$ to $\cC_x\times \T_n$ is a
diffeomorphism onto the dense open submanifold
\begin{equation}
\hat M_c^o=\{z\in\hat M_c\mid\prod_{j=1}^nz_j\neq 0\}.
\label{3.97}
\end{equation}
It verifies
\begin{equation}
\cZ_x^\ast(\hat\omega_c)=\hat\omega=\sum_{j=1}^nd\hat q_j\wedge d\hat p_j,
\label{3.98}
\end{equation}
which means that $\cZ_x$ is a symplectic embedding of $(\cC_x\times\T_n,\hat\omega)$
into $(\hat M_c,\hat\omega_c)$. The inverse
$\cZ_x^{-1}\colon\hat M_c^o\to\cC_x\times\T_n$ operates according to
\begin{equation}
\begin{gathered}
\hat p_1(z)=\log(1-|z_n|^2),\quad
\hat p_j(z)=\log(1-|z_n|^2)-\sum_{k=1}^{j-1}(|z_k|^2+|x|/2)\quad (j=2,\dots,n)\\
e^{\ri\hat q_1}(z)=\frac{z_n\bar z_1}{|z_n\bar z_1|},\quad
e^{\ri\hat q_m}(z)=\frac{z_{m-1}\bar z_m}{|z_{m-1}\bar z_m|}\quad (m=2,\dots,n-1),\quad
e^{\ri\hat q_n}(z)=\frac{z_{n-1}}{|z_{n-1}|}.
\end{gathered}
\label{3.99}
\end{equation}
It is important to remark that the $\hat p_k(z)$ $(k=1,\dots,n)$ given above yield
smooth functions on the whole of $\hat M_c$, while the angles $\hat q_k$ are of
course not well-defined on the complementary locus of $\hat M_c^o$.
Our construction of the global cross-section will rely on the building blocks
collected in the following long definition.

\begin{definition}
\label{def:3.12}
For any $(z_1,\dots,z_{n-1})\in\C^{n-1}$ consider the smooth
functions
\begin{equation}
Q_{jk}(x,z)=\bigg[\frac{\sinh(\sum_{\ell=j}^{k-1}z_\ell\bar z_\ell+(k-j)|x|/2-x/2)}
{\sinh(\sum_{\ell=j}^{k-1}z_\ell\bar z_\ell+(k-j)|x|/2)}\bigg]^{\tfrac{1}{2}},\quad
1\leq j<k\leq n,
\label{3.100}
\end{equation}
and set $Q_{jk}(x,z):=Q_{kj}(-x,z)$ for $j>k$.
Applying these as well as the real analytic function
\begin{equation}
J(y):=\sqrt{\frac{\sinh(y)}{y}},\quad y\neq 0,\quad J(0):=1,
\label{3.101}
\end{equation}
and recalling \eqref{3.61}, introduce the $n\times n$ matrix $\hat\zeta(x,z)$ by the
formulae
\begin{equation}
\begin{gathered}
\hat\zeta(x,z)_{aa}=r(x,\hat p(z))_a,\quad
\hat\zeta(x,z)_{aj}=-\overline{\hat\zeta(x,z)_{ja}},\quad j\neq a, \\
\hat\zeta(x,z)_{jn}=\sqrt{\frac{\sinh(\frac{x}{2})}{\sinh(\frac{nx}{2})}}
\frac{z_jJ(z_j\bar z_j)}{\sinh(z_j\bar z_j+\frac{x}{2})}
\prod_{\substack{\ell=1\\(\ell\neq j,j+1)}}^nQ_{j\ell}(x,z),\quad x>0,\quad j\neq n,\\
\hat\zeta(x,z)_{j1}=\sqrt{\frac{\sinh(\frac{x}{2})}{\sinh(\frac{nx}{2})}}
\frac{\bar z_{j-1}J(z_{j-1}\bar z_{j-1})}{\sinh(z_{j-1}\bar z_{j-1}-\frac{x}{2})}
\prod_{\substack{\ell=1\\(\ell\neq j-1,j)}}^nQ_{j\ell}(x,z),\quad x<0,\quad j\neq 1,\\
\hat\zeta(x,z)_{jk}=\delta_{j,k}
+\frac{\hat\zeta(x,z)_{ja}\hat\zeta(x,z)_{ak}}{1+\hat\zeta(x,z)_{aa}},\quad j,k\neq a,
\end{gathered}
\label{3.102}
\end{equation}
where $a=n$ if $x>0$ and $a=1$ if $x<0$.
Then introduce the matrix $\hat\theta(x,z)$ for $x>0$ as
\begin{equation}
\begin{gathered}
\hat\theta(x,z)_{jk}=\frac{\sinh(\frac{nx}{2})\sgn(k-j-1)
\hat\zeta(x,z)_{jn}\hat\zeta(-x,\bar{z})_{1k}}
{\sinh(\sum_{\ell=\min(j,k)}^{\max(j,k)-1}z_\ell\bar z_\ell+|k-j-1|\frac{x}{2})},\quad
k\neq j+1,\\
\hat\theta(x,z)_{j,j+1}=\frac{-\sinh(\frac{x}{2})}{\sinh(z_j\bar z_j+\frac{x}{2})}
\prod_{\substack{\ell=1\\(\ell\neq j,j+1)}}^nQ_{j\ell}(x,z)Q_{j+1,\ell}(-x,z),
\end{gathered}
\label{3.103}
\end{equation}
and for $x<0$ as
\begin{equation}
\hat{\theta}(x,z)=\hat{\theta}(-x,\bar{z})^\dag.
\label{3.104}
\end{equation}
Finally, for any $z\in\hat M_c$ define the matrix
$\hat\gamma(x,z)=\diag(\hat\gamma_1,\dots,\hat\gamma_n)$ with
\begin{equation}
\hat\gamma(z)_1=z_n\sqrt{2-z_n\bar z_n},\quad
\hat\gamma(x,z)_j=\bigg[1-(1-z_n\bar z_n)^2
e^{-2\sum_{\ell=1}^{j-1}(z_\ell\bar z_\ell+|x|/2)}\bigg]^{\tfrac{1}{2}},\quad j=2,\dots,n,
\label{3.105}
\end{equation}
and the matrix
\begin{equation}
\hat\alpha(x,u,v,z)=-\ri\big[\sqrt{e^{-2u}e^{-2\hat p(z)}
-e^{-2v}\1_n}]\,\hat\theta(-x,\bar{z})
-e^ve^{-\hat p(z)}\hat\gamma(x,z)^\dag\big],
\label{3.106}
\end{equation}
using the constants $x,u=(b-a)/2,v=(a+b)/2$ subject to \eqref{3.3}.
\end{definition}

Although the variable $z_n$ appears only in $\hat\gamma_1$, we can regard all
objects defined above as smooth functions on $\hat M_c$, and we shall do so below.

The key properties of the matrices $\hat\zeta$, $\hat\theta$, $\hat\alpha$ and
$\hat\gamma$ are given by the following lemma, which can be verified by
straightforward inspection. The role of these identities and their origin will be
enlightened by Theorem \ref{thm:3.14}.

\begin{lemma}
\label{lem:3.13}
Prepare the notations
\begin{equation}
\tau_{(x)}:=\diag(\tau_2,\dots,\tau_n,1)\quad \text{if}\quad x>0
\quad\text{and}\quad
\tau_{(x)}:=\diag(1,\tau_2^{-1},\dots,\tau_n^{-1})\quad\text{if}\quad x<0,
\label{3.107}
\end{equation}
\begin{equation}
\tilde\tau_{(x)}:=\diag(1,\tau_2,\dots,\tau_n)\quad \text{if}\quad x>0
\quad\text{and}\quad
\tilde\tau_{(x)}:=\diag(\tau_2^{-1},\dots,\tau_n^{-1},1)\quad \text{if}\quad x<0
\label{3.108}
\end{equation}
with
\begin{equation}
\tau_j=\prod_{k=j}^ne^{\ri\hat q_k}.
\label{3.109}
\end{equation}
Then the following identities hold for all
$(\hat p,e^{\ri\hat q})\in\bar\cC_x\times\T_n$:
\begin{align}
\hat\zeta(x,z(\hat p,e^{\ri \hat q}))
&=\tau_{(x)}\zeta(x,\hat p)\tau_{(x)}^{-1},
\label{3.110}\\
\hat\theta(x,z(\hat p,e^{\ri \hat q}))
&=\tau_{(x)}\theta(x,\hat p)\tilde\tau_{(x)}^{-1},
\label{3.111}\\
\hat\gamma(x,z(\hat p,e^{\ri \hat q}))
&=e^{\ri\hat q}\tau_{(x)}\tilde\tau_{(x)}^{-1}\sqrt{\1_n-e^{2\hat p}},
\label{3.112}\\
\hat\alpha(x,u,v, z(\hat p,e^{\ri \hat q}))
&=e^{-\ri \hat q} \tilde\tau_{(x)}
\alpha(x,u,v,\hat p, e^{\ri \hat q})\tau_{(x)}^{-1}.
\label{3.113}
\end{align}
Here we use Definition \ref{def:3.12} and the functions on $\bar\cC_x\times\T_n$
introduced in Subsection \ref{subsec:3.2.2}.
\end{lemma}

For the verification of the above identities, we remark that the vector $r$
\eqref{3.61} can be expressed as a smooth function of the complex variables as
\begin{equation}
r(x,\hat p(z))_j=\sqrt{\frac{\sinh(\frac{x}{2})}{\sinh(\frac{nx}{2})}}
\prod_{\substack{k=1\\(k\neq j)}}^nQ_{jk}(x,z),\quad j=1,\dots,n.
\label{3.114}
\end{equation}

With all necessary preparations now done, we state the main new result of the chapter.

\begin{theorem}
\label{thm:3.14}
The image of the smooth map
$\hat{K}\colon\hat{M}_c\to\SL(2n,\C)$ given by the formula
\begin{equation}
\hat{K}(z)=\begin{bmatrix}\kappa(x)\hat\zeta(x,z)^{-1} &\0_n\\\0_n&\1_n\end{bmatrix}
\begin{bmatrix}\hat\gamma(x,z)&\ri e^{\hat p(z)}\\\ri e^{\hat p(z)}&
\hat\gamma(x,z)^\dag\end{bmatrix}
\begin{bmatrix}e^{-v}\1_n&\hat\alpha(x,u,v,z)\\\0_n&e^v\1_n\end{bmatrix}
\label{3.115}
\end{equation}
lies in $\Phi_+^{-1}(\mu)$, intersects every gauge orbit in precisely one point, and $\hat K$ is injective.
The pull-back of the Alekseev-Malkin $2$-form $\omega$ \eqref{3.16} by $\hat K$ is
$\hat\omega_c$ \eqref{3.94}. Consequently, $\pi_\mu\circ\hat K\colon\hat M_c\to M$
is a symplectomorphism, whereby $(\hat M_c,\hat\omega_c)$ provides a model of the
reduced phase space $(M,\omega_M)$ defined in Subsection \ref{subsec:3.1.2}.
\end{theorem}

\begin{proof}
The proof is based upon the identity
\begin{equation}
\hat K(z(\hat p,e^{\ri\hat q}))
=\begin{bmatrix}\kappa(x)\tau_{(x)}\kappa(x)^{-1}&\0_n\\
\0_n&\tilde\tau_{(x)}e^{-\ri\hat q}\end{bmatrix}K(\hat p,e^{\ri\hat q})
\begin{bmatrix}\tilde\tau_{(x)}e^{-\ri\hat q}&\0_n\\
\0_n&\tau_{(x)}\end{bmatrix}^{-1}, \quad \forall (\hat p, e^{\ri \hat q}) \in \bar\cC_x \times \T_n,
\label{3.116}
\end{equation}
which is readily seen to be equivalent to the set of identities displayed in Lemma
\ref{lem:3.13}. It means that $\hat K(z(\hat p,e^{\ri\hat q}))$ is a gauge transform of
$K(\hat p, e^{\ri \hat q})$ in \eqref{3.67}. Indeed, the above transformation of
$K(\hat p,e^{\ri\hat q})$ has the form \eqref{3.31} with
\begin{equation}
\eta_L=c\begin{bmatrix}\kappa(x)\tau_{(x)}\kappa(x)^{-1}&\0_n\\
\0_n&\tilde\tau_{(x)}e^{-\ri\hat q}\end{bmatrix},
\qquad
\eta_R=c\begin{bmatrix}\tilde\tau_{(x)}e^{-\ri\hat q}&\0_n\\
\0_n&\tau_{(x)}\end{bmatrix},
\label{3.117}
\end{equation}
where $c$ is a harmless scalar inserted to ensure $\det(\eta_L)=\det(\eta_R)=1$.
Using \eqref{3.65} and \eqref{3.107}, one can check that $\kappa(x)\tau_{(x)}\kappa(x)^{-1}\hat v(x)=\hat v(x)$
for the vector $\hat v(x)$ in \eqref{3.64}, which implies via the relation \eqref{3.55}
that $(\eta_L,\eta_R)$ belongs to the isotropy group $G_\mu$ \eqref{3.36},
the gauge group acting on $\Phi_+^{-1}(\mu)$.

It follows from Proposition \ref{prop:3.3} and the identity \eqref{3.116} that the set
\begin{equation}
\hat S:=\{\hat K(z)\mid z\in\hat M_c\}
\label{3.118}
\end{equation}
lies in $\Phi_+^{-1}(\mu)$ and intersects every gauge orbit. Since the dense subset
\begin{equation}
\hat S^o:=\{\hat K(z)\mid z\in\hat M_c^o\}
\label{3.119}
\end{equation}
is gauge equivalent to $S^o$ in \eqref{3.79}, we obtain the equality
\begin{equation}
\hat K^\ast(\omega)=\hat\omega_c
\label{3.120}
\end{equation}
by using Theorem \ref{thm:3.11} and equation \eqref{3.98}. More precisely, we here also
utilized that $\hat K^\ast(\omega)$ is (obviously) smooth and $\hat M_c^o$ is dense
in $\hat M_c$.

The only statements that remain to be proved are that the intersection of $\hat S$
with any gauge orbit consists of a single point and that $\hat K$ is injective. (These are already clear
for $\hat S^o\subset \hat S$ and for $\hat K\vert_{\hat M^o_c}$.) Now suppose that
\begin{equation}
\hat K(z')=\begin{bmatrix}\eta_L(1)&\0_n\\\0_n&\eta_L(2)\end{bmatrix}
\hat K(z)\begin{bmatrix}\eta_R(1) &\0_n\\\0_n&\eta_R(2)\end{bmatrix}^{-1}
\label{3.121}
\end{equation}
for some gauge transformation and $z,z'\in\hat M_c$.
Let us observe from the definitions that we can write
\begin{equation}
\begin{bmatrix}\hat\gamma(x,z)&\ri e^{\hat p(z)}\\\ri e^{\hat p(z)}&
\hat\gamma(x,z)^\dag\end{bmatrix}
=D(z)\begin{bmatrix}\cos q(z)&\ri\sin q(z)\\\ri\sin q(z)&\cos q(z)\end{bmatrix}D(z),
\label{3.122}
\end{equation}
where $\sin q(z)=e^{\hat p(z)}$, with $\frac{\pi}{2}\geq q_1>\dots>q_n>0$,
and $D(z)$ is a diagonal unitary matrix of the form
$D(z)=\diag(d_1,\1_{n-1},\bar d_1,\1_{n-1})$.
Then the uniqueness properties of the Iwasawa decomposition of $\SL(2n,\C)$ and the
generalized Cartan decomposition \eqref{3.41} of $\SU(2n)$ allow to establish the
following consequences of \eqref{3.121}. First,
\begin{equation}
\hat p(z)=\hat p(z').
\label{3.123}
\end{equation}
Second, using Lemma \ref{lem:3.6},
\begin{equation}
\begin{bmatrix}\eta_R(1)&\0_n\\\0_n&\eta_R(2)\end{bmatrix}
=\begin{bmatrix}m_2&\0_n\\\0_n&m_1\end{bmatrix}
\label{3.124}
\end{equation}
for some diagonal unitary matrices of the form \eqref{3.73}. Third, we have
\begin{equation}
\hat\alpha(z')=\eta_R(1)\hat\alpha(z)\eta_R(2)^{-1}=m_2\hat\alpha(z)m_1^{-1}.
\label{3.125}
\end{equation}
For definiteness, let us focus on the case $x>0$. Then we see from the definitions
that the components $\hat \alpha_{k+1,k}$ and $\hat\alpha_{1,n}$ depend only on
$\hat p(z)$ and are non-zero. By using this, we find from \eqref{3.125} that
$m_1=m_2=C\1_n$ with a scalar $C$, and therefore
\begin{equation}
\hat\alpha(z')=\hat\alpha(z).
\label{3.126}
\end{equation}
Inspection of the components $(1,2),\dots,(1,n-1)$ of this matrix equality and
\eqref{3.123} permit to conclude that $z'_2=z_2,\dots,z'_{n-1}=z_{n-1}$, respectively.
Then, the equality of the $(2,n)$ entries in \eqref{3.126} gives $z'_1=z_1$ which
used in the $(1,1)$ position implies $z'_n=z_n$. Thus we see that $z'=z$ and
the proof is complete. (Everything written below \eqref{3.125} is quite similar for
$x<0$.)
\end{proof}

\begin{remark}
\label{rem:3.15}
Let us hint at the way the global structure was found. The first point to notice was
that all or some of the phases $e^{\ri \hat q_j}$ cannot encode gauge invariant
quantities if $\hat p$ belongs to the boundary of $\bar \cC_x$, as was already
mentioned in Remark \ref{rem:3.5}. Motivated by \cite{FK11}, then we searched for
complex variables by requiring that a suitable gauge transform of
$K(\hat p,e^{\ri \hat q})$ in \eqref{3.67} should be expressible as a smooth function
of those variables. Given the similarities to \cite{FK11}, only the definition of
$z_n$ was a true open question. After trial and error, the idea came in a flash that
the gauge transformation at issue should be constructed from a transformation that
appears in Lemma \ref{lem:C.2}. Then it was not difficult to find the correct result.
\end{remark}

\begin{remark}
\label{rem:3.16}
Let us elaborate on how the trajectories $\hat p(t)$ corresponding to the flows of the reduced 
free Hamiltonians, arising from $\cH_k$ \eqref{3.90} for $k=1,\ldots, n$, can be obtained.
Recall that for $k=1$ the reduction of $\cH_1$ completes the main Hamiltonian $H$ \eqref{3.1}.
Since $\cH_k(K)= h_k(b_R)$ with $h_k(b) = \frac{1}{2k}\tr (b^\dag b)^k$, the free flow
generated by $\cH_k$ through the initial value $K(0) = g_L(0) b_R^{-1}(0)$ is given by
\eqref{3.21} with $d^R h_k(b) = \ri (b^\dagger b)^k$.
Thus the curve $g_L(t)$ \eqref{3.21} has the form
\begin{equation}
g_L(t)=g_L(0)\exp\!\big(-\ri t\big[\cL(0)^k- \frac{1}{2n}\tr (\cL(0)^k) \1_{2n}\big]\big)
\quad\text{with}\quad
\cL(0)= b_R(0)^\dag b_R(0).
\label{3.127}
\end{equation}
The reduced flow results by the usual projection algorithm.
This starts by picking an initial value $z(0) \in \hat M_c$ and setting 
$K(0)=\hat K(z(0))$
by applying \eqref{3.115}, which directly determines $g_L(0)$ and $b_R(0)$ as well.
Then the map $\cP$ \eqref{3.84} gives rise to $\hat p(t)$ via the decomposition of
$g_L(t)\in\SU(2n)$ as displayed in \eqref{3.41}, that is
\begin{equation}
\hat p(t)=\cP(K(t)).
\label{3.128}
\end{equation}
More explicitly, if $\cD(t)$ stands for the (11) block of $g_L(t)$, then the eigenvalues
of $\cD(t)\cD(t)^\dag$ are
\begin{equation}
\sigma(\cD(t)\cD(t)^\dag)=\{\cos^2q_j(t)\mid j=1,\dots,n\},
\label{3.129}
\end{equation}
from which $\hat p_j(t)$ can be obtained using \eqref{3.53}.
In particular, the `particle positions' evolve according to an `eigenvalue dynamics' similarly to other
many-body systems. This involves the one-parameter group $e^{-\ri t\cL(0)^k}$, where
$\cL(0)$ 
is the initial value of the Lax matrix (cf. \eqref{3.91})
\begin{equation}
\cL(z)=\begin{bmatrix}
e^{2v}\1_n&-e^v\hat\alpha(z)\\
-e^v\hat\alpha(z)^\dag&(e^{-2v}\1_n+\hat\alpha(z)^\dag\hat\alpha(z))
\end{bmatrix},
\label{3.130}
\end{equation}
where we suppressed the dependence of $\hat \alpha$ \eqref{3.106} on the parameters $x,u,v$.
A more detailed characterization of the dynamics will be provided elsewhere.
\end{remark}

\section{Full phase space of the hyperbolic version}
\label{sec:3.4}

In this section, we complete the earlier derivation \cite{Ma15} of the hyperbolic analogue
of the Ruijsenaars type system that we studied in the previous sections.
This hyperbolic version arises as a reduction of the natural free system
on the Heisenberg double of $\SU(n,n)$. The previous analysis by Marshall
focused on a dense open submanifold of the reduced phase space, and here
we describe the full phase space wherein Liouville integrability of the
system holds by construction.

\subsection{Definitions and first steps}
\label{subsec:3.4.1}

The starting point in this case is the group
\begin{equation}
\SU(n,n)=\{g\in\SL(2n,\C)\mid g^\dag\bJ g=\bJ\}
\label{1.131}
\end{equation}
with $\bJ=\diag(\1_n,-\1_n)$. Then consider the open submanifold
$\SL(2n,\C)'\subset\SL(2n,\C)$ consisting of those elements, $K$,
that admit both Iwasawa-like decompositions of the form
\begin{equation}
K=g_Lb_R^{-1}=b_Lg_R^{-1},\qquad g_L,g_R\in\SU(n,n),\quad b_L,b_R\in\SB(2n),
\label{3.132}
\end{equation}
where $\SB(2n)<\SL(2n,\C)$ is the group of upper triangular matrices having
positive entries along the diagonal. Both decompositions are unique and the
constituent factors depend smoothly on $K\in\SL(2n,\C)'$. The manifold
$\SL(2n,\C)'$ inherits a symplectic form $\omega$ \cite{AM94}, which has the
same form as the one seen in \eqref{3.16}. On this symplectic manifold
$(\SL(2n,\C)',\omega)$, which is a symplectic leaf of the Heisenberg double
of the Poisson-Lie group $\SU(n,n)$, one has the pairwise Poisson commuting
Hamiltonians
\begin{equation}
\cH_j(K)=\frac{1}{2j}\tr(K\bJ K^\dag\bJ)^j,\qquad j\in\Z^\ast.
\label{3.133}
\end{equation}
They generate complete flows that can be written down explicitly (see
Section \ref{sec:3.5}). We are concerned with a reduction of these Hamiltonians
based on the symmetry group $G_+\times G_+$, where $G_+$ is the block-diagonal
subgroup \eqref{3.26}. Throughout, we refer to the obvious $2\times 2$
block-matrix structure corresponding to $\bJ$. The action of $G_+\times G_+$
on $\SL(2n,\C)'$ is given by the map
\begin{equation}
G_+\times G_+\times\SL(2n,\C)'\to\SL(2n,\C)'
\label{3.134}
\end{equation}
that works according to
\begin{equation}
(\eta_L,\eta_R,K)\mapsto\eta_LK\eta_R^{-1}.
\label{3.135}
\end{equation}
One can check that this map is well-defined, i.e. $\eta_LK\eta_R^{-1}$ stays in
$\SL(2n,\C)'$, and has the Poisson property with respect to the product Poisson
structure on the left-hand side \cite{Se85,AM94}, where on $G_+$ the standard
Sklyanin bracket is used and the Poisson structure on $\SL(2n,\C)'$ is engendered
by $\omega$. Moreover, this $G_+\times G_+$ action is associated with a momentum
map in the sense of Lu \cite{Lu91}.
The momentum map
\begin{equation}
\Phi_+\colon\SL(2n,\C)'\to G_+^\ast\times G_+^\ast
\label{3.136}
\end{equation}
works the same way as \eqref{3.30} did. Namely, it takes an element $K\in\SL(2n,\C)'$
\eqref{3.132} and maps it to a pair of matrices obtained from $(b_L,b_R)$ by replacing
the off-diagonal blocks by null matrices. The Hamiltonians $\cH_j$ \eqref{3.133} are
invariant with respect to the symmetry group $G_+\times G_+$ and $\Phi_+$ is constant
along their flows.

The general theory \cite{Lu91} ensures that one can now perform Marsden-Weinstein
type reduction. This amounts to imposing the constraint $\Phi_+(K)=\mu=(\mu_L,\mu_R)$
with some constant $\mu\in G_+^\ast\times G_+^\ast$ and then taking the quotient
of $\Phi_+^{-1}(\mu)$ by the corresponding isotropy group, denoted below as $G_\mu$.

We pick the value of the momentum map to be $\mu=(\mu_L,\mu_R)$ \eqref{3.33}.
For simplicity, we now assume that the parameter $x$ in \eqref{3.33} is positive.
The corresponding isotropy group is $G_\mu$ \eqref{3.36}.
It will turn out that the reduced phase space
\begin{equation}
M=\Phi_+^{-1}(\mu)/G_\mu
\label{3.137}
\end{equation}
is a smooth manifold. Our task is to characterize this manifold, which carries the
reduced symplectic form $\omega_M$ defined by the relation
\begin{equation}
\iota_\mu^\ast\omega =\pi_\mu^\ast\omega_M,
\label{3.138}
\end{equation}
where $\iota_\mu\colon\Phi_+^{-1}(\mu)\to\SL(2n,\C)'$ is the tautological
injection and $\pi_\mu\colon\Phi_+^{-1}(\mu)\to M$ is the natural projection.

Consider the following central subgroup $\Z_{2n}$ of $G_+ \times G_+$,
\begin{equation}
\Z_{2n}=\{(w\1_{2n},w\1_{2n})\mid w\in\C,\ w^{2n}=1\},
\label{3.139}
\end{equation}
which acts trivially according to \eqref{3.135} and is contained in $G_\mu$.
Later we shall refer to the factor group
\begin{equation}
\bar{G}_\mu=G_\mu/\Z_{2n}
\label{3.140}
\end{equation}
as the `effective gauge group'. Obviously, we have
$\Phi_+^{-1}(\mu)/G_\mu=\Phi_+^{-1}(\mu)/\bar{G}_\mu$.

Our aim is to obtain a model of the quotient space $M$ by explicitly
exhibiting a global cross-section of the orbits of $G_\mu$ in $\Phi_+^{-1}(\mu)$.
The construction uses the generalized Cartan decomposition of $\SU(n,n)$, which
says that every $g\in \SU(n,n)$ can be written as
\begin{equation}
g=g_+\begin{bmatrix}\cosh q&\sinh q\\ \sinh q&\cosh q\end{bmatrix}h_+,
\label{3.141}
\end{equation}
where $g_+,h_+\in G_+$ and $q=\diag(q_1,\dots,q_n)$ is a real diagonal matrix verifying
\begin{equation}
q_1\geq\dots\geq q_n\geq 0.
\label{3.142}
\end{equation}
The components $q_j$ are uniquely determined by $g$, and yield smooth functions
on the locus where they are all distinct. In what follows we shall often identify
diagonal matrices like $q$ with the corresponding elements of $\R^n$.

As the first step towards describing $M$, we apply the decomposition \eqref{3.141}
to $g_L$ in $K=g_Lb_R^{-1}$ and impose the right-handed momentum constraint
$\pi(b_R)=\mu_R$. It is then easily seen that up to $G_\mu$-transformations every
element of $\Phi_+^{-1}(\mu)$ can be represented in the following form:
\begin{equation}
K=\begin{bmatrix}\rho&\0_n\\\0_n&\1_n\end{bmatrix}
\begin{bmatrix}\cosh q& \sinh q\\ \sinh q&\cosh q\end{bmatrix}
\begin{bmatrix}e^{-v}\1_n&\alpha\\\0_n&e^v\1_n\end{bmatrix}.
\label{3.143}
\end{equation}
Here $\rho\in\SU(n)$ and $\alpha$ is an $n\times n$ complex matrix. Referring to
the $2\times 2$ block-matrix notation, we introduce $\Omega=K_{22}$ and record
from \eqref{3.143} that
\begin{equation}
\Omega=(\sinh q)\alpha+e^v\cosh q.
\label{3.144}
\end{equation}
Just as in \eqref{3.45}, it proves to be advantageous to seek for $\Omega$ in the polar-decomposed form,
\begin{equation}
\Omega=\Lambda T,
\label{3.145}
\end{equation}
where $T\in\UN(n)$ and $\Lambda$ is a Hermitian, positive semi-definite matrix.

The next step is to implement the left-handed momentum constraint $\pi(b_L)=\mu_L$
by writing $K=b_Lg_R^{-1}$ with
\begin{equation}
b_L=\begin{bmatrix}e^u\nu(x)&\chi\\\0_n&e^{-u}\1_n\end{bmatrix},
\label{3.146}
\end{equation}
where $\chi$ is an unknown $n\times n$ matrix. Then we inspect the components of
the $2\times 2$ block-matrix identity
\begin{equation}
K\bJ K^\dag=b_L\bJ b_L^\dag,
\label{3.147}
\end{equation}
which results by substituting $K$ from \eqref{3.143}. We find that the (22) component
of this identity is equivalent to
\begin{equation}
\Omega\Omega^\dag=\Lambda^2=e^{-2u}\1_n+e^{-2v}(\sinh q)^2.
\label{3.148}
\end{equation}
This uniquely determines $\Lambda$ in terms of $q$ and also shows that
$\Lambda$ is invertible. As in the trigonometric case, the condition
$u+v\neq 0$ ensures that
\begin{equation}
q_n>0,
\label{3.149}
\end{equation}
and therefore $\sinh q$ is an invertible diagonal matrix.

By using the above relations, it is simple algebra to convert the (12)
and the (21) components of the identity \eqref{3.147} into the equation
\begin{equation}
\chi=\rho(\sinh q)^{-1}[e^{-u}\cosh q-e^{u+v}\Omega^\dag].
\label{3.150}
\end{equation}
Finally, the (11) entry of the identity \eqref{3.147} translates into the following
crucial equation (which is the hyperbolic analogue of \eqref{3.51}):
\begin{equation}
\rho(\sinh q)^{-1}T^\dag(\sinh q)^2T(\sinh q)^{-1}\rho^\dag=\nu(x)\nu(x)^\dag.
\label{3.151}
\end{equation}
This is to be satisfied by $q$ subject to \eqref{3.142}, \eqref{3.149} and $T\in\UN(n)$,
$\rho\in\SU(n)$. After finding $q$, $T$ and $\rho$, one can reconstruct $K$
\eqref{3.143} by applying the formulas derived above.

From our viewpoint, a key observation is that \eqref{3.151} coincides completely with
equation (5.7) in the paper \cite{FK11}, where its general solution was found.
The correspondence between the notations used here and in \cite{FK11} is
\begin{equation}
(\rho,T,\sinh q)\Longleftrightarrow(k_L,k_R^\dag,e^{\hat p}).
\label{3.152}
\end{equation}
For this reason, we introduce the new variable $\hat{p}\in\R^n$ by the definition
\begin{equation}
\sinh q_k=e^{\hat p_k},\quad k=1,\dots,n.
\label{3.153}
\end{equation}
Because of \eqref{3.142} and \eqref{3.149}, the variables $\hat{p}_k$ satisfy
\begin{equation}
\hat{p}_1\geq\dots\geq\hat{p}_n.
\label{3.154}
\end{equation}

We do not see an a priori reason why the very different reduction procedures led
to the same equation \eqref{3.151} here and in \cite{FK11}. However, we are going to
take full advantage of this situation. We note that essentially every formula written in
this section appears in \cite{Ma15} as well (with slightly different notations),
but in Marshall's work the previously obtained results about the solutions of
\eqref{3.151} were not used.

\subsection{The reduced phase space}
\label{subsec:3.4.2}

The statement of Proposition \ref{prop:3.19} characterizes a submanifold of $M$
\eqref{3.137}, which was erroneously claimed in \cite{Ma15} to be equal to $M$.
After describing this `local picture', we shall present a globally valid model
of $M$.

\subsubsection{The local picture}
\label{subsubsec:3.4.2.1}

By applying results of \cite{FK11,FK12} in the same way as
we did in the trigonometric case, one can prove the following lemma.

\begin{lemma}
\label{lem:3.17}
The constraint surface $\Phi_+^{-1}(\mu)$ contains an element of the form
\eqref{3.143} if and only if $\hat p$ defined by \eqref{3.153} lies in the closed
polyhedron
\begin{equation}
\bar{\cC}'_x=\{\hat{p}\in\R^n\mid\hat{p}_k-\hat{p}_{k+1}\geq x/2\ (k=1,\dots,n-1)\}.
\label{3.155}
\end{equation}
\end{lemma}

The polyhedron $\bar{\cC}'_x$ is the closure of its interior, $\cC'_x$, defined
by strict inequalities. We note that in \cite{Ma15} the elements of the boundary
$\bar{\cC}'_x\setminus\cC'_x$ were omitted.

For any fixed $\hat{p}\in\bar\cC'_x$, one can write down the solutions of \eqref{3.151}
for $T$ and $\rho$ explicitly \cite{FK11}. By inserting those into the formula
\eqref{3.143}, using the relations \eqref{3.144}, \eqref{3.145}, \eqref{3.148} to determine
the matrix $\alpha$, one arrives at the next lemma. It refers to the $n\times n$
real matrices  $\theta(x,\hat p)$, $\zeta(x,\hat p)$, $\kappa(x)$ displayed in
\eqref{3.61}-\eqref{3.65}, which belong to the group $\SO(n)$.

\begin{proposition}
\label{prop:3.18}
For any parameters $u,v,x$ subject to $u+v\neq 0$, $x>0$, and variables
$\hat p\in\bar\cC'_x$ and $e^{\ri\hat q}$ from the $n$-torus $\T_n$,
define the matrix
\begin{equation}
K(\hat p,e^{\ri \hat q})
=\begin{bmatrix}\rho &\0_n\\\0_n&\1_n\end{bmatrix}
\begin{bmatrix}\sqrt{\1_n+e^{2\hat p}} & e^{\hat p}\\
e^{\hat p}&\sqrt{\1_n+e^{2\hat p}}\end{bmatrix}
\begin{bmatrix}e^{-v}\1_n&\alpha\\\0_n&e^v\1_n\end{bmatrix}
\label{3.156}
\end{equation}
by employing
\begin{equation}
\rho=\rho(x,\hat p)=\kappa(x)\zeta(x,\hat p)^{-1}
\label{3.157}
\end{equation}
and
\begin{equation}
\alpha=\alpha(x,u,v,\hat p, e^{\ri \hat q})
=e^{\ri\hat q}\sqrt{e^{-2u}e^{-2\hat p}+ e^{-2v}\1_n}\,\theta(x,\hat p)^{-1}
-e^v\sqrt{e^{-2\hat p}+\1_n}.
\label{3.158}
\end{equation}
Then $K(\hat p,e^{\ri\hat q})$ resides in the constraint surface $\Phi_+^{-1}(\mu)$
and the set
\begin{equation}
S=\{K(\hat p,e^{\ri\hat q})\mid(\hat p,e^{\ri\hat q})\in\bar\cC'_x\times\T_n\}
\label{3.159}
\end{equation}
intersects every orbit of $G_\mu$ in $\Phi_+^{-1}(\mu)$.
\end{proposition}

By arguing verbatim along the lines of the previous section, and referring to
\cite{Ma15} for the calculation of the reduced symplectic form, one can establish
the validity of the subsequent proposition.

\begin{proposition}
\label{prop:3.19}
The effective gauge group $\bar G_\mu$ \eqref{3.140} acts freely
on $\Phi_+^{-1}(\mu)$ and thus the quotient space $M$
\eqref{3.137} is a smooth manifold. The restriction of the natural projection
$\pi_\mu\colon\Phi_+^{-1}(\mu)\to M$ to
\begin{equation}
S^o=\{K(\hat p,e^{\ri\hat q})\mid(\hat p,e^{\ri\hat q})\in\cC'_x\times\T_n\}
\label{3.160}
\end{equation}
gives rise to a diffeomorphism between $\cC'_x\times\T_n$ and the open, dense
submanifold of $M$ provided by $\pi_\mu(S^o)$. Taking $S^o$ as model of
$\pi_\mu(S^o)$, the corresponding restriction of the reduced symplectic
form $\omega_M$ becomes the Darboux form
\begin{equation}
\omega_{S^o}=\sum_{k=1}^nd\hat q_k\wedge d\hat p_k.
\label{3.161}
\end{equation}
\end{proposition}

\begin{remark}
\label{rem:3.20}
In the formula \eqref{3.156} $K(\hat p,e^{\ri\hat q})$ appears in the decomposed form
$K=g_Lb_R^{-1}$ and it is not immediately obvious that it belongs to $\SL(2n,\C)'$,
i.e. that it can be decomposed alternatively as $b_Lg_R^{-1}$.
However, by defining $b_L(\hat p, e^{\ri \hat q})\in \SB(2n)$ by the formula
\eqref{3.146} using $\chi$ in \eqref{3.150} with the change of variables
$\sinh q=e^{\hat p}$, the matrix $\rho$ as given above, and
$T=e^{\ri\hat q}\theta(x,\hat p)^{-1}$ that enters \eqref{3.156}, we can verify that for
these elements $g_R^{-1}=b_L^{-1}K$ satisfies the defining relation of $\SU(n,n)$
\eqref{1.131}, as required. The reader may perform this verification, which relies
only on the constraint equations displayed in Subsection \ref{subsec:3.4.1}.
\end{remark}

\subsubsection{The global picture}
\label{subsubsec:3.4.2.2}

The train of thought leading to the construction below can be outlined as follows.
Proposition \ref{prop:3.19} tells us, in particular, that any $G_\mu$-orbit passing
through $S^o$ intersects $S^o$ in a single point. Direct inspection shows that the
analogous statement is false for $S\setminus S^o$, which corresponds to
$(\bar\cC'_x\setminus\cC'_x)\times\T_n$ in a one-to-one manner.
Thus a global model of $M$ should result by identifying those points of
$S\setminus S^o$ that lie on the same $G_\mu$-orbit. By using the bijective map
from $\bar\cC'_x\times\T_n$ onto $S$ given by the formula \eqref{3.156}, the desired
identification will be achieved by constructing such complex variables out of
$(\hat p,e^{\ri\hat q})\in\bar\cC'_x\times\T_n$ that coincide precisely for gauge
equivalent elements of $S$.

Turning to the implementation of the above plan, we introduce the space of complex
variables
\begin{equation}
\hat M_c=\C^{n-1}\times\C^\times,\qquad(\C^\times=\C\setminus\{0\}),
\label{3.162}
\end{equation}
carrying the symplectic form
\begin{equation}
\hat\omega_c=\ri\sum_{j=1}^{n-1}dz_j\wedge d\bar z_j
+\frac{\ri dz_n\wedge d\bar z_n}{2z_n\bar z_n}.
\label{3.163}
\end{equation}
We also define the surjective map
\begin{equation}
\hat\cZ_x\colon\bar\cC'_x\times\T_n\to\hat M_c,\quad
(\hat p,e^{\ri\hat q})\mapsto z(\hat p,e^{\ri\hat q})
\label{3.164}
\end{equation}
by setting
\begin{equation}
\begin{split}
z_j(\hat p,e^{\ri\hat q})
&=(\hat p_j-\hat p_{j+1}-x/2)^{\tfrac{1}{2}}\prod_{k=j+1}^ne^{\ri\hat q_k},
\quad j=1,\dots,n-1,\\
z_n(\hat p,e^{\ri\hat q})&=e^{-\hat p_1}\prod_{k=1}^ne^{\ri\hat q_k}.
\end{split}
\label{3.165}
\end{equation}
The restriction $\cZ_x$ of $\hat\cZ_x$ to $\cC'_x\times\T_n$ is a diffeomorphism
onto the open subset
\begin{equation}
\hat M_c^o=\bigg\{z\in\hat M_c\,\bigg\vert\,\prod_{j=1}^{n-1} z_j\neq 0\bigg\},
\label{3.166}
\end{equation}
and it verifies the relation
\begin{equation}
\cZ_x^\ast\hat\omega_c=\sum_{k=1}^nd\hat q_k\wedge d\hat p_k.
\label{3.167}
\end{equation}
Thus we manufactured a change of variables $\cC'_x\times\T_n\longleftrightarrow\hat M_c^o$.
The inverse $\cZ_x^{-1}\colon\hat M_c^o\to\cC'_x\times\T_n$ involves the functions
\begin{equation}
\hat p_1(z)=-\log|z_n|,\quad
\hat p_j(z)=-\log|z_n|-\sum_{k=1}^{j-1}(|z_k|^2+x/2)\quad (j=2,\dots,n).
\label{3.168}
\end{equation}
These extend smoothly to $\hat M_c$ wherein $\hat M_c^o$ sits as a dense
submanifold.

Now we state a lemma, which is a simple adaptation from \cite{FK11,Ru95}.

\begin{lemma}
\label{lem:3.21}
By using the shorthand $\sigma_j=\prod_{k=j+1}^ne^{\ri\hat q_k}$ for $j=1,\dots, n-1$ (cf.~\eqref{3.165}), let us define
\begin{equation}
\sigma_+( e^{\ri \hat q})=\diag(\sigma_1,\dots,\sigma_{n-1},1)
\quad\text{and}\quad
\sigma_-(e^{\ri \hat q})=\diag(1,\sigma_1^{-1},\dots,\sigma_{n-1}^{-1}).
\label{3.169}
\end{equation}
Then there exist unique smooth functions $\hat\zeta(x,z)$, $\hat\theta(x,z)$ and
$\hat{\alpha}(x,u,v,z)$ of $z\in\hat{M}_c$ that satisfy the following identities
for any $(\hat p,e^{\ri\hat q})\in\bar\cC'_x\times\T_n$:
\begin{align}
\hat\zeta(x,z(\hat p,e^{\ri \hat q}))
&=\sigma_+(e^{\ri \hat q})\zeta(x,\hat p)\sigma_+( e^{\ri \hat q})^{-1},
\label{3.170}\\
\hat\theta(x,z(\hat p,e^{\ri \hat q}))
&=\sigma_+(e^{\ri \hat q})\theta(x,\hat p) \sigma_-(e^{\ri \hat q}),
\label{3.171}\\
\hat\alpha(x,u,v, z(\hat p,e^{\ri \hat q}))
&=\sigma_+(e^{\ri \hat q})
\alpha(x,u,v,\hat p, e^{\ri \hat q}) \sigma_+(e^{\ri \hat q})^{-1}.
\label{3.172}
\end{align}
Here we refer to the functions on $\bar\cC'_x\times\T_n$
displayed in equations \eqref{3.61}-\eqref{3.65}, and \eqref{3.158}.
\end{lemma}

The explicit formulas of the functions on $\hat M_c$ that appear in the above
identities are easily found by first determining them on $\hat M_c^o$ using the
change of variables $\cZ_x$, and then noticing that they automatically extend to
$\hat M_c$. The expressions of the functions $\hat\zeta$ and $\hat\theta$, which
depend only on $z_1,\dots,z_{n-1}$, are the same as given in \cite[Definition 3.3]{FK11}.
(For most purposes the above definitions and the formulas \eqref{3.61}-\eqref{3.65} suffice.) As for $\hat\alpha$, by defining
\begin{equation}
\Delta(z)=\diag(z_n,e^{-\hat p_2(z)},\dots,e^{-\hat p_n(z)})
\label{3.173}
\end{equation}
we have
\begin{equation}
\hat\alpha(x,u,v,z)=\sqrt{e^{-2v}e^{2\hat p(z)}+e^{-2u}\1_n}\,
\Delta(z)\hat \theta(x,z)^{-1}-e^v\sqrt{e^{-2\hat p(z)}+\1_n}
\label{3.174}
\end{equation}
that satisfies relation \eqref{3.172} due to the identity
\begin{equation}
\Delta(z(\hat p, e^{\ri \hat q}))=e^{-\hat p}e^{\ri \hat q}\sigma_+(e^{\ri \hat q})\sigma_-(e^{\ri \hat q}),
\qquad\forall (\hat p,e^{\ri\hat q})\in\bar\cC'_x\times\T_n.
\label{3.175}
\end{equation}

With these preparations at hand, we can formulate the main result of this section.

\begin{theorem}
\label{thm:3.22}
Define the smooth map $\hat K\colon\hat M_c\to\SL(2n,\C)'$ by the formula
\begin{equation}
\hat K(z)
=\begin{bmatrix}\kappa(x)\hat\zeta(x,z)^{-1}&\0_n\\\0_n&\1_n\end{bmatrix}
\begin{bmatrix}\sqrt{\1_n+e^{2\hat p(z)}}&e^{\hat p(z)}\\
e^{\hat p(z)}&\sqrt{\1_n+e^{2\hat p(z)}}\end{bmatrix}
\begin{bmatrix}e^{-v}\1_n&\hat\alpha(x,u,v,z)\\\0_n&e^v\1_n\end{bmatrix}.
\label{3.176}
\end{equation}
The image of $\hat{K}$ belongs to the submanifold $\Phi_+^{-1}(\mu)$ and
the induced mapping $\pi_\mu \circ \hat K$, obtained by using the natural
projection $\pi_\mu\colon\Phi_+^{-1}(\mu)\to M=\Phi^{-1}(\mu)/G_\mu$,
is a symplectomorphism between $(\hat M_c,\hat\omega_c)$, defined by \eqref{3.162}, \eqref{3.163},
and the reduced phase space $(M,\omega_M)$.
\end{theorem}

\begin{proof}
We start by pointing out that for any
$(\hat p,e^{\ri\hat q})\in\bar\cC'_x\times\T_n$ the identity
\begin{equation}
\hat K(z(\hat p,e^{\ri\hat q}))
=\begin{bmatrix}\kappa(x) \sigma_+(e^{\ri \hat q})\kappa(x)^{-1}&\0_n\\
\0_n& \sigma_+(e^{\ri \hat q}) \end{bmatrix}K(\hat p,e^{\ri\hat q})
\begin{bmatrix} \sigma_+(e^{\ri \hat q}) &\0_n\\
\0_n&\ \sigma_+(e^{\ri \hat q}) \end{bmatrix}^{-1}
\label{3.177}
\end{equation}
is equivalent to the identities listed in Lemma \ref{lem:3.21}.
We see from this that $\hat K(z(\hat p,e^{\ri\hat q}))$ is a $G_\mu$-transform
of $K(\hat p, e^{\ri \hat q})$ \eqref{3.156}, and thus $\hat K(z)$ belongs to
$\Phi_+^{-1}(\mu)$. Indeed, the right-hand side of \eqref{3.177} can be written
as $\eta_L K(\hat p,e^{\ri\hat q}) \eta_R^{-1}$ with
\begin{equation}
\eta_L=c\begin{bmatrix}\kappa(x)\sigma_+(e^{\ri \hat q})\kappa(x)^{-1}&\0_n\\
\0_n&\sigma_+(e^{\ri \hat q})\end{bmatrix},
\qquad
\eta_R=c\begin{bmatrix} \sigma_+(e^{\ri \hat q})&\0_n\\
\0_n&\sigma_+(e^{\ri \hat q})\end{bmatrix},
\label{3.178}
\end{equation}
where $c$ is a scalar ensuring $\det(\eta_L)=\det(\eta_R)=1$, and one can check
that this $(\eta_L,\eta_R)$ lies in the isotropy group $G_\mu$. Indeed, both
$\kappa(x)$ and $\zeta(x,\hat p)$ are orthogonal matrices of determinant 1.
The main feature of $\kappa(x)$ is that the matrix
$\kappa(x)^{-1}\nu(x)\nu(x)^\dag \kappa(x)$ (with $\nu(x)$ in \eqref{3.34}) is diagonal.
This implies that $\eta_L(1)=\kappa(x)\tau\kappa(x)^{-1}\in\UN(n)$
satisfies \eqref{3.38} for any $\tau\in\T_n$.

To proceed further, we let $\hat K_o$ denote the restriction
of $\hat K$ to the dense open subset $\hat M_c^o$ and also let
$K_o\colon\cC'_x\times\T_n\to\SL(2n,\C)'$ denote the map
defined by the corresponding restriction of the formula \eqref{3.156}.
Notice that, in addition to \eqref{3.138}, we have the relations
\begin{equation}
\pi_\mu\circ\hat K_0=\pi_\mu\circ K_o \circ\cZ_x^{-1}
\quad\hbox{and}\quad
(\pi_\mu\circ K_o)^\ast\omega_M=\sum_{k=1}^nd\hat q_k\wedge d\hat p_k,
\label{3.179}
\end{equation}
which follow from \eqref{3.177} and the last sentence of Proposition \ref{prop:3.19}.
By using \eqref{3.167} (together with $\hat K_o=\iota_\mu\circ\hat K_o$ and
$K_o=\iota_\mu\circ K_o$) the above relations imply the restriction of the equality
\begin{equation}
(\pi_\mu\circ\hat K)^\ast\omega_M=\hat\omega_c
\label{3.180}
\end{equation}
on $\hat M_c^o$. This equality is then valid on the full $\hat M_c$ since the
$2$-forms concerned are smooth.

It is a direct consequence of \eqref{3.177} and Proposition \ref{prop:3.18} that
$\pi_\mu\circ\hat{K}$ is surjective. Since, on account of \eqref{3.180}, it is a
local diffeomorphism, it only remains to demonstrate that the map $\pi_\mu\circ\hat{K}$
is injective. The relation $\pi_\mu(\hat K(z))=\pi_\mu(\hat K(z'))$ for $z,z'\in\hat M_c$
requires that
\begin{equation}
\hat K(z')=\begin{bmatrix}\eta_L(1)&\0_n\\\0_n&\eta_L(2)\end{bmatrix}
\hat K(z)\begin{bmatrix}\eta_R(1) &\0_n\\\0_n&\eta_R(2)\end{bmatrix}^{-1}
\label{3.181}
\end{equation}
for some $(\eta_L,\eta_R)\in G_\mu$. Supposing that \eqref{3.181} holds,
application of the decomposition $\hat{K}(z)=g_L(z)b_R(z)^{-1}$ to the
formula \eqref{3.176} implies that
\begin{equation}
\hat\alpha(z')=\eta_R(1)\hat\alpha(z)\eta_R(2)^{-1}
\label{3.182}
\end{equation}
and
\begin{equation}
g_L(z')=\eta_Lg_L(z)\eta_R^{-1}.
\label{3.183}
\end{equation}
The matrices on the two sides of \eqref{3.183} appear in the form \eqref{3.141},
and standard uniqueness properties of the constituents in this generalized Cartan
decomposition now imply that
\begin{equation}
\hat p(z')=\hat p(z)
\label{3.184}
\end{equation}
and
\begin{equation}
\eta_R(1)=\eta_R(2)=m\in\T_n.
\label{3.185}
\end{equation}
We continue by looking at the $(k+1,k)$ components of the equality \eqref{3.182}
for $k=1,\dots,n-1$ using that $\hat\alpha_{k+1,k}$ depends on $z$ only through
$\hat p(z)$ and it never vanishes. (This follows from \eqref{3.173}-\eqref{3.174}
by utilizing that $\hat \theta(x,z)_{k,k+1} =\theta(x,\hat p(z))_{k,k+1}$ by
\eqref{3.171}, which is nonzero for each $\hat p(z)\in\bar\cC'_x$ as seen from
\eqref{3.59}.) Putting \eqref{3.185} into \eqref{3.182}, we obtain that $m=C\1_n$
with a scalar $C$, and therefore
\begin{equation}
\hat\alpha(z')=\hat\alpha(z).
\label{3.186}
\end{equation}
The rest is an inspection of this matrix equality. In view of \eqref{3.184} and
the forms of $\Delta(z)$ \eqref{3.173} and $\hat \alpha(z)$ \eqref{3.174},
the last column of the equality \eqref{3.186} entails that
\begin{equation}
\hat \theta(x,z)_{nk} =\hat \theta(x,z')_{nk},\quad k=2,\dots,n,
\label{3.187}
\end{equation}
where we re-instated the dependence on $x$ that was suppressed above.
One can check directly from the formulas \eqref{3.165}, \eqref{3.171} and
\eqref{3.59}, \eqref{3.60} that
\begin{equation}
\hat\theta(x,z)_{nk}=\bar z_{k-1}F_k(x,\hat p(z)),\quad k=2,\dots,n,
\label{3.188}
\end{equation}
where $F_k(x,\hat p(z))$ is a smooth, strictly positive function.
Hence we obtain that $z_j=z_j'$ for $j=1,\dots,n-1$.
With this in hand,
since the variable $z_n$ appears only in $\Delta(z)$, we conclude from \eqref{3.186}
that $\Delta(z)=\Delta(z')$. This plainly implies that $z_n=z'_n$, whereby
the proof is complete.

We note in passing that by continuing the above line of arguments the
free action of $G_\mu$ is easily confirmed. Indeed, for $z'=z$ \eqref{3.183}
also implies, besides \eqref{3.185}, the equalities $\eta_L(2)=m$ and
$\eta_L(1)\kappa(x)\hat\zeta(x,z)^{-1}=\kappa(x)\hat\zeta(x,z)^{-1}m$.
Since $m=C\1_n$, as was already established, we must have
$(\eta_L,\eta_R)=C(\1_{2n},\1_{2n})\in\Z_{2n}$ \eqref{3.139}.
By using that the image of $\hat K$ intersects every $G_\mu$-orbit, we can conclude that
$\bar G_\mu$ \eqref{3.140} acts freely on $\Phi_+^{-1}(\mu)$.
\end{proof}

\begin{remark}
\label{rem:3.23}
Observe from Theorem \ref{thm:3.22} that $\hat S=\{\hat K(z)\mid z\in\hat M_c\}$
is a global cross-section for the action of $G_\mu$ on $\Phi_+^{-1}(\mu)$.
Hence $\hat S$ carrying the pull-back of $\omega$ as well as
$(\hat M_c, \hat \omega_c)$ yield globally valid models of the reduced phase
space $(M,\omega_M)$. The submanifold of $\hat S$ corresponding to $\hat M_c^o$
\eqref{3.166} is gauge equivalent to $S^o$ \eqref{3.160} that features in
Proposition \ref{prop:3.19}.
\end{remark}

\section{Discussion}
\label{sec:3.5}

In this chapter we derived a deformation of the trigonometric $\BC_n$ Sutherland system
by means of Hamiltonian reduction of a free system on the Heisenberg double of
$\SU(2n)$. Our main result is the global characterization of the reduced phase space
given by Theorem \ref{thm:3.14}. The Liouville integrability of our system holds on
this phase space, wherein the reduced free flows are complete. These flows can be
obtained by the usual projection method applied to the original free flows described
in Section \ref{sec:3.1}.

The local form of our reduced `main Hamiltonian' \eqref{3.1} is similar to the
Hamiltonian derived in \cite{Ma15}, which deforms the hyperbolic $\BC_n$ Sutherland
system. However, besides a sign difference corresponding to the difference of the
undeformed Hamiltonians, the local domain of our system, $\cC_x\times\T_n$ in
\eqref{3.4}, is different from the local domain appearing in \cite{Ma15}, which in effect
has the form $\cC_x'\times\T_n$ with the open polyhedron\footnote{The notational
correspondence with \cite{Ma15} is:
$(q,p,\alpha,x,y)\leftrightarrow(\hat p,\hat q,e^{-\frac{x}{2}},e^{-v},e^{-u})$.}
\begin{equation}
\cC'_x:=\{\hat p\in\R^n\mid\hat p_k-\hat p_{k+1}>|x|/2\
(k=1,\dots,n-1)\}.
\label{3.189}
\end{equation}
We here wish to point out once more that $\cC_x'\times\T_n$ is \emph{not} the
full reduced phase space that arises from the reduction considered in \cite{Ma15}.
In fact, similarly to our case, the constraint surface contains a submanifold of the form
$\bar\cC_x'\times\T_n$ in the case of \cite{Ma15}, where $\bar\cC_x'$ is the closure of
$\cC_x'$. Then a global model of the reduced phase space can be constructed by
introducing complex variables suitably accommodating the procedure that we utilized
in Subsection \ref{subsec:3.3.4}. In Section \ref{sec:3.4} we clarified the global
structure of the reduced phase space $M$ \eqref{3.137}, and thus completed the
previous analysis \cite{Ma15} that dealt with the submanifold parametrized by
$\cC'_x\times\T_n$. In terms of the model $\hat M_c$ \eqref{3.162} of $M$, the
complement of the submanifold in question is simply the zero set of the product
of the complex variables. The phase space $\hat M_c$ and the embedding of
$\cC'_x\times\T_n$ into it coincides with what occurs for the so-called
$\widetilde{\mathrm{III}}$-system of Ruijsenaars \cite{Ru95,FK11}, which is the
action-angle dual of the standard trigonometric Ruijsenaars-Schneider system.
This circumstance is not surprising in light of the fact \cite{Ma15} that the
reduced `main Hamiltonian' arising from $\cH_1$ \eqref{3.133} is a
$\widetilde{\mathrm{III}}$-type Hamiltonian coupled to external fields.
We display this Hamiltonian below after exhibiting the corresponding Lax matrices.

The unreduced free Hamiltonians $\cH_j$, $j\in\Z^\ast$ \eqref{3.133}, mentioned in
Section \ref{sec:3.4}, can be written alternatively as
\begin{equation}
\cH_j(K)=\frac{1}{2j}\tr(K\bJ K^\dag\bJ)^j=\frac{1}{2j}\tr(K^\dag\bJ K\bJ)^j.
\label{3.190}
\end{equation}

One can verify (for example by using the standard $r$-matrix formula of the Poisson
bracket on the Heisenberg double \cite{Se85}) that the Hamiltonian flow generated by
$\cH_j$ reads
\begin{equation}
\begin{split}
K(t_j)&=\exp\bigg[\ri t_j\bigg(
(K(0)\bJ K(0)^\dag\bJ)^j
-\frac{1}{2n}\tr(K(0)\bJ K(0)^\dag\bJ)^j\1_{2n}\bigg)\bigg]K(0)\\
&=K(0)\exp\bigg[\ri t_j\bigg((\bJ K(0)^\dag\bJ K(0))^j
-\frac{1}{2n}\tr(\bJ K(0)^\dag\bJ K(0))^j\1_{2n}\bigg)\bigg].
\end{split}
\label{3.191}
\end{equation}

Since the exponentiated elements reside in the Lie algebra $\su(n,n)$,
these alternative formulas show that the flow stays in $\SL(2n,\C)'$,
as it must, and imply that the building blocks $g_L$ and $g_R$ of
$K=b_Lg_R^{-1}=g_Lb_R^{-1}$ follow geodesics on $\SU(n,n)$, while
$b_L$ and $b_R$ provide constants of motion.
Equivalently, the last statement means that
\begin{equation}
K\bJ K^\dag\bJ=b_L\bJ b_L^\dag\bJ
\quad\text{and}\quad
K^\dag\bJ K\bJ=(b_R^{-1})^\dag\bJ b_R^{-1}\bJ
\label{3.192}
\end{equation}
stay constant along the unreduced free flows.

To elaborate the reduced Hamiltonians, note that for an element $K$ of the form
\eqref{3.143} we have
\begin{equation}
(b_R^{-1})^\dag\bJ b_R^{-1}\bJ
=\begin{bmatrix}
e^{-2v} \1_n &-e^{-v} \alpha\\
e^{-v}\alpha^\dag&e^{2v}\1_n-\alpha^\dag\alpha
\end{bmatrix}.
\label{3.193}
\end{equation}
By using this, as explained in Appendix \ref{sec:C.5}, one can prove that on
$\Phi_+^{-1}(\mu)$ the Hamiltonians $\cH_j$ can be written (for all $j$),
up to additive constants, as linear combinations of the expressions
\begin{equation}
h_k=\tr(\alpha^\dag\alpha)^k,\qquad k=1,\dots,n.
\label{3.194}
\end{equation}
Since in this way the Hermitian matrix $L=\alpha\alpha^\dag$ generates the
commuting reduced Hamiltonians, it provides a Lax matrix for the reduced system.
By inserting $\alpha$ from \eqref{3.158}, we obtain the explicit formula
\begin{equation}
\begin{split}
L(\hat p,e^{\ri\hat q})=
&(e^{2v}+e^{-2u})e^{-2\hat p}+(e^{2v}+e^{-2v})\1_n\\
&-\sqrt{e^{-2u}e^{-2\hat p}+ e^{-2v}\1_n}e^{\ri\hat q}\theta(x,\hat p)^{-1}
e^v\sqrt{e^{-2\hat p}+\1_n}\\
&- e^v\sqrt{e^{-2\hat p}+\1_n}\theta(x,\hat p)e^{-\ri\hat q}
\sqrt{e^{-2u}e^{-2\hat p}+ e^{-2v}\1_n}.
\end{split}
\label{3.195}
\end{equation}
On the other hand, the Lax matrix of Ruijsenaars's $\widetilde{\mathrm{III}}$-system
can be taken to be \cite{Ru95,FK11}
\begin{equation}
\tilde L(\hat p,e^{\ri\hat q})
=e^{\ri\hat q}\theta(x,\hat p)^{-1}+\theta(x,\hat p)e^{-\ri\hat q}.
\label{3.196}
\end{equation}
The similarity of the structures of these Lax matrices as well as the presence
of the external field couplings in \eqref{3.195} is clear upon comparison.
The extension of the Lax matrix $\alpha\alpha^\dag$ \eqref{3.195} to the
full phase space $M\simeq\hat{M}_c$ is of course given by $\hat\alpha\hat\alpha^\dag$
by means of \eqref{3.174}.

The main reduced Hamiltonian found in \cite{Ma15} reads as follows:
\begin{multline}
 \cH_1(K(\hat p,e^{\ri\hat q}))=
-\frac{e^{-2u}+e^{2v}}{2}\sum_{j=1}^ne^{-2\hat p_j}+\\
\qquad+\sum_{j=1}^n\cos(\hat q_j)\big[1+ (1+e^{2(v-u)})e^{-2\hat p_j}
+e^{2(v-u)}e^{-4\hat p_j}\big]^{\tfrac{1}{2}}\prod_{\substack{k=1\\(k\neq j)}}^n
\bigg[1-\frac{\sinh^2\big(\frac{x}{2}\big)}{\sinh^2(\hat p_j-\hat p_k)}
\bigg]^{\tfrac{1}{2}}.
\label{3.197}
\end{multline}
Liouville integrability holds since the functional independence of the involutive family
obtained by reducing $\cH_1,\dots,\cH_n$ \eqref{3.190} is readily established and the
projections of the free flows \eqref{3.191} to $M$ are automatically complete. Similarly
to its trigonometric analogue the Hamiltonian \eqref{3.197} can be identified as an
Inozemtsev type limit of a specialization of van Diejen's $5$-coupling deformation of
the hyperbolic $\BC_n$ Sutherland Hamiltonian \cite{vD94}. This fact suggests that it
should be possible to extract the local form of dual Hamiltonians from \cite{vDE16}
and references therein, which contain interesting results about closely related quantum
mechanical systems and their bispectral properties. Indeed, in several examples,
classical Hamiltonians enjoying action-angle duality correspond to bispectral pairs
of Hamiltonian operators after quantization.

Throughout the text we assumed that $n>1$, but we now note that the reduced system
can be specialized to $n=1$ and the reduction procedure works in this case as well.
The assumption was made merely to save words. The formalism actually simplifies for
$n=1$ since the Poisson structure on $G_+=\mathrm{S}(\UN(1)\times\UN(1))<\SU(2)$ is
trivial.

As explained in Appendix \ref{sec:C.1}, the Hamiltonian \eqref{3.1} is a singular limit
of a specialization of the trigonometric van Diejen Hamiltonian \cite{vD94}, which (in
addition to the deformation parameter) contains $5$ coupling constants. As a result,
at least classically, van Diejen's system can be degenerated into the trigonometric
$\BC_n$ Sutherland system either directly, as described in \cite{vD94}, or in two stages,
going through our system. Of course, a similar statement holds in relation to hyperbolic
$\BC_n$ Sutherland and the system of \cite{Ma15}.

Until recently, no Lax matrix was known that would generate van Diejen's commuting
Hamiltonians, except in the rational limit \cite{Pu12}. In the reduction approach
a Lax matrix arises automatically, in our case it features in equations \eqref{3.91}
and \eqref{3.130}. This might be helpful in further investigations for a Lax matrix
behind van Diejen's $5$-coupling Hamiltonian. The search would be easy if one could
derive van Diejen's system by Hamiltonian reduction. It is a long standing open
problem to find such derivation. Perhaps one should consider some `classical analogue'
of the quantum group interpretation of the Koornwinder ($\BC_n$ Macdonald) polynomials
found in \cite{OS05}, since those polynomials diagonalize van Diejen's quantized
Hamiltonians \cite{vD95}.

Another problem is to construct action-angle duals of the deformed $\BC_n$
Sutherland systems. Duality relations are not only intriguing on their own right,
but are also very useful for extracting information about the dynamics
\cite{Ru88,Ru95,Ru99,Pu13}.
The duality was used in \citepalias{AFG12,FG14} to show that the hyperbolic $\BC_n$
Sutherland system is maximally superintegrable, the trigonometric $\BC_n$ Sutherland
system has precisely $n$ constants of motion, and the relevant dual systems are
maximally superintegrable in both cases.

We mention that based on the results of this chapter Feh\'{e}r and Marshall \cite{FM17} recently explored the action-angle dual of the Hamiltonian \eqref{3.197} in the reduction framework. The question of duality for the system derived in \cite{Ma15} is still open.

Finally, we wish to mention the recent paper \cite{vDE16} dealing with the
quantum mechanics of a lattice version of a $4$-parameter Inozemtsev type limit of
van Diejen's trigonometric/hyperbolic system. The systems studied in \cite{Ma15} and
in this chapter correspond to further limits of specializations of this one. The
statements about quantum mechanical dualities contained in \cite{vDE16} and its
references should be related to classical dualities. We hope to return to this question in the future.

\addtocontents{toc}{\vspace{-1em}}
\cleardoublepage
\part{Developments in the Ruijsenaars-Schneider family}
\label{part:2}

\chapter{Lax representation of the hyperbolic BC${}_{\textit{n}}$ van Diejen system}
\label{chap:4}

In this chapter, which follows \citepalias{PG16}, we construct a Lax pair for the classical hyperbolic
van Diejen system with two independent coupling parameters. Built upon
this construction, we show that the dynamics can be solved by a projection
method, which in turn allows us to initiate the study of the scattering
properties. As a consequence, we prove the equivalence between the first
integrals provided by the eigenvalues of the Lax matrix and the family of
van Diejen's commuting Hamiltonians. Also, at the end of the chapter, we
propose a candidate for the Lax matrix of the hyperbolic van Diejen
system with three independent coupling constants.

The Ruijsenaars-Schneider-van Diejen (RSvD) systems, or simply van Diejen 
systems \cite{vD95,vD94,vD95-2}, 
are multi-parametric generalisations of the translation invariant 
Ruijsenaars-Schneider (RS) models \cite{RS86,Ru87}. 
Moreover, in the so-called `non-relativistic' limit, they reproduce the 
Ca\-lo\-gero-Moser-Sutherland (CMS) models \cite{Ca71,Su71,Mo75,OP76} associated with the $\BC$-type root systems. 
However, compared to the translation invariant $\rA$-type models, the 
geometrical picture underlying the most general classical van Diejen models 
is far less developed. The most probable explanation of this fact is the lack 
of Lax representation for the van Diejen dynamics. For this reason, working 
mainly in a symplectic reduction framework, in the last couple of years Pusztai
undertook the study of the $\BC$-type rational van Diejen models 
\cite{Pu11-2,Pu12,Pu13,Pu15} \citepalias{FG14,GF15}. By going one stage up, in this chapter 
we wish to report on our first results about the hyperbolic variants of the 
van Diejen family.

In order to describe the Hamiltonian systems of our interest, let us recall 
that the configuration space of the hyperbolic $n$-particle van Diejen model 
is the open subset
\be
Q=\{\lambda=(\lambda_1,\dots,\lambda_n)\in\R^n|\lambda_1>\dots>\lambda_n>0\}\subseteq\R^n,
\label{Q}
\ee
that can be seen as an open Weyl chamber of type $\BC_n$. The cotangent bundle 
of $Q$ is trivial, and it can be naturally identified with the open subset
\be
P=Q\times\R^n=\{(\lambda,\theta)
=(\lambda_1,\dots,\lambda_n,\theta_1,\dots,\theta_n)\in\R^{2n}|
\lambda_1>\dots>\lambda_n>0\}\subseteq\R^{2n}.
\label{P}
\ee
Following the widespread custom, throughout the chapter we shall occasionally 
think of the letters $\lambda_a$ and $\theta_a$ $(1\leq a\leq n)$ as 
globally defined coordinate functions on $P$. For example, using this latter 
interpretation, the canonical symplectic form on the phase space 
$P\cong T^\ast Q$ can be written as
\be
    \omega = \sum_{c = 1}^n \dd \lambda_c \wedge \dd \theta_c,
\label{symplectic_form}
\ee
whereas the fundamental Poisson brackets take the form
\be
    \{ \lambda_a, \lambda_b \} = 0,
    \quad
    \{ \theta_a, \theta_b \} = 0,
    \quad
    \{ \lambda_a, \theta_b \} = \delta_{a, b}
    \qquad
    (1 \leq a, b \leq n).
\label{PBs}
\ee
The principal goal of this chapter is to study the dynamics generated by the 
smooth Hamiltonian function
\be
    H 
    = \sum_{a = 1}^n \cosh(\theta_a)
        \left[1 + \frac{\sin(\nu)^2}{\sinh(2 \lambda_a)^2} \right]^\half
        \prod_{\substack{c = 1 \\ (c \neq a)}}^n
            \left[ 
                1 + \frac{\sin(\mu)^2}{\sinh(\lambda_a - \lambda_c)^2} 
            \right]^\half
            \left[ 
                1 + \frac{\sin(\mu)^2}{\sinh(\lambda_a + \lambda_c)^2} 
            \right]^\half,
\label{H}
\ee
where $\mu, \nu \in \R$ are arbitrary coupling constants satisfying the 
conditions
\be
    \sin(\mu) \neq 0 \neq \sin(\nu).
\label{mu_nu_conds}
\ee
Note that $H$ \eqref{H} does belong to the family of the hyperbolic 
$n$-particle van Diejen Hamiltonians with \emph{two independent parameters} 
$\mu$ and $\nu$ (cf. \eqref{H1_vs_H}). Of course, the values of the parameters 
$\mu$ and $\nu$ really matter only modulo $\pi$.

Now, we briefly outline the content of the chapter. In the subsequent 
section, we start with a short overview on some relevant facts and notations 
from Lie theory. Having equipped with the necessary background material, in 
Section \ref{sec:4.2} we define our Lax matrix \eqref{L} 
for the van Diejen system \eqref{H}, and also investigate its main algebraic 
properties. In Section \ref{sec:4.3} we turn to the study of the 
Hamiltonian flow generated by \eqref{H}. As the first step, in Theorem 
\ref{THEOREM_Completeness} we formulate the completeness of the corresponding 
Hamiltonian vector field. Most importantly, in Theorem 
\ref{THEOREM_Lax_representation} we provide a Lax representation of the 
dynamics, whereas in Theorem \ref{THEOREM_eigenvalue_dynamics} we establish 
a solution algorithm of purely algebraic nature. Making use of the projection 
method formulated in Theorem \ref{THEOREM_eigenvalue_dynamics}, we also 
initiate the study of the scattering properties of the system \eqref{H}. 
Our rigorous results on the temporal asymptotics of the maximally defined 
trajectories are summarized in Lemma \ref{LEMMA_asymptotics}. Section 
\ref{sec:4.4} serves essentially two purposes. In 
Subsection \ref{subsec:4.4.1} we elaborate the link between our 
special $2$-parameter family of Hamiltonians \eqref{H} and the most general 
$5$-parameter family of hyperbolic van Diejen systems \eqref{H_vD}. At the 
level of the coupling parameters the relationship can be read off from the 
equation \eqref{2parameters}. Furthermore, in Lemma 
\ref{LEMMA_linear_relation} we affirm the equivalence between van Diejen's
commuting family of Hamiltonians and the coefficients of the characteristic 
polynomial of the Lax matrix \eqref{L}. Based on this technical result, in 
Theorem \ref{THEOREM_commuting_eigenvalues} we can infer that the eigenvalues 
of the proposed Lax matrix \eqref{L} provide a commuting family of first 
integrals for the Hamiltonian system \eqref{H}. We conclude the chapter with 
Section \ref{sec:4.5}, where we discuss the potential applications, 
and also offer some open problems and conjectures. In particular, in 
\eqref{L_tilde} we propose a Lax matrix for the $3$-parameter family of 
hyperbolic van Diejen systems defined in \eqref{H_1_mu_nu_kappa}.

\section{Preliminaries from group theory}
\label{sec:4.1}

This section has two main objectives. Besides fixing the notations used
throughout the chapter, we also provide a brief account on some relevant facts
from Lie theory underlying our study of the $2$-parameter family of 
hyperbolic van Diejen systems \eqref{H}. For convenience, our conventions
closely follow Knapp's book \cite{Kn02}.
 
As before, by $n \in \bN = \{1, 2, \dots \}$ we denote the number of particles.
Let $N = 2 n$, and also introduce the shorthand notations
\be 
    \bN_n = \{ 1, \dots, n \}
    \quad \text{and} \quad
    \bN_N = \{1, \dots, N \}.
\label{bN_n} 
\ee
With the aid of the $N \times N$ matrix
\be
    C 
    = \begin{bmatrix}
    	0_n & \bsone_n \\
    	\bsone_n & 0_n
    \end{bmatrix}
\label{C} 
\ee
we define the non-compact real reductive matrix Lie group
\be
G =\UN(n, n) = \{ y \in GL(N, \bC) \, | \, y^* C y = C \},
\label{G}
\ee
in which the set of unitary elements
\be
    K = \{ y \in G \, | \, y^* y = \bsone_N \}
    \cong
    \UN(n) \times\UN(n)
\label{K}
\ee
forms a maximal compact subgroup. The Lie algebra of $G$ \eqref{G} takes 
the form
\be
	\mfg 
	= \mfu(u, n) 
	= \{ Y \in \mfgl(N, \bC) \, | \, Y^* C + C Y = 0 \},
\label{mfg}
\ee 
whereas for the Lie subalgebra corresponding to $K$ \eqref{K} we have
the identification
\be
    \mfk 
    = \{ Y \in \mfg \, | \, Y^* + Y = 0 \} \cong \mfu(n) \oplus \mfu(n).
\label{mfk}
\ee
Upon introducing the subspace
\be
    \mfp = \{ Y \in \mfg \, | \, Y^* = Y \},
\label{mfp}
\ee
we can write the decomposition $\mfg = \mfk \oplus \mfp$, which is orthogonal 
with respect to the usual trace pairing defined on the matrix Lie algebra 
$\mfg$. Let us note that the restriction of the exponential map onto the 
complementary subspace $\mfp$ \eqref{mfp} is injective. Moreover, the image 
of $\mfp$ under the exponential map can be identified with the set of the 
positive definite elements of the group $\UN(n, n)$; that is, 
\be
\exp(\mfp)=\{y\in\UN(n,n)\mid y>0\}.
\label{exp_mfp_identification}
\ee
Notice that, due to the Cartan decomposition $G = \exp(\mfp) K$, the above 
set can be also naturally identified with the non-compact symmetric space 
associated with the pair $(G, K)$, i.e.,
\be
\exp(\mfp)\cong\UN(n,n)/(\UN(n)\times\UN(n))\cong\SU(n,n)/\mathrm{S}(\UN(n)\times\UN(n)).
\label{symm_space}
\ee

To get a more detailed picture about the structure of the reductive Lie 
group $\UN(n, n)$, in $\mfp$ \eqref{mfp} we introduce the maximal Abelian 
subspace
\be
\mfa=\{X=\diag(x_1,\dots,x_n,-x_1,\dots,-x_n)\mid x_1,\dots,x_n\in\R\}.
\label{mfa}
\ee
Let us recall that we can attain every element of $\mfp$ by conjugating the 
elements of $\mfa$ with the elements of the compact subgroup $K$ \eqref{K}. 
More precisely, the map
\be
\mfa\times K\ni (X,k)\mapsto kXk^{-1}\in\mfp
\label{mfp_and_mfa_and_K}
\ee
is well-defined and onto. As for the centralizer of $\mfa$ inside $K$ 
\eqref{K}, it turns out to be the Abelian Lie group
\be
M=Z_K(\mfa)=\{\diag(e^{\ri\chi_1},\dots,e^{\ri\chi_n},
e^{\ri\chi_1},\dots,e^{\ri \chi_n})\mid\chi_1,\dots,\chi_n\in\R\}
\label{M}
\ee
with Lie algebra
\be
\mfm=\{\diag(\ri\chi_1,\dots,\ri\chi_n,\ri\chi_1,\dots,\ri\chi_n)\mid
\chi_1,\dots,\chi_n\in\R\}.
\label{mfm}
\ee
Let $\mfm^\perp$ and $\mfa^\perp$ denote the sets of the off-diagonal 
elements in the subspaces $\mfk$ and $\mfp$, respectively; then clearly
we can write the refined orthogonal decomposition
\be
    \mfg = \mfm \oplus \mfm^\perp \oplus \mfa \oplus \mfa^\perp.
\label{refined_decomposition}
\ee
To put it simple, each Lie algebra element $Y \in \mfg$ can be decomposed as
\be
Y=Y_\mfm+Y_{\mfm^\perp}+Y_\mfa+Y_{\mfa^\perp}
\label{Y_decomp}
\ee
with unique components belonging to the subspaces indicated by the subscripts.

Throughout our work the commuting family of linear operators
\be
\ad_X\colon\mfgl(N,\C)\rightarrow\mfgl(N,\C),\quad Y\mapsto [X,Y]
\label{ad}
\ee
defined for the diagonal matrices $X \in \mfa$ plays a distinguished role. 
Let us note that the (real) subspace 
$\mfm^\perp \oplus \mfa^\perp \subseteq \mfgl(N, \bC)$ is invariant under 
$\ad_X$, whence the restriction
\be
\wad_X=\ad_X|_{\mfm^\perp\oplus\mfa^\perp}\in\mfgl(\mfm^\perp\oplus\mfa^\perp)
\label{wad}
\ee
is a well-defined operator for each 
$X=\diag(x_1,\dots,x_n,-x_1,\dots,-x_n)\in\mfa$ with spectrum
\be
\mathrm{Spec}(\wad_X)=\{x_a-x_b,\pm(x_a+x_b),\pm 2x_c\mid a,b,c\in\N_n,\,a\neq b\}.
\label{wad_X_spectrum}
\ee
Now, recall that the regular part of the Abelian subalgebra $\mfa$ \eqref{mfa} is defined by the subset
\be
\mfa_\reg=\{X\in\mfa\mid\wad_X \text{ is invertible}\},
\label{mfa_reg}
\ee
in which the standard open Weyl chamber
\be
    \mfc 
    = \{ X = \diag(x_1, \ldots, x_n, -x_1, \ldots, -x_n) \in \mfa
        \, | \,
        x_1 > \ldots > x_n > 0 \}
\label{mfc}
\ee
is a connected component. Let us observe that it can be naturally identified 
with the configuration space $Q$ \eqref{Q}; that is, $Q \cong \mfc$. Finally, 
let us recall that the regular part of $\mfp$ \eqref{mfp} is defined as
\be
    \mfp_\reg 
    = \{ k X k^{-1} \in \mfp 
        \, | \, 
        X \in \mfa_\reg \text{ and } k \in K \}.
\label{mfp_reg}
\ee
As a matter of fact, from the map \eqref{mfp_and_mfa_and_K} we can derive
a particularly useful characterization for the open subset 
$\mfp_\reg \subseteq \mfp$. Indeed, the map
\be
    \mfc \times (K / M) \ni (X, k M) \mapsto k X k^{-1} \in \mfp_\reg
\label{mfp_reg_identification}
\ee
turns out to be a diffeomorphism, providing the identification
$\mfp_\reg \cong \mfc \times (K / M)$.

\section{Algebraic properties of the Lax matrix}
\label{sec:4.2}

Having reviewed the necessary notions and notations from Lie theory, in this 
section we propose a Lax matrix for the hyperbolic van Diejen system of our 
interest \eqref{H}. To make the presentation simpler, with any
$\lambda = (\lambda_1, \dots, \lambda_n) \in\R^n$ and 
$\theta = (\theta_1, \dots, \theta_n) \in \R^n$ we associate the real 
$N$-tuples
\be
    \Lambda = (\lambda_1, \dots, \lambda_n, -\lambda_1, \dots, -\lambda_n)
    \quad \text{and} \quad
    \Theta = (\theta_1, \dots, \theta_n, -\theta_1, \dots, -\theta_n),
\label{Lambda_Theta}
\ee
respectively, and also define the $N \times N$ diagonal matrix
\be
    \bsLambda 
    = \diag(\Lambda_1, \dots, \Lambda_N)
    = \diag(\lambda_1, \dots, \lambda_n, -\lambda_1, \dots, -\lambda_n)
    \in \mfa.
\label{bsLambda}
\ee
Notice that if $\lambda \in \R^n$ is a regular element in the sense that 
the corresponding diagonal matrix $\bsLambda$ \eqref{bsLambda} belongs to 
$\mfa_\reg$ \eqref{mfa_reg}, then for each $j \in \bN_N$ the complex
number
\be
    z_j = - \frac{\sinh(\ri \nu + 2 \Lambda_j)}{\sinh(2 \Lambda_j)}
            \prod_{\substack{c = 1 \\ (c \neq j, j - n)}}^n
                \frac{\sinh(\ri \mu + \Lambda_j - \lambda_c)}
                    {\sinh(\Lambda_j - \lambda_c)}
                \frac{\sinh(\ri \mu + \Lambda_j + \lambda_c)}
                    {\sinh(\Lambda_j + \lambda_c)}
\label{z_a}
\ee
is well-defined. Thinking of $z_j$ as a function of $\lambda$, let us observe 
that its modulus $u_j = \vert z_j \vert$ takes the form
\be
    u_j
    = \left( 1 + \frac{\sin(\nu)^2}{\sinh(2 \Lambda_j)^2} \right)^\half
        \prod_{\substack{c = 1 \\ (c \neq j, j - n)}}^n
            \left( 
                1 + \frac{\sin(\mu)^2}{\sinh(\Lambda_j - \lambda_c)^2} 
            \right)^\half
            \left( 
                1 + \frac{\sin(\mu)^2}{\sinh(\Lambda_j + \lambda_c)^2} 
            \right)^\half,   
\label{u_a}
\ee
and the property $z_{n + a} = \bar{z}_a$ ($a \in \bN_n$) is also clear. 
Next, built upon the functions $z_j$ and $u_j$, we introduce the column 
vector $F \in \bC^N$ with components
\be
    F_a = e^{\frac{\theta_a}{2}} u_a^{\frac{1}{2}}
    \quad \text{and} \quad
    F_{n + a} = e^{-\frac{\theta_a}{2}} \bar{z}_a u_a^{-\frac{1}{2}}
    \qquad
    (a \in \bN_n).
\label{F}
\ee
At this point we are in a position to define our Lax matrix 
$L \in \mfgl(N, \bC)$ with the entries
\be
    L_{k, l} 
    = \frac{\ri \sin(\mu) F_k \bar{F}_l 
        + \ri \sin(\mu - \nu) C_{k, l}}
        {\sinh(\ri \mu + \Lambda_k - \Lambda_l)}
    \qquad
    (k, l \in \bN_N).
\label{L}
\ee
Note that the matrix valued function $L$ is well-defined at each point
$(\lambda, \theta) \in \R^{N}$ satisfying the regularity condition
$\bsLambda \in \mfa_\reg$. Since $\mfc \subseteq \mfa_\reg$ \eqref{mfc}, 
$L$ makes sense at each point of the phase space $P$ \eqref{P} as well.
To give a motivation for the definition of $L = L(\lambda, \theta; \mu, \nu)$ 
\eqref{L}, let us observe that in its `rational limit' we get back the Lax 
matrix of the rational van Diejen system with two parameters. Indeed, up to 
some irrelevant numerical factors caused by a slightly different convention, 
in the $\alpha \to 0^+$ limit the matrix 
$L(\alpha \lambda, \theta; \alpha \mu, \alpha \nu)$ tends to the rational 
Lax matrix $\cA = \cA(\lambda, \theta; \mu, \nu)$ as defined in \cite[eqs. (4.2)-(4.5)]{Pu11-2}. In \cite{Pu11-2} the matrix $\cA$ has many peculiar algebraic properties, that we wish to 
generalize for the proposed hyperbolic Lax matrix $L$ in the rest of this 
section.

\subsection{Lax matrix: explicit form, inverse, and positivity}
\label{subsec:4.2.1}

By inspecting the matrix entries \eqref{L}, it is obvious that $L$ is 
Hermitian. However, it is a less trivial fact that $L$ is closely tied with 
the non-compact Lie group $\UN(n, n)$ \eqref{G}. The purpose of this subsection 
is to explore this surprising relationship.

\begin{proposition}
\label{PROPOSITION_L_in_G}
The matrix $L$ \eqref{L} obeys the quadratic equation $L C L = C$. In other 
words, the matrix valued function $L$ takes values in the Lie group $\UN(n, n)$.
\end{proposition}

\begin{proof}
Take an arbitrary element $(\lambda, \theta) \in \R^N$ 
satisfying the regularity condition $\bsLambda \in \mfa_\reg$. We start 
by observing that for each $a \in \bN_n$ the complex conjugates of $z_a$ 
\eqref{z_a} and $F_{n + a}$ \eqref{F} can be obtained by changing the 
sign of the single component $\lambda_a$ of $\lambda$. Therefore, if 
$a, b \in \bN_n$ are arbitrary indices, then by interchanging the components 
$\lambda_a$ and $\lambda_b$ of the $n$-tuple $\lambda$, the expression 
$(L C L)_{a, b} F_{a}^{-1} \bar{F}_{b}^{-1}$ readily transforms into 
$(L C L)_{n + a, n + b} F_{n + a}^{-1} \bar{F}_{n + b}^{-1}$. We capture 
this fact by writing
\be
    \frac{(L C L)_{a, b}}{F_{a} \bar{F}_{b}}
    \underset{\lambda_a \leftrightarrow \lambda_b}{\leadsto}
    \frac{(L C L)_{n + a, n + b}}{F_{n + a} \bar{F}_{n + b}}
    \qquad
    (a, b \in \bN_n).
\label{L1_1}
\ee
Similarly, if $a \neq b$, then from 
$(L C L)_{a, b} F_{a}^{-1} \bar{F}_{b}^{-1}$ we can recover 
$(L C L)_{n + a, b} F_{n + a}^{-1} \bar{F}_{b}^{-1}$
by exchanging $\lambda_a$ for $-\lambda_a$. Schematically, we have
\be
    \frac{(L C L)_{a, b}}{F_{a} \bar{F}_{b}}
    \underset{\lambda_a \leftrightarrow -\lambda_a}{\leadsto}
    \frac{(L C L)_{n + a, b}}{F_{n + a} \bar{F}_{b}}
    \qquad
    (a, b \in \bN_n, \, a \neq b).
\label{L1_2}
\ee
Furthermore, the expression $(L C L)_{a, b} F_{a}^{-1} \bar{F}_{b}^{-1}$ 
reproduces $(L C L)_{a, n + b} F_{a}^{-1} \bar{F}_{n + b}^{-1}$ upon 
swapping $\lambda_b$ for $-\lambda_b$, i.e.,
\be
    \frac{(L C L)_{a, b}}{F_{a} \bar{F}_{b}}
    \underset{\lambda_b \leftrightarrow -\lambda_b}{\leadsto}
    \frac{(L C L)_{a, n + b}}{F_{a} \bar{F}_{n + b}}
    \qquad(a, b \in \bN_n, a \neq b).
\label{L1_3}
\ee
Finally, the relationship between the remaining entries is given by the 
exchange 
\be
    (L C L)_{a, n + a}
    \underset{\lambda_a \leftrightarrow -\lambda_a}{\leadsto}
    (L C L)_{n + a, a}
    \qquad (a \in \bN_n).
\label{L1_4}
\ee
The message of the above equations \eqref{L1_1}-\eqref{L1_4} is quite 
evident. Indeed, in order to prove the desired matrix equation $L C L = C$, 
it does suffice to show that $(L C L)_{a, b} = 0$ for all $a, b \in \bN_n$, 
and also that $(L C L)_{a, n + a} = 1$ for all $a \in \bN_n$.

Recalling the formulae \eqref{F} and \eqref{L}, it is clear that
for all $a \in \bN_n$ we can write
\be
    \frac{(L C L)_{a, a}}{F_a \bar{F}_a}
    = 2 \Real 
    \bigg( 
        \frac{\ri \sin(\mu) z_a + \ri \sin(\mu - \nu)}
            {\sinh(\ri \mu + 2 \lambda_a)}
        -\sum_{\substack{c = 1 \\ (c \neq a)}}^n
            \frac{\sin(\mu)^2 z_c}
                {\sinh(\ri \mu + \lambda_a + \lambda_c) 
                \sinh(\ri \mu - \lambda_a + \lambda_c)}
    \bigg).
\label{LCL_aa}
\ee
To proceed further, we introduce a complex valued function $f_a$ depending 
on a single complex variable $w$ obtained simply by replacing $\lambda_a$ 
with $\lambda_a + w$ in the right-hand side of the above equation 
\eqref{LCL_aa}. Remembering \eqref{z_a}, it is obvious that the resulting 
function is meromorphic with at most first order poles at the points 
 \be
    w \equiv -\lambda_a, \, 
    w \equiv \pm \ri \mu / 2 - \lambda_a, \, 
    w \equiv \Lambda_j - \lambda_a \, (j \in \bN_N)
    \pmod{\ri \pi}.
\label{poles_aa}
\ee
However, by inspecting the terms appearing in the explicit expression
of $f_a$, a straightforward computation reveals immediately that the 
residue of $f_a$ at each of these points is zero, i.e., the singularities 
are in fact removable. As a consequence, $f_a$ can be uniquely extended 
onto the whole complex plane as a periodic entire function with period 
$2 \pi \ri$. Moreover, since $f_a(w)$ vanishes as $\Real(w) \to \infty$, 
the function $f_a$ is clearly bounded. By invoking Liouville's theorem,
we conclude that $f_a(w) = 0$ for all $w \in \bC$, and so
\be
    \frac{(L C L)_{a, a}}{F_a \bar{F}_a} = f_a(0) = 0.
\label{LCL_aa-zero}
\ee

Next, let $a, b \in \bN_n$ be arbitrary indices satisfying $a \neq b$. Keeping 
in mind the definitions \eqref{F} and \eqref{L}, we find at once that
\be
\begin{split}
    \frac{(L C L)_{a, b}}{F_a \bar{F}_b}
    = & \frac{\ri \sin(\mu) \big( \ri \sin(\mu) z_a 
                                    + \ri \sin(\mu - \nu) \big)}
            {\sinh(\ri \mu + \lambda_a - \lambda_b) 
            \sinh(\ri \mu + 2 \lambda_a)}
    + \frac{\ri \sin(\mu) \big( \ri \sin(\mu) \bar{z}_b
                                    + \ri \sin(\mu - \nu) \big)}
            {\sinh(\ri \mu + \lambda_a - \lambda_b)
            \sinh(\ri \mu - 2 \lambda_b)}
    \\
    & + \frac{\ri \sin(\mu) \bar{z}_a}
            {\sinh(\ri \mu - \lambda_a - \lambda_b)}
    + \frac{\ri \sin(\mu) z_b}{\sinh(\ri \mu + \lambda_a + \lambda_b)}
    \\
    & - \sum_{\substack{j = 1 \\ (j \neq a, b, n + a, n + b)}}^N
            \frac{\sin(\mu)^2 z_j}
                {\sinh(\ri \mu + \lambda_a + \Lambda_j)
                \sinh(\ri \mu - \lambda_b + \Lambda_j)}.
\end{split}
\label{LCL_ab}
\ee
Although this equation looks considerably more complicated than 
\eqref{LCL_aa}, it can be analyzed by the same techniques. Indeed, by 
replacing $\lambda_a$ with $\lambda_a + w$ in the right-hand side of 
\eqref{LCL_ab}, we may obtain a meromorphic function $f_{a, b}$ of 
$w \in \bC$ that has at most first order poles at the points
\be
    w \equiv -\lambda_a, \, 
    w \equiv -\ri \mu / 2 - \lambda_a, \, 
    w \equiv -\ri \mu - \lambda_a + \lambda_b, \, 
    w \equiv \Lambda_j - \lambda_a \, (j \in \bN_N)
    \pmod{\ri \pi}.
\label{poles_ab}
\ee
However, the residue of $f_{a, b}$ at each of these points turns out 
to be zero, and $f_{a, b}(w)$ also vanishes as $\Real(w) \to \infty$. Due
to Liouville's theorem we get $f_{a, b}(w) = 0$ for all $w \in \bC$, thus
\be
    \frac{(L C L)_{a, b}}{F_a \bar{F}_b} = f_{a, b}(0) = 0.
\label{LCL_ab-zero}
\ee

Finally, by taking an arbitrary $a \in \bN_n$, from \eqref{F} and \eqref{L}
we see that
\be
    (L C L)_{a, n + a} 
    = u_a^2
        + \frac{(\ri \sin(\mu) z_a + \ri \sin(\mu - \nu))^2}
            {\sinh(\ri \mu + 2 \lambda_a)^2}
        - \sum_{\substack{j = 1 \\ (j \neq a, n + a)}}^N
            \frac{\sin(\mu)^2 z_a z_j}
                {\sinh(\ri \mu + \lambda_a + \Lambda_j)^2}.
\label{LCL_n+a,a}
\ee
By replacing $\lambda_a$ with $\lambda_a + w$ in the right-hand side of 
\eqref{LCL_n+a,a}, we end up with a meromorphic function $f_{n + a}$ of 
the complex variable $w$ that has at most second order poles at the points
\be
    w \equiv -\lambda_a, \, 
    w \equiv -\ri \mu / 2 - \lambda_a, \, 
    w \equiv \Lambda_j - \lambda_a \, (j \in \bN_N)
    \pmod{\ri \pi}.
\label{poles_n+a,a}
\ee
Though the calculations are a bit more involved as in the previous cases, 
one can show that the singularities of $f_{n + a}$ are actually removable. 
Moreover, it is evident that $f_{n + a}(w) \to 1$ as $\Real(w) \to \infty$. 
Liouville's theorem applies again, implying that $f_{n + a}(w) = 1$ for 
all $w \in \bC$. Thus the relationship
\be
    (L C L)_{a, n + a} = f_{n + a}(0) = 1
\label{L_n+a,a-one}
\ee
also follows, whence the proof is complete.
\end{proof}

In the earlier paper \cite{Pu11-2} we saw that the rational 
analogue of $L$ \eqref{L} takes values in the symmetric space $\exp(\mfp)$ 
\eqref{symm_space}. We find it reassuring that the proof of Lemma 7 of 
paper \cite{Pu11-2} allows a straightforward generalization into 
the present hyperbolic context, too.

\begin{lemma}
\label{LEMMA_L_in_exp_p}
At each point of the phase space we have $L \in \exp(\mfp)$.
\end{lemma}

\begin{proof}
Recalling the identification \eqref{exp_mfp_identification} and Proposition 
\ref{PROPOSITION_L_in_G}, it is enough to prove that the Hermitian matrix
$L$ \eqref{L} is positive definite. For this reason, take an arbitrary point 
$(\lambda, \theta) \in P$ and keep it fixed. To prove the Lemma, below we 
offer a standard continuity argument by analyzing the dependence of $L$ 
solely on the coupling parameters.

In the very special case when the pair $(\mu, \nu)$ formed by the coupling 
parameters obey the relationship $\sin(\mu - \nu) = 0$, the Lax matrix $L$ 
\eqref{L} becomes a hyperbolic Cauchy-like matrix and the generalized Cauchy 
determinant formula (see e.g. \cite[eq. (1.2)]{Ru88}) readily 
implies the positivity of all its leading principal minors. Thus, recalling 
Sylvester's criterion, we conclude that $L$ is positive definite.

Turning to the general case, suppose that the pair $(\mu, \nu)$ is restricted 
only by the conditions displayed in \eqref{mu_nu_conds}. It is clear that in 
the $2$-dimensional space of the admissible coupling parameters characterized
by \eqref{mu_nu_conds} one can find a continuous curve with endpoints 
$(\mu, \nu)$ and $(\mu_0, \nu_0)$, where $\mu_0$ and $\nu_0$ satisfy the 
additional requirement $\sin(\mu_0 - \nu_0) = 0$. Since the dependence of the 
Hermitian matrix $L$ on the coupling parameters is smooth, along this curve 
the smallest eigenvalue of $L$ moves continuously. However, it cannot cross 
zero, since by Proposition \ref{PROPOSITION_L_in_G} the matrix $L$ remains 
invertible during this deformation. Therefore, since the eigenvalues of $L$ 
are strictly positive at the endpoint $(\mu_0, \nu_0)$, they must be strictly 
positive at the other endpoint $(\mu, \nu)$ as well.
\end{proof}

\subsection{Commutation relation and regularity}
\label{subsec:4.2.2}

As Ruijsenaars has observed in his seminal paper on the translation 
invariant CMS and RS type pure soliton systems, one of the key ingredients 
in their analysis is the fact that their Lax matrices obey certain 
non-trivial commutation relations with some diagonal matrices (for details,
see equation (2.4) and the surrounding ideas in \cite{Ru88}). 
As a momentum map constraint, an analogous commutation relation has also 
played a key role in the geometric study of the rational $\rC_n$ and $\BC_n$ 
RSvD systems (see \cite{Pu11-2,Pu12} \citepalias{FG14}). Due to its importance, our first goal in this 
subsection is to set up a Ruijsenaars type commutation relation for the 
proposed Lax matrix $L$ \eqref{L}, too. As a technical remark, we mention 
in passing that from now on we shall apply frequently the standard functional 
calculus on the linear operators $\ad_{\bsLambda}$ \eqref{ad} and 
$\wad_{\bsLambda}$ \eqref{wad} associated with the diagonal matrix 
$\bsLambda \in \mfc$ \eqref{bsLambda}.

\begin{lemma}
\label{LEMMA_commut_rel}
The matrix $L$ \eqref{L} and the diagonal matrix $e^{\bsLambda}$ obey the 
Ruijsenaars type commutation relation
\be
    e^{\ri \mu} e^{\ad_{\bsLambda}} L
    - e^{- \ri \mu} e^{-\ad_{\bsLambda}} L
    = 2 \ri \sin(\mu) F F^* + 2 \ri \sin(\mu - \nu) C.
\label{commut_rel} 
\ee
\end{lemma}

\begin{proof}
Recalling the matrix entries of $L$, for all $k, l \in \bN_N$ we can 
write that
\be
\begin{split}
    & \left(
        e^{\ri \mu} e^{\ad_{\bsLambda}} L
        - e^{- \ri \mu} e^{-\ad_{\bsLambda}} L
    \right)_{k, l}
    = \left(
        e^{\ri \mu} e^{\bsLambda} L e^{- \bsLambda}
        - e^{- \ri \mu} e^{- \bsLambda} L e^{\bsLambda}
    \right)_{k, l} 
    \\
    & = e^{\ri \mu} e^{\Lambda_k} L_{k, l} e^{- \Lambda_l}
        - e^{- \ri \mu} e^{- \Lambda_k} L_{k, l} e^{\Lambda_l} 
    = 2 \sinh(\ri \mu + \Lambda_k - \Lambda_l) L_{k, l}
    \\
    & = 2 \ri \sin(\mu) F_k \bar{F}_l + 2 \ri \sin(\mu - \nu) C_{k, l} 
    = \left(
        2 \ri \sin(\mu) F F^* + 2 \ri \sin(\mu - \nu) C
    \right)_{k, l},
\end{split}
\label{commut_rel_proof}
\ee
thus \eqref{commut_rel} follows at once. 
\end{proof}

Though the proof of Lemma \ref{LEMMA_commut_rel} is almost trivial, it 
proves to be quite handy in the forthcoming calculations. In particular, 
based on the commutation relation \eqref{commut_rel}, we shall now prove 
that the spectrum of $L$ is \emph{simple}. Heading toward our present goal, 
first let us recall that Lemma \ref{LEMMA_L_in_exp_p} tells us that 
$L \in \exp(\mfp)$. Therefore, as we can infer easily from 
\eqref{mfp_and_mfa_and_K}, one can find some $y \in K$ and a real $n$-tuple 
$\htheta = (\htheta_1, \ldots, \htheta_n) \in \R^n$ satisfying
\be
    \htheta_1 \geq \ldots \geq \htheta_n \geq 0,
\label{htheta_assumptions}
\ee
such that with the diagonal matrix
\be
    \hbsTheta 
    = \diag(\hTheta_1, \ldots, \hTheta_N)
    = \diag(\htheta_1, \ldots, \htheta_n, 
            -\htheta_1, \ldots, -\htheta_n)
    \in \mfa   
\label{hbsTheta}
\ee
we can write
\be
    L = y e^{2 \hbsTheta} y^{-1}.
\label{L_diagonalized}
\ee
Now, upon defining
\be
    \hat{L} = y^{-1} e^{2 \bsLambda} y \in \exp(\mfp)
    \quad \text{and} \quad
    \hat{F} = e^{-\hbsTheta} y^{-1} e^{\bsLambda} F \in \bC^N,
\label{hat_L_and_hat_F}
\ee
for these new objects we can also set up a commutation relation analogous 
to \eqref{commut_rel}. Indeed, from \eqref{commut_rel} one can derive that 
\be
    e^{\ri \mu} e^{- \hbsTheta} \hat{L} e^{\hbsTheta}
    - e^{- \ri \mu} e^{\hbsTheta} \hat{L} e^{- \hbsTheta}
    = 2 \ri \sin(\mu) \hat{F} \hat{F}^* + 2 \ri \sin(\mu - \nu) C.
\label{hat_commut_rel}
\ee
Componentwise, from \eqref{hat_commut_rel} we conclude that 
\be
    \hat{L}_{k, l} =
    \frac{\ri \sin(\mu) \hat{F}_k \bar{\hat{F}}_l 
            + \ri \sin(\mu - \nu) C_{k, l}}
        {\sinh(\ri \mu - \hTheta_k + \hTheta_l)}
    \qquad
    (k, l \in \bN_N).
\label{hat_L_entries}
\ee
Since $\hat{L}$ \eqref{hat_L_and_hat_F} is a positive definite matrix, 
its diagonal entries are strictly positive. Therefore, by exploiting
\eqref{hat_L_entries}, we can write
\be
    0 < \hat{L}_{k, k} = \vert \hat{F}_k \vert^2.
\label{hat_F_comp}
\ee
The upshot of this trivial observation is that $\hat{F}_k \neq 0$ for all 
$k \in \bN_N$.

To proceed further, notice that for the inverse matrix 
$\hat{L}^{-1} = C \hat{L} C$ we can also cook up an equation analogous to 
\eqref{hat_commut_rel}. Indeed, by simply multiplying both sides of 
\eqref{hat_commut_rel} with the matrix $C$ \eqref{C}, we obtain
\be
    e^{\ri \mu} e^{\hbsTheta} \hat{L}^{-1} e^{- \hbsTheta}
    - e^{- \ri \mu} e^{- \hbsTheta} \hat{L}^{-1} e^{\hbsTheta}
    = 2 \ri \sin(\mu) (C \hat{F}) (C \hat{F})^* + 2 \ri \sin(\mu - \nu) C,
\label{hat_commut_rel_inv}
\ee
that leads immediately to the matrix entries
\be
    (\hat{L}^{-1})_{k, l} =
    \frac{\ri \sin(\mu) (C \hat{F})_k \overline{(C \hat{F})}_l 
            + \ri \sin(\mu - \nu) C_{k, l}}
        {\sinh(\ri \mu + \hTheta_k - \hTheta_l)}
    \qquad
    (k, l \in \bN_N).
\label{hat_L_inv_entries}
\ee
For further reference, we now spell out the trivial equation
\be 
    \delta_{k, l} = \sum_{j = 1}^N \hat{L}_{k, j} (\hat{L}^{-1})_{j, l}
\label{hat_trivial}
\ee
for certain values of $k, l \in \bN_N$. First, by plugging the explicit 
formulae \eqref{hat_L_entries} and \eqref{hat_L_inv_entries} into the 
relationship \eqref{hat_trivial}, with the special choice of indices 
$k = l = a \in \bN_n$ one finds that
\be
\begin{split}
    0 
    = & 1 + \frac{\sin(\mu - \nu)^2}{\sinh(\ri \mu - 2 \htheta_a)^2}
        + \frac{2 \sin(\mu) \sin(\mu - \nu) 
                \hat{F}_a \bar{\hat{F}}_{n + a}}
            {\sinh(\ri \mu - 2 \htheta_a)^2} \\
    & + \sin(\mu)^2 \hat{F}_a \bar{\hat{F}}_{n + a}
        \sum_{c = 1}^n 
            \left(
                \frac{\bar{\hat{F}}_c \hat{F}_{n + c}}
                    {\sinh(\ri \mu - \htheta_a + \htheta_c)^2}
                + \frac{\hat{F}_c \bar{\hat{F}}_{n + c}}
                    {\sinh(\ri \mu - \htheta_a - \htheta_c)^2}
            \right).
\end{split}
\label{rel_1}
\ee
Second, if $k = a$ and $l = n + a$ with some $a \in \bN_n$, then from 
\eqref{hat_trivial} we obtain
\be
\begin{split}
    & \sin(\mu)^2 
        \sum_{c = 1}^n 
            \left(
                \frac{\bar{\hat{F}}_c \hat{F}_{n + c}}
                    {\sinh(\ri \mu - \htheta_a + \htheta_c)
                    \sinh(\ri \mu + \htheta_c + \htheta_a)}
                + \frac{\hat{F}_c \bar{\hat{F}}_{n + c}}
                    {\sinh(\ri \mu - \htheta_a - \htheta_c)
                    \sinh(\ri \mu - \htheta_c + \htheta_a)}
            \right) 
    \\
    & \quad
    = \ri \sin(\mu - \nu) 
    \left(
        \frac{1}{\sinh(\ri \mu - 2 \htheta_a)}
        + \frac{1}{\sinh(\ri \mu + 2 \htheta_a)}
    \right). 
\end{split}
\label{rel_2}
\ee
Third, if $k = a$ and $l = b$ with some $a, b \in \bN_n$ satisfying 
$a \neq b$, then the relationship \eqref{hat_trivial} immediately leads to 
the equation
\be
\begin{split}
    & \sin(\mu)^2 
        \sum_{c = 1}^n 
            \left(
                \frac{\bar{\hat{F}}_c \hat{F}_{n + c}}
                    {\sinh(\ri \mu - \htheta_a + \htheta_c)
                    \sinh(\ri \mu + \htheta_c - \htheta_b)}
                + \frac{\hat{F}_c \bar{\hat{F}}_{n + c}}
                    {\sinh(\ri \mu - \htheta_a - \htheta_c)
                    \sinh(\ri \mu - \htheta_c - \htheta_b)}
            \right) 
    \\
    & \quad
    = - \frac{\sin(\mu) \sin(\mu - \nu)}
        {\sinh(\ri \mu - \htheta_a - \htheta_b)}  
    \left(
        \frac{1}{\sinh(\ri \mu - 2 \htheta_a)}
        + \frac{1}{\sinh(\ri \mu - 2 \htheta_b)}
    \right). 
\end{split}
\label{rel_3}
\ee
At this point we wish to emphasize that during the derivation of the 
last two equations \eqref{rel_2} and \eqref{rel_3} it proves to be 
essential that each component of the column vector $\hat{F}$ 
\eqref{hat_L_and_hat_F} is nonzero, as we have seen in \eqref{hat_F_comp}.

\begin{lemma}
\label{LEMMA_regularity}
Under the additional assumption on the coupling parameters
\be
    \sin(2 \mu - \nu) \neq 0,
\label{coupling_cond_extra}
\ee
the spectrum of the matrix $L$ \eqref{L} is simple of the form
\be
    \mathrm{Spec}(L) 
    = \{ e^{\pm 2 \htheta_a} \, | \, a \in \bN_n \},
\label{spec_L}
\ee
where $\htheta_1 > \ldots > \htheta_n > 0$. In other words, $L$ is regular 
in the sense that $L \in \exp(\mfp_\reg)$.
\end{lemma}

\begin{proof}
First, let us \emph{suppose} that $\htheta_a = 0$ for some $a \in \bN_n$. 
With this particular index $a$, from equation \eqref{rel_1} we infer that
\be
\begin{split}
    0 
    = & 1 - \frac{\sin(\mu - \nu)^2}{\sin(\mu)^2}
        - \frac{2 \sin(\mu - \nu) 
                \hat{F}_a \bar{\hat{F}}_{n + a}}
            {\sin(\mu)} \\
    & + \sin(\mu)^2 \hat{F}_a \bar{\hat{F}}_{n + a}
        \sum_{c = 1}^n 
            \left(
                \frac{\bar{\hat{F}}_c \hat{F}_{n + c}}
                    {\sinh(\ri \mu + \htheta_c)^2}
                + \frac{\hat{F}_c \bar{\hat{F}}_{n + c}}
                    {\sinh(\ri \mu - \htheta_c)^2}
            \right),
\end{split}
\label{L2_rel_1}
\ee
while \eqref{rel_2} leads to the relationship
\be
    \sin(\mu)^2 
        \sum_{c = 1}^n 
            \left(
                \frac{\bar{\hat{F}}_c \hat{F}_{n + c}}
                    {\sinh(\ri \mu + \htheta_c)^2}
                + \frac{\hat{F}_c \bar{\hat{F}}_{n + c}}
                    {\sinh(\ri \mu - \htheta_c)^2}
            \right)
    = \frac{2 \sin(\mu - \nu)}{\sin(\mu)}.
\label{L2_rel_2}
\ee
Now, by plugging \eqref{L2_rel_2} into \eqref{L2_rel_1}, we obtain
\be
    0 
    = 1 - \frac{\sin(\mu - \nu)^2}{\sin(\mu)^2}
    = \frac{\sin(\mu)^2 - \sin(\mu - \nu)^2 }{\sin(\mu)^2}
    = \frac{\sin(\nu) \sin(2 \mu - \nu)}{\sin(\mu)^2},
\label{trig_contradiction_1}
\ee
which clearly contradicts the assumptions imposed in the equations 
\eqref{mu_nu_conds} and \eqref{coupling_cond_extra}. Thus, we are 
forced to conclude that for all $a \in \bN_n$ we have $\htheta_a \neq 0$.

Second, let us \emph{suppose} that $\htheta_a = \htheta_b$ for some 
$a, b \in \bN_n$ satisfying $a \neq b$. With these particular indices 
$a$ and $b$, equation \eqref{rel_3} takes the form
\be
    \sin(\mu)^2 
        \sum_{c = 1}^n 
            \left(
                \frac{\bar{\hat{F}}_c \hat{F}_{n + c}}
                    {\sinh(\ri \mu - \htheta_a + \htheta_c)^2}
                + \frac{\hat{F}_c \bar{\hat{F}}_{n + c}}
                    {\sinh(\ri \mu - \htheta_a - \htheta_c)^2}
            \right) 
    = - \frac{2 \sin(\mu) \sin(\mu - \nu)}{\sinh(\ri \mu - 2 \htheta_a)^2}.
\label{L2_rel_3}
\ee
Now, by plugging this formula into \eqref{rel_1}, we obtain immediately that
\be
    0 = 1 + \frac{\sin(\mu - \nu)^2}{\sinh(\ri \mu - 2 \htheta_a)^2},
\label{L2_rel_4}
\ee
which in turn implies that
\be
\begin{split}
    & \sin(\mu - \nu)^2 
    = - \sinh(\ri \mu - 2 \htheta_a)^2 \\
    & \quad
    = \sin(\mu)^2 \cosh(2 \htheta_a)^2 
        - \cos(\mu)^2 \sinh(2 \htheta_a)^2
        + \ri \sin(\mu) \cos(\mu) \sinh(4 \htheta_a).
\end{split}
\label{L2_ab_key}
\ee
Since $\htheta_a \neq 0$ and since $\sin(\mu) \neq 0$, the imaginary part 
of the above equation leads to the relation $\cos(\mu) = 0$, whence 
$\sin(\mu)^2 = 1$ also follows. Now, by plugging these observations into 
the real part of \eqref{L2_ab_key}, we end up with the contradiction
\be
    1 \geq \sin(\mu - \nu)^2 = \cosh(2 \htheta_a)^2 > 1.
\label{L2_contradiction}
\ee
Thus, if $a, b \in \bN_n$ and $a \neq b$, then necessarily we have 
$\htheta_a \neq \htheta_b$.
\end{proof}

Since the spectrum of $L$ \eqref{L} is simple, it follows that the
dependence of the eigenvalues on the matrix entries is smooth. Therefore,
recalling \eqref{spec_L}, it is clear that each $\htheta_c$ $(c \in \bN_n)$
can be seen as a smooth function on $P$ \eqref{P}, i.e.,
\be
    \htheta_c \in C^\infty(P).
\label{htheta_smooth}
\ee
To conclude this subsection, we also offer a few remarks on the additional
constraint appearing in \eqref{coupling_cond_extra}, that we keep in
effect in the rest of the chapter. Naively, this assumption excludes a 
$1$-dimensional subset from the $2$-dimensional space of the parameters 
$(\mu, \nu)$. However, looking back to the Hamiltonian $H$ \eqref{H},
it is clear that the effective coupling constants of our van Diejen systems
are rather the positive numbers $\sin(\mu)^2$ and $\sin(\nu)^2$. Therefore, 
keeping in mind \eqref{mu_nu_conds}, on the parameters $\mu$ and $\nu$ we 
could have imposed the requirement, say,
\be
    (\mu, \nu) 
    \in \left( (0, \pi / 4) \times [ -\pi / 2, 0 ) \right)
        \cup \left( [\pi / 4, \pi / 2] \times ( 0, \pi / 2] \right),
\label{mu_nu_spec}
\ee 
at the outset. The point is that, under the requirement \eqref{mu_nu_spec}, 
the equation $\sin(2 \mu - \nu) = 0$ is equivalent to the pair of equations 
$\sin(\mu)^2 = 1 / 2$ and $\sin(\nu)^2 = 1$. To put it differently, our 
observation is that under the assumptions \eqref{mu_nu_conds} and 
\eqref{coupling_cond_extra} the pair $(\sin(\mu)^2, \sin(\nu)^2)$ formed by 
the relevant coupling constants can take on any values from the `square' 
$(0, 1] \times (0, 1]$, except the \emph{single point} $(1 / 2, 1)$.
From the proof of Lemma \ref{LEMMA_regularity}, especially from equation 
\eqref{L2_rel_2}, one may get the impression that even this very slight 
technical assumption can be relaxed by further analyzing the properties of 
column vector $\hat{F}$ \eqref{hat_L_and_hat_F}. However, we do not wish to 
pursue this direction in this chapter.

\section{Analyzing the dynamics}
\label{sec:4.3}

In this section we wish to study the dynamics generated by the Hamiltonian
$H$ \eqref{H}. Recalling the formulae \eqref{u_a} and \eqref{L}, by the
obvious relationship
\be
    H = \sum_{c = 1}^n \cosh(\theta_c) u_c = \half \tr(L)
\label{H_and_u_and_L}
\ee
we can make the first contact of our van Diejen system with the proposed Lax 
matrix $L$ \eqref{L}. As an important ingredient of the forthcoming analysis,
let us introduce the Hamiltonian vector field $\bsX_H \in \mfX(P)$ with the
usual definition
\be
    \bsX_H [f] = \{ f, H \}
    \qquad
    (f \in C^\infty(P)).
\label{bsX_H}
\ee
Working with the convention \eqref{PBs}, for the time evolution of the 
global coordinate functions $\lambda_a$ and $\theta_a$ $(a \in \bN_n)$ 
we can clearly write
\begin{align}
    & \dot{\lambda}_a 
        = \bsX_H [\lambda_a]
        = \PD{H}{\theta_a}
        = \sinh(\theta_a) u_a,
    \label{lambda_dot}
    \\
    & \dot{\theta}_a 
        = \bsX_H [\theta_a]
        = -\PD{H}{\lambda_a}
        = - \sum_{c = 1}^n \cosh(\theta_c) u_c \PD{\ln(u_c)}{\lambda_a}.
    \label{theta_dot}
\end{align}
To make the right-hand side of \eqref{theta_dot} more explicit, let us 
display the logarithmic derivatives of the constituent functions $u_c$. 
Notice that for all $a \in \bN_n$ we can write
\be
    \PD{\ln(u_a)}{\lambda_a}
    = -\Real 
        \bigg(
            \frac{2 \ri \sin(\nu)}
                {\sinh(2 \lambda_a) \sinh(\ri \nu + 2 \lambda_a)}
            +\sum_{\substack{j = 1 \\ (j \neq a, n + a)}}^N
                \frac{\ri \sin(\mu)}
                    {\sinh(\lambda_a - \Lambda_j) 
                    \sinh(\ri \mu + \lambda_a - \Lambda_j)}
        \bigg),
\label{PD_u_a_lambda_a}
\ee
while if $c \in \bN_n$ and $c \neq a$, then we have
\be
    \PD{\ln(u_c)}{\lambda_a}
    = \Real
        \bigg(
            \frac{\ri \sin(\mu)}
                {\sinh(\lambda_a - \lambda_c) 
                \sinh(\ri \mu + \lambda_a - \lambda_c)}
            -\frac{\ri \sin(\mu)}
                {\sinh(\lambda_a + \lambda_c) 
                \sinh(\ri \mu + \lambda_a + \lambda_c)}
        \bigg).
\label{PD_u_c_lambda_a}
\ee
The rest of this section is devoted to the study of the Hamiltonian 
dynamical system \eqref{lambda_dot}-\eqref{theta_dot}.

\subsection{Completeness of the Hamiltonian vector field}
\label{subsec:4.3.1}

Undoubtedly, the Hamiltonian \eqref{H} does not take the usual form one finds
in the standard textbooks on classical mechanics. It is thus inevitable that
we have even less intuition about the generated dynamics than in the case of
the `natural systems' characterized by a kinetic term plus a potential. To
get a finer picture about the solutions of the Hamiltonian dynamics 
\eqref{lambda_dot}-\eqref{theta_dot}, we start our study with a brief analysis
on the completeness of the Hamiltonian vector field $\bsX_H$ \eqref{bsX_H}.

As the first step, we introduce the strictly positive constant
\be
    \cS = \min \{ \vert \sin(\mu) \vert, \vert \sin(\nu) \vert \} \in (0, 1].
\label{cS}
\ee
Giving a glance at \eqref{u_a}, it is evident that
\be
    u_n 
    > \left( 1 + \frac{\sin(\nu)^2}{\sinh(2 \lambda_n)^2} \right)^\half
    > \frac{\vert \sin(\nu) \vert}{\sinh(2 \lambda_n)}
    \geq \frac{\cS}{\sinh(2 \lambda_n)},
\label{u_n_estim}
\ee
while for all $c \in \bN_{n - 1}$ we can write
\be
    u_c 
    > \left( 
        1 + \frac{\sin(\mu)^2}{\sinh(\lambda_c - \lambda_{c + 1})^2} 
    \right)^\half
    > \frac{\vert \sin(\mu) \vert}{\sinh(\lambda_c - \lambda_{c + 1})}
    \geq \frac{\cS}{\sinh(\lambda_c - \lambda_{c + 1})}.
\label{u_c_estim}
\ee
Keeping in mind the above trivial inequalities, we are ready to prove the 
following result.

\begin{theorem}
\label{THEOREM_Completeness}
The Hamiltonian vector field $\bsX_H$ \eqref{bsX_H} generated by the van Diejen
type Hamiltonian function $H$ \eqref{H} is complete. That is, the maximum
interval of existence of each integral curve of $\bsX_H$ is the whole real 
axis $\R$.
\end{theorem}

\begin{proof}
Take an arbitrary point
\be
    \gamma_0 = (\lambda^{(0)}, \theta^{(0)}) \in P,
\label{gamma_0}
\ee
and let
\be
    \gamma \colon (\alpha, \beta) \to P,
    \qquad
    t \mapsto \gamma(t) = (\lambda(t), \theta(t))
\label{gamma}
\ee
be the unique maximally defined integral curve of $\bsX_H$ with 
$-\infty \leq \alpha < 0 < \beta \leq \infty$ satisfying the initial condition
$\gamma(0) = \gamma_0$. Since the Hamiltonian $H$ is smooth, the existence,
the uniqueness, and also the smoothness of such a maximal solution are obvious.
Our goal is to show that for the domain of the  maximally defined trajectory
$\gamma$ \eqref{gamma} we have  $(\alpha, \beta) = \R$; that is, 
$\alpha = -\infty$ and $\beta = \infty$.

Arguing by contradiction, first let us \emph{suppose} that $\beta < \infty$. 
Since the Hamiltonian $H$ is a first integral, for all $t \in (\alpha, \beta)$ 
and for all $a \in \bN_n$ we can write
\be
    H(\gamma_0) 
    = H(\gamma(t)) 
    = \sum_{c = 1}^n \cosh(\theta_c(t)) u_c(\lambda(t))
    > \cosh(\theta_a(t)) u_a(\lambda(t)),
\label{estim_1}
\ee
whence the estimation
\be
    H(\gamma_0) 
    > \cosh(\vert \theta_a(t) \vert) \geq \half e^{\vert \theta_a(t) \vert}
\label{estim_2}
\ee
is also immediate. Thus, upon introducing the cube
\be
    \cC
    = [-\ln(2 H(\gamma_0)), \ln(2 H(\gamma_0))]^n 
        \subseteq \R^n,
\label{cube}
\ee
from \eqref{estim_2} we infer at once that
\be
    \theta(t) \in \cC
    \qquad
    (t \in (\alpha, \beta)).
\label{theta_in_cube}
\ee
Turning to the equations \eqref{lambda_dot} and \eqref{estim_1}, we can 
cook up an estimation on the growing of the vector $\lambda(t)$, too. 
Indeed, we see that
\be
    \vert \dot{\lambda}_1(t) \vert
    = \sinh(\vert \theta_1(t) \vert) u_1(\lambda(t))
    \leq \cosh(\vert \theta_1(t) \vert) u_1(\lambda(t))
    < H(\gamma_0)
    \qquad
    (t \in (\alpha, \beta)),
\label{lambda_dot_estim}
\ee 
that implies immediately that for all $t \in [0, \beta)$ we have
\be
    \vert \lambda_1(t) - \lambda_1^{(0)} \vert
    = \vert \lambda_1(t) - \lambda_1(0) \vert
    = \left \vert \int_0^t \dot{\lambda}_1(s) \, \dd s \right \vert
    \leq \int_0^t \vert \dot{\lambda}_1(s) \vert \, \dd s
    \leq t H(\gamma_0) < \beta H(\gamma_0).
\label{lambda_1_estim}
\ee
Therefore, with the aid of the strictly positive constant 
\be
    \rho = \lambda_1^{(0)} + \beta H(\gamma_0) \in (0, \infty), 
\label{rho}
\ee    
we end up with the estimation
\be
    \lambda_1(t) 
    = \vert \lambda_1(t) \vert
    = \vert \lambda_1^{(0)} + \lambda_1(t) - \lambda_1^{(0)} \vert
    \leq \vert \lambda_1^{(0)} \vert 
        + \vert \lambda_1(t) - \lambda_1^{(0)} \vert
    < \rho
    \qquad
    (t \in [0, \beta)).
\label{lambda_1_estim_OK}
\ee
Since $\lambda(t)$ moves in the configuration space $Q$ \eqref{Q}, the
above observation entails that
\be
    \rho > \lambda_1(t) > \ldots > \lambda_n(t) > 0
    \qquad
    (t \in [0, \beta)).
\label{lambda_estim}
\ee

To proceed further, now for all $\eps > 0$ we define the subset 
$Q_\eps \subseteq \R^n$ consisting of those real $n$-tuples 
$x = (x_1, \dots, x_n) \in \R^n$ that satisfy the inequalities
\be
    \rho \geq x_1
    \text{ and }
    2 x_n \geq \eps
    \text{ and }
    x_c \geq x_{c + 1} + \eps
    \text{ for all } c \in \bN_{n - 1},
\label{Q_eps_def}
\ee
simultaneously. In other words, 
\be
    Q_\eps 
    = \{ x \in \R^n \, | \, \rho \geq x_1 \} 
        \cap \{ x \in \R^n \, | \, 2 x_n \geq \eps \}
        \cap
            \bigcap_{c = 1}^{n - 1} 
                \{ x \in \R^n \, | \, x_c - x_{c + 1} \geq \eps \}.
\label{Q_eps_OK}
\ee
Notice that $Q_\eps$ is a bounded and closed subset of $\R^n$. Moreover, by
comparing the definitions \eqref{Q} and \eqref{Q_eps_def}, it is evident that
$Q_\eps \subseteq Q$. Since the cube $\cC$ \eqref{cube} is also a compact 
subset of $\R^n$, we conclude that the Cartesian product $Q_\eps \times \cC$ 
is a \emph{compact} subset of the phase space $P$ \eqref{P}. Therefore, due 
to the assumption $\beta < \infty$, after some time the maximally defined
trajectory $\gamma$ \eqref{gamma} escapes from $Q_\eps \times \cC$, as can
be read off from any standard reference on dynamical systems (see e.g.
\cite[Theorem 2.1.18]{AM78}). More precisely, there is some 
$\tau_\eps \in [0,\beta)$ such that
\be
    (\lambda(t), \theta(t)) 
    \in 
    P \setminus (Q_\eps \times \cC)
    = \left( (Q \setminus Q_\eps) \times \cC \right)
        \cup
        \left( Q \times (\R^n \setminus \cC) \right)    
    \qquad
    (t \in (\tau_\eps, \beta)),
\label{gamma_escapes}
\ee
where the union above is actually a disjoint union. For instance,
due to the relationship \eqref{theta_in_cube}, at the mid-point
\be
    t_\eps = \frac{\tau_\eps + \beta}{2} \in (\tau_\eps, \beta)
\label{t_eps}
\ee
we can write that
\be
    \lambda(t_\eps) 
    \in 
    Q \setminus Q_\eps \subseteq \R^n \setminus Q_\eps.
\label{lambda_escapes}
\ee 
Therefore, simply by taking the complement of $Q_\eps$ \eqref{Q_eps_OK},
and also keeping in mind \eqref{lambda_1_estim_OK}, it is evident that 
\be
    \min 
    \{ \lambda_1(t_\eps) - \lambda_2(t_\eps),
        \dots,
        \lambda_{n - 1}(t_\eps) - \lambda_n(t_\eps),
        2 \lambda_n(t_\eps)
    \} < \eps,
\label{key_for_completeness}
\ee
which in turn implies that
\be
    \max 
    \left\{ 
        \frac{1}{\sinh(\lambda_1(t_\eps) - \lambda_2(t_\eps))},
        \dots,
        \frac{1}{\sinh(\lambda_{n - 1}(t_\eps) - \lambda_n(t_\eps))},
        \frac{1}{\sinh(2 \lambda_n(t_\eps))}
    \right\} 
    > \frac{1}{\sinh(\eps)}.
\label{key_for_completeness_OK}
\ee
Now, since $\eps > 0$ was arbitrary, the estimations \eqref{u_n_estim} and 
\eqref{u_c_estim} immediately lead to the contradiction
\be
\begin{split}
    H(\gamma_0) 
    & = H(\gamma(t_\eps)) 
    = \sum_{c = 1}^n \cosh(\theta_c(t_\eps)) u_c(\lambda(t_\eps)) 
    \\
    & \geq \sum_{c = 1}^n u_c(\lambda(t_\eps))
    > \frac{\cS}{\sinh(2 \lambda_n(t_\eps))}
        + \sum_{c = 1}^{n - 1} 
            \frac{\cS}{\sinh(\lambda_c(t_\eps) - \lambda_{c + 1}(t_\eps))}
    > \frac{\cS}{\sinh(\eps)}.
\end{split}
\label{contradiction}
\ee
Therefore, necessarily, $\beta = \infty$.

Either by repeating the above ideas, or by invoking a time-reversal argument, 
one can also show that $\alpha = -\infty$, whence the proof is complete.
\end{proof}

\subsection{Dynamics of the vector $F$}
\label{subsec:4.3.2}

Looking back to the definition \eqref{L}, we see that the column vector $F$ 
\eqref{F} is important building block of the matrix $L$. Therefore, the study 
of the derivative of $L$ along the Hamiltonian vector field $\bsX_H$ 
\eqref{bsX_H} does require close control over the derivative of the components 
of $F$, too. Upon introducing the auxiliary functions
\be
    \varphi_k = \frac{1}{F_k} \bsX_H [F_k]
    \qquad
    (k \in \bN_N), 
\label{varphi}
\ee
for all $a \in \bN_n$ we can write
\be
\begin{split}
    2 \varphi_a  
    & = \bsX_H [\ln (F_a^2)]
    = \bsX_H [\theta_a + \ln(u_a)]
    = \left\{ \theta_a + \ln(u_a), H \right\} 
    \\
    & = \sum_{c = 1}^n
        \left(
            \sinh(\theta_c) u_c \PD{\ln(u_a)}{\lambda_c}
            - \cosh(\theta_c) u_c \PD{\ln(u_c)}{\lambda_a}
        \right).
\end{split}
\label{varphi_a}
\ee
Therefore, due to the explicit formulae \eqref{PD_u_a_lambda_a} and 
\eqref{PD_u_c_lambda_a}, we have complete control over the first $n$
components of \eqref{varphi}. Turning to the remaining components, from 
the definition \eqref{F} it is evident that 
$F_{n + a} = F_a^{-1} \bar{z}_a$, whence the relationship
\be
    \varphi_{n + a}
    = -\varphi_a + \frac{1}{\bar{z}_a} \bsX_H [\bar{z}_a]
    = -\varphi_a 
        + \sum_{c = 1}^n 
            \sinh(\theta_c) u_c \frac{1}{\bar{z}_a} \PD{\bar{z}_a}{\lambda_c}
\label{varphi_n+a}
\ee
follows immediately. Notice that for all $a \in \bN_n$ we can write that
\be
    \frac{1}{z_a} \PD{z_a}{\lambda_a}
    = -\frac{2 \ri \sin(\nu)}
            {\sinh(2 \lambda_a) \sinh(\ri \nu + 2 \lambda_a)}
        -\sum_{\substack{j = 1 \\ (j \neq a, n + a)}}^N
            \frac{\ri \sin(\mu)}
                {\sinh(\lambda_a - \Lambda_j) 
                \sinh(\ri \mu + \lambda_a -\Lambda_j)},
\label{PD_z_a_lambda_a}
\ee
whereas if $c \in \bN_n$ and $c \neq a$, then we find immediately that
\be
    \frac{1}{z_a} \PD{z_a}{\lambda_c}
    = \frac{\ri \sin(\mu)}
            {\sinh(\lambda_a - \lambda_c) 
            \sinh(\ri \mu + \lambda_a - \lambda_c)}
        -\frac{\ri \sin(\mu)}
            {\sinh(\lambda_a + \lambda_c) 
            \sinh(\ri \mu + \lambda_a + \lambda_c)}.
\label{PD_z_a_lambda_c}
\ee
The above observations can be summarized as follows.

\begin{proposition}
\label{PROPOSITION_varphi}
For the derivative of the components of the function $F$ \eqref{F} along the 
Hamiltonian vector field $\bsX_H$ \eqref{bsX_H} we have
\be
    \bsX_H [F_k] = \varphi_k F_k
    \qquad
    (k \in \bN_N),
\label{F_der}
\ee
where for each $a \in \bN_n$ we can write
\be
    \varphi_a 
    = \Real
        \bigg(
            \frac{\ri \sin(\nu) e^{-\theta_a} u_a}
                {\sinh(2 \lambda_a) \sinh(\ri \nu + 2 \lambda_a)}
            + \half \sum_{\substack{j = 1 \\ (j \neq a, n + a)}}^N
                \frac{\ri \sin(\mu) (e^{-\theta_a} u_a + e^{\Theta_j} u_j)}
                    {\sinh(\lambda_a - \Lambda_j) 
                    \sinh(\ri \mu + \lambda_a - \Lambda_j)}
        \bigg),
\label{varphi_a_OK}
\ee
whereas
\be
    \varphi_{n + a} 
    = -\varphi_a
    - \frac{2 \ri \sin(\nu) \sinh(\theta_a) u_a}
        {\sinh(2 \lambda_a) \sinh(\ri \nu - 2 \lambda_a)}
    -\sum_{\substack{j = 1 \\ (j \neq a, n + a)}}^N
        \frac{\ri \sin(\mu) (\sinh(\theta_a) u_a - \sinh(\Theta_j) u_j)}
            {\sinh(\lambda_a - \Lambda_j) 
            \sinh(\ri \mu - \lambda_a + \Lambda_j)}.
\label{varphi_n+a_OK}
\ee
\end{proposition}

By invoking Proposition \ref{PROPOSITION_L_in_G}, let us observe that for the 
inverse of the matrix $L$ \eqref{L} we can write that $L^{-1} = C L C$, whence
for the Hermitian matrix $L - L^{-1}$ we have
\be
    (L - L^{-1}) C + C (L - L^{-1})
    = L C - C L C^2 + C L - C^2 L C 
    = 0.
\label{L-L_inv_Hermitian}
\ee
Thus, the matrix valued smooth function $(L - L^{-1}) / 2$ defined on the
phase space $P$ \eqref{P} takes values in the subspace $\mfp$ \eqref{mfp}.
Therefore, by taking its projection onto the Abelian subspace $\mfa$ 
\eqref{mfa}, we obtain the diagonal matrix
\be
    D = (L - L^{-1})_\mfa / 2 \in \mfa
\label{D}
\ee
with diagonal entries
\be
    D_{j, j} = \sinh(\Theta_j) u_j
    \qquad
    (j \in \bN_N).
\label{D_entries}
\ee
Next, by projecting the function $(L - L^{-1}) / 2$ onto the complementary
subspace $\mfa^\perp$, we obtain the off-diagonal matrix
\be
    Y=(L - L^{-1})_{\mfa^\perp} / 2 \in \mfa^\perp,
\label{Y}
\ee
which in turn allows us to introduce the matrix valued smooth function
\be
    Z = \sinh(\wad_{\bsLambda})^{-1} Y \in \mfm^\perp,
\label{Z}
\ee
too. Since $\lambda \in Q$ \eqref{Q}, the corresponding diagonal matrix 
$\bsLambda$ \eqref{bsLambda} is regular in the sense that it takes values 
in the open Weyl chamber $\mfc \subseteq \mfa_\reg$ \eqref{mfc}. Therefore, 
$Z$ is indeed a well-defined off-diagonal $N \times N$ matrix, and its 
non-trivial entries take the form
\be
    Z_{k, l} 
    = \frac{Y_{k, l}}{\sinh(\Lambda_k - \Lambda_l)}
    = \frac{L_{k, l} - (L^{-1})_{k, l}}
        {2 \sinh(\Lambda_k - \Lambda_l)}
    \qquad
    (k, l \in \bN_N, \, k \neq l).
\label{Z_entries}
\ee
Utilizing $Z$, for each $a \in \bN_n$ we also define the function
\be
    \cM_a 
    = \frac{\ri}{F_a} \Imag((Z F)_a) 
    = \frac{\ri}{F_a} 
        \Imag \left( \sum_{j = 1}^N Z_{a, j} F_j \right)
        \in C^\infty(P).
\label{cM_a}
\ee
Recalling the subspace $\mfm$ \eqref{mfm}, it is clear that
\be
    B_\mfm = \diag(\cM_1, \dots, \cM_n, \cM_1, \dots, \cM_n) \in \mfm
\label{B_mfm}
\ee
is a well-defined function. Having the above objects at our disposal, the 
content of Proposition \ref{PROPOSITION_varphi} can be recast into a more 
convenient matrix form as follows.

\begin{lemma}
\label{LEMMA_F_derivative}
With the aid of the smooth functions $Z$ \eqref{Z} and $B_\mfm$ \eqref{B_mfm}, 
for the derivative of the column vector $F$ \eqref{F} along the Hamiltonian 
vector field $\bsX_H$ \eqref{bsX_H} we can write
\be
    \bsX_H [F] = (Z - B_{\mfm}) F.
\label{F_der_matrix_form}
\ee
\end{lemma}

\begin{proof}
Upon introducing the column vector 
\be
    J = \bsX_H [F] + B_{\mfm} F - Z F \in \bC^N,
\label{J}
\ee 
it is enough to prove that $J_k = 0$ for all $k \in \bN_N$, at each point
$(\lambda, \theta)$ of the phase space $P$ \eqref{P}. Starting with the
upper $n$ components of $J$, notice that by Proposition 
\ref{PROPOSITION_varphi} and the formulae \eqref{Z_entries}-\eqref{B_mfm} we 
can write that
\be
    J_a = \half e^{-\frac{\theta_a}{2}} u_a^{\frac{3}{2}} G_a
    \qquad
    (a \in \bN_n),
\label{J_a}
\ee
where $G_a$ is an appropriate function depending only on $\lambda$. More
precisely, it has the form
\be
\begin{split}
    G_a 
    = \Real
        \bigg(
            & \frac{2 \ri \sin(\nu)}{\sinh(2 \lambda_a) 
                \sinh(\ri \nu + 2 \lambda_a)}
            + \sum_{\substack{j = 1 \\ (j \neq a, n + a)}}^N
                \frac{\ri \sin(\mu) (1 + \bar{z}_j \bar{z}_a^{-1})}
                    {\sinh(\lambda_a - \Lambda_j) 
                    \sinh(\ri \mu + \lambda_a - \Lambda_j)}
            \\
            & + \frac{\ri \sin(\mu) (z_a \bar{z}_a^{-1} - 1)
                        + \ri \sin(\mu - \nu)(\bar{z}_a^{-1} - z_a^{-1})}
                {\sinh(2 \lambda_a) \sinh(\ri \mu + 2 \lambda_a)}
        \bigg),
\end{split}
\label{G_a}
\ee
that can be made quite explicit by exploiting the definition of the 
constituent functions $z_j$ \eqref{z_a}. Now, following the same strategy 
we applied in the proof of Proposition \ref{PROPOSITION_L_in_G}, let us 
introduce a complex valued function $g_a$ depending only on a single 
complex variable $w$, obtained simply by replacing $\lambda_a$ with 
$\lambda_a + w$ in the explicit expression of right-hand side of the 
above equation \eqref{G_a}. In mod $\ri \pi$ sense this meromorphic
function has at most first order poles at the points
\be
    w \equiv -\lambda_a, \, 
    w \equiv \pm \ri \mu / 2 - \lambda_a, \, 
    w \equiv \pm \ri \nu / 2 - \lambda_a, \,
    w \equiv \Lambda_j - \lambda_a, \, 
    w \equiv \pm (\ri \mu + \Lambda_j) - \lambda_a \, (j \in \bN_N).
\label{poles-a}
\ee
However, at each of these points the residue of $g_a$ turns out to be zero. 
Moreover, it is obvious that $g_a(w)$ vanishes as $\Real(w) \to \infty$, 
therefore Liouville's theorem implies that $g_a(w) = 0$ for all $w \in \bC$. 
In particular $G_a = g_a(0) = 0$, and so by \eqref{J_a} we conclude that 
$J_a = 0$.

Turning to the lower $n$ components of the column vector $J$ \eqref{J}, let 
us note that our previous result $J_a = 0$ allows us to write that
\be
    J_{n+a}
    = -\sinh(\theta_a) e^{-\frac{\theta_a}{2}} 
        u_a^{\half} \bar{z}_a G_{n + a}
    \qquad
    (a \in \bN_n),
\label{J_n+a}
\ee
where $G_{n + a}$ is again an appropriate smooth function depending only on 
$\lambda$, as can be seen from the formula
\begin{align}
    G_{n + a} 
    = & \frac{2 \ri \sin(\nu)}
            {\sinh(2 \lambda_a) \sinh(\ri \nu - 2 \lambda_a)}
    - \frac{\ri \sin(\mu) + \ri \sin(\mu - \nu) \bar{z}_a^{-1}}
            {\sinh(2 \lambda_a) \sinh(\ri \mu - 2 \lambda_a)}
    + \bar{z}_a^{-1} 
        \frac{\ri \sin(\mu) z_a + \ri \sin(\mu-\nu)}
            {\sinh(2 \lambda_a) \sinh(\ri \mu + 2 \lambda_a)}
    \nonumber \\
    & + \sum_{\substack{j = 1 \\ (j \neq a, n + a)}}^N
            \frac{1}{\sinh(\lambda_a - \Lambda_j)}
            \left(
                \frac{\ri \sin(\mu)}{\sinh(\ri \mu - \lambda_a + \Lambda_j)}
                + \frac{\ri \sin(\mu) \bar{z}_a^{-1} \bar{z}_j}
                    {\sinh(\ri \mu + \lambda_a - \Lambda_j)}
            \right).
\label{G_n+a}
\end{align}
Next, let us plug the definition of $z_j$ \eqref{z_a} into the above
expression \eqref{G_n+a} and introduce the complex valued function 
$g_{n + a}$ of $w \in \bC$ by replacing $\lambda_a$ with $\lambda_a + w$ 
in the resulting formula. Note that $g_{n + a}$ has at most first order 
poles at the points
\be
    w \equiv -\lambda_a, \,
    w \equiv \pm \ri \mu / 2 - \lambda_a, \, 
    w \equiv \Lambda_j - \lambda_a \, (j \in \bN_N)
    \pmod{\ri \pi},
\label{poles-n+a}
\ee
but all these singularities are removable. Since $g_{n + a}(w) \to 0$ 
as $\Real(w) \to \infty$, the boundedness of the periodic function 
$g_{n + a}$ is also obvious. Thus, Liouville's theorem entails that 
$g_{n + a}= 0$ on the whole complex plane, whence the relationship 
$G_{n + a} = g_{n + a}(0) = 0$ also follows. Now, looking back to the
equation \eqref{J_n+a}, we end up with the desired equation 
$J_{n + a} = 0$.
\end{proof}

\subsection{Lax representation of the dynamics}
\label{subsec:4.3.3}

Based on our proposed Lax matrix \eqref{L}, in this subsection we wish to
construct a Lax representation for the dynamics of the van Diejen system 
\eqref{H}. As it turns out, Lemmas \ref{LEMMA_commut_rel} and 
\ref{LEMMA_F_derivative} prove to be instrumental in our approach. As the 
first step, by applying the Hamiltonian vector field $\bsX_H$ \eqref{bsX_H}
on the Ruijsenaars type commutation relation \eqref{commut_rel}, let us 
observe that the Leibniz rule yields
\be
\begin{split}
    & e^{\ri \mu} e^{\ad_{\bsLambda}} 
        \left( 
            \bsX_H [L] 
            - \left[ L, e^{-\bsLambda} \bsX_H [e^{\bsLambda}] \right] 
        \right) 
    - e^{-\ri \mu} e^{-\ad_{\bsLambda}} 
        \left( 
            \bsX_H [L] 
            + \left[ L, \bsX_H [e^{\bsLambda}] e^{-\bsLambda} \right] 
        \right)
    \\
    & \quad = 2 \ri \sin(\mu) 
                \left( 
                    \bsX_H [F] F^* + F (\bsX_H [F])^* 
                \right).
\end{split}
\label{commut_rel_der_1}
\ee
By comparing the formula appearing in \eqref{lambda_dot} with the matrix
entries \eqref{D_entries} of the diagonal matrix $D$, it is clear that
\be
    \bsX_H [\bsLambda] = D,
\label{bsLambda_der}
\ee 
which in turn implies that
\be
    e^{-\bsLambda} \bsX_H [e^{\bsLambda}] 
    = \bsX_H [e^{\bsLambda}] e^{-\bsLambda} 
    = D.
\label{e_bsLambda_der}
\ee
Thus, the above equation \eqref{commut_rel_der_1} can be cast into the 
fairly explicit form
\be
\begin{split}
    & e^{\ri \mu} e^{\ad_{\bsLambda}} 
        \left( 
            \bsX_H [L] - \left[ L, D \right] 
        \right)
    - e^{-\ri \mu} e^{-\ad_{\bsLambda}} 
        \left( 
            \bsX_H [L] + \left[ L, D \right] 
        \right) 
    \\
    & \quad = 2 \ri \sin(\mu) 
                \left( 
                    \bsX_H [F] F^* + F (\bsX_H [F])^* 
                \right),
\end{split}
\label{commut_rel_der_2}
\ee
which serves as the starting point in our analysis on the derivative 
$\bsX_H [L]$. Before formulating the main result of this subsection, over
the phase space $P$ \eqref{P} we define the matrix valued function
\be
    B_{\mfm^\perp} = -\coth(\wad_{\bsLambda}) Y \in \mfm^\perp.
\label{B_mfm_perp}
\ee
Recalling the definition \eqref{Y}, we see that $B_{\mfm^\perp}$ is actually 
an off-diagonal matrix. Furthermore, for its non-trivial entries we have the 
explicit expressions
\be
    (B_{\mfm^\perp})_{k, l} 
    = - \coth(\Lambda_k - \Lambda_l) \frac{L_{k, l} - (L^{-1})_{k, l}}{2}
    \qquad
    (k, l \in \bN_N, \, k \neq l).
\label{B_mfm_perp_entries}
\ee
Finally, with the aid of the diagonal matrix $B_{\mfm}$ \eqref{B_mfm}, over
the phase space $P$ \eqref{P} we also define the $\mfk$-valued smooth function
\be
    B = B_{\mfm} + B_{\mfm^\perp} \in \mfk.
\label{B}
\ee

\begin{theorem}
\label{THEOREM_Lax_representation}
The derivative of the matrix valued function $L$ \eqref{L} along the 
Hamiltonian vector field $\bsX_H$ \eqref{bsX_H} takes the Lax form
\be
    \bsX_H [L] = [L, B].
\label{Lax_representation}
\ee
In other words, the matrices $L$ \eqref{L} and $B$ \eqref{B} provide a Lax 
pair for the dynamics generated by the Hamiltonian \eqref{H}.
\end{theorem}

\begin{proof}
For simplicity, let us introduce the matrix valued smooth functions
\be
    \Psi = \bsX_H [L] - [L, B]
    \quad \text{and} \quad
    R = \sinh(\ri \mu \Id_{\mfgl(N, \bC)} + \ad_{\bsLambda}) \Psi
\label{Psi_and_R}
\ee
defined on the phase space $P$ \eqref{P}. Our goal is to prove that $\Psi = 0$. 
However, since $\sin(\mu) \neq 0$, the linear operator
\be
    \sinh(\ri \mu \Id_{\mfgl(N, \bC)} + \ad_{\bsLambda})
    \in \End(\mfgl(N, \bC))
\label{Lax_repr_linear_op}
\ee
is invertible at each point of $P$, whence it is enough to show that $R = 0$. 
For this reason, notice that from the relationship \eqref{commut_rel_der_2}
we can infer that
\be
\begin{split}
    2 R 
    = & 
        e^{\ri \mu} e^{\ad_{\bsLambda}} \Psi
        - e^{-\ri \mu} e^{-\ad_{\bsLambda}} \Psi
    \\
    = & 2 \ri \sin(\mu) \left( \bsX_H [F] F^* + F (\bsX_H [F])^* \right) 
        -\left( 
            e^{\ri \mu} e^{\ad_{\bsLambda}} [L, B_{\mfm^\perp}]
            - e^{-\ri \mu} e^{-\ad_{\bsLambda}} [L, B_{\mfm^\perp}] 
        \right)
    \\ 
    &   -\left( 
            e^{\ri \mu} e^{\ad_{\bsLambda}} [L, B_{\mfm}]
            - e^{-\ri \mu} e^{-\ad_{\bsLambda}} [L, B_{\mfm}] 
        \right)
        -\left( 
            e^{\ri \mu} e^{\ad_{\bsLambda}} [D, L]
            + e^{-\ri \mu} e^{-\ad_{\bsLambda}} [D, L] 
        \right).
\end{split}
\label{R_expanded}
\ee
Our strategy is to inspect the right-hand side of the above equation 
term-by-term.

As a preparatory step, from the definitions of $D$ \eqref{D} and
$Y$ \eqref{Y} we see that
\be
    (L - L^{-1}) / 2 = D + Y,
\label{D+Y}
\ee
thus the commutation relation
\be
    [L, Y] = [ L, -D + (L - L^{-1}) / 2 ] = [D, L]
\label{L_Y_D_comm_rel}
\ee
readily follows. Keeping in mind the relationship \eqref{L_Y_D_comm_rel}
and the standard hyperbolic functional equations
\be
    \coth(w) \pm 1 = \frac{e^{\pm w}}{\sinh(w)}
    \qquad
    (w \in \bC),
\label{coth_iden}
\ee
from the definitions of $Z$ \eqref{Z} and $B_{\mfm^\perp}$ \eqref{B_mfm_perp}
we infer that
\be
\begin{split}
    e^{\ad_{\bsLambda}} [L, B_{\mfm^\perp}]
    = & -e^{\ad_{\bsLambda}} [L, \coth(\wad_{\bsLambda}) Y]
    \\
    = & -e^{\ad_{\bsLambda}}
        \left(
            [L, (\coth(\wad_{\bsLambda}) 
                - \Id_{\mfm^\perp \oplus \mfa^\perp}) Y]
            + [L, Y]
        \right)
    \\
    = & -e^{\ad_{\bsLambda}}
        \left(
            [L, e^{-\wad_{\bsLambda}} \sinh(\wad_{\bsLambda})^{-1} Y]
            + [D, L]
        \right)
    \\
    = & -[e^{\ad_{\bsLambda}} L, Z] - e^{\ad_{\bsLambda}} [D, L].
\end{split}
\label{term_1_first_half}
\ee
Along the same lines, one finds immediately that
\be
    e^{-\ad_{\bsLambda}} [L, B_{\mfm^\perp}] 
    = -[e^{-\ad_{\bsLambda}} L, Z] + e^{-\ad_{\bsLambda}} [D, L].
\label{term_1_second_half}
\ee
At this point let us recall that $Z$ \eqref{Z} takes values in the subspace 
$\mfm^\perp \subseteq \mfk$, thus it is anti-Hermitian and commutes with the 
matrix $C$ \eqref{C}. Therefore, by utilizing equations
\eqref{term_1_first_half} and \eqref{term_1_second_half}, the application of 
the commutation relation \eqref{commut_rel} leads to the relationship
\be
\begin{split}
    & e^{\ri \mu} e^{\ad_{\bsLambda}} [L, B_{\mfm^\perp}]
    - e^{-\ri \mu} e^{-\ad_{\bsLambda}} [L, B_{\mfm^\perp}]
    \\
    & \quad 
    = -[e^{\ri \mu} e^{\ad_{\bsLambda}} L 
        - e^{-\ri \mu} e^{-\ad_{\bsLambda}} L, Z]
        - \left( 
            e^{\ri \mu} e^{\ad_{\bsLambda}} [D, L]
            + e^{-\ri \mu} e^{-\ad_{\bsLambda}} [D, L] 
        \right)
    \\
    & \quad
    = -[2 \ri \sin(\mu) F F^* + 2 \ri \sin(\mu - \nu) C, Z]
        - \left( 
            e^{\ri \mu} e^{\ad_{\bsLambda}} [D, L]
            + e^{-\ri \mu} e^{-\ad_{\bsLambda}} [D, L] 
        \right)
    \\
    & \quad
    = 2 \ri \sin(\mu) \left( (Z F) F^* + F (Z F)^* \right)
        - \left( 
            e^{\ri \mu} e^{\ad_{\bsLambda}} [D, L]
            + e^{-\ri \mu} e^{-\ad_{\bsLambda}} [D, L] 
        \right).
\end{split}
\label{term_1}
\ee
To proceed further, let us recall that $B_\mfm$ \eqref{B_mfm} takes values 
in $\mfm \subseteq \mfk$, whence it is also anti-Hermitian and also commutes 
with the matrix $C$ \eqref{C}. Thus, by applying commutation relation
\eqref{commut_rel} again, we obtain at once that
\be
\begin{split}
    & e^{\ri \mu} e^{\ad_{\bsLambda}} [L, B_{\mfm}]
    - e^{-\ri \mu} e^{-\ad_{\bsLambda}} [L, B_{\mfm}]
    \\
    & \quad 
    = e^{\ri \mu} [e^{\ad_{\bsLambda}} L, e^{\ad_{\bsLambda}} B_\mfm]
        -e^{-\ri \mu} [e^{-\ad_{\bsLambda}} L, e^{-\ad_{\bsLambda}} B_\mfm]
    = [e^{\ri \mu} e^{\ad_{\bsLambda}} L 
        - e^{-\ri \mu} e^{-\ad_{\bsLambda}} L, B_\mfm]
    \\
    & \quad
    = [2 \ri \sin(\mu) F F^* + 2 \ri \sin(\mu - \nu) C, B_\mfm]
    = -2 \ri \sin(\mu) \left( (B_\mfm F) F^* + F (B_\mfm F)^* \right).
\end{split}
\label{term_2}
\ee
Now, by plugging the expressions \eqref{term_1} and \eqref{term_2}
into \eqref{R_expanded}, we obtain that
\be
    R = \ri \sin(\mu) 
        \left(
            (\bsX_H [F] - Z F + B_\mfm F) F^*
            + F (\bsX_H [F] - Z F + B_\mfm F)^*
        \right).
\label{R_OK}
\ee
Giving a glance at Lemma \ref{LEMMA_F_derivative}, we conclude that $R = 0$, 
thus the Theorem follows.
\end{proof}

At this point we wish to make a short comment on matrix 
$B = B(\lambda, \theta; \mu, \nu)$ \eqref{B} appearing in the Lax 
representation \eqref{Lax_representation} of the dynamics \eqref{H}. It is 
an important fact that by taking its `rational limit' we can recover the 
second member of the Lax pair of the rational $C_n$ van Diejen system with 
two parameters $\mu$ and $\nu$. More precisely, up to some irrelevant 
numerical factors, in the $\alpha \to 0^+$ limit the matrix 
$\alpha B(\alpha \lambda, \theta; \alpha \mu, \alpha \nu)$ tends to the 
second member $\hat{\cB}(\lambda, \theta; \mu, \nu, \kappa = 0)$ of the 
rational Lax pair, that first appeared in equation (4.60) of the recent paper 
\cite{Pu15}. In other words, matrix $B$ \eqref{B} is an appropriate 
hyperbolic generalization of the `rational' matrix $\hat{\cB}$ with two 
coupling parameters. We can safely state that the results presented in 
\cite{Pu15} has played a decisive role in our present work. As a 
matter of fact, most probably we could not have guessed the form of the 
non-trivial building blocks \eqref{B_mfm} and \eqref{B_mfm_perp} without the 
knowledge of rational analogue of $B$. 

In order to harvest some consequences of the Lax representation 
\eqref{Lax_representation}, we continue with a simple corollary of Theorem 
\ref{THEOREM_Lax_representation}, that proves to be quite handy in the 
developments of the next subsection.

\begin{proposition}
\label{PROPOSITION_D_and_Y_derivatives}
For the derivatives of the matrix valued smooth functions $D$ \eqref{D}
and $Y$ \eqref{Y} along the Hamiltonian vector field $\bsX_H$ \eqref{bsX_H}
we have
\be
    \bsX_H [D] = [Y, B_{\mfm^\perp}]_{\mfa}
    \quad \text{and} \quad
    \bsX_H [Y] 
    = [Y, B_{\mfm^\perp}]_{\mfa^\perp} 
        + [D, B_{\mfm^\perp}] 
        + [Y, B_{\mfm}].
\label{D_and_Y_der}
\ee
\end{proposition}

\begin{proof}
As a consequence of Proposition \ref{PROPOSITION_L_in_G}, for the inverse of 
$L$ we can write that $L^{-1} = C L C$. Since the matrix valued function $B$ 
\eqref{B} takes values in $\mfk$ \eqref{mfk}, from Theorem 
\ref{THEOREM_Lax_representation} we infer that
\be
    \bsX [L^{-1}] 
    = C \bsX_H [L] C = C [L, B] C = [C L C, C B C] = [L^{-1}, B],
\label{L_inv_der}
\ee
thus the equation
\be
    \bsX_H [(L - L^{-1}) / 2] = [(L - L^{-1}) / 2, B]
\label{bsX_H_on_L-L_inv}
\ee
is immediate. Due to the relationship \eqref{D+Y}, by simply projecting of 
the above equation onto the subspaces $\mfa$ and $\mfa^\perp$, respectively, 
the derivatives displayed in \eqref{D_and_Y_der} follow at once.
\end{proof}

\subsection{Geodesic interpretation}
\label{subsec:4.3.4}

The geometric study of the CMS type integrable systems goes back to the 
fundamental works of Olshanetsky and Perelomov (see e.g. 
\cite{OP76,OP81}). Since their landmark papers the 
so-called projection method has been vastly generalized to cover many 
variants of the CMS type particle systems. By now some result are
available in the context of the RSvD models, too. For details, see
e.g. \cite{KKS78,FP06,FP07,FK09,Pu11-2,Pu12}. The primary goal 
of this subsection is to show that the Hamiltonian flow generated by the 
Hamiltonian \eqref{H} can be also obtained by an appropriate `projection 
method' from the geodesic flow of the Lie group $\UN(n ,n)$. In order to 
make this statement more precise, take the maximal integral curve
\be
    \R \ni 
        t 
        \mapsto (
        \lambda(t), \theta(t)) 
        = (\lambda_1(t), \ldots, \lambda_n(t), 
            \theta_1(t), \ldots, \theta_n(t))
    \in P
\label{trajectory}
\ee
of the Hamiltonian vector field $\bsX_H$ \eqref{bsX_H} satisfying the initial 
condition
\be
    \gamma(0) = \gamma_0,
\label{trajectory_init_cond}
\ee
where $\gamma_0 \in P$ is an arbitrary point. By exploiting Proposition 
\ref{PROPOSITION_D_and_Y_derivatives}, we start our analysis with the 
following observation.

\begin{proposition}
\label{PROPOSITION_ddot_bsLambda_etc}
Along the maximally defined trajectory \eqref{trajectory}, the time evolution
of the diagonal matrix $\bsLambda = \bsLambda(t) \in \mfc$ \eqref{bsLambda} 
obeys the second order differential equation
\be
    \ddot{\bsLambda} + [Y, \coth(\wad_{\bsLambda}) Y]_\mfa = 0,
\label{bsLambda_DE}
\ee
whilst for the evolution of $Y = Y(t)$ \eqref{Y} we have the first order 
equation
\be
    \dot{Y} + [Y, \coth(\wad_{\bsLambda}) Y]_{\mfa^\perp} - [Y, B_{\mfm}]
        + [\dot{\bsLambda}, \coth(\wad_{\bsLambda}) Y] = 0.
\label{Y_DE}
\ee
\end{proposition}

\begin{proof}
Due to equation \eqref{bsLambda_der}, along the solution curve 
\eqref{trajectory} we can write 
\be
    \dot{\bsLambda} = D, 
\label{bsLambda_dot} 
\ee
whereas from the relationships displayed in \eqref{D_and_Y_der} we get
\be
    \dot{D} = [Y, B_{\mfm^\perp}]_{\mfa}
    \quad \text{and} \quad
    \dot{Y} 
    = [Y, B_{\mfm^\perp}]_{\mfa^\perp} 
        + [D, B_{\mfm^\perp}] 
        + [Y, B_{\mfm}].
\label{D_dot_and_Y_dot}
\ee
Recalling the definition \eqref{B_mfm_perp}, equations \eqref{bsLambda_DE}
and \eqref{Y_DE} clearly follow.
\end{proof}

Next, by evaluating the matrices $Z$ \eqref{Z} and $B_\mfm$ \eqref{B_mfm}
along the fixed trajectory \eqref{trajectory}, for all $t \in \R$ we define
\be
    \cK(t) = B_{\mfm}(t) - Z(t) \in \mfk.
\label{cK}
\ee
Since the dependence of $\cK$ on $t$ is smooth, there is a unique maximal
smooth solution 
\be
    \R \ni t \mapsto k(t) \in GL(N, \bC)
\label{k}
\ee
of the first order differential equation
\be
    \dot{k}(t) = k(t) \cK(t)
    \qquad
    (t \in \R)
\label{k_DE}
\ee
satisfying the initial condition 
\be
    k(0) = \bsone_N. 
\label{k_initial_cond}
\ee
Since \eqref{k_DE} is a linear differential equation for $k$, the existence 
of such a global fundamental solution is obvious. Moreover, since $\cK$ 
\eqref{cK} takes values in the Lie algebra $\mfk$ \eqref{mfk}, the trivial 
observations
\be
    \frac{\dd (k C k^*)}{\dd t}
    = \dot{k} C k^* + k C \dot{k}^*
    = k (\cK C + C \cK^*) k^*
    = 0
    \quad \text{and} \quad
    k(0) C k(0)^* = C
\label{k_quadratic_eq}
\ee
imply immediately that $k$ \eqref{k} actually takes values in the subgroup 
$K$ \eqref{K}; that is,
\be
    k(t) \in K
    \qquad
    (t \in \R).
\label{k_in_K}
\ee
Utilizing $k$, we can formulate the most important technical result of this 
subsection.

\begin{lemma}
\label{LEMMA_A}
The smooth function
\be
    \R \ni t 
        \mapsto 
        A(t) = k(t) e^{2 \bsLambda(t)} k(t)^{-1}
    \in \exp(\mfp_\reg)
\label{A}
\ee
satisfies the second order geodesic differential equation
\be
    \frac{\dd}{\dd t} \left( \frac{\dd A(t)}{\dd t} A(t)^{-1} \right) = 0
    \qquad
    (t \in \R).
\label{A_geodesic_eqn}
\ee
\end{lemma}

\begin{proof}
First, let us observe that \eqref{A} is a well-defined map. Indeed, since 
along the trajectory \eqref{trajectory} we have $\bsLambda(t) \in \mfc$, 
from \eqref{mfp_reg_identification} we see that $A$ does take values in 
$\exp(\mfp_\reg)$. Continuing with the proof proper, notice that for all 
$t \in \R$ we have $A^{-1} = k e^{-2 \bsLambda} k^{-1}$ and
\be
    \dot{A} 
    = \dot{k} e^{2 \bsLambda} k^{-1}
        + k e^{2 \bsLambda} 2 \dot{\bsLambda} k^{-1}
        - k e^{2 \bsLambda} k^{-1} \dot{k} k^{-1},
\label{A_dot}
\ee
thus the formulae
\be
    \dot{A} A^{-1}  
    = k \left( 
            2 \dot{\bsLambda} - e^{2 \ad_{\bsLambda}} \cK + \cK 
        \right) k^{-1}
    \quad \text{and} \quad
    A^{-1} \dot{A} 
    = k \left( 
            2 \dot{\bsLambda} + e^{-2 \ad_{\bsLambda}} \cK - \cK 
        \right) k^{-1}
\label{A_inv_and_A_dot}
\ee
are immediate. Upon introducing the shorthand notations
\begin{align}
    & \cL(t) = \dot{\bsLambda}(t) + \cosh(\wad_{\bsLambda(t)}) Y(t) \in \mfp,
    \label{cL} \\
    & \cN(t) = \sinh(\wad_{\bsLambda(t)}) Y(t) \in \mfk,
    \label{cN}
\end{align}
from \eqref{A_inv_and_A_dot} we conclude that
\be
\begin{split}
    \frac{\dot{A} A^{-1} + A^{-1} \dot{A}}{4} 
    & = k \left( 
            \dot{\bsLambda} - \half \sinh(2 \ad_{\bsLambda}) \cK 
            \right) k^{-1} 
    = k \left( 
            \dot{\bsLambda} - \half \sinh(2 \wad_{\bsLambda}) \cK_{\mfm^\perp} 
            \right) k^{-1}
    \\
    & = k \left( 
            \dot{\bsLambda} 
            + \cosh(\wad_{\bsLambda}) \sinh(\wad_{\bsLambda}) Z
            \right) k^{-1}
    = k \cL k^{-1},
\end{split}
\label{A_dot_+}
\ee
and the relationship
\be
\begin{split}
    \frac{\dot{A} A^{-1} - A^{-1} \dot{A}}{4} 
    & = k \frac{\cK - \cosh(2 \ad_{\bsLambda}) \cK}{2} k^{-1} 
    = - k \left( \sinh(\ad_{\bsLambda})^2 \cK \right) k^{-1} 
    \\
    & = k \left( \sinh(\wad_{\bsLambda})^2 Z \right) k^{-1}
    = k \cN k^{-1}
\end{split}
\label{A_dot_-}
\ee
also follows. 

Now, by differentiating \eqref{A_dot_+} with respect to time $t$, we get
\be
    \frac{\dd}{\dd t} \frac{\dot{A} A^{-1} + A^{-1} \dot{A}}{4}
    = k \left (\dot{\cL} - [\cL, \cK] \right) k^{-1}.
\label{A_ddot_+}
\ee
Recalling the definition \eqref{cL}, Leibniz rule yields
\be
    \dot{\cL}
    = \ddot{\bsLambda} 
        + [\dot{\bsLambda}, \sinh(\wad_{\bsLambda}) Y] 
        + \cosh(\wad_{\bsLambda}) \dot{Y},
\label{cL_dot}
\ee
and the commutator
\be
    [\cL, \cK] 
    = - [\dot{\bsLambda}, \sinh(\wad_{\bsLambda})^{-1} Y]
        + [\cosh(\wad_{\bsLambda}) Y, B_{\mfm}]
        - [\cosh(\wad_{\bsLambda}) Y, \sinh(\wad_{\bsLambda})^{-1} Y]
\label{cL_and_cK_commutator}
\ee
is also immediate. By inspecting the right-hand side of the above equation, 
for the second term one can easily derive that
\be
\begin{split}
    [\cosh(\wad_{\bsLambda}) Y, B_{\mfm}] 
    & = \half [e^{\ad_{\bsLambda}} Y, B_{\mfm}] 
        + \half [e^{-\ad_{\bsLambda}} Y, B_{\mfm}]
    = \half e^{\ad_{\bsLambda}} [Y, B_{\mfm}]
        + \half e^{-\ad_{\bsLambda}} [Y, B_{\mfm}]
    \\
    & = \cosh(\ad_{\bsLambda}) [Y, B_{\mfm}]
    = \cosh(\wad_{\bsLambda}) [Y, B_{\mfm}].
\end{split}
\label{comm_term_1}
\ee
Furthermore, bearing in mind the identities appearing in \eqref{coth_iden}, 
a slightly longer calculation also reveals that the third term in 
\eqref{cL_and_cK_commutator} can be cast into the form
\be
\begin{split}
    & [\cosh(\wad_{\bsLambda}) Y, \sinh(\wad_{\bsLambda})^{-1} Y]
    = \half [e^{\ad_{\bsLambda}} Y, \sinh(\wad_{\bsLambda})^{-1} Y]
        + \half [e^{-\ad_{\bsLambda}} Y, \sinh(\wad_{\bsLambda})^{-1} Y]
    \\
    & \quad
    = \half e^{\ad_{\bsLambda}} 
                [Y, e^{-\wad_{\bsLambda}} \sinh(\wad_{\bsLambda})^{-1} Y]
        + \half e^{-\ad_{\bsLambda}} 
                [Y, e^{\wad_{\bsLambda}} \sinh(\wad_{\bsLambda})^{-1} Y]
    \\
    & \quad
    = \cosh(\ad_{\bsLambda}) [Y, \coth(\wad_{\bsLambda}) Y]
    = [Y, \coth(\wad_{\bsLambda}) Y]_{\mfa} 
        + \cosh(\wad_{\bsLambda}) [Y, \coth(\wad_{\bsLambda}) Y]_{\mfa^\perp}. 
\end{split}
\label{comm_term_2}
\ee
Now, by plugging the expressions \eqref{comm_term_1} and \eqref{comm_term_2}
into \eqref{cL_and_cK_commutator}, and by applying the hyperbolic identity
\be
    \sinh(w) + \frac{1}{\sinh(w)} = \cosh(w) \coth(w)
    \qquad
    (w \in \bC),
\label{sinh_iden}
\ee
one finds immediately that
\be
\begin{split}
    \dot{\cL} - [\cL, \cK]
    = & \ddot{\bsLambda} + [Y, \coth(\wad_{\bsLambda}) Y]_{\mfa}
    \\ 
    & + \cosh(\wad_{\bsLambda}) 
        \left( 
            \dot{Y} 
            + [Y, \coth(\wad_{\bsLambda}) Y]_{\mfa^\perp}
            - [Y, B_{\mfm}]
            + [\dot{\bsLambda}, \coth(\wad_{\bsLambda}) Y]
        \right).
\end{split}
\label{cL_key}
\ee
Looking back to Proposition \ref{PROPOSITION_ddot_bsLambda_etc}, we see
that $\dot{\cL} - [\cL, \cK] = 0$, thus by \eqref{A_ddot_+} we end up with 
the equation
\be
    \frac{\dd}{\dd t} \frac{\dot{A} A^{-1} + A^{-1} \dot{A}}{4} = 0.
\label{A_ddot_+_OK}
\ee

Next, upon differentiating \eqref{A_dot_-} with respect to $t$, we see that
\be
    \frac{\dd}{\dd t} \frac{\dot{A} A^{-1} - A^{-1} \dot{A}}{4}
    = k \left (\dot{\cN} - [\cN, \cK] \right) k^{-1}.
\label{A_ddot_-}
\ee
Remembering the form of $\cN$ \eqref{cN}, Leibniz rule yields
\be
    \dot{\cN} 
    = \cosh(\wad_{\bsLambda}) [\dot{\bsLambda}, Y] 
        + \sinh(\wad_{\bsLambda}) \dot{Y}
    = \sinh(\wad_{\bsLambda})
        \left(
            \coth(\wad_{\bsLambda}) [\dot{\bsLambda}, Y] + \dot{Y}
        \right),
\label{cN_dot}
\ee
and the formula
\be
    [\cN, \cK] 
    = [\sinh(\wad_{\bsLambda}) Y, B_{\mfm}] 
        - [\sinh(\wad_{\bsLambda}) Y, \sinh(\wad_{\bsLambda})^{-1} Y]
\label{cN_and_cK_commutator}
\ee
is also immediate. Now, let us observe that the first term on the 
right-hand side of the above equation can be transformed into the 
equivalent form
\be
\begin{split}
    [\sinh(\wad_{\bsLambda}) Y, B_{\mfm}] 
    & = \half [e^{\ad_{\bsLambda}} Y, B_{\mfm}] 
        - \half [e^{-\ad_{\bsLambda}} Y, B_{\mfm}]
    = \half e^{\ad_{\bsLambda}} [Y, B_{\mfm}]
        - \half e^{-\ad_{\bsLambda}} [Y, B_{\mfm}]
    \\
    & = \sinh(\ad_{\bsLambda}) [Y, B_{\mfm}]
    = \sinh(\wad_{\bsLambda}) [Y, B_{\mfm}],
\end{split}
\label{comm_term_3}
\ee
while for the second term we get
\be
\begin{split}
    & [\sinh(\wad_{\bsLambda}) Y, \sinh(\wad_{\bsLambda})^{-1} Y]
    = \half [e^{\ad_{\bsLambda}} Y, \sinh(\wad_{\bsLambda})^{-1} Y]
        - \half [e^{-\ad_{\bsLambda}} Y, \sinh(\wad_{\bsLambda})^{-1} Y]
    \\
    & \quad
    = \half e^{\ad_{\bsLambda}} 
                [Y, e^{-\wad_{\bsLambda}} \sinh(\wad_{\bsLambda})^{-1} Y]
        - \half e^{-\ad_{\bsLambda}} 
                [Y, e^{\wad_{\bsLambda}} \sinh(\wad_{\bsLambda})^{-1} Y]
    \\
    & \quad
    = \sinh(\ad_{\bsLambda}) [Y, \coth(\wad_{\bsLambda}) Y]
    = \sinh(\wad_{\bsLambda}) [Y, \coth(\wad_{\bsLambda}) Y]_{\mfa^\perp}. 
\end{split}
\label{comm_term_4}
\ee
Taking into account the above expressions, we obtain that
\be
\begin{split}
    \dot{\cN} - [\cN, \cK]
    = \sinh(\wad_{\bsLambda}) 
        \left( 
            \dot{Y} 
            + [Y, \coth(\wad_{\bsLambda}) Y]_{\mfa^\perp}
            - [Y, B_{\mfm}]
            + [\dot{\bsLambda}, \coth(\wad_{\bsLambda}) Y]
        \right),
\end{split}
\label{cN_key}
\ee
whence by Proposition \ref{PROPOSITION_ddot_bsLambda_etc} we are entitled 
to write that $\dot{\cN} - [\cN, \cK] = 0$. Giving a glance at the 
relationship \eqref{A_ddot_-}, it readily follows that
\be
    \frac{\dd}{\dd t} \frac{\dot{A} A^{-1} - A^{-1} \dot{A}}{4} = 0.
\label{A_ddot_-_OK}
\ee
To complete the proof, observe that the desired geodesic equation 
\eqref{A_geodesic_eqn} is a trivial consequence of the equations
\eqref{A_ddot_+_OK} and \eqref{A_ddot_-_OK}.
\end{proof}

To proceed further, let us observe that by integrating the differential
equation \eqref{A_geodesic_eqn}, we obtain immediately that
\be
    \dot{A}(t) A(t)^{-1} = \dot{A}(0) A(0)^{-1}
    \qquad
    (t \in \R).
\label{A_first_order_DE}
\ee
However, recalling the definitions \eqref{cL} and \eqref{cN}, and also the 
relationships \eqref{bsLambda_dot} and \eqref{D+Y}, from the equations 
\eqref{A_dot_+}, \eqref{A_dot_-} and \eqref{k_initial_cond} we infer that
\be
\begin{split}
    \dot{A}(0) A(0)^{-1} 
    & = 2 k(0) (\cL(0) + \cN(0)) k(0)^{-1}
    = 2 (\dot{\bsLambda}(0) + e^{\ad_{\bsLambda(0)}} Y(0))
    \\
    & = 2 e^{\ad_{\bsLambda(0)}} (D(0) + Y(0))
    = e^{\bsLambda(0)} (L(0) - L(0)^{-1}) e^{-\bsLambda(0)}.
\end{split}
\label{A_dot_t=0}
\ee
Moreover, remembering \eqref{k_initial_cond} and the definition \eqref{A}, 
at $t = 0$ we can also write that
\be
    A(0) = k(0) e^{2 \bsLambda(0)} k(0)^{-1} = e^{2 \bsLambda(0)}.
\label{A(0)}
\ee
Putting the above observations together, it is now evident that the unique 
maximal solution of the first order differential equation 
\eqref{A_first_order_DE} with the initial condition \eqref{A(0)} is the 
smooth curve
\be
    A(t) 
    = e^{t e^{\bsLambda(0)} (L(0) - L(0)^{-1}) e^{-\bsLambda(0)}} 
        e^{2 \bsLambda(0)}
    = e^{\bsLambda(0)} e^{t (L(0) - L(0)^{-1})} e^{\bsLambda(0)}
    \qquad
    (t \in \R).
\label{A(t)_exp_form}
\ee
Comparing this formula with \eqref{A}, the following result is immediate.

\begin{theorem}
\label{THEOREM_eigenvalue_dynamics}
Take an arbitrary maximal solution \eqref{trajectory} of the van Diejen 
system \eqref{H}, then at each $t \in \R$ it can be recovered uniquely  
from the spectral identification
\be
    \{ e^{\pm 2 \lambda_a(t)} \, | \, a \in \bN_n \}
    = \mathrm{Spec} 
        (e^{\bsLambda(0)} e^{t (L(0) - L(0)^{-1})} e^{\bsLambda(0)}).
\label{spectral_identification}
\ee
\end{theorem}

The essence of the above theorem is that any solution \eqref{trajectory} of 
the van Diejen system \eqref{H} can be obtained by a purely algebraic process 
based on the diagonalization of a matrix flow. Indeed, once one finds the 
evolution of $\lambda(t)$ from \eqref{spectral_identification}, the evolution 
of $\theta(t)$ also becomes accessible by the formula
\be
    \theta_a(t) 
    = \mathrm{arcsinh} 
        \left( \frac{\dot{\lambda}_a(t)}{u_a(\lambda(t))} \right)
    \qquad
    (a \in \bN_n), 
\label{theta_evolution}
\ee
as dictated by the equation of motion \eqref{lambda_dot}.

\subsection{Temporal asymptotics}
\label{subsec:4.3.5}

One of the immediate consequences of the projection method formulated
in the previous subsection is that the Hamiltonian \eqref{H} describes 
a `repelling' particle system, thus it is fully justified to inquire
about its scattering properties. Although rigorous scattering theory is 
in general a hard subject, a careful study of the algebraic solution 
algorithm described in Theorem \ref{THEOREM_eigenvalue_dynamics} allows 
us to investigate the asymptotic properties of any maximally defined 
trajectory \eqref{trajectory} as $t \to \pm \infty$. In this respect our 
main tool is Ruijsenaars' theorem on the spectral asymptotics of exponential 
type matrix flows (see \cite[Theorem A2]{Ru88}). To make it work, 
let us look at the relationship \eqref{L_diagonalized} and Lemma 
\ref{LEMMA_regularity}, from where we see that there is a group element 
$y \in K$ and a unique real $n$-tuple 
$\htheta = (\htheta_1, \ldots, \htheta_n) \in \R^n$ satisfying 
\be
    \htheta_1 > \ldots > \htheta_n > 0, 
\label{theta_order}
\ee
such that with the (regular) diagonal matrix $\hbsTheta \in \mfc$ defined 
in \eqref{hbsTheta} we can write that
\be
    L(0) = y e^{2 \hbsTheta} y^{-1}.
\label{L(0)_diagonalized}
\ee
Following the notations of the previous subsection, here $L(0)$ still stands 
for the Lax matrix \eqref{L} evaluated along the trajectory \eqref{trajectory} 
at $t = 0$. Since
\be
    L(0) - L(0)^{-1} = 2 y \sinh(2 \hbsTheta) y^{-1},
\label{exponent}
\ee
with the aid of the positive definite matrix
\be
    \hat{L} = y^{-1} e^{2 \bsLambda(0)} y \in \exp(\mfp)
\label{hat_L_0}
\ee
for the spectrum of the matrix flow appearing in \eqref{A(t)_exp_form} 
we obtain at once that
\be
    \mathrm{Spec} 
        (e^{\bsLambda(0)} e^{t (L(0) - L(0)^{-1})} e^{\bsLambda(0)})
    = \mathrm{Spec} (\hat{L} e^{2 t \sinh(2 \hbsTheta)}).
\label{spectrum_of_matrix_flow}
\ee
In order to make a closer contact with Ruijsenaars' theorem, let us also
introduce the Hermitian $n \times n$ matrix $\cR$ with entries
\be
    \cR_{a, b} = \delta_{a + b, n + 1}.
\label{cR}
\ee
Since $\cR^2 = \bsone_n$, we have $\cR^{-1} = \cR$, whence the 
block-diagonal matrix
\be
    \cW = \begin{bmatrix}
        \bsone_n & 0_n \\
        0_n & \cR_n
    \end{bmatrix} 
    \in GL(N, \bC),
\label{cW}
\ee
also satisfies the relations $\cW^{-1} = \cW = \cW^*$. As the most important 
ingredients of our present analysis, now we introduce the matrices
\be
    \bsTheta^+ = 2 \cW \hbsTheta \cW^{-1}
    \quad \text{and} \quad
    \tilde{L} = \cW \hat{L} \cW^{-1}.
\label{bsTheta_+_and_L_tilde}
\ee
Recalling the relationships \eqref{spectral_identification} and 
\eqref{spectrum_of_matrix_flow}, it is clear that for all $t \in \R$
we can write that
\be
    \{ e^{\pm 2 \lambda_a(t)} \, | \, a \in \bN_n \}
    = \mathrm{Spec} 
        (
            \tilde{L} e^{2 t \sinh(\bsTheta^+)}
        ).
\label{spectral_identification_OK}
\ee
However, upon performing the conjugations with the unitary matrix $\cW$
\eqref{cW} in the defining equations displayed in 
\eqref{bsTheta_+_and_L_tilde}, we find immediately that
\be
    \bsTheta^+ 
    = \diag(\theta_1^+, \ldots, \theta_n^+, 
            - \theta_n^+, \ldots, -\theta_1^+),
\label{bsTheta_+}
\ee
where
\be
    \theta_a^+ = 2 \htheta_a
    \qquad
    (a \in \bN_n).
\label{theta_+}
\ee
The point is that, due to our regularity result formulated in Lemma
\ref{LEMMA_regularity}, the diagonal matrix \eqref{bsTheta_+} has a simple 
spectrum, and its eigenvalues are in strictly decreasing order along the 
diagonal (see \eqref{theta_order}). Moreover, since $\hat{L}$ 
\eqref{hat_L_0} is positive definite, so is $\tilde{L}$. In particular, the 
leading principal minors of matrix $\tilde{L}$ are all strictly positive. 
So, the exponential type matrix flow
\be
    \R \ni t \mapsto \tilde{L} e^{2 t \sinh(\bsTheta^+)} \in GL(N, \bC)
\label{matrix_flow}
\ee
does meet all the requirements of Ruijsenaars' aforementioned theorem.
Therefore, essentially by taking the logarithm of the quotients of 
the consecutive leading principal minors of the $n \times n$ submatrix taken 
from the upper-left-hand corner of $\tilde{L}$, one finds a unique real 
$n$-tuple
\be
    \lambda^+ = (\lambda_1^+, \dots, \lambda_n^+) \in \R^n
\label{lambda_+}
\ee
such that for all $a \in \bN_n$ we can write 
\be
    \lambda_a(t) \sim t \sinh(\theta_a^+) + \lambda_a^+
    \quad \text{and} \quad
    \theta_a(t) \sim \theta_a^+,
\label{lambda_asymptotics_t_to_infty}
\ee
up to exponentially vanishing small terms as $t \to \infty$. It is obvious 
that the same ideas work for the case $t \to -\infty$, too, with complete
control over the asymptotic momenta $\theta_a^-$ and the asymptotic phases 
$\lambda_a^-$ as well. The above observations can be summarized as follows.

\begin{lemma}
\label{LEMMA_asymptotics}
For an arbitrary maximal solution \eqref{trajectory} of the hyperbolic 
$n$-particle van Diejen system \eqref{H} the particles move asymptotically 
freely as $\vert t \vert \to \infty$. More precisely, for all $a \in \bN_n$ 
we have the asymptotics
\be
    \lambda_a(t) \sim t \sinh(\theta_a^\pm) + \lambda_a^\pm
    \quad \text{and} \quad
    \theta_a(t) \sim \theta_a^\pm
    \qquad
    (t \to \pm \infty),
\label{asymptotic_result}
\ee
where the asymptotic momenta obey
\be
    \theta_a^- = -\theta_a^+
    \quad \text{and} \quad
    \theta_1^+ > \ldots > \theta_n^+ > 0.
\label{asymptotic_momenta_conds}
\ee
\end{lemma}

We find it quite remarkable that, up to an overall sign, the asymptotic 
momenta are preserved \eqref{asymptotic_momenta_conds}. Following 
Ruijsenaars' terminology \cite{Ru88,Ru90}, we may say that the $2$-parameter family of van Diejen systems \eqref{H} are
\emph{finite dimensional pure soliton systems}. Now, let us remember that 
for each pure soliton system analysed in the earlier literature, the 
scattering map has a factorized form. That is, the $n$-particle scattering 
can be completely reconstructed from the $2$-particle processes, and also 
by the $1$-particle scattering on the external potential (see e.g. 
\cite{Ku76,Mo77,Ru88,Ru90,Pu13}). Albeit the results we shall present in rest of this chapter 
do not rely on this peculiar feature of the scattering process, still, it 
would be of considerable interest to prove this property for the hyperbolic 
van Diejen systems \eqref{H}, too. However, because of its subtleties, we 
wish to work out the details of the scattering theory in a later publication.

\section{Spectral invariants of the Lax matrix}
\label{sec:4.4}

The ultimate goal of this section is to prove that the eigenvalues of the 
Lax matrix $L$ \eqref{L} are in involution. Superficially, one could say
that it follows easily from the scattering theoretical results presented
in the previous section. A convincing argument would go as follows. Recalling 
the notations \eqref{trajectory} and \eqref{trajectory_init_cond}, let us 
consider the flow
\be
    \Phi \colon \R \times P \rightarrow P,
    \qquad
    (t, \gamma_0) \mapsto \Phi_t (\gamma_0) = \gamma(t)
\label{Phi_flow}
\ee
generated by the Hamiltonian vector field $\bsX_H$ \eqref{bsX_H}. Since 
for all $t \in \R$ the map $\Phi_t \colon P \rightarrow P$ is a 
symplectomorphism, for all $a ,b \in \bN_n$ we can write that
\be
    \{ \theta_a \circ \Phi_t, \theta_b \circ \Phi_t \}
    = \{ \theta_a, \theta_b \} \circ \Phi_t 
    = 0.
\label{theta_and_Phi}
\ee
On the other hand, from \eqref{asymptotic_result} it is also clear that 
at each point of the phase space $P$, for all $c \in \bN_n$ we have
\be
    \theta_c \circ \Phi_t \to \theta_c^+
    \qquad
    (t \to \infty).    
\label{theta_and_theta_+}
\ee
Recalling \eqref{htheta_smooth} and \eqref{theta_+}, it is evident that 
$\theta_c^+ \in C^\infty(P)$. Therefore, by a `simple interchange of limits', 
from \eqref{theta_and_Phi} and \eqref{theta_and_theta_+} one could infer 
that the asymptotic momenta $\theta_c^+$ $(c \in \bN_n)$ Poisson commute. 
Bearing in mind the relationships \eqref{theta_+} and \eqref{spec_L}, it 
would also follow that the eigenvalues of $L$ \eqref{L} generate a maximal 
Abelian Poisson subalgebra. However, to justify the interchange of limits, 
one does need a deeper knowledge about the scattering properties than the 
pointwise limit formulated in \eqref{theta_and_theta_+}. Since we wish 
to work out the full scattering theory elsewhere, in this chapter we choose an 
alternative approach by merging the temporal asymptotics of the trajectories 
with van Diejen's earlier results \cite{vD95,vD94,vD95-2}. 

\subsection{Link to the $5$-parameter family of van Diejen systems}
\label{subsec:4.4.1}

As is known from the seminal papers \cite{vD94,vD95-2}, the definition of the classical hyperbolic van Diejen 
system is based on the smooth functions 
$v, w \colon \R \setminus \{ 0 \} \rightarrow \bC$ defined by the formulae
\be
    v(x) = \frac{\sinh(\ri g + x)}{\sinh(x)},
    \quad
    w(x) = \frac{\sinh(\ri g_0 + x)}{\sinh(x)} 
            \frac{\cosh(\ri g_1 + x)}{\cosh(x)}
            \frac{\sinh(\ri g'_0 + x)}{\sinh(x)}
            \frac{\cosh(\ri g'_1 + x)}{\cosh(x)},
\label{v-w}
\ee
where the five independent real numbers $g$, $g_0$, $g_1$, $g_0'$, $g_1'$ 
are the coupling constants. Parameter $g$ in the `potential' function $v$ 
controls the strength of inter-particle interaction, whereas the remaining 
four constants appearing in the `external potential' $w$ are responsible 
for the influence of the ambient field. Conforming to the notations 
introduced in the aforementioned papers, let us recall that the set of 
Poisson commuting functions found by van Diejen can be succinctly written as
\be
    H_l 
    = \sum_{\substack{J \subseteq \bN_n, \ \vert J \vert \leq l 
            \\ \eps_j = \pm 1, \ j \in J}}
                \cosh(\theta_{\eps J}) 
                \vert V_{\eps J; J^c} \vert
                U_{J^c, l - \vert J \vert}
    \qquad
    (l \in \bN_n),
\label{H_vD}
\ee
where the various constituents are defined by the formulae
\be
    \theta_{\eps J} = \sum_{j \in J} \eps_j \theta_j,
    \quad
    V_{\eps J; J^c} 
    = \prod_{j \in J} w(\eps_j \lambda_j)
        \prod_{\substack{j, j' \in J \\ (j < j')}}
            v(\eps_j \lambda_j + \eps_{j'} \lambda_{j'})^2
        \prod_{\substack{j \in J \\ k \in J^c}}
            v(\eps_j \lambda_j + \lambda_k) 
            v(\eps_j \lambda_j - \lambda_k),
\label{V}
\ee
together with the expression
\be
    U_{J^c, l - \vert J \vert}
    = (-1)^{l - \vert J \vert} 
    \mkern-25mu 
    \sum_{\substack{I \subseteq J^c, \ \vert I \vert = l - \vert J \vert 
                    \\ \eps_i = \pm 1, \ i \in I}}
        \prod_{i \in I} w(\eps_i \lambda_i)
        \mkern-5mu 
        \prod_{\substack{i, i' \in I \\ (i < i')}}
            \vert v(\eps_i \lambda_i + \eps_{i'} \lambda_{i'}) \vert^2
        \mkern-5mu 
        \prod_{\substack{i \in I \\ k \in J^c \setminus I}}
        \mkern-5mu
            v(\eps_i \lambda_i + \lambda_k) 
            v(\eps_i \lambda_i-\lambda_k).
\label{U}
\ee
At this point two short technical remarks are in order. First, we extend 
the family of the first integrals \eqref{H_vD} with the constant function 
$H_0 = 1$. Analogously, in the last equation \eqref{U} it is understood 
that $U_{J^c, 0} = 1$. 

To make contact with the $2$-parameter family of van Diejen systems of 
our interest \eqref{H}, for the coupling parameters of the potential 
functions \eqref{v-w} we make the special choice 
\be
    g = \mu,
    \quad
    g_0 = g_1 = \frac{\nu}{2},
    \quad
    g'_0 = g'_1 = 0.
\label{2parameters}
\ee
Under this assumption, from the definitions \eqref{z_a} and \eqref{V} it is
evident that with the singleton $J = \{ a \}$ we can write that
\be
    V_{\{ a \}; \{ a \}^c} = -z_a
    \qquad 
    (a \in \bN_n).
\label{V-z}
\ee
Giving a glance at \eqref{U}, it is also clear that the term corresponding 
to $J = \emptyset$ in the defining sum of $H_1$ \eqref{H_vD} is a constant
function of the form
\be
    U_{\bN_n, 1} = 2 \sum_{a = 1}^n \Real(z_a)
    = - 2 \cos \left( \nu + (n - 1) \mu \right) 
        \frac{\sin(n \mu)}{\sin(\mu)}.
\label{U-z}
\ee
Plugging the above formulae into van Diejen's main Hamiltonian $H_1$ 
\eqref{H_vD}, one finds immediately that
\be
    H_1 + 2 \cos \big(\nu + (n - 1) \mu \big) \frac{\sin(n \mu)}{\sin(\mu)}
    = 2 H 
    =\tr(L).
\label{H1_vs_H}
\ee
That is, up to some irrelevant constants, our Hamiltonian $H$ \eqref{H} 
can be identified with $H_1$ \eqref{H_vD}, provided the coupling parameters 
are related by the equations displayed in \eqref{2parameters}. At this 
point one may suspect that the quantities $\tr(L^l)$ are also expressible 
with the aid of the Poisson commuting family of functions $H_l$ \eqref{H_vD}. 
Clearly, it would imply immediately that the eigenvalues of the Lax matrix 
$L$ \eqref{L} are in involution. However, due to the complexity of the 
underlying objects \eqref{V}-\eqref{U}, this naive approach would lead to 
a formidable combinatorial task, that we do not wish to pursue in this 
chapter. To circumvent the difficulties, below we rather resort to a clean 
analytical approach by exploiting the scattering theoretical results 
formulated in the previous section.

\subsection{Poisson brackets of the eigenvalues of $L$}
\label{subsec:4.4.2}

Take an arbitrary point $\gamma_0 \in P$ and consider the unique maximal 
integral curve 
\be
    \R \ni t \mapsto \gamma(t) = (\lambda(t), \theta(t)) \in P
\label{curve_gamma}
\ee
of the Hamiltonian vector field $\bsX_H$ \eqref{bsX_H} satisfying the 
initial condition 
\be
    \gamma(0) = \gamma_0.
\label{init_cond}
\ee 
Since the functions $H_l$ \eqref{H_vD} are first integrals of the dynamics, 
their values at the point $\gamma_0$ can be recovered by inspecting the 
limit of $H_l (\gamma(t))$ as $t \to \infty$. Now, recalling the potentials 
\eqref{H_vD} and the specialization of the coupling parameters 
\eqref{2parameters}, it is evident that
\be
    \lim_{x \to \pm \infty} v(x) = e^{\pm \ri \mu}
    \quad \text{and} \quad
    \lim_{x \to \pm \infty} w(x) = e^{\pm \ri \nu}.
\label{lim_v_w}
\ee
Therefore, taking into account the regularity properties 
\eqref{asymptotic_momenta_conds} of the asymptotic momenta $\theta_c^+$
\eqref{asymptotic_result}, from Lemma \ref{LEMMA_asymptotics} and the 
definitions \eqref{H_vD}-\eqref{U} one finds immediately that
\be
    H_l(\gamma_0)
    = \lim_{t \to \infty} H_l (\gamma(t))
    = \sum_{\substack{J \subseteq \bN_n, \ \vert J \vert \leq l 
                        \\ \eps_j = \pm 1, \ j \in J}}
        \cosh(\theta_{\eps J}^+) \ \cU_{J^c, l - \vert J \vert}
    \qquad
    (l \in \bN_n),
\label{H_l_gamma_0}
\ee
where
\be
    \cU_{J^c, l - \vert J \vert}
    = (-1)^{l - \vert J \vert}
        \sum_{\substack{I \subseteq J^c, \ \vert I \vert = l - \vert J \vert 
                        \\ \eps_j = \pm 1, \ j \in I}}
        \prod_{j \in I} e^{\eps_j \ri \nu}
        \prod_{\substack{j \in I, \ k \in J^c \setminus I \\ (j < k)}} 
            e^{\eps_j 2 \ri \mu}.
\label{cU}
\ee
By inspecting the above expression, let us observe that the value of 
$\cU_{J^c, l - \vert J \vert}$ does \emph{not} depend on the specific 
choice of the subset $J$, but only on its cardinality $\vert J \vert$. 
More precisely, if $J \subseteq \bN_n$ is an arbitrary subset of cardinality 
$\vert J \vert = k$ $(0 \leq k \leq l - 1)$, then we can write that
\be
    \cU_{J^c, l - \vert J \vert}
    = (-1)^{l - k}
        \sum_{\substack{1 \leq j_1 < \dots < j_{l - k} \leq n - k 
                        \\ \eps_1 = \pm 1, \dots, \eps_{l - k} = \pm 1}}
            \exp 
                \left(
                    \ri \sum_{m = 1}^{l - k} 
                        \eps_m \left( \nu + 2(n - l + m - j_m) \mu \right)
                \right).
\label{cU_ell-k}
\ee

To proceed further, let us now turn to the study of the Lax matrix $L$ 
\eqref{L}. Due to the Lax representation of the dynamics that we established 
in Theorem \ref{THEOREM_Lax_representation}, the eigenvalues of $L$ are 
conserved quantities. Consequently, the coefficients 
$K_0, K_1, \ldots, K_N \in C^\infty(P)$ of the characteristic polynomial 
\be
    \det(L - y \bsone_N) = \sum_{m = 0}^N K_{N - m} y^m
    \qquad
    (y \in \bC)
\label{L_kar_poly}
\ee
are also first integrals. As expected, the special algebraic properties 
of $L$ formulated in Proposition \ref{PROPOSITION_L_in_G} and Lemma 
\ref{LEMMA_L_in_exp_p} have a profound impact on these coefficients 
as well, as can be seen from the relations 
\be
    K_{N - m} = K_m
    \qquad 
    (m = 0, 1, \dots, N).
\label{K_m-symmetry}
\ee
So, it is enough to analyze the properties of the members 
$K_0 = 1, K_1, \ldots, K_n$. In this respect the most important ingredient 
is the relationship
\be
    \lim_{t \to \infty} L(\gamma(t)) = \exp(\bsTheta^+),
\label{asymptotic-Lax}
\ee
where $\bsTheta^+$ is the $N \times N$ diagonal matrix \eqref{bsTheta_+} 
containing the asymptotic momenta. Therefore, looking back to the definition
\eqref{L_kar_poly}, for any $m = 0, 1, \dots, n$ we obtain at once that 
\be
    K_m(\gamma_0)
    = \lim_{t \to \infty} K_m(\gamma(t))
    = (-1)^m
        \sum_{a = 0}^{\genfrac{\lfloor}{\rfloor}{}{}{m}{2}}
        \sum_{\substack{J \subseteq \bN_n, \ \vert J \vert = m - 2 a 
                        \\ \eps_j = \pm 1, \ j \in J}}
        \binom{n - \vert J \vert}{a} \cosh(\theta_{\eps J}^+).
\label{K_m_gamma_0}
\ee
Based on the formulae \eqref{H_l_gamma_0} and \eqref{K_m_gamma_0}, we can
prove the following important technical result.

\begin{lemma}
\label{LEMMA_linear_relation}
The two distinguished families of first integrals $\{ H_l \}_{l = 0}^n$ 
and $\{ K_m \}_{m = 0}^n$ are connected by an invertible linear relation with 
purely numerical coefficients depending only on the coupling parameters 
$\mu$ and $\nu$.
\end{lemma}

\begin{proof}
For brevity, let us introduce the notation
\be
    \cA_k 
    = \sum_{\substack{J \subseteq \bN_n, \ \vert J \vert = k 
                        \\ \eps_j = \pm 1, \ j \in J}}
            \cosh(\theta_{\eps J}^+)
    \qquad 
    (k = 0, 1, \dots, n).
\label{cA}
\ee
As we have seen in \eqref{cU_ell-k}, the coefficients 
$\cU_{J^c, l - \vert J \vert}$ appearing in the formula \eqref{H_l_gamma_0} 
depend only on the cardinality of $J$, whence for any 
$l \in \{ 0, 1, \dots, n \}$ we can write that
\be
    H_l(\gamma_0) = \sum_{k = 0}^l \cU_{\bN_{n - k}, l - k} \cA_k.
\label{H_vs_cA}
\ee
Since $\cU_{\bN_{n - l}, 0} = 1$, the matrix transforming
$\{ \cA_k \}_{k = 0}^n$ into $\{ H_l(\gamma_0) \}_{l = 0}^n$ is lower 
triangular with plus ones on the diagonal, whence the above linear 
relation \eqref{H_vs_cA} is invertible. Comparing the formulae 
\eqref{K_m_gamma_0} and \eqref{cA}, it is also clear that 
\be
    K_m(\gamma_0)
    = (-1)^m \sum_{a = 0}^{\genfrac{\lfloor}{\rfloor}{}{}{m}{2}}
                \binom{n - (m - 2 a)}{a} \cA_{m - 2 a},
\label{K_vs_cA}
\ee
which in turn implies that the matrix relating $\{ \cA_k \}_{k = 0}^n$ to 
$\{ K_m(\gamma_0) \}_{m = 0}^n$ is lower triangular with diagonal entries 
$\pm 1$. Hence the linear relationship \eqref{K_vs_cA} is also invertible. 
Putting together the above observations, it is clear that there is an 
invertible $(n + 1) \times (n + 1)$ matrix $\cC$ with purely numerical 
entries $\cC_{m, l}$ depending only on $\mu$ and $\nu$ such that
\be
    K_m(\gamma_0) = \sum_{l = 0}^n \cC_{m, l} H_l(\gamma_0).
\label{K_vs_H}
\ee
Since $\gamma_0$ is an arbitrary point of the phase space $P$ \eqref{P}, 
the Lemma follows.
\end{proof}

The scattering theoretical idea in the proof the above Lemma goes back 
to the fundamental works of Moser (see e.g. \cite{Mo75}). However, 
in the recent paper \citepalias{GF15} it has been revitalized in 
the context of the rational $\BC_n$ van Diejen model, too. Compared to the 
rational case, it is a significant difference that our coefficients 
$\cU_{J^c, l - \vert J \vert}$ \eqref{cU} do depend on the parameters 
$\mu$ and  $\nu$ in a non-trivial manner, whence the observations surrounding 
the derivations of formula \eqref{cU_ell-k} turns out to be crucial in our 
presentation.

Since the family of functions $\{ H_l \}_{l = 0}^n$ Poisson commute, Lemma 
\ref{LEMMA_linear_relation} readily implies that the first integrals 
$\{ K_m \}_{m = 0}^n$ are also in involution. Now, let us recall that the 
spectrum of the Lax matrix $L$ is simple, as we have seen in Lemma 
\ref{LEMMA_regularity}. As a consequence, the eigenvalues of $L$ can be 
realized as smooth functions of the coefficients of the characteristic 
polynomial \eqref{L_kar_poly}, thus the following result is immediate.

\begin{theorem}
\label{THEOREM_commuting_eigenvalues}
The eigenvalues of the Lax matrix $L$ \eqref{L} are in involution.
\end{theorem}

To conclude this section, let us note that the proof of Theorem 
\ref{THEOREM_commuting_eigenvalues} is quite indirect in the sense that
it hinges on the commutativity of the family of functions \eqref{H_vD}.
However, the only available proof of this highly non-trivial fact is
based on the observation that the Hamiltonians \eqref{H_vD} can be
realized as classical limits of van Diejen's commuting analytic difference 
operators \cite{vD95}. As a more elementary approach, let 
us note that Theorem \ref{THEOREM_commuting_eigenvalues} would also follow 
from the existence of an $r$-matrix encoding the tensorial Poisson bracket 
of the Lax matrix $L$ \eqref{L}. Due to Lemma \ref{LEMMA_linear_relation}, 
it would imply the commutativity of the family \eqref{H_vD}, too, at least 
under the specialization \eqref{2parameters}. To find such an $r$-matrix, 
one may wish to generalize the analogous results on the rational system 
\cite{Pu15}.

\section{Discussion}
\label{sec:4.5}

One of the most important objects in the study of integrable systems is the
Lax representation of the dynamics. By generalizing the earlier results on 
the rational $\BC_n$ RSvD models \cite{Pu11-2,Pu15},
in this chapter we succeeded in constructing a Lax pair for the $2$-parameter 
family of hyperbolic van Diejen systems \eqref{H}. Making use of this 
construction, we showed that the dynamics can be solved by a projection 
method, which in turn allowed us to initiate the study of the scattering 
properties of \eqref{H}. Moreover, by combining our scattering theoretical 
results with the ideas of the recent paper \citepalias{GF15}, we 
proved that the first integrals provided by the eigenvalues of the proposed 
Lax matrix \eqref{L} are in fact in involution. To sum up, it is fully 
justified to say that the matrices $L$ \eqref{L} and $B$ \eqref{B} form a 
Lax pair for the hyperbolic van Diejen system \eqref{H}. 

Apart from taking a non-trivial step toward the construction of Lax matrices 
for the most general hyperbolic van Diejen many-particle systems \eqref{H_vD}, 
let us not forget about the potential applications of our results. In analogy 
with the translation invariant RS systems, we expect that the van Diejen 
models may play a crucial role in clarifying the particle-soliton picture 
in the context of integrable boundary field theories. While the relationship 
between the $A$-type RS models and the soliton equations defined on the whole 
line is under control (see e.g. \cite{RS86,Ru88,BB93,Ru94,Ru95}), the link between the van Diejen 
models and the soliton systems defined on the half-line is less understood 
(see e.g. \cite{SSW95,KS94}). As in the translation 
invariant case, the Lax matrices of the van Diejen systems could turn out to 
be instrumental for elaborating this correspondence.

Turning to the more recent activities surrounding the CMS and the RS 
many-particle models, let us recall the so-called classical/quantum duality 
(see e.g. \cite{MTV11,ALTZ14,GZZ14,TZZ15,BLZZ16}), which 
relates the spectra of certain quantum spin chains with the Lax matrices of 
the classical CMS and RS systems. An equally remarkable development is the 
emergence of new integrable tops based on the Lax matrices of the CMS and 
the RS systems \cite{AASZ14,LOZ14}. Relatedly, 
it would be interesting to see whether the Lax matrix \eqref{L} of the 
hyperbolic van Diejen system \eqref{H} can be fit into these frameworks.

One of the most interesting aspects of the CMS and the RSvD systems we 
have not addressed in this chapter is the so-called Ruijsenaars duality, or 
action-angle duality. Based on hard analytical techniques, this remarkable 
property was first exhibited by Ruijsenaars \cite{Ru88} in the 
context of the translation invariant non-elliptic models. Let us note 
that in the recent papers \cite{FK09,FA10,FK11,FK12} almost all of these duality 
relationships have been successfully reinterpreted in a nice geometrical 
framework provided by powerful symplectic reduction methods. Moreover, 
by now some duality results are available also for the CMS and the RSvD 
models associated with the $\BC$-type root systems \cite{Pu11-2,Pu12} \citepalias{FG14}.

As for the key player of this chapter, we have no doubt that the $2$-parameter
family of hyperbolic van Diejen systems \eqref{H} is self-dual. Indeed, upon 
diagonalizing the Lax matrix $L$ \eqref{L}, we see that the transformed 
objects defined in \eqref{L_diagonalized}-\eqref{hat_L_and_hat_F} obey the 
relationship \eqref{hat_commut_rel}, that has the same form as the Ruijsenaars 
type commutation relation \eqref{commut_rel} we set up in Lemma 
\ref{LEMMA_commut_rel}. Based on the method presented in \cite{Ru88}, 
we expect that the transformed matrix $\hat{L}$ \eqref{hat_L_and_hat_F} shall 
provide a Lax matrix for the dual system. Therefore, comparing the matrix 
entries displayed in \eqref{L} and \eqref{hat_L_entries}, the self-duality 
of the system \eqref{H} seems to be more than plausible. Admittedly, many 
subtle details are still missing for a complete proof. As for filling these 
gaps, the immediate idea is that either one could mimic Ruijsenaars' 
scattering theoretical approach, or invent an appropriate symplectic reduction 
framework. However, notice that the non-standard form of the Hamiltonian 
\eqref{H} poses severe analytical difficulties on the study of the scattering 
theory, whereas the weakness of the geometrical approach lies in the fact that 
up to now even the translation invariant hyperbolic RS model has not been 
derived from symplectic reduction. Nevertheless, by taking the analytical 
continuation of the Lax matrix $L$ \eqref{L}, it is conceivable that the 
self-duality of the compactified trigonometric version of \eqref{H} can be 
proved by adapting the quasi-Hamiltonian reduction approach advocated by 
Feh\'er and Klim\v{c}\'ik \cite{FK12}. For further 
motivation, let us recall that the duality properties are indispensable in 
the study of the recently introduced integrable random matrix ensembles 
\cite{BGS09,BGS11,FG15-2}, too.

\chapter[{Trigonometric and elliptic Ruijsenaars-Schneider models on $\mathbb{CP}^{n-1}$}]{Trigonometric and elliptic Ruijsenaars-\\Schneider models on $\mathbb{CP}^{n-1}$}
\label{chap:5}

Following \citepalias{FG16-2}, we present a direct construction of compact real forms of the trigonometric and elliptic $n$-particle Ruijsenaars-Schneider systems whose completed center-of-mass phase space is the complex projective space $\CP^{n-1}$ with the Fubini-Study symplectic structure. These systems are labelled by an integer $p\in\{1,\dots,n-1\}$ relative prime to $n$ and a coupling parameter $y$, which can vary in a certain punctured interval around $p\pi/n$. Our work extends Ruijsenaars's pioneering study of compactifications that imposed the restriction $0<y<\pi/n$, and also builds on an earlier derivation of more general compact trigonometric systems by Hamiltonian reduction.

The phase spaces of the particle systems we encountered so far are usually the cotangent bundles of the configuration spaces, hence they are never compact due to the infinite range
of the canonical momenta. For example, the standard Ruijsenaars-Schneider Hamiltonian depends on the momenta $\phi_k$ through the function $\cosh(\phi_k)$, but by analytic continuation this may be replaced by $\cos(\phi_k)$, which effectively compactifies the momenta on a circle. If the dependence on the position variables $x_k$ is also through a periodic function, then the phase space of the system can be taken to be bounded. This possibility was examined in \cite{Ru95}, where the Hamiltonian
\begin{equation}
H=\sum_{k=1}^n\cos(\phi_k)\sqrt{\prod_{\substack{j=1\\(j\neq k)}}^n
\bigg[1-\frac{\sin^2y}{\sin^2(x_j-x_k)}\bigg]}
\label{5.1}
\end{equation}
containing a real coupling parameter $0<y<\pi/2$ was considered.
Ruijsenaars called this the III$_\b$ system, with III referring to the trigonometric character
of the interaction, as in \cite{OP81}, and the suffix standing for `bounded'. (One may also
introduce the deformation parameter $\beta$ into the III$_\b$ system, by replacing $\phi_k$ by
$\beta\phi_k$.) The domain of the `angular position variables'
$\{(x_1,\dots,x_n)\}\subset[0,\pi]^n$ must be restricted in such a way that the
Hamiltonian \eqref{5.1} is real and smooth. This may be ensured by prescribing
\begin{equation}
x_{i+1}-x_i>y\quad(i=1,\dots,n-1),\quad x_n-x_1<\pi-y,
\label{I2}
\end{equation}
which obviously implies Ruijsenaars's condition
\begin{equation}
0<y<\frac{\pi}{n}.
\label{I3}
\end{equation}
Although the Hamiltonian is then real, its flow is not complete on the naive phase space, because
it may reach the boundary $x_{k+1}-x_k=y$ (with $x_{k+n}\equiv x_k+\pi$) at finite time \cite{Ru95}.
Completeness of the commuting flows is a crucial property of any bona fide integrable system, but
one cannot directly add the boundary to the phase space because that would not yield
a smooth manifold. One of the seminal results of \cite{Ru95} is the solution of this
conundrum. In fact, Ruijsenaars constructed a symplectic embedding of the
center-of-mass phase space of the system into the complex projective space $\CP^{n-1}$, such
that the image of the embedding is a dense open submanifold and the Hamiltonian \eqref{5.1} as well as its
commuting family extend to smooth functions on the full $\CP^{n-1}$.
As $\CP^{n-1}$ is compact, the corresponding Hamiltonian flows are complete.
The resulting `compactified trigonometric RS system' has been studied at the classical level
in detail \cite{Ru95}, and after an initial exploration of the rank 1 case \cite{Ru90},
its quantum mechanical version was also solved \cite{vDV98}. These classical systems are self-dual
in the sense that their position and action variables can be exchanged by a canonical transformation
of order 4, somewhat akin to the mapping $(x,\phi) \mapsto (-\phi, x)$ for a free particle, and
their quantum mechanical versions enjoy the bispectral property \cite{Ru90,vDV98}.

The possibility of an analogous compactification
of the elliptic RS system having the Hamiltonian
\begin{equation}
H=\sum_{k=1}^n\cos(\phi_k)\sqrt{\prod_{\substack{j=1\\(j\neq k)}}^n
\big[\ws(y)^2\big(\wp(y)-\wp(x_j-x_k)\big)\big]}
\label{I5}
\end{equation}
with
 functions
$\wp$ \eqref{wp} and $\ws$ \eqref{sigma}
was pointed out in \cite{Ru90,Ru99}, but it was not
described in detail.

Even though it was only proved  \cite{Ru95} that the restrictions \eqref{I2}, \eqref{I3}
are sufficient to allow compactification, equation \eqref{I3} was customarily mentioned
in the literature  \cite{FK12,GN95,Ru90-2,Ru99,vDV98} as a necessary condition for the systems to make sense.
However, in a recent work \cite{FKl14} a completion of the III$_\b$ system on a compact phase space was
obtained for any generic parameter
\begin{equation}
0<y<\pi.
\label{I4}
\end{equation}
The paper \cite{FKl14} relied on deriving compactified RS systems in the center-of-mass frame
via reduction of a `free system' on the quasi-Hamiltonian \cite{AMM98} double $\SU(n) \times \SU(n)$.
This was achieved by setting
the relevant group-valued moment map equal to the constant matrix
$\mu_0(y)=\diag(e^{2\ri y},\dots,e^{2\ri y},e^{-2(n-1)\ri y})$,
and it makes perfect sense for any (generic) $y$.
The corresponding domain of the position variables depends on $y$ and differs
from the one posited in \eqref{I2}.
The possibility to relax the condition \eqref{I3} on $y$  also appeared in \cite{BGS09}.

The principal motivation for our present work comes from the classification of the coupling parameter
$y$ found in \cite{FKl14}. Namely, it turned out that the reduction is applicable except for a finite
set of $y$-values, and the rest of the set $(0,\pi)$ decomposes into two subsets, containing so-called
type (i) and type (ii) $y$-values. The `main reduced Hamiltonian' always takes the III$_\b$ form
\eqref{5.1} on a dense open subset of the reduced phase space. In the type (i) cases the particles
cannot collide and the action variables of the reduced system naturally engender an isomorphism with
the Hamiltonian toric manifold $\CP^{n-1}$. In type (ii) cases, that exist for any $n>3$, the reduction
constraints admit solutions $(a,b)\in\SU(n)\times\SU(n)$ for which the eigenvalues of $a$ or $b$ are
not all distinct, entailing that the particles of the reduced system can collide.
For a detailed exposition of these succinct statements, the reader may consult \cite{FKl14}.
We here only add the remark that the connected domain of the positions
always contains the equal-distance configuration
$x_{k+1} - x_k = \pi/n$ ($\forall k$)  for which
the number of negative factors in each product under the square root in \eqref{5.1}
is $2 \lfloor n y/\pi\rfloor$
 if $0<y < \pi/2$ and
 $2 \lfloor n (\pi -y)/\pi \rfloor$  if $\pi/2 < y <\pi$.

\section{Embedding of the local phase space into $\mathbb{CP}^{n-1}$}
\label{sec:5.1}

In this section we first recall the local phase space of the III$_\b$ model from \cite{FKl14},
and then present its symplectic embedding into $\CP^{n-1}$ in every type (i) case.

The III$_\b$ model can be thought of as $n$ interacting particles on the unit circle
with positions $\delta_k=e^{2\ri x_k}$.
We impose the condition $\prod_{k=1}^n\delta_k=1$, which means that we
work in the `center-of-mass frame',
and parametrize the positions as
\begin{equation}
\delta_1(\xi)=e^{\frac{2\ri}{n}\sum_{j=1}^nj\xi_j},\qquad
\delta_{k}(\xi)=e^{2\ri\xi_{k-1}}\delta_{k-1}(\xi),\quad k=2,\dots,n,
\label{delta}
\end{equation}
where $\xi$ belongs to a certain open subset $\cA_y^+$ inside the `Weyl alcove'
\begin{equation}
\cA=\{\xi\in\R^n\mid\xi_k\geq 0\ (k=1,\dots,n),\ \xi_1+\dots+\xi_n=\pi\}.
\label{Weyl}
\end{equation}
Note that $\cA$ is a simplex in the $(n-1)$-dimensional affine space
\begin{equation}
E=\{\xi\in\R^n\mid\xi_1+\dots+\xi_n=\pi\}.
\label{E}
\end{equation}
The local phase space can be described as the product manifold
\begin{equation}
P_y^\loc=\{(\xi,e^{\ri\theta})\mid\xi\in\cA_y^+,\ e^{\ri\theta}\in\T^{n-1}\},
\label{P-loc}
\end{equation}
where $\T^{n-1}$ is the $(n-1)$-torus, equipped with the standard symplectic form
\begin{equation}
\omega^\loc=\sum_{k=1}^{n-1}d\theta_k\wedge d\xi_k.
\label{om-loc}
\end{equation}
The dynamics is governed by the Hamiltonian
\begin{equation}
H_y^\loc(\xi,\theta)=\sum_{j=1}^n\cos(\theta_j-\theta_{j-1})\sqrt{ \prod_{m=j+1}^{j+n-1}
\biggl[ 1-\frac{\sin^2y}{\sin^2(\sum_{k=j}^{m-1}\xi_k)}\biggr]}.
\label{H_y^loc}
\end{equation}
Here, $\theta_0=\theta_n=0$ have been introduced and the indices are understood
modulo $n$, i.e.
\begin{equation}
\xi_{m+n}=\xi_m, \quad\forall m.
\label{perconv}
\end{equation}
The product under the square
root is positive for every $\xi\in\cA_y^+$, and thus $H_y^\loc\in C^\infty(P_y^\loc)$.
This model was considered in \cite{FKl14} for any $y$ chosen from the interval $(0,\pi)$
except the excluded values that satisfy $e^{2\ri my}=1$ for some $m=1,\dots,n$.

According to \cite{FKl14}, there are two different kinds of intervals for $y$ to be in,
named type (i) and (ii). The type (i) couplings can be described as follows. For a
fixed positive integer $n\geq 2$, choose $p\in\{1,\dots,n-1\}$ to be a coprime to
$n$, i.e. $\gcd(n,p)=1$, and let $q$ denote the multiplicative inverse of $p$ in
the ring $\Z_n$, that is $pq\equiv 1\pmod{n}$. Then the parameter $y$ can take its
values according to either
\begin{equation}
\bigg(\frac{p}{n}-\frac{1}{nq}\bigg)\pi<y<\frac{p\pi}{n}
\qquad\text{or}\qquad
\frac{p\pi}{n}<y<\bigg(\frac{p}{n}+\frac{1}{(n-q)n}\bigg)\pi.
\label{typeI-y}
\end{equation}
For such a type (i) parameter $y$, the local configuration space $\cA_y^+$
is the interior of a simplex $\cA_y$ in $E$ \eqref{E}
bounded by the hyperplanes
\begin{equation}
\xi_j+\dots+\xi_{j+p-1}=y,\quad j=1,\dots,n,
\label{hyperplanes}
\end{equation}
where \eqref{perconv} is understood.
To give a more detailed description of $\cA_y$, we introduce
\begin{equation}
M=p\pi-ny,
\label{M}
\end{equation}
and note that \eqref{typeI-y} gives $M>0$ and $M<0$, respectively.
Then any $\xi\in\cA_y$ must satisfy
\begin{equation}
\sgn(M)(\xi_j+\dots+\xi_{j+p-1}-y)\geq 0,\quad j=1,\dots,n.
\label{5.11}
\end{equation}
In terms of the particle coordinates $x_k$, which are ordered as $x_{k+1} \geq x_k$ and
extended by the convention $x_{k+n} = x_k + \pi$, the above condition says that
\begin{equation}
x_{j+p}-x_j\geq y\quad\text{if}\ M>0
\quad\text{and}\quad
x_{j+p}-x_j\leq y\quad\text{if}\ M<0
\label{5.12}
\end{equation}
for every $j$.
Therefore the distances of the $p$-th neighbouring particles on the circle are constrained.
The $n$ vertices of the simplex $\cA_y$ are explicitly given in \cite[Proposition 11
and Lemma 8 \emph{op. cit.}]{FKl14}. Every vertex and thus $\cA_y$ itself lies inside the larger
simplex $\cA$ \eqref{Weyl}, entailing that $x_{j+1}-x_j$ possesses a positive lower bound
in each type (i) case.

The type (ii) cases correspond to those admissible $y$-values that do not satisfy
\eqref{typeI-y} for any $p$ relative prime to $n$. In such cases $\cA_y^+$ has a different
structure \cite{FKl14}. Type (ii) cases exist for every $n\geq 4$. See Figure \ref{fig:7}
for an illustration.

\begin{figure}[h!]
\centering
\begin{tikzpicture}
\def\s{.9\textwidth}
\def\r{.2em}
\draw(1em,1em) node{$n=4$};
\draw (0,0)--({\s/3},0) ({2*\s/3},0)--({\s},0);
\draw[dashed] ({\s/3},0)--({2*\s/3},0);
\draw[black,fill=white]
(0,0) circle(\r) node[below,yshift=-1mm]{$0$}
(\s/4,0) circle(\r) node[below,yshift=-1mm]{$\displaystyle\frac{1}{4}$}
(\s/3,0) circle(\r) node[below,yshift=-1mm]{$\displaystyle\frac{1}{3}$}
(\s/2,0) circle(\r) node[below,yshift=-1mm]{$\displaystyle\frac{1}{2}$}
(2*\s/3,0) circle(\r) node[below,yshift=-1mm]{$\displaystyle\frac{2}{3}$}
(3*\s/4,0) circle(\r) node[below,yshift=-1mm]{$\displaystyle\frac{3}{4}$}
(\s,0) circle(\r) node[below,yshift=-1mm]{$1$};
\draw(1em,-3em) node{$n=5$};
\draw (0,-4em)--(\s/4,-4em) (\s/3,-4em)--(2*\s/3,-4em) (3*\s/4,-4em)--(\s,-4em);
\draw[dashed] (\s/4,-4em)--(\s/3,-4em) (2*\s/3,-4em)--(3*\s/4,-4em);
\draw[black,fill=white]
(0,-4em) circle(\r) node[below,yshift=-1mm]{$0$}
(\s/5,-4em) circle(\r) node[below,yshift=-1mm]{$\displaystyle\frac{1}{5}$}
(\s/4,-4em) circle(\r) node[below,yshift=-1mm]{$\displaystyle\frac{1}{4}$}
(\s/3,-4em) circle(\r) node[below,yshift=-1mm]{$\displaystyle\frac{1}{3}$}
(2*\s/5,-4em) circle(\r) node[below,yshift=-1mm]{$\displaystyle\frac{2}{5}$}
(\s/2,-4em) circle(\r) node[below,yshift=-1mm]{$\displaystyle\frac{1}{2}$}
(3*\s/5,-4em) circle(\r) node[below,yshift=-1mm]{$\displaystyle\frac{3}{5}$}
(2*\s/3,-4em) circle(\r) node[below,yshift=-1mm]{$\displaystyle\frac{2}{3}$}
(3*\s/4,-4em) circle(\r) node[below,yshift=-1mm]{$\displaystyle\frac{3}{4}$}
(4*\s/5,-4em) circle(\r) node[below,yshift=-1mm]{$\displaystyle\frac{4}{5}$}
(\s,-4em) circle(\r) node[below,yshift=-1mm]{$1$};
\draw(1em,-7em) node{$n=6$};
\draw (0,-8em)--(\s/5,-8em) (4*\s/5,-8em)--(\s,-8em);
\draw[dashed] (\s/5,-8em)--(4*\s/5,-8em);
\draw[black,fill=white]
(0,-8em) circle(\r) node[below,yshift=-1mm]{$0$}
(\s/6,-8em) circle(\r) node[below,yshift=-1mm]{$\displaystyle\frac{1}{6}$}
(\s/5,-8em) circle(\r) node[below,yshift=-1mm]{$\displaystyle\frac{1}{5}$}
(\s/4,-8em) circle(\r) node[below,yshift=-1mm]{$\displaystyle\frac{1}{4}$}
(\s/3,-8em) circle(\r) node[below,yshift=-1mm]{$\displaystyle\frac{1}{3}$}
(2*\s/5,-8em) circle(\r) node[below,yshift=-1mm]{$\displaystyle\frac{2}{5}$}
(\s/2,-8em) circle(\r) node[below,yshift=-1mm]{$\displaystyle\frac{1}{2}$}
(3*\s/5,-8em) circle(\r) node[below,yshift=-1mm]{$\displaystyle\frac{3}{5}$}
(2*\s/3,-8em) circle(\r) node[below,yshift=-1mm]{$\displaystyle\frac{2}{3}$}
(3*\s/4,-8em) circle(\r) node[below,yshift=-1mm]{$\displaystyle\frac{3}{4}$}
(4*\s/5,-8em) circle(\r) node[below,yshift=-1mm]{$\displaystyle\frac{4}{5}$}
(5*\s/6,-8em) circle(\r) node[below,yshift=-1mm]{$\displaystyle\frac{5}{6}$}
(\s,-8em) circle(\r) node[below,yshift=-1mm]{$1$};
\draw(1em,-11em) node{$n=7$};
\draw (0,-12em)--(\s/6,-12em) (\s/4,-12em)--(\s/3,-12em) (2*\s/5,-12em)--(3*\s/5,-12em)
(2*\s/3,-12em)--(3*\s/4,-12em) (5*\s/6,-12em)--(\s,-12em);
\draw[dashed] (\s/6,-12em)--(\s/4,-12em) (\s/3,-12em)--(2*\s/5,-12em)
(3*\s/5,-12em)--(2*\s/3,-12em) (3*\s/4,-12em)--(5*\s/6,-12em);
\draw[black,fill=white]
(0,-12em) circle(\r) node[below,yshift=-1mm]{$0$}
(\s/7,-12em) circle(\r) node[below,yshift=-1mm]{$\displaystyle\frac{1}{7}$}
(\s/6,-12em) circle(\r) node[below,yshift=-1mm]{$\displaystyle\frac{1}{6}$}
(\s/5,-12em) circle(\r) node[below,yshift=-1mm]{$\displaystyle\frac{1}{5}$}
(\s/4,-12em) circle(\r) node[below,yshift=-1mm]{$\displaystyle\frac{1}{4}$}
(2*\s/7,-12em) circle(\r) node[below,yshift=-1mm]{$\displaystyle\frac{2}{7}$}
(\s/3,-12em) circle(\r) node[below,yshift=-1mm]{$\displaystyle\frac{1}{3}$}
(2*\s/5,-12em) circle(\r) node[below,yshift=-1mm]{$\displaystyle\frac{2}{5}$}
(3*\s/7,-12em) circle(\r) node[below,yshift=-1mm]{$\displaystyle\frac{3}{7}$}
(\s/2,-12em) circle(\r) node[below,yshift=-1mm]{$\displaystyle\frac{1}{2}$}
(4*\s/7,-12em) circle(\r) node[below,yshift=-1mm]{$\displaystyle\frac{4}{7}$}
(3*\s/5,-12em) circle(\r) node[below,yshift=-1mm]{$\displaystyle\frac{3}{5}$}
(2*\s/3,-12em) circle(\r) node[below,yshift=-1mm]{$\displaystyle\frac{2}{3}$}
(5*\s/7,-12em) circle(\r) node[below,yshift=-1mm]{$\displaystyle\frac{5}{7}$}
(3*\s/4,-12em) circle(\r) node[below,yshift=-1mm]{$\displaystyle\frac{3}{4}$}
(4*\s/5,-12em) circle(\r) node[below,yshift=-1mm]{$\displaystyle\frac{4}{5}$}
(5*\s/6,-12em) circle(\r) node[below,yshift=-1mm]{$\displaystyle\frac{5}{6}$}
(6*\s/7,-12em) circle(\r) node[below,yshift=-1mm]{$\displaystyle\frac{6}{7}$}
(\s,-12em) circle(\r) node[below,yshift=-1mm]{$1$};
\end{tikzpicture}
\caption{The range of $y/\pi$ for $n=4,5,6,7$. The displayed numbers are excluded values.
Admissible values of $y$ form intervals of type (i) (solid) and type (ii) (dashed) couplings.}
\label{fig:7}
\end{figure}
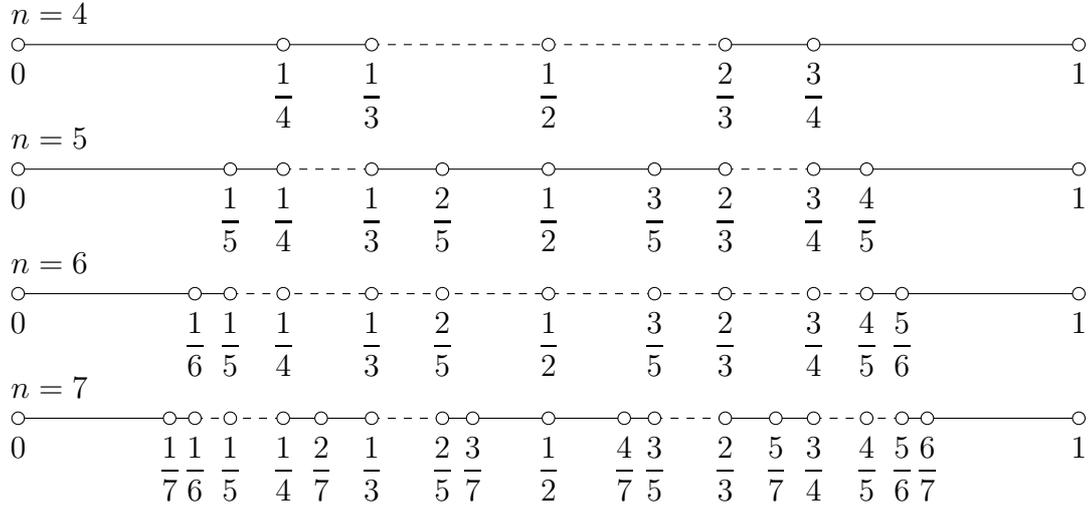

We further continue with the assumption that $y$ satisfies \eqref{typeI-y}.
Motivated by \cite{Ru95,FK12}, we now introduce the map
\begin{equation}
\cE\colon\cA_y^+\times\T^{n-1}\to\C^n,\quad
(\xi,e^{\ri\theta})\mapsto (u_1,\dots,u_n)
\label{Emap}
\end{equation}
with the complex coordinates having the squared absolute values
\begin{equation}
|u_j|^2=\sgn(M)(\xi_j+\dots+\xi_{j+p-1}-y),\quad j=1,\dots,n,
\label{u-abs-squared}
\end{equation}
and the arguments
\begin{equation}
\arg(u_j)=\sgn(M)\sum_{k=1}^{n-1}\Omega_{j,k}\theta_k,\quad j=1,\dots,n-1,
\qquad\arg(u_n)=0,
\label{argu}
\end{equation}
where the $\Omega_{j,k}$ ($j,k=1,\dots,n-1$) are integers chosen in such a way that
\begin{equation}
\cE^\ast\bigg(\ri\sum_{j=1}^nd\bar u_j\wedge du_j\bigg)
=\sum_{k=1}^{n-1}d\theta_k\wedge d\xi_k.
\label{cE-sympl}
\end{equation}
In order for \eqref{cE-sympl} to be achieved $\Omega$ has to be the inverse transpose of
the $(n-1)\times(n-1)$ coefficient matrix of $\xi_1,\dots,\xi_{n-1}$ extracted from eqs.
\eqref{u-abs-squared} by applying $\xi_1+\dots+\xi_n=\pi$. In other words, the squared
absolute values $|u_j|^2$ are written as
\begin{equation}
|u_j|^2=\begin{cases}
\sgn(M)\big(\sum_{k=1}^{n-1}A_{j,k}\xi_k-y\big),&\text{if}\ 1\leq j\leq n-p,\\
\sgn(M)\big(\sum_{k=1}^{n-1}A_{j,k}\xi_k-y+\pi\big),&\text{if}\ n-p<j\leq n-1,
\end{cases}
\label{uabs}
\end{equation}
where $A$ stands for the above-mentioned coefficient matrix, which has the components
\begin{equation}
A_{j,k}=\begin{cases}
+1,&\text{if}\ 1\leq j\leq n-p\ \text{and}\ j\leq k<j+p,\\
-1,&\text{if}\ n-p<j\leq n-1\ \text{and}\ j+p-n\leq k<j,\\
0,&\text{otherwise}.
\end{cases}
\label{A_j,k}
\end{equation}
A close inspection of the structure of $A$ reveals that
\begin{equation}
\det(A)=(-1)^{(n-p)(p-1)}\prod_{j=1}^{n-p}A_{j,j+p-1}\prod_{k=1}^{p-1}A_{n-p+k,k}
=(-1)^{(n-p+1)(p-1)}=+1,
\label{det(A)}
\end{equation}
therefore $\Omega=(A^{-1})^\top$ exists and consists of integers, as required in \eqref{argu}.
Next, we give $\Omega$ explicitly.

\begin{proposition}
\label{prop:5.1}
The transpose of the inverse of the matrix $A$ \eqref{A_j,k} can be written as
\begin{equation}
\Omega=B-C,
\label{B-C}
\end{equation}
where $B$ is a $(0,1)$-matrix of size $(n-1)$ with zeros along certain diagonals given by
\begin{equation}
B_{m,k}=\begin{cases}
0,&\text{if}\ k-m\equiv \ell p\pmod{n}\ \text{for some}\
\ell\in\{1,\dots,n-q\},\\
1,&\text{otherwise},
\end{cases}
\label{B_m,k}
\end{equation}
and $C$ is also a binary matrix of size $(n-1)$ with zeros along columns given by
\begin{equation}
C_{m,k}=\begin{cases}
0,&\text{if}\ k\equiv \ell p\pmod{n}\ \text{for some}\
\ell\in\{1,\dots,n-q\},\\
1,&\text{otherwise}.
\end{cases}
\label{C_m,k}
\end{equation}
\end{proposition}
\begin{proof}
We start by presenting a useful auxiliary statement.
Let us introduce the subsets $S$ and $S_i$ of the ring $\Z_n$ as
\begin{equation}
S=\{\ell p\ (\bmod\ n)\mid\ell=1,\dots,n-q\},
\quad
S_i=\{i+\ell\ (\bmod\ n)\mid\ell=0,\dots,p-1\},
\end{equation}
for any $i\in\Z_n$.
Then define $I_i\in\N$ to be the number of elements in the intersection $S_i\cap S$.
Notice that $i\in S$ if and only if $(i+p)\in S$ except for $i\equiv(n-1)\equiv(n-q)p\pmod{n}$,
for which $(n-1)+p\equiv(n-q+1)p\pmod{n}$ does not belong to $S$.
It follows that
\begin{equation}
I_1=\dots=I_{n-1}=I_n+1.
\label{shiftid}
\end{equation}

Our aim is to show that $(A\Omega^\top)_{j,m}=\delta_{j,m}$ ($\forall j,m$)
with $\Omega$ defined by \eqref{B-C}-\eqref{C_m,k}. First, by the formula of $A$
\eqref{A_j,k} for any $1\leq j\leq n-p$ and $1\leq m\leq n-1$ we have
\begin{equation}
(A\Omega^\top)_{j,m}
=\sum_{k=1}^{n-1}A_{j,k}\Omega_{m,k}
=\sum_{k=j}^{j+p-1}\Omega_{m,k}
=\sum_{k=j}^{j+p-1}(B_{m,k}-C_{m,k}).
\label{}
\end{equation}
The definition of the matrices $B$ \eqref{B_m,k} and $C$ \eqref{C_m,k} gives directly that
\begin{equation}
\sum_{k=j}^{j+p-1}B_{m,k}=p-I_{j-m},
\qquad
\sum_{k=j}^{j+p-1}C_{m,k}=p-I_j.
\label{}
\end{equation}
By using \eqref{shiftid}, this readily implies that
$(A\Omega^\top)_{j,m}=\delta_{j,m}$ holds for the case at hand.

Second, for any $n-p<j\leq n-1$ and $1\leq m\leq n-1$ we have
\begin{equation}
(A\Omega^\top)_{j,m}
=\sum_{k=1}^{n-1}A_{j,k}\Omega_{m,k}
=\sum_{k=j+p-n}^{j-1}(-1)\Omega_{m,k}
=\sum_{k=j+p-n}^{j-1}(C_{m,k}-B_{m,k}).
\label{}
\end{equation}
From this point on the reasoning is quite similar to the previous case, and we obtain that
$(A\Omega^\top)_{j,m}=\delta_{j,m}$ always holds.
\end{proof}

To enlighten the geometric meaning of the map $\cE$ \eqref{Emap}, notice from
\eqref{u-abs-squared} that
\begin{equation}
\sum_{j=1}^n|u_j|^2=\sgn(M)\big(p(\xi_1+\dots+\xi_n)-ny\big)
=\sgn(M)\big(p\pi-ny\big)=|M|.
\label{}
\end{equation}
Then represent the complex projective space $\CP^{n-1}$ as
\begin{equation}
\CP^{n-1}=S_{|M|}^{2n-1}/\UN(1)
\label{coset}
\end{equation}
with
\begin{equation}
S_{|M|}^{2n-1}=\{(u_1,\dots,u_n)\in\C^n\mid|u_1|^2+\dots+|u_n|^2=|M|\}.
\label{sphere}
\end{equation}
Correspondingly, let
\begin{equation}
\pi_{|M|}\colon S_{|M|}^{2n-1}\to\CP^{n-1}
\label{}
\end{equation}
denote the natural projection and equip $\CP^{n-1}$ with the rescaled Fubini-Study
symplectic form $|M|\omega_{\text{FS}}$ characterized by the relation
\begin{equation}
\pi_{|M|}^\ast(|M|\omega_{\text{FS}})=\ri\sum_{j=1}^nd\bar u_j\wedge du_j,
\label{5.37}
\end{equation}
where the $u_j$'s are regarded as functions on $S^{2n-1}_{\vert M\vert }$.
It is readily seen from the definitions that the map
\begin{equation}
\pi_{|M|}\circ\cE\colon\cA_y^+\times\T^{n-1}\to\CP^{n-1}
\label{pi+cE}
\end{equation}
is smooth, injective and its image is the open submanifold for which
$\prod_{j=1}^n \vert u_j \vert^2 \neq 0$.
Equations \eqref{om-loc}, \eqref{cE-sympl} and \eqref{5.37}  together imply the symplectic property
\begin{equation}
(\pi_{|M|}\circ\cE)^\ast (\vert M\vert \omega_{\text{FS}}) = \omega^\loc,
\end{equation}
from which it follows that this map is an \emph{embedding}.

To summarize, in this section we have constructed the
symplectic diffeomorphism $\pi_{|M|}\circ\cE$ between the local phase space
$P_y^\loc$ \eqref{P-loc}
and the dense open submanifold of $\CP^{n-1}$ on which the product of the homogeneous
coordinates is nowhere zero.
If desired, the explicit formula of the smooth inverse mapping can be easily found as well.

\section{Global extension of the trigonometric Lax matrix}
\label{sec:5.2}

It was proved in \cite{FKl14} with the aid of quasi-Hamiltonian reduction that the global
phase space of the III$_\b$ model is $\CP^{n-1}$ for the type (i) couplings, which we
continue to consider. Here, we utilize the symplectic embedding \eqref{pi+cE} to
construct a global Lax matrix on $\CP^{n-1}$ explicitly, starting from the local
RS Lax matrix defined on $\cA_y^+\times\T^{n-1}$. This issue was not investigated
previously except for the $p=1$ case of \eqref{typeI-y}, see \cite{Ru95,FK12,FKl14}.

The local Lax matrix $L_y^\loc(\xi,e^{\ri \theta})\in\SU(n)$ used in \cite{FKl14} contains the
 trigonometric Cauchy matrix $C_y$ given with the help of \eqref{delta} by
\begin{equation}
C_y(\xi)_{j,\ell}=\frac{e^{\ri y}-e^{-\ri y}}
{e^{\ri y}\delta_j(\xi)^{1/2}\delta_\ell(\xi)^{-1/2}
-e^{-\ri y}\delta_j(\xi)^{-1/2}\delta_\ell(\xi)^{1/2}}.
\label{C_y}
\end{equation}
Thanks to the relation $\delta_k(\xi)= e^{2\ri x_k}$, this is equivalent to
\begin{equation}
C_y(\xi)_{j,\ell}=\frac{\sin(y)}{\sin(x_j-x_\ell+y)}.
\label{C_y-2}
\end{equation}
Then we have
\begin{equation}
L_y^\loc(\xi,e^{\ri\theta})_{j,\ell}=
C_y(\xi)_{j,\ell}v_j(\xi,y)v_\ell(\xi,-y)\rho(\theta)_\ell,
\qquad
\forall (\xi, e^{\ri \theta})\in \cA_y^+\times\T^{n-1},
\label{L_y^loc}
\end{equation}
where $\rho(\theta)_\ell=e^{\ri(\theta_{\ell-1}-\theta_\ell)}$
(applying $\theta_0=\theta_n=0$) and
\begin{equation}
v_\ell(\xi,\pm y)=\sqrt{z_\ell(\xi,\pm y)}\quad\text{with}\quad
z_\ell(\xi,\pm y)=\sgn(\sin(ny))\prod_{m=\ell+1}^{\ell+n-1}
\frac{\sin(\sum_{k=\ell}^{m-1}\xi_k\mp y)}{\sin(\sum_{k=\ell}^{m-1}\xi_k)}.
\label{z_ell(xi,pmy)}
\end{equation}
A key point \cite{FKl14} (which is detailed below) is that $z_\ell(\xi,\pm y)$ is
positive for any $\xi \in \cA_y^+$. We note for clarity that $z_\ell$ and $v_\ell$
above differ from those in \cite{FKl14} by a harmless multiplicative constant,
and also mention that $L_y^\loc$ is a specialization of (a similarity transform of)
the standard RS Lax matrix \cite{Ru99}.

The spectral invariants of $L_y^\loc$ \eqref{L_y^loc} yield a Poisson
commuting family of functional dimension $(n-1)$ \cite{Ru99,FKl14}, containing
the Hamiltonian $H_y^\loc$ \eqref{H_y^loc} due to the equation
\begin{equation}
\Re\big(\tr L_y^\loc(\xi,e^{\ri\theta})\big)=H_y^\loc(\xi,\theta).
\label{Re-tr-L_y^loc}
\end{equation}
There are two important observations to be made here. First, for each
$1\leq\ell\leq n$, there is only one factor in $z_\ell(\xi,\pm y)$
\eqref{z_ell(xi,pmy)} that (up to sign) contains the sine of the squared
absolute value \eqref{u-abs-squared} of one of the complex variables in
its numerator:
\begin{itemize}
\item For $z_\ell(\xi,y)$, it is the factor corresponding to $m=\ell+p$,
whose numerator is
\begin{equation}
\sgn(M)\sin(|u_\ell|^2).
\label{+y-num}
\end{equation}
\item For $z_\ell(\xi,-y)$, it is the factor with $m=\ell+n-p$, whose the numerator is either
\begin{equation}
\sin(\pi-\sgn(M)|u_{\ell+n-p}|^2)=\sgn(M)\sin(|u_{\ell+n-p}|^2),
\quad\text{if}\ 1\leq\ell\leq p,
\label{-y-num-1}
\end{equation}or
\begin{equation}
\sin(\pi-\sgn(M)|u_{\ell-p}|^2)=\sgn(M)\sin(|u_{\ell-p}|^2),
\quad\text{if}\ p<\ell\leq n.
\label{-y-num-2}
\end{equation}
Here we made use of $\xi_1+\dots+\xi_n=\pi$, $\sin(\pi-\alpha)=\sin(\alpha)$
and $\sin(-\alpha)=-\sin(\alpha)$.
\end{itemize}
Second, the $(p-1)$ factors in $z_\ell(\xi,\pm y)$ with
$m<\ell+p$ and $m>\ell+n-p$, respectively, are strictly negative and the
factors corresponding to $m>\ell+p$ and $m<\ell+n-p$, respectively, are
strictly positive for all $\xi$ in the \emph{closed} simplex $\cA_y$.
In particular, for any $\xi\in\cA_y^+$ the sign of the $\xi$-dependent product in
\eqref{z_ell(xi,pmy)} equals $(-1)^{p-1}\sgn(M)=\sgn(\sin(ny))$, and therefore
\begin{equation}
z_\ell(\xi,\pm y)\geq 0,\quad\forall\xi\in\cA_y,\quad\ell=1,\dots,n.
\label{z-ell-pos}
\end{equation}
We saw that $z_\ell$ can only vanish due to the numerators \eqref{+y-num} and
\eqref{-y-num-1}, \eqref{-y-num-2}, respectively.
Consequently, in \eqref{z_ell(xi,pmy)} the positive square
root of $z_\ell(\xi,\pm y)$ can be taken for any $\xi\in\cA_y^+$.

Now notice that, for all $\xi\in \cA_y^+$, we have
\begin{equation}
v_j(\xi,y)=|u_j|w_j(\xi,y),\quad 1\leq j\leq n,
\label{v_ell(xi,y)}
\end{equation}
where the $w_j(\xi,y)$ are positive and smooth functions of the form
\begin{equation}
w_j(\xi,y)=\bigg[
\frac{\sin(|u_j|^2)}{|u_j|^2}
\frac{(-1)^{p-1}}{\sin(\sum_{k=j}^{j+p-1}\xi_k)}
\prod_{\substack{m=j+1\\(m\neq j+p)}}^{j+n-1}
\frac{\sin(\sum_{k=j}^{m-1}\xi_k-y)}{\sin(\sum_{k=j}^{m-1}\xi_k)}
\bigg]^{\tfrac{1}{2}}.
\label{w_j(xi,y)}
\end{equation}
Similarly, we have
\begin{equation}
v_\ell(\xi,-y)=\begin{cases}
|u_{\ell+n-p}|w_\ell(\xi,-y),&\text{if}\ 1\leq\ell\leq p,\\
|u_{\ell-p}|w_\ell(\xi,-y),&\text{if}\ p<\ell\leq n
\end{cases}
\label{v_ell(xi,-y)}
\end{equation}
with the positive and smooth functions
\begin{equation}
w_\ell(\xi,-y)=\bigg[
\frac{\sin(|u_{\ell+n-p}|^2)}{|u_{\ell+n-p}|^2}
\frac{(-1)^{p-1}}{\sin(\sum_{k=\ell}^{\ell+n-p-1}\xi_k)}
\prod_{\substack{m=\ell+1\\(m\neq\ell+n-p)}}^{\ell+n-1}
\frac{\sin(\sum_{k=\ell}^{m-1}\xi_k+y)}{\sin(\sum_{k=\ell}^{m-1}\xi_k)}
\bigg]^{\tfrac{1}{2}}
\label{w_ell(xi,-y)-1}
\end{equation}
for $1\leq\ell\leq p$, and
\begin{equation}
w_\ell(\xi,-y)=\bigg[
\frac{\sin(|u_{\ell-p}|^2)}{|u_{\ell-p}|^2}
\frac{(-1)^{p-1}}{\sin(\sum_{k=\ell}^{\ell+n-p-1}\xi_k)}
\prod_{\substack{m=\ell+1\\(m\neq\ell+n-p)}}^{\ell+n-1}
\frac{\sin(\sum_{k=\ell}^{m-1}\xi_k+ y)}{\sin(\sum_{k=\ell}^{m-1}\xi_k)}
\bigg]^{\tfrac{1}{2}}
\label{w_ell(xi,-y)-2}
\end{equation}
for $p<\ell\leq n$.

The relation \eqref{uabs} allows us to express the $\xi_k$ in terms
of the complex variables for $k=1,\dots,n-1$ as
\begin{equation}
\xi_k(u)=
\sum_{j=1}^{n-1}\Omega_{j,k}\big(\sgn(M)|u_j|^2+c_j\big),\quad\text{with}\
c_j=\begin{cases}y,&\text{if}\ 1\leq j\leq n-p,\\
y-\pi,&\text{if}\ n-p<j\leq n-1,
\end{cases}
\label{xi(u)}
\end{equation}
and $\xi_n(u)=\pi-\xi_1(u)-\dots-\xi_{n-1}(u)$.
These formulas extend to $\UN(1)$-invariant smooth functions on $S_{|M|}^{2n-1}$,
which represent smooth functions on $\CP^{n-1}$ on account of \eqref{coset}.
By applying these, the above expressions $w_j(\xi(u), \pm y)$ $(j=1,\dots,n)$
\emph{give rise to smooth functions on $\CP^{n-1}$}.

\begin{definition}
\label{def:5.2}
By setting $\theta_k=0$ $(\forall k)$ in the local Lax matrix $L_y^\loc$
\eqref{L_y^loc} with $y$ \eqref{typeI-y}, we define the functions
$\Lambda_{j,\ell}^y\colon\cA_y^+\to\R$ ($j,\ell=1,\dots,n$) via the equations
\begin{equation}
\Lambda_{j,j+p}^y(\xi)=L_y^\loc(\xi,\1_{n-1})_{j,j+p},\quad 1\leq j\leq n-p,
\label{L_y^loc_j,j+p}
\end{equation}
\begin{equation}
\Lambda_{j,j+p-n}^y(\xi)=L_y^\loc(\xi,\1_{n-1})_{j,j+p-n},\quad n-p<j\leq n,
\label{L_y^loc_j,j+p-n}
\end{equation}
\begin{equation}
\Lambda_{j,\ell}^y(\xi)=L_y^\loc(\xi,\1_{n-1})_{j,\ell}(|u_j||u_{\ell+n-p}|)^{-1},\quad
1\leq j\leq n,\ 1\leq\ell\leq p\quad (\ell\neq j+p-n),
\label{L_y^loc_j,ell-1}
\end{equation}
\begin{equation}
\Lambda_{j,\ell}^y(\xi)=L_y^\loc(\xi,\1_{n-1})_{j,\ell}(|u_j||u_{\ell-p}|)^{-1},\quad
1\leq j\leq n,\ p<\ell\leq n\quad (\ell\neq j+p).
\label{L_y^loc_j,ell-2}
\end{equation}
\end{definition}

The foregoing results lead to explicit formulas for $\Lambda_{j,\ell}^y$
(see Appendix \ref{sec:E.1}). Using the identification \eqref{coset}
and \eqref{xi(u)}, it is readily seen that the $\Lambda^y_{j,\ell}(\xi(u))$
given by Definition \ref{def:5.2} extend to smooth functions on $\CP^{n-1}$.

\begin{remark}
\label{rem:5.3}
The explicit formulas of $\Lambda_{j,\ell}^y(\xi(u))$ contain products of square roots
of strictly positive functions depending on $|u_k|^2\in C^\infty(S^{2n-1}_{|M|})^{\UN(1)}$
for $k=1,\dots,n$. In particular, they contain the square root of the function $J$ given by
\begin{equation}
J(|u_k|^2)=\frac{\sin(|u_k|^2)}{|u_k|^2},
\label{J1}
\end{equation}
which remains smooth (even real-analytic) at $|u_k|^2=0$ and is positive since
we have $0\leq|u_k|^2\leq|M|<\pi$. Indeed, $|M|<\pi/q$ and $|M|<\pi/(n-q)$,
respectively, for the two intervals of the type (i) couplings in \eqref{typeI-y}.
\end{remark}

The above observations allow us to introduce the following functions,
which will be used to construct the global Lax matrix.
\begin{definition}
\label{def:5.4}
For $M>0$ \eqref{M}, define the smooth functions $L_{j,\ell}^{y,+}\colon\CP^{n-1}\to\C$ by
\begin{align}
L_{j,\ell}^{y,+}\circ \pi_{\vert M\vert}(u)&=\begin{cases}
\Lambda_{j,j+p}^y(\xi(u)),&\text{if}\ 1\leq j\leq n-p,\ \ell=j+p,\\
\Lambda_{j,j+p-n}^y(\xi(u)),&\text{if}\ n-p<j\leq n,\ \ell=j+p-n,\\
\bar u_ju_{\ell+n-p}\Lambda_{j,\ell}^y(\xi(u)),&\text{if}\ 1\leq j\leq n,\ 1\leq\ell\leq p,\ \ell\neq j+p-n,\\
\bar u_ju_{\ell-p}\Lambda_{j,\ell}^y(\xi(u)),&\text{if}\ 1\leq j\leq n,\ p<\ell\leq n,\ \ell\neq j+p,
\end{cases}
\label{L^y+}
\end{align}
where $u$ varies in $S^{2n-1}_{\vert M\vert}$. Then, for $M<0$, define
$L_{j,\ell}^{y,-}\colon\CP^{n-1}\to\C$ by
\begin{equation}
L_{j,\ell}^{y,-}\circ \pi_{\vert M\vert} (u)=L_{j,\ell}^{y,+}\circ \pi_{\vert M\vert}(\bar u),
\label{L^y-}
\end{equation}
referring to the right-hand-side of \eqref{L^y+}
with the understanding that now $y>p\pi/n$.
\end{definition}
Next, we prove that the matrices $L_y^\loc$ and $L^{y,\pm}\circ\pi_{|M|}\circ\cE$,
are similar and can be
transformed into each other by a unitary matrix. This is one of our main results.

\begin{theorem}
\label{thm:5.5}
The smooth matrix function $L^{y,\pm}\colon\CP^{n-1}\to\C^{n\times n}$ with components
$L_{j,\ell}^{y,\pm}$ given by \eqref{L^y+},\eqref{L^y-} satisfies the following identity
\begin{equation}
(L^{y,\pm}\circ\pi_{|M|}\circ\cE)(\xi,e^{\ri\theta})
=\Delta(e^{\ri\theta})^{-1}L_y^\loc(\xi,e^{\ri\theta})\Delta(e^{\ri\theta}),\quad
\forall(\xi,e^{\ri\theta})\in\cA_y^+\times\T^{n-1},
\label{L^y-circ-cE}
\end{equation}
where $\Delta(e^{\ri\theta})=\diag(\Delta_1,\dots,\Delta_n)\in\UN(n)$ with
\begin{equation}
\Delta_j=\exp\bigg(\ri\sum_{k=1}^{n-1}\Omega_{j,k}\theta_k\bigg),
\quad j=1,\dots,n-1,\quad \Delta_n=1.
\label{Delta}
\end{equation}
Consequently, $L^{y,\pm}(\pi_{\vert M\vert}( u))\in\SU(n)$ for every
$u\in S^{2n-1}_{\vert M\vert} $, and $L^{y,\pm}$ provides an extension
of the local Lax matrix $L_y^\loc$ \eqref{L_y^loc} to the global phase
space $\CP^{n-1}$.
\end{theorem}

\begin{proof}
The form of the local Lax matrix $L_y^\loc$ \eqref{L_y^loc} and Definitions
\ref{def:5.2} and \ref{def:5.4} show that \eqref{L^y-circ-cE} is equivalent
to the equations
\begin{equation}
\Delta_j=\begin{cases}
\Delta_{j+p}\rho_{j+p},&\text{if}\quad 1\leq j\leq n-p,\\
\Delta_{j+p-n}\rho_{j+p-n},&\text{if}\quad n-p<j\leq n.
\end{cases}
\label{Delta-recursion}
\end{equation}
The two sides of \eqref{Delta-recursion} can be written as exponentials of linear
combinations of the variables $\theta_k$ $(1\leq k\leq n-1)$. We next spell out the
relations that ensure the exact matching of the coefficients of the $\theta_k$ in
these exponentials. Plugging the components of $\Delta$ and $\rho$ into
\eqref{Delta-recursion}, the case $1\leq j<n-p$ gives
\begin{equation}
\begin{split}
\Omega_{j,j+p-1}&=\Omega_{j+p,j+p-1}+1,\quad(\text{coefficients of}\ \theta_{j+p-1})\\
\Omega_{j,j+p}&=\Omega_{j+p,j+p}-1,\quad(\text{coefficients of}\ \theta_{j+p})\\
\Omega_{j,k}&=\Omega_{j+p,k},\quad(\text{coefficients of}\ \theta_k,\ k\neq j+p-1,j+p),
\end{split}
\label{Omega-1}
\end{equation}
while for $j=n-p$ we get
\begin{equation}
\begin{split}
\Omega_{n-p,n-1}&=1,\quad(\text{coefficients of}\ \theta_{n-1})\\
\Omega_{n-p,k}&=0,\quad(\text{coefficients of}\ \theta_k,\ k\neq n-1).
\end{split}
\label{Omega-2}
\end{equation}
The case $n-p<j<n$ (and $p>1$) leads to
\begin{equation}
\begin{split}
\Omega_{j,j+p-n-1}&=\Omega_{j+p-n,j+p-n-1}+1,\quad(\text{coefficients of}\ \theta_{j+p-n-1})\\
\Omega_{j,j+p-n}&=\Omega_{j+p-n,j+p-n}-1,\quad(\text{coefficients of}\ \theta_{j+p-n})\\
\Omega_{j,k}&=\Omega_{j+p-n,k},\quad(\text{coefficients of}\ \theta_k,\ k\neq j+p-n-1,j+p-n).
\end{split}
\label{Omega-3}
\end{equation}
For $j=n$ there are two possibilities. If $p=1$ then we obtain
\begin{equation}
\begin{split}
\Omega_{1,1}&=1,\quad(\text{coefficients of}\ \theta_1)\\
\Omega_{1,k}&=0,\quad(\text{coefficients of}\ \theta_k,\ k\neq 1),
\end{split}
\label{Omega-4}
\end{equation}
and if $p>1$ then we require
\begin{equation}
\begin{split}
\Omega_{p,p-1}&=-1,\quad(\text{coefficients of}\ \theta_{p-1})\\
\Omega_{p,p}&=1,\quad(\text{coefficients of}\ \theta_p)\\
\Omega_{p,k}&=0,\quad(\text{coefficients of}\ \theta_k,\ k\neq p-1,p).
\end{split}
\label{Omega-5}
\end{equation}
Using the explicit formula given by Proposition \ref{prop:5.1},
we now show that $\Omega$ satisfies \eqref{Omega-1}.
Since $\Omega_{j,k}=B_{j,k}-C_{j,k}$ for all $j,k$, where $B_{j,k}$ \eqref{B_m,k}
depends on $(k-j)$ and $C_{j,k}$ \eqref{C_m,k} depends only on $k$, the equations
\eqref{Omega-1} reduce to
\begin{equation}
\begin{split}
B_{j,j+p-1}&=B_{j+p,j+p-1}+1,\\
B_{j,j+p}&=B_{j+p,j+p}-1,\\
B_{j,k}&=B_{j+p,k},\quad k\neq j+p-1,j+p.
\end{split}
\label{B-1}
\end{equation}
The first equation holds, because $(j+p-1)-j=p-1\equiv (n-q+1)p\pmod{n}$
implies $B_{j,j+p-1}=1$ and $(j+p-1)-(j+p)=-1\equiv (n-q)p\pmod{n}$ implies
$B_{j+p,j+p-1}=0$. For the second equation, we plainly have
$B_{j,j+p}=0$, and $(j+p)-(j+p)=0\equiv np\pmod{n}$ gives $B_{j+p,j+p}=1$.
Regarding the third equation, notice that $B_{j,k}=0$ in \eqref{B-1} when
$k-j\equiv\ell p\pmod{n}$ for some $\ell\in\{2,\dots,n-q\}$, and then
$B_{j+p,k}=0$ holds, too. Conversely, $B_{j+p,k}=0$ in \eqref{B-1} means
that $(k-j)-p\equiv\ell p\pmod{n}$ for some $\ell\in\{1,\dots,n-q-1\}$,
from which $(k-j)\equiv(\ell+1)p\pmod{n}$ and thus $B_{j,k}=0$ follows.
As $B$ is a $(0,1)$-matrix, we conclude that \eqref{B-1} is valid.
Proceeding in a similar manner, we have verified the rest of the relations
\eqref{Omega-2}-\eqref{Omega-5} as well. Since the relations
\eqref{Omega-1}-\eqref{Omega-5} imply \eqref{Delta-recursion},
the proof is complete.
\end{proof}

It is an immediate consequence of Theorem \ref{thm:5.5} that the spectral invariants of the global Lax matrix
$L^{y, \pm} \in C^\infty(\CP^{n-1}, \SU(n))$ yield a Liouville integrable system.
Because of \eqref{Re-tr-L_y^loc} the corresponding Poisson commuting family contains
the extension of the III$_\b$ Hamiltonian $H_y^\loc$ to $\CP^{n-1}$ for any type (i)
coupling. The self-duality of this compactified RS system was established in \cite{FKl14},
and it will be studied in more detail elsewhere.

\section{New compact forms of the elliptic Ruijsenaars-Schneider system}
\label{sec:5.3}

In this section we explain that type (i) compactifications of the elliptic RS system
can be constructed in exactly the same way as we saw for the trigonometric system.
This is due to the fact that the local elliptic Lax matrix is built from
the $\ws$-function \eqref{sigma} similarly as its trigonometric counterpart is built
from the sine function, and on the real axis these two functions have the same
zeros, signs, parity and antiperiodicity property.

We start by recalling some formulas of the relevant elliptic functions.
First, let $\omega,\omega'$ stand for the half-periods of the Weierstrass $\wp$
function defined by
\begin{equation}
\wp(z;\omega,\omega')
=\frac{1}{z^2}+\sum_{\substack{m,m'=-\infty\\(m,m')\neq(0,0)}}^\infty
\bigg[\frac{1}{(z-\omega_{m,m'})^2}-\frac{1}{\omega_{m,m'}^2}\bigg],
\label{wp}
\end{equation}
with $\omega_{m,m'}=2m\omega+2m'\omega'$. We adopt the convention
$\omega,-\ri\omega'\in(0,\infty)$, which ensures that $\wp$ is positive
on the real axis. Next, introduce the following `$\ws$-function':
\begin{equation}
\ws(z;\omega,\omega')
=\frac{2\omega}{\pi}\sin\!\Bigl(\frac{\pi z}{2\omega}\Bigr)
\prod_{m=1}^\infty\bigg[1+\frac{\sin^2(\pi z/(2\omega))}{\sinh^2(m\pi|\omega'| /\omega)}\bigg],
\label{sigma}
\end{equation}
related to the Weierstrass $\sigma$ and $\zeta$ functions by
$\ws(z)=\sigma(z)\exp(-\eta z^2/(2\omega))$ with the constant $\eta=\zeta(\omega)$.
A useful identity connecting $\wp$ and $\ws$ is
\begin{equation}
\frac{\ws(z+z')\ws(z-z')}{\ws^2(z)\ws^2(z')}=\wp(z')-\wp(z),
\qquad z, z'\in \C.
\label{wp-sigma}
\end{equation}
The $\ws$-function is odd, has simple zeros at $\omega_{m,m'}$ $(m,m'\in\Z)$ and
enjoys the scaling property $\ws(t z;t\omega,t\omega')=t\ws(z;\omega,\omega')$.
From now on we take
\begin{equation}
\omega=\frac{\pi}{2},
\end{equation}
whereby $\ws(z+\pi)=-\ws(z)$ holds as well. The trigonometric limit is obtained according to
\begin{equation}
\lim_{-\ri\omega'\to\infty}\wp(z;\pi/2,\omega')=\frac{1}{\sin^2(z)}-\frac{1}{3},\quad
\lim_{-\ri\omega'\to\infty}\ws(z;\pi/2,\omega')=\sin(z).
\label{trig-limit}
\end{equation}

Let us now pick a type (i) coupling parameter $y$ \eqref{typeI-y} and choose
the domain of the dynamical variables to be the same $\cA_y^+\times\T^{n-1}$
as in the trigonometric case. Then consider the following IV$_\b$ variant of
the standard \cite{Ru87,Ru99} elliptic RS Lax matrix:
\begin{equation}
L_y^\loc(\xi,e^{\ri\theta}\vert\lambda)_{j,\ell}
=\frac{\ws(y)}{\ws(\lambda)}
\frac{\ws(x_j-x_\ell+\lambda)}{\ws(x_j-x_\ell+y)}
v_j(\xi,y)v_\ell(\xi,-y)\rho(\theta)_\ell,
\,\,\, \forall(\xi,e^{\ri\theta})\in\cA_y^+\times\T^{n-1},
\label{L_y^loc-IV}
\end{equation}
where $\lambda\in\C\setminus\{\omega_{m,m'}:m,m'\in\Z\}$ is a spectral parameter
and $v_\ell(\xi,\pm y)=\sqrt{z_\ell(\xi,\pm y)}$ with
\begin{equation}
z_\ell(\xi,\pm y)=\sgn(\ws(ny))\prod_{m=\ell+1}^{\ell+n-1}
\frac{\ws(\sum_{k=\ell}^{m-1}\xi_k\mp y)}{\ws(\sum_{k=\ell}^{m-1}\xi_k)}.
\label{z_ell(xi,pmy)-IV}
\end{equation}
These formulas are to be compared with the trigonometric case. Since $\ws(z)$ and
$\sin(z)$ have matching properties on the real line, we can repeat the arguments
presented in Section \ref{sec:5.2} to verify that $z_\ell(\xi,\pm y)>0$ for every
$\xi\in\cA_y^+$. Taking positive square roots, and applying the relation $x_{k+1}-x_k=\xi_k$
to express $x_j-x_\ell$ in terms of $\xi$, we conclude that the above local Lax matrix
is a smooth function on $\cA_y^+\times\T^{n-1}$ for every allowed value of the spectral
parameter. The fact that it is a specialization of the standard elliptic Lax matrix
ensures \cite{Ru87,Ru99} that its characteristic polynomial generates $(n-1)$
independent real Hamiltonians in involution with respect to the symplectic form \eqref{om-loc}.
Indeed, the characteristic polynomial has the form
\begin{equation}
\det\big(L_y^\loc(\xi,e^{\ri\theta}\vert\lambda)-\alpha\1_n\big)
=\sum_{k=0}^n(-\alpha)^{n-k}c_k(\lambda,y)\cS_k^\loc(\xi,e^{\ri\theta},y),
\label{char1}\end{equation}
where the functions $\cS_k^\loc$ as well as their real and imaginary parts Poisson
commute, and $\Re(\cS_{k}^\loc)$ for $k=1,\dots, n-1$ are functionally independent.
Explicit formulas of the $c_k$ (that do not depend on the phase space variables)
and $\cS_k^\loc$ (that do not depend on $\lambda$) can be found in \cite{Ru87,Ru99}.
The function $\Re(\cS^\loc_1)$ is the RS Hamiltonian of IV$_\b$ type
\begin{equation}
\Re\big(\tr L_y^\loc(\xi,e^{\ri\theta}\vert\lambda)\big)
=\sum_{j=1}^n\cos(\theta_j-\theta_{j-1})\sqrt{ \prod_{m=j+1}^{j+n-1}
\left[\ws(y)^2(\wp(y)-\wp(\textstyle\sum_{k=j}^{m-1}\xi_k))\right]}.
\label{H_y^loc-IV}
\end{equation}

We note in passing that in Ruijsenaars's papers \cite{Ru87,Ru99} one finds
the elliptic Lax matrix $VL_y^\loc V^{-1}$, where $V$ is the diagonal matrix
$V=\rho(\theta)\diag(v_1(\xi,-y),\dots,v_n(\xi,-y))$. This difference is
irrelevant, since it has no effect on the generated spectral invariants.
Another difference is that we work in the center-of-mass frame.

Now the complete train of thought applied in the previous section
remains valid if we simply replace the sine function with the $\ws$-function everywhere.
In particular, the direct analogues of the formulas \eqref{v_ell(xi,y)}-\eqref{w_ell(xi,-y)-2}
hold with smooth functions $w_k(\xi,\pm y)>0$, for $\xi\in\cA_y$. Due to this fact,
we can introduce a smooth elliptic Lax matrix defined on the global phase space $\CP^{n-1}$.
The subsequent definition refers to the explicit formulas of Appendix \ref{sec:E.1},
which in the elliptic case contain the function
\begin{equation}
\cJ(|u_k|^2)=\frac{\ws(|u_k|^2)}{|u_k|^2}.
\label{J2}\end{equation}
This has the same smoothness and positivity properties at and around zero as $J$ \eqref{J1} does.
We also use $\xi(u)$ \eqref{xi(u)} and the functions $(x_j-x_\ell)(\xi)$ determined by $x_{k+1}-x_k=\xi_k$.

\begin{definition}
\label{def:5.6}
Take a type (i) $y$ from \eqref{typeI-y} and represent the points of $\CP^{n-1}$ as
$\pi_{|M|}(u)$ with $u\in S^{2n-1}_{|M|}$. For $M>0$ \eqref{M}, define the smooth
functions $\cL_{j,\ell}^{y,+}$ on $\CP^{n-1}$ by
\begin{align}
\cL_{j,\ell}^{y,+}(\pi_{\vert M\vert}(u))&=\begin{cases}
\Lambda_{j,j+p}^y(\xi(u)),&\text{if}\ 1\leq j\leq n-p,\ \ell=j+p,\\
\Lambda_{j,j+p-n}^y(\xi(u)),&\text{if}\ n-p<j\leq n,\ \ell=j+p-n,\\
\bar u_ju_{\ell+n-p}\Lambda_{j,\ell}^y(\xi(u)),&\text{if}\ 1\leq j\leq n,\ 1\leq\ell\leq p,\ \ell\neq j+p-n,\\
\bar u_ju_{\ell-p}\Lambda_{j,\ell}^y(\xi(u)),&\text{if}\ 1\leq j\leq n,\ p<\ell\leq n,\ \ell\neq j+p,
\end{cases}
\label{L^y+IV}
\end{align}
with $\Lambda_{j,\ell}^y$ given in Appendix \ref{sec:E.1}. For $M<0$, set $\cL_{j,\ell}^{y,-}$
to be
\begin{equation}
\cL_{j,\ell}^{y,-}(\pi_{\vert M\vert}(u))=\cL_{j,\ell}^{y,+}(\pi_{\vert M\vert}(\bar u))
\label{L^y-IV}
\end{equation}
with the understanding that in this case $y>p\pi/n$. Finally,
define the $\lambda$-dependent elliptic Lax matrix $L^{y,\pm}$ on $\CP^{n-1}$ by
\begin{equation}
L_{j,\ell}^{y,\pm}(\pi_{\vert M\vert}(u)\vert\lambda)=
\frac{\ws((x_j-x_\ell)(\xi(u))+\lambda)}{\ws(\lambda)}
\cL_{j,\ell}^{y,\pm}(\pi_{\vert M\vert}(u)),
\label{L^y}
\end{equation}
where $u$ runs over $S^{2n-1}_{\vert M\vert}$ and the spectral parameter $\lambda$
varies in $\C\setminus\{\omega_{m,m'}:m,m'\in\Z\}$.
\end{definition}

\begin{theorem}
\label{thm:5.7}
The spectral parameter dependent elliptic Lax matrix $L^{y,\pm}(\pi_{\vert M\vert}(u)\vert\lambda)$
\eqref{L^y} is a smooth global extension of $L_y^\loc(\xi,e^{\ri \theta}\vert\lambda)$
\eqref{L_y^loc-IV} to the complex projective space $\CP^{n-1}$ since it satisfies
\begin{equation}
L^{y,\pm}((\pi_{|M|}\circ \cE)(\xi,e^{\ri\theta})\vert\lambda)
=\Delta(e^{\ri\theta})^{-1}L_y^\loc(\xi,e^{\ri\theta}\vert\lambda)
\Delta(e^{\ri\theta}),\quad
\forall(\xi,e^{\ri\theta})\in\cA_y^+\times\T^{n-1},
\label{L^y-circ-pi_M-circ-cE}
\end{equation}
where $\Delta$ is given by \eqref{Delta} and
$\pi_{\vert M\vert}\circ\cE\colon\cA_y^+\times\T^{n-1}\to \CP^{n-1}$
is the symplectic embedding defined in Section \ref{sec:5.1}.
\end{theorem}

The proof of Theorem \ref{thm:5.7} follows the lines of the proof of Theorem \ref{thm:5.5}.
The characteristic polynomial $\det\big(L^{y,\pm}(\pi_{\vert M\vert}(u)\vert\lambda)-\alpha\1_n\big)$
of the global Lax matrix depends smoothly on $\pi_{\vert M\vert}(u)\in \CP^{n-1}$ and as a
consequence of \eqref{L^y-circ-pi_M-circ-cE} it satisfies
\begin{equation}
\det\big(L^{y,\pm}((\pi_{|M|}\circ\cE)(\xi,e^{\ri\theta})\vert\lambda)
-\alpha \1_n\big)=
\det\big(L_y^\loc(\xi,e^{\ri\theta}\vert\lambda)-\alpha\1_n\big).
\end{equation}
Since this holds for all $\alpha$ and $\lambda$, we see that the local IV$_\b$
Hamiltonian \eqref{H_y^loc-IV} together with its constants of motion $\Re(\cS_k^\loc)$,
$k=2,\dots,n-1$ extends to an integrable system on $\CP^{n-1}$. This was pointed out
previously \cite{Ru99} for the special case $0<y<\pi/n$ in \eqref{typeI-y}.

In the trigonometric limit $-\ri\omega'\to\infty$ the $\ws$-function becomes the sine function,
and we obtain a spectral parameter dependent trigonometric Lax matrix from the elliptic one.
Then, setting the spectral parameter to be on the imaginary axis and taking the limit
$-\ri\lambda\to\infty$ reproduces, up to conjugation by a diagonal matrix, the trigonometric
global Lax matrix of Definition \ref{def:5.4}. Correspondingly, the global extension of the
IV$_\b$ Hamiltonian \eqref{H_y^loc-IV} and its commuting family reduces to the global extension
of the III$_\b$ Hamiltonian \eqref{H_y^loc} and its constants of motion.

\section{Discussion}
\label{sec:5.4}

In this chapter we have demonstrated by direct construction that the local phase space
$\cA_y^+\times\T^{n-1}$ of the III$_\b$ and IV$_\b$ RS models (where $\cA_y^+$ is the
interior of the simplex \eqref{5.11}) can be embedded into $\CP^{n-1}$ for any type (i)
coupling $y$ \eqref{typeI-y} in such a way that a suitable conjugate of the local Lax
matrix extends to a smooth (actually real-analytic) function. Theorems \ref{thm:5.5}
and \ref{thm:5.7} together with Appendix \ref{sec:E.1} provide explicit formulas for
the resulting global Lax matrices. Their characteristic polynomials give rise to Poisson
commuting real Hamiltonians on $\CP^{n-1}$ that yield the Liouville integrable compactified
trigonometric and elliptic RS systems.

Our direct construction was inspired by the earlier derivation of compactified
III$_\b$ systems by quasi-Hamiltonian reduction \cite{FKl14}. The reduction identifies
the III$_\b$ system with a topological Chern-Simons field theory for any generic
coupling parameter $y$. It appears natural to ask if an analogous derivation and
relation to some topological field theory could exist for IV$_\b$ systems, too.
We also would like to obtain a better understanding of the type (ii) trigonometric
systems and their possible elliptic analogues.

Besides further studying the systems that we described, it would be also interesting
to search for compactifications of generalized RS systems. We have in mind especially
the $\BC_n$ systems due to van Diejen \cite{vD94} and the recently introduced
supersymmetric systems \cite{BDM15}. Regarding the former case, and even for general
root systems, the results of \cite{vDE14} could be relevant, as well as the construction
of Lax matrices for some of the $\BC_n$ systems reported in Chapter \ref{chap:4}.

Throughout the text, we worked in the `center-of-mass frame' and now we end by a
comment on how the center-of-mass coordinate can be introduced into our systems.
One possibility is to take the full phase space to be the Cartesian product
of $\CP^{n-1}$ with $\UN(1)\times\UN(1)=\{(e^{2\ri X},e^{\ri\Phi})\}$
endowed with the symplectic form $|M|\,\omega_{\mathrm{FS}}+dX\wedge d\Phi$.
Here, $e^{2\ri X}$ is interpreted as a center-of-mass variable for the $n$
particles on the circle. Then $n$ functions in involution result by adding an
arbitrary function of $e^{\ri\Phi}$ to the $(n-1)$ commuting Hamiltonians
generated by the `total Lax matrix' $e^{-\ri \Phi}L^{y,\pm}$. On the dense open
domain the total Lax matrix is obtained by replacing $\rho(\theta)$ in \eqref{L_y^loc-IV}
by $\rho(\theta)e^{-\ri\Phi}$. By setting $e^{\ri\Phi}$ to $1$ and quotienting
by the canonical transformations generated by the functions of $e^{\ri\Phi}$ one
recovers the phase space of the relative motion, $\CP^{n-1}$. There are also several
other possibilities, as was discussed for analogous situations in \cite{Ru95,FA10}.
For example, one may replace $\UN(1)\times\UN(1)$ by its covering space $\R\times\R$.

In the near future, we wish to explore the classical dynamics and quantization of
the III$_\b$ systems. For arbitrary type (i) couplings, geometric quantization yields
the joint spectra of the quantized action variables effortlessly \cite{FK12-2}.
(It is necessary to introduce a second parameter into the systems before quantization,
which can be achieved by taking an arbitrary multiple of the symplectic form.)
The joint eigenfunctions of the quantized Ruijsenaars-Schneider Hamiltonian and
its commuting family should be derived by generalizing the results of van Diejen
and Vinet \cite{vDV98}.

\appendix

\cleardoublepage
\phantomsection
\addcontentsline{toc}{part}{Appendices}
\refstepcounter{appe}
\part*{Appendices}
\label{part:app}

\chapter{Appendix to Chapter \ref{chap:1}}
\label{chap:A}

\section{An alternative proof of Theorem \ref{thm:1.2}}
\label{sec:A.1}

In this appendix, we give an alternative proof for Theorem \ref{thm:1.2}, which is based on the scattering behaviour of particles in the rational Calogero-Moser model (see Figure \ref{fig:8}).

Recall the Lax pair found by Moser \cite{Mo75}
\begin{equation}
L_{jk}=p_j\delta_{jk}+\ri g\frac{1-\delta_{jk}}{q_j-q_k},\quad
B_{jk}=\ri g\delta_{jk}\sum_{\substack{l=1\\(l\neq k)}}^n
\frac{1}{(q_k-q_l)^2}-\ri g\frac{1-\delta_{jk}}{(q_j-q_k)^2},
\label{A.1}
\end{equation}
and consider $Q=\diag(q_1,\dots,q_n)$. These matrices satisfy the commutation relation
\begin{equation}
[z\1_n-L,Q]=\ri g(vv^\dag-\1_n),
\label{A.2}
\end{equation}
for any scalar $z\in\C$, where $v=(1\dots 1)^\dag$ and $\1_n$ stands for the $n\times n$ identity matrix. They also enjoy the following relations along solutions
\begin{equation}
\dot L=[L,B],\quad 
\dot Q=[Q,B]+L.
\label{A.3}
\end{equation}
The asymptotic form of solutions of the Calogero-Moser system is
\begin{equation}
q_k(t)\sim p_k^\pm t+q_k^\pm,\quad
p_k(t)\sim p_k^\pm,\quad t\to\pm\infty.
\label{A.4}
\end{equation}
Theorem \ref{thm:1.2} connects the following functions
\begin{equation}
C(z)=\tr(Q(z\1_n-L)^\vee vv^\dag),\quad
D(z)=\tr(Q(z\1_n-L)^\vee),
\label{A.5}
\end{equation}
where $M^\vee$ denotes the adjugate of $M$, i.e. the transpose of its cofactor matrix.

\begin{theorem}
\label{thm:A.1}
For any $n\in\N$, $p,q\in\R^n$ with $q_j\neq q_k$ $(j\neq k)$, and $z\in\C$ we have
\begin{equation}
C(z)=D(z)+\frac{\ri g}{2}\frac{d^2}{dz^2}\det(z\1_n-L).
\label{A.6}
\end{equation}
\end{theorem}

\begin{proof}
Pick any point in the phase space and consider the solution passing through it. The Lax matrix $L$ is isospectral along the solutions, thus its characteristic polynomial is constant, viz. $\frac{d}{dt}\det(z\1_n-L)=0$. Consequently, its second derivative w.r.t. $z$ is also constant, that is
\begin{equation}
\frac{d}{dt}\bigg(\frac{d^2}{dz^2}\det(z\1_n-L)\bigg)=0.
\label{A.7}
\end{equation}
The difference of the functions $C(z)$ and $D(z)$ reads
\begin{equation}
C(z)-D(z)=\tr\big(Q(z\1_n-L)^\vee(vv^\dag-\1_n)\big).
\label{A.8}
\end{equation}
By utilizing the commutation relation \eqref{A.2} of $L$ and $Q$, the above casts into
\begin{equation}
\ri g(C(z)-D(z))
=\det(z\1_n-L)\tr(Q^2)-\tr\big(Q(z\1_n-L)^\vee Q(z\1_n-L)\big),
\label{A.9}
\end{equation}
where we made use of the matrix identity $MM^\vee=\det(M)$. By applying \eqref{A.3} and
\begin{equation}
\frac{d}{dt}(z\1_n-L)^\vee=[(z\1_n-L)^\vee,B],
\label{A.10}
\end{equation}
as well as, the Leibniz rule and the cyclic property of the trace one finds that
\begin{equation}
\frac{d}{dt}\big(C(z)-D(z)\big)=0.
\label{A.11}
\end{equation}
Putting \eqref{A.7} and \eqref{A.11} together shows that $C(z)-D(z)+\frac{\ri g}{2}\frac{d^2}{dz^2}\det(z\1_n-L)$ is constant. However, due to \eqref{A.4}, in the asymptotic limit a closer inspection of \eqref{A.8} shows that
\begin{equation}
\lim_{t\to\infty}\bigg(C(z)-D(z)+\frac{\ri g}{2}\frac{d^2}{dz^2}\det(z\1_n-L)\bigg)=0.
\label{A.12}
\end{equation}
This concludes the proof.
\end{proof}

\begin{figure}[h!]
\centering
\begin{tikzpicture}
\begin{axis}[width=.88\textwidth,height=.48\textwidth,axis lines=center,
xlabel=$q$,ylabel={$t$},xlabel style={below},ylabel style={left},
legend cell align={left},xmin=-5.5,xmax=5.5,ymin=-3.,ymax=3.,xtick={-5,-4,-3,-2,-1,0,1,2,3,4,5},ytick={-2,-1,0,1,2},
after end axis/.code={\path (axis cs:-.25,-.32) node [] {0}
(axis cs:4.64782,3.) node [right] {$q_1(t)$}
(axis cs:1.,3.) node [right] {$q_2(t)$}
(axis cs:-5.2,3.) node [right] {$q_3(t)$};}
]

\addplot[domain=-5:5.,color=midnightblue,thick]
coordinates{
(5.88655,-3.)(5.86761,-2.99)(5.84867,-2.98)(5.82973,-2.97)(5.81079,-2.96)(5.79186,-2.95)(5.77292,-2.94)(5.75399,-2.93)(5.73505,-2.92)(5.71612,-2.91)(5.69719,-2.9)(5.67826,-2.89)(5.65933,-2.88)(5.6404,-2.87)(5.62148,-2.86)(5.60255,-2.85)(5.58362,-2.84)(5.5647,-2.83)(5.54578,-2.82)(5.52686,-2.81)(5.50794,-2.8)(5.48902,-2.79)(5.4701,-2.78)(5.45119,-2.77)(5.43227,-2.76)(5.41336,-2.75)(5.39445,-2.74)(5.37553,-2.73)(5.35662,-2.72)(5.33772,-2.71)(5.31881,-2.7)(5.2999,-2.69)(5.281,-2.68)(5.2621,-2.67)(5.2432,-2.66)(5.2243,-2.65)(5.2054,-2.64)(5.1865,-2.63)(5.16761,-2.62)(5.14871,-2.61)(5.12982,-2.6)(5.11093,-2.59)(5.09204,-2.58)(5.07315,-2.57)(5.05427,-2.56)(5.03538,-2.55)(5.0165,-2.54)(4.99762,-2.53)(4.97874,-2.52)(4.95986,-2.51)(4.94099,-2.5)(4.92211,-2.49)(4.90324,-2.48)(4.88437,-2.47)(4.8655,-2.46)(4.84663,-2.45)(4.82777,-2.44)(4.8089,-2.43)(4.79004,-2.42)(4.77118,-2.41)(4.75233,-2.4)(4.73347,-2.39)(4.71462,-2.38)(4.69577,-2.37)(4.67692,-2.36)(4.65807,-2.35)(4.63923,-2.34)(4.62038,-2.33)(4.60154,-2.32)(4.5827,-2.31)(4.56387,-2.3)(4.54503,-2.29)(4.5262,-2.28)(4.50737,-2.27)(4.48855,-2.26)(4.46972,-2.25)(4.4509,-2.24)(4.43208,-2.23)(4.41326,-2.22)(4.39445,-2.21)(4.37564,-2.2)(4.35683,-2.19)(4.33802,-2.18)(4.31922,-2.17)(4.30042,-2.16)(4.28162,-2.15)(4.26282,-2.14)(4.24403,-2.13)(4.22524,-2.12)(4.20645,-2.11)(4.18767,-2.1)(4.16889,-2.09)(4.15011,-2.08)(4.13133,-2.07)(4.11256,-2.06)(4.09379,-2.05)(4.07503,-2.04)(4.05627,-2.03)(4.03751,-2.02)(4.01875,-2.01)(4.,-2.)(3.98125,-1.99)(3.96251,-1.98)(3.94377,-1.97)(3.92503,-1.96)(3.9063,-1.95)(3.88757,-1.94)(3.86884,-1.93)(3.85012,-1.92)(3.8314,-1.91)(3.81269,-1.9)(3.79398,-1.89)(3.77527,-1.88)(3.75657,-1.87)(3.73787,-1.86)(3.71918,-1.85)(3.70049,-1.84)(3.68181,-1.83)(3.66313,-1.82)(3.64445,-1.81)(3.62578,-1.8)(3.60712,-1.79)(3.58846,-1.78)(3.56981,-1.77)(3.55116,-1.76)(3.53251,-1.75)(3.51387,-1.74)(3.49524,-1.73)(3.47661,-1.72)(3.45799,-1.71)(3.43937,-1.7)(3.42076,-1.69)(3.40216,-1.68)(3.38356,-1.67)(3.36497,-1.66)(3.34638,-1.65)(3.3278,-1.64)(3.30923,-1.63)(3.29066,-1.62)(3.2721,-1.61)(3.25355,-1.6)(3.23501,-1.59)(3.21647,-1.58)(3.19794,-1.57)(3.17941,-1.56)(3.1609,-1.55)(3.14239,-1.54)(3.12389,-1.53)(3.1054,-1.52)(3.08691,-1.51)(3.06844,-1.5)(3.04997,-1.49)(3.03151,-1.48)(3.01306,-1.47)(2.99463,-1.46)(2.9762,-1.45)(2.95777,-1.44)(2.93936,-1.43)(2.92096,-1.42)(2.90257,-1.41)(2.88419,-1.4)(2.86582,-1.39)(2.84747,-1.38)(2.82912,-1.37)(2.81078,-1.36)(2.79246,-1.35)(2.77415,-1.34)(2.75585,-1.33)(2.73756,-1.32)(2.71929,-1.31)(2.70103,-1.3)(2.68278,-1.29)(2.66455,-1.28)(2.64633,-1.27)(2.62813,-1.26)(2.60994,-1.25)(2.59176,-1.24)(2.5736,-1.23)(2.55546,-1.22)(2.53734,-1.21)(2.51923,-1.2)(2.50114,-1.19)(2.48306,-1.18)(2.46501,-1.17)(2.44697,-1.16)(2.42895,-1.15)(2.41095,-1.14)(2.39298,-1.13)(2.37502,-1.12)(2.35708,-1.11)(2.33917,-1.1)(2.32128,-1.09)(2.30341,-1.08)(2.28557,-1.07)(2.26775,-1.06)(2.24996,-1.05)(2.23219,-1.04)(2.21445,-1.03)(2.19674,-1.02)(2.17906,-1.01)(2.16141,-1.)(2.14378,-0.99)(2.12619,-0.98)(2.10863,-0.97)(2.0911,-0.96)(2.07361,-0.95)(2.05615,-0.94)(2.03873,-0.93)(2.02135,-0.92)(2.00401,-0.91)(1.9867,-0.9)(1.96944,-0.89)(1.95222,-0.88)(1.93504,-0.87)(1.91792,-0.86)(1.90083,-0.85)(1.8838,-0.84)(1.86682,-0.83)(1.84989,-0.82)(1.83301,-0.81)(1.8162,-0.8)(1.79943,-0.79)(1.78273,-0.78)(1.7661,-0.77)(1.74952,-0.76)(1.73302,-0.75)(1.71658,-0.74)(1.70022,-0.73)(1.68393,-0.72)(1.66772,-0.71)(1.65159,-0.7)(1.63554,-0.69)(1.61958,-0.68)(1.60371,-0.67)(1.58793,-0.66)(1.57226,-0.65)(1.55668,-0.64)(1.5412,-0.63)(1.52584,-0.62)(1.51059,-0.61)(1.49545,-0.6)(1.48044,-0.59)(1.46556,-0.58)(1.4508,-0.57)(1.43619,-0.56)(1.42171,-0.55)(1.40738,-0.54)(1.39321,-0.53)(1.37919,-0.52)(1.36534,-0.51)(1.35166,-0.5)(1.33816,-0.49)(1.32484,-0.48)(1.3117,-0.47)(1.29877,-0.46)(1.28604,-0.45)(1.27352,-0.44)(1.26121,-0.43)(1.24913,-0.42)(1.23728,-0.41)(1.22567,-0.4)(1.2143,-0.39)(1.20319,-0.38)(1.19233,-0.37)(1.18174,-0.36)(1.17142,-0.35)(1.16137,-0.34)(1.15161,-0.33)(1.14214,-0.32)(1.13296,-0.31)(1.12408,-0.3)(1.11551,-0.29)(1.10724,-0.28)(1.09928,-0.27)(1.09164,-0.26)(1.08431,-0.25)(1.0773,-0.24)(1.07061,-0.23)(1.06424,-0.22)(1.05819,-0.21)(1.05246,-0.2)(1.04704,-0.19)(1.04194,-0.18)(1.03716,-0.17)(1.03269,-0.16)(1.02853,-0.15)(1.02467,-0.14)(1.02112,-0.13)(1.01786,-0.12)(1.01489,-0.11)(1.01221,-0.1)(1.00981,-0.09)(1.00769,-0.08)(1.00584,-0.07)(1.00426,-0.06)(1.00293,-0.05)(1.00186,-0.04)(1.00104,-0.03)(1.00046,-0.02)(1.00011,-0.01)(1.,0.)(1.00011,0.01)(1.00044,0.02)(1.00099,0.03)(1.00174,0.04)(1.0027,0.05)(1.00385,0.06)(1.0052,0.07)(1.00674,0.08)(1.00846,0.09)(1.01036,0.1)(1.01244,0.11)(1.01469,0.12)(1.01711,0.13)(1.01968,0.14)(1.02242,0.15)(1.02532,0.16)(1.02837,0.17)(1.03157,0.18)(1.03492,0.19)(1.03841,0.2)(1.04205,0.21)(1.04582,0.22)(1.04973,0.23)(1.05378,0.24)(1.05796,0.25)(1.06227,0.26)(1.06671,0.27)(1.07127,0.28)(1.07596,0.29)(1.08077,0.3)(1.08571,0.31)(1.09076,0.32)(1.09594,0.33)(1.10123,0.34)(1.10664,0.35)(1.11216,0.36)(1.1178,0.37)(1.12355,0.38)(1.12941,0.39)(1.13539,0.4)(1.14147,0.41)(1.14767,0.42)(1.15397,0.43)(1.16038,0.44)(1.1669,0.45)(1.17353,0.46)(1.18026,0.47)(1.18709,0.48)(1.19403,0.49)(1.20108,0.5)(1.20823,0.51)(1.21548,0.52)(1.22283,0.53)(1.23029,0.54)(1.23785,0.55)(1.2455,0.56)(1.25326,0.57)(1.26112,0.58)(1.26908,0.59)(1.27714,0.6)(1.28529,0.61)(1.29354,0.62)(1.30189,0.63)(1.31034,0.64)(1.31888,0.65)(1.32752,0.66)(1.33626,0.67)(1.34508,0.68)(1.35401,0.69)(1.36302,0.7)(1.37213,0.71)(1.38134,0.72)(1.39063,0.73)(1.40002,0.74)(1.40949,0.75)(1.41906,0.76)(1.42872,0.77)(1.43846,0.78)(1.44829,0.79)(1.45821,0.8)(1.46822,0.81)(1.47832,0.82)(1.48849,0.83)(1.49876,0.84)(1.50911,0.85)(1.51954,0.86)(1.53005,0.87)(1.54065,0.88)(1.55133,0.89)(1.56208,0.9)(1.57292,0.91)(1.58384,0.92)(1.59483,0.93)(1.6059,0.94)(1.61705,0.95)(1.62827,0.96)(1.63957,0.97)(1.65094,0.98)(1.66238,0.99)(1.6739,1.)(1.68549,1.01)(1.69715,1.02)(1.70888,1.03)(1.72068,1.04)(1.73254,1.05)(1.74447,1.06)(1.75647,1.07)(1.76854,1.08)(1.78067,1.09)(1.79286,1.1)(1.80512,1.11)(1.81744,1.12)(1.82982,1.13)(1.84226,1.14)(1.85476,1.15)(1.86732,1.16)(1.87994,1.17)(1.89261,1.18)(1.90534,1.19)(1.91813,1.2)(1.93097,1.21)(1.94387,1.22)(1.95682,1.23)(1.96982,1.24)(1.98287,1.25)(1.99597,1.26)(2.00913,1.27)(2.02233,1.28)(2.03559,1.29)(2.04889,1.3)(2.06223,1.31)(2.07563,1.32)(2.08907,1.33)(2.10255,1.34)(2.11608,1.35)(2.12966,1.36)(2.14327,1.37)(2.15693,1.38)(2.17063,1.39)(2.18437,1.4)(2.19816,1.41)(2.21198,1.42)(2.22584,1.43)(2.23974,1.44)(2.25367,1.45)(2.26765,1.46)(2.28166,1.47)(2.29571,1.48)(2.30979,1.49)(2.32391,1.5)(2.33806,1.51)(2.35225,1.52)(2.36647,1.53)(2.38072,1.54)(2.395,1.55)(2.40932,1.56)(2.42367,1.57)(2.43804,1.58)(2.45245,1.59)(2.46689,1.6)(2.48136,1.61)(2.49586,1.62)(2.51038,1.63)(2.52493,1.64)(2.53951,1.65)(2.55412,1.66)(2.56876,1.67)(2.58342,1.68)(2.5981,1.69)(2.61281,1.7)(2.62755,1.71)(2.64231,1.72)(2.65709,1.73)(2.6719,1.74)(2.68674,1.75)(2.70159,1.76)(2.71647,1.77)(2.73137,1.78)(2.74629,1.79)(2.76124,1.8)(2.7762,1.81)(2.79119,1.82)(2.80619,1.83)(2.82122,1.84)(2.83627,1.85)(2.85133,1.86)(2.86642,1.87)(2.88153,1.88)(2.89665,1.89)(2.91179,1.9)(2.92695,1.91)(2.94213,1.92)(2.95733,1.93)(2.97254,1.94)(2.98777,1.95)(3.00302,1.96)(3.01829,1.97)(3.03357,1.98)(3.04886,1.99)(3.06418,2.)(3.07951,2.01)(3.09485,2.02)(3.11021,2.03)(3.12558,2.04)(3.14097,2.05)(3.15638,2.06)(3.17179,2.07)(3.18723,2.08)(3.20267,2.09)(3.21813,2.1)(3.23361,2.11)(3.24909,2.12)(3.26459,2.13)(3.2801,2.14)(3.29563,2.15)(3.31117,2.16)(3.32672,2.17)(3.34228,2.18)(3.35785,2.19)(3.37344,2.2)(3.38904,2.21)(3.40465,2.22)(3.42027,2.23)(3.4359,2.24)(3.45154,2.25)(3.4672,2.26)(3.48286,2.27)(3.49854,2.28)(3.51422,2.29)(3.52992,2.3)(3.54563,2.31)(3.56134,2.32)(3.57707,2.33)(3.59281,2.34)(3.60855,2.35)(3.62431,2.36)(3.64007,2.37)(3.65585,2.38)(3.67163,2.39)(3.68742,2.4)(3.70322,2.41)(3.71903,2.42)(3.73485,2.43)(3.75068,2.44)(3.76651,2.45)(3.78236,2.46)(3.79821,2.47)(3.81407,2.48)(3.82994,2.49)(3.84581,2.5)(3.8617,2.51)(3.87759,2.52)(3.89349,2.53)(3.90939,2.54)(3.92531,2.55)(3.94123,2.56)(3.95716,2.57)(3.97309,2.58)(3.98903,2.59)(4.00498,2.6)(4.02094,2.61)(4.0369,2.62)(4.05287,2.63)(4.06885,2.64)(4.08483,2.65)(4.10082,2.66)(4.11682,2.67)(4.13282,2.68)(4.14883,2.69)(4.16484,2.7)(4.18087,2.71)(4.19689,2.72)(4.21292,2.73)(4.22896,2.74)(4.24501,2.75)(4.26106,2.76)(4.27711,2.77)(4.29317,2.78)(4.30924,2.79)(4.32531,2.8)(4.34139,2.81)(4.35747,2.82)(4.37356,2.83)(4.38966,2.84)(4.40575,2.85)(4.42186,2.86)(4.43797,2.87)(4.45408,2.88)(4.4702,2.89)(4.48632,2.9)(4.50245,2.91)(4.51859,2.92)(4.53472,2.93)(4.55087,2.94)(4.56701,2.95)(4.58316,2.96)(4.59932,2.97)(4.61548,2.98)(4.63165,2.99)(4.64782,3.)};

\addplot[domain=-5:5.,color=midnightblue,thick]
coordinates{
(-0.248582,-3.)(-0.246116,-2.99)(-0.243649,-2.98)(-0.241182,-2.97)(-0.238715,-2.96)(-0.236247,-2.95)(-0.233778,-2.94)(-0.231309,-2.93)(-0.22884,-2.92)(-0.226371,-2.91)(-0.2239,-2.9)(-0.22143,-2.89)(-0.218959,-2.88)(-0.216488,-2.87)(-0.214016,-2.86)(-0.211544,-2.85)(-0.209071,-2.84)(-0.206598,-2.83)(-0.204125,-2.82)(-0.201651,-2.81)(-0.199177,-2.8)(-0.196702,-2.79)(-0.194227,-2.78)(-0.191752,-2.77)(-0.189276,-2.76)(-0.186799,-2.75)(-0.184322,-2.74)(-0.181845,-2.73)(-0.179368,-2.72)(-0.17689,-2.71)(-0.174411,-2.7)(-0.171932,-2.69)(-0.169453,-2.68)(-0.166973,-2.67)(-0.164493,-2.66)(-0.162012,-2.65)(-0.159531,-2.64)(-0.15705,-2.63)(-0.154568,-2.62)(-0.152086,-2.61)(-0.149603,-2.6)(-0.14712,-2.59)(-0.144636,-2.58)(-0.142152,-2.57)(-0.139668,-2.56)(-0.137183,-2.55)(-0.134698,-2.54)(-0.132213,-2.53)(-0.129727,-2.52)(-0.12724,-2.51)(-0.124753,-2.5)(-0.122266,-2.49)(-0.119779,-2.48)(-0.117291,-2.47)(-0.114802,-2.46)(-0.112314,-2.45)(-0.109824,-2.44)(-0.107335,-2.43)(-0.104845,-2.42)(-0.102355,-2.41)(-0.0998642,-2.4)(-0.0973731,-2.39)(-0.0948818,-2.38)(-0.09239,-2.37)(-0.0898979,-2.36)(-0.0874055,-2.35)(-0.0849127,-2.34)(-0.0824196,-2.33)(-0.0799261,-2.32)(-0.0774323,-2.31)(-0.0749382,-2.3)(-0.0724437,-2.29)(-0.0699489,-2.28)(-0.0674538,-2.27)(-0.0649585,-2.26)(-0.0624628,-2.25)(-0.0599668,-2.24)(-0.0574705,-2.23)(-0.054974,-2.22)(-0.0524772,-2.21)(-0.0499802,-2.2)(-0.0474828,-2.19)(-0.0449853,-2.18)(-0.0424875,-2.17)(-0.0399895,-2.16)(-0.0374913,-2.15)(-0.0349928,-2.14)(-0.0324942,-2.13)(-0.0299954,-2.12)(-0.0274964,-2.11)(-0.0249973,-2.1)(-0.022498,-2.09)(-0.0199986,-2.08)(-0.017499,-2.07)(-0.0149994,-2.06)(-0.0124996,-2.05)(-0.0099998,-2.04)(-0.0074999,-2.03)(-0.00499996,-2.02)(-0.00249998,-2.01)(2.25273*10^-8,-2.)(0.00250002,-1.99)(0.005,-1.98)(0.00749994,-1.97)(0.00999982,-1.96)(0.0124996,-1.95)(0.0149993,-1.94)(0.0174989,-1.93)(0.0199984,-1.92)(0.0224977,-1.91)(0.0249968,-1.9)(0.0274957,-1.89)(0.0299944,-1.88)(0.0324928,-1.87)(0.0349909,-1.86)(0.0374887,-1.85)(0.0399861,-1.84)(0.0424832,-1.83)(0.0449799,-1.82)(0.0474761,-1.81)(0.0499719,-1.8)(0.0524671,-1.79)(0.0549618,-1.78)(0.0574559,-1.77)(0.0599495,-1.76)(0.0624423,-1.75)(0.0649344,-1.74)(0.0674259,-1.73)(0.0699165,-1.72)(0.0724063,-1.71)(0.0748952,-1.7)(0.0773832,-1.69)(0.0798702,-1.68)(0.0823562,-1.67)(0.0848411,-1.66)(0.0873249,-1.65)(0.0898074,-1.64)(0.0922887,-1.63)(0.0947687,-1.62)(0.0972474,-1.61)(0.0997245,-1.6)(0.1022,-1.59)(0.104674,-1.58)(0.107147,-1.57)(0.109617,-1.56)(0.112086,-1.55)(0.114553,-1.54)(0.117018,-1.53)(0.119481,-1.52)(0.121942,-1.51)(0.1244,-1.5)(0.126856,-1.49)(0.12931,-1.48)(0.131761,-1.47)(0.134209,-1.46)(0.136655,-1.45)(0.139097,-1.44)(0.141537,-1.43)(0.143974,-1.42)(0.146407,-1.41)(0.148837,-1.4)(0.151263,-1.39)(0.153685,-1.38)(0.156104,-1.37)(0.158519,-1.36)(0.160929,-1.35)(0.163335,-1.34)(0.165737,-1.33)(0.168134,-1.32)(0.170526,-1.31)(0.172912,-1.3)(0.175294,-1.29)(0.17767,-1.28)(0.18004,-1.27)(0.182405,-1.26)(0.184763,-1.25)(0.187114,-1.24)(0.189459,-1.23)(0.191797,-1.22)(0.194128,-1.21)(0.196451,-1.2)(0.198766,-1.19)(0.201073,-1.18)(0.203372,-1.17)(0.205661,-1.16)(0.207942,-1.15)(0.210213,-1.14)(0.212474,-1.13)(0.214724,-1.12)(0.216964,-1.11)(0.219193,-1.1)(0.22141,-1.09)(0.223615,-1.08)(0.225807,-1.07)(0.227986,-1.06)(0.230152,-1.05)(0.232303,-1.04)(0.234439,-1.03)(0.23656,-1.02)(0.238665,-1.01)(0.240753,-1.)(0.242824,-0.99)(0.244877,-0.98)(0.246911,-0.97)(0.248925,-0.96)(0.250919,-0.95)(0.252891,-0.94)(0.254841,-0.93)(0.256768,-0.92)(0.258671,-0.91)(0.260549,-0.9)(0.2624,-0.89)(0.264225,-0.88)(0.26602,-0.87)(0.267786,-0.86)(0.269521,-0.85)(0.271224,-0.84)(0.272893,-0.83)(0.274527,-0.82)(0.276124,-0.81)(0.277683,-0.8)(0.279202,-0.79)(0.28068,-0.78)(0.282114,-0.77)(0.283502,-0.76)(0.284843,-0.75)(0.286135,-0.74)(0.287375,-0.73)(0.288561,-0.72)(0.28969,-0.71)(0.290761,-0.7)(0.291771,-0.69)(0.292716,-0.68)(0.293594,-0.67)(0.294403,-0.66)(0.295138,-0.65)(0.295797,-0.64)(0.296377,-0.63)(0.296873,-0.62)(0.297283,-0.61)(0.297602,-0.6)(0.297827,-0.59)(0.297953,-0.58)(0.297976,-0.57)(0.297892,-0.56)(0.297696,-0.55)(0.297384,-0.54)(0.29695,-0.53)(0.29639,-0.52)(0.295699,-0.51)(0.294871,-0.5)(0.293901,-0.49)(0.292783,-0.48)(0.291511,-0.47)(0.290081,-0.46)(0.288486,-0.45)(0.286721,-0.44)(0.284779,-0.43)(0.282655,-0.42)(0.280343,-0.41)(0.277837,-0.4)(0.275132,-0.39)(0.272221,-0.38)(0.2691,-0.37)(0.265764,-0.36)(0.262207,-0.35)(0.258425,-0.34)(0.254414,-0.33)(0.25017,-0.32)(0.245689,-0.31)(0.240968,-0.3)(0.236006,-0.29)(0.230799,-0.28)(0.225347,-0.27)(0.219649,-0.26)(0.213706,-0.25)(0.207517,-0.24)(0.201085,-0.23)(0.194411,-0.22)(0.187498,-0.21)(0.18035,-0.2)(0.172972,-0.19)(0.165369,-0.18)(0.157545,-0.17)(0.149509,-0.16)(0.141268,-0.15)(0.132828,-0.14)(0.1242,-0.13)(0.115393,-0.12)(0.106416,-0.11)(0.0972809,-0.1)(0.0879983,-0.09)(0.0785803,-0.08)(0.0690396,-0.07)(0.0593893,-0.06)(0.0496431,-0.05)(0.0398155,-0.04)(0.0299214,-0.03)(0.0199765,-0.02)(0.00999703,-0.01)(0.,0.)(-0.00999697,0.01)(-0.0199755,0.02)(-0.0299165,0.03)(-0.0398001,0.04)(-0.0496055,0.05)(-0.0593114,0.06)(-0.0688953,0.07)(-0.078334,0.08)(-0.0876037,0.09)(-0.0966795,0.1)(-0.105536,0.11)(-0.114147,0.12)(-0.122485,0.13)(-0.130525,0.14)(-0.138237,0.15)(-0.145596,0.16)(-0.152574,0.17)(-0.159145,0.18)(-0.165283,0.19)(-0.170965,0.2)(-0.176168,0.21)(-0.180872,0.22)(-0.18506,0.23)(-0.188717,0.24)(-0.191833,0.25)(-0.194399,0.26)(-0.196411,0.27)(-0.197869,0.28)(-0.198777,0.29)(-0.199141,0.3)(-0.198971,0.31)(-0.198282,0.32)(-0.197088,0.33)(-0.195407,0.34)(-0.193261,0.35)(-0.190671,0.36)(-0.187658,0.37)(-0.184247,0.38)(-0.180461,0.39)(-0.176323,0.4)(-0.171857,0.41)(-0.167084,0.42)(-0.162029,0.43)(-0.156711,0.44)(-0.151151,0.45)(-0.145369,0.46)(-0.139384,0.47)(-0.133212,0.48)(-0.126871,0.49)(-0.120377,0.5)(-0.113743,0.51)(-0.106983,0.52)(-0.100111,0.53)(-0.0931377,0.54)(-0.0860752,0.55)(-0.0789335,0.56)(-0.0717224,0.57)(-0.064451,0.58)(-0.0571276,0.59)(-0.04976,0.6)(-0.0423555,0.61)(-0.0349208,0.62)(-0.0274623,0.63)(-0.0199858,0.64)(-0.0124966,0.65)(-0.00499975,0.66)(0.00250001,0.67)(0.0099984,0.68)(0.0174914,0.69)(0.0249752,0.7)(0.0324466,0.71)(0.0399021,0.72)(0.0473389,0.73)(0.0547542,0.74)(0.0621456,0.75)(0.0695107,0.76)(0.0768472,0.77)(0.0841533,0.78)(0.0914271,0.79)(0.098667,0.8)(0.105871,0.81)(0.113039,0.82)(0.120168,0.83)(0.127258,0.84)(0.134307,0.85)(0.141316,0.86)(0.148281,0.87)(0.155204,0.88)(0.162084,0.89)(0.168919,0.9)(0.175709,0.91)(0.182454,0.92)(0.189153,0.93)(0.195806,0.94)(0.202413,0.95)(0.208974,0.96)(0.215488,0.97)(0.221955,0.98)(0.228376,0.99)(0.23475,1.)(0.241077,1.01)(0.247357,1.02)(0.253591,1.03)(0.259778,1.04)(0.265918,1.05)(0.272013,1.06)(0.278062,1.07)(0.284064,1.08)(0.290022,1.09)(0.295934,1.1)(0.301801,1.11)(0.307624,1.12)(0.313403,1.13)(0.319138,1.14)(0.324829,1.15)(0.330477,1.16)(0.336082,1.17)(0.341645,1.18)(0.347167,1.19)(0.352646,1.2)(0.358085,1.21)(0.363483,1.22)(0.368841,1.23)(0.374159,1.24)(0.379438,1.25)(0.384678,1.26)(0.38988,1.27)(0.395043,1.28)(0.40017,1.29)(0.405259,1.3)(0.410312,1.31)(0.415329,1.32)(0.42031,1.33)(0.425257,1.34)(0.430168,1.35)(0.435046,1.36)(0.43989,1.37)(0.4447,1.38)(0.449478,1.39)(0.454224,1.4)(0.458937,1.41)(0.46362,1.42)(0.468271,1.43)(0.472892,1.44)(0.477482,1.45)(0.482044,1.46)(0.486575,1.47)(0.491079,1.48)(0.495553,1.49)(0.5,1.5)(0.504419,1.51)(0.508812,1.52)(0.513177,1.53)(0.517516,1.54)(0.52183,1.55)(0.526117,1.56)(0.53038,1.57)(0.534618,1.58)(0.538831,1.59)(0.543021,1.6)(0.547187,1.61)(0.55133,1.62)(0.555449,1.63)(0.559547,1.64)(0.563622,1.65)(0.567675,1.66)(0.571706,1.67)(0.575717,1.68)(0.579706,1.69)(0.583675,1.7)(0.587624,1.71)(0.591553,1.72)(0.595462,1.73)(0.599352,1.74)(0.603222,1.75)(0.607074,1.76)(0.610908,1.77)(0.614724,1.78)(0.618521,1.79)(0.622301,1.8)(0.626064,1.81)(0.62981,1.82)(0.633539,1.83)(0.637251,1.84)(0.640947,1.85)(0.644627,1.86)(0.648291,1.87)(0.65194,1.88)(0.655573,1.89)(0.659191,1.9)(0.662795,1.91)(0.666384,1.92)(0.669958,1.93)(0.673518,1.94)(0.677064,1.95)(0.680597,1.96)(0.684116,1.97)(0.687621,1.98)(0.691113,1.99)(0.694593,2.)(0.698059,2.01)(0.701513,2.02)(0.704955,2.03)(0.708385,2.04)(0.711802,2.05)(0.715207,2.06)(0.718601,2.07)(0.721984,2.08)(0.725355,2.09)(0.728715,2.1)(0.732064,2.11)(0.735402,2.12)(0.738729,2.13)(0.742046,2.14)(0.745352,2.15)(0.748648,2.16)(0.751934,2.17)(0.75521,2.18)(0.758477,2.19)(0.761733,2.2)(0.764981,2.21)(0.768218,2.22)(0.771447,2.23)(0.774666,2.24)(0.777877,2.25)(0.781078,2.26)(0.784271,2.27)(0.787455,2.28)(0.790631,2.29)(0.793798,2.3)(0.796957,2.31)(0.800108,2.32)(0.803251,2.33)(0.806385,2.34)(0.809512,2.35)(0.812632,2.36)(0.815743,2.37)(0.818848,2.38)(0.821944,2.39)(0.825034,2.4)(0.828116,2.41)(0.831191,2.42)(0.834259,2.43)(0.83732,2.44)(0.840375,2.45)(0.843422,2.46)(0.846463,2.47)(0.849498,2.48)(0.852525,2.49)(0.855547,2.5)(0.858562,2.51)(0.861571,2.52)(0.864574,2.53)(0.867571,2.54)(0.870561,2.55)(0.873546,2.56)(0.876525,2.57)(0.879498,2.58)(0.882466,2.59)(0.885428,2.6)(0.888384,2.61)(0.891335,2.62)(0.89428,2.63)(0.89722,2.64)(0.900155,2.65)(0.903085,2.66)(0.906009,2.67)(0.908929,2.68)(0.911843,2.69)(0.914752,2.7)(0.917657,2.71)(0.920556,2.72)(0.923451,2.73)(0.926341,2.74)(0.929227,2.75)(0.932108,2.76)(0.934984,2.77)(0.937856,2.78)(0.940723,2.79)(0.943586,2.8)(0.946445,2.81)(0.9493,2.82)(0.95215,2.83)(0.954996,2.84)(0.957838,2.85)(0.960675,2.86)(0.963509,2.87)(0.966339,2.88)(0.969164,2.89)(0.971986,2.9)(0.974804,2.91)(0.977619,2.92)(0.980429,2.93)(0.983236,2.94)(0.986039,2.95)(0.988838,2.96)(0.991634,2.97)(0.994426,2.98)(0.997215,2.99)(1.,3.)};

\addplot[domain=-5:5.,color=midnightblue,thick]
coordinates{
(-5.63797,-3.)(-5.6215,-2.99)(-5.60502,-2.98)(-5.58855,-2.97)(-5.57208,-2.96)(-5.55561,-2.95)(-5.53914,-2.94)(-5.52268,-2.93)(-5.50621,-2.92)(-5.48975,-2.91)(-5.47329,-2.9)(-5.45683,-2.89)(-5.44037,-2.88)(-5.42391,-2.87)(-5.40746,-2.86)(-5.39101,-2.85)(-5.37455,-2.84)(-5.3581,-2.83)(-5.34165,-2.82)(-5.32521,-2.81)(-5.30876,-2.8)(-5.29232,-2.79)(-5.27587,-2.78)(-5.25943,-2.77)(-5.243,-2.76)(-5.22656,-2.75)(-5.21012,-2.74)(-5.19369,-2.73)(-5.17726,-2.72)(-5.16083,-2.71)(-5.1444,-2.7)(-5.12797,-2.69)(-5.11155,-2.68)(-5.09512,-2.67)(-5.0787,-2.66)(-5.06228,-2.65)(-5.04587,-2.64)(-5.02945,-2.63)(-5.01304,-2.62)(-4.99663,-2.61)(-4.98022,-2.6)(-4.96381,-2.59)(-4.9474,-2.58)(-4.931,-2.57)(-4.9146,-2.56)(-4.8982,-2.55)(-4.8818,-2.54)(-4.86541,-2.53)(-4.84901,-2.52)(-4.83262,-2.51)(-4.81623,-2.5)(-4.79985,-2.49)(-4.78346,-2.48)(-4.76708,-2.47)(-4.7507,-2.46)(-4.73432,-2.45)(-4.71794,-2.44)(-4.70157,-2.43)(-4.6852,-2.42)(-4.66883,-2.41)(-4.65246,-2.4)(-4.6361,-2.39)(-4.61974,-2.38)(-4.60338,-2.37)(-4.58702,-2.36)(-4.57067,-2.35)(-4.55431,-2.34)(-4.53796,-2.33)(-4.52162,-2.32)(-4.50527,-2.31)(-4.48893,-2.3)(-4.47259,-2.29)(-4.45625,-2.28)(-4.43992,-2.27)(-4.42359,-2.26)(-4.40726,-2.25)(-4.39093,-2.24)(-4.37461,-2.23)(-4.35829,-2.22)(-4.34197,-2.21)(-4.32566,-2.2)(-4.30934,-2.19)(-4.29304,-2.18)(-4.27673,-2.17)(-4.26043,-2.16)(-4.24413,-2.15)(-4.22783,-2.14)(-4.21153,-2.13)(-4.19524,-2.12)(-4.17895,-2.11)(-4.16267,-2.1)(-4.14639,-2.09)(-4.13011,-2.08)(-4.11383,-2.07)(-4.09756,-2.06)(-4.08129,-2.05)(-4.06503,-2.04)(-4.04877,-2.03)(-4.03251,-2.02)(-4.01625,-2.01)(-4.,-2.)(-3.98375,-1.99)(-3.96751,-1.98)(-3.95127,-1.97)(-3.93503,-1.96)(-3.91879,-1.95)(-3.90256,-1.94)(-3.88634,-1.93)(-3.87012,-1.92)(-3.8539,-1.91)(-3.83768,-1.9)(-3.82147,-1.89)(-3.80526,-1.88)(-3.78906,-1.87)(-3.77286,-1.86)(-3.75667,-1.85)(-3.74048,-1.84)(-3.72429,-1.83)(-3.70811,-1.82)(-3.69193,-1.81)(-3.67576,-1.8)(-3.65959,-1.79)(-3.64342,-1.78)(-3.62726,-1.77)(-3.61111,-1.76)(-3.59496,-1.75)(-3.57881,-1.74)(-3.56267,-1.73)(-3.54653,-1.72)(-3.5304,-1.71)(-3.51427,-1.7)(-3.49815,-1.69)(-3.48203,-1.68)(-3.46592,-1.67)(-3.44981,-1.66)(-3.43371,-1.65)(-3.41761,-1.64)(-3.40152,-1.63)(-3.38543,-1.62)(-3.36935,-1.61)(-3.35328,-1.6)(-3.33721,-1.59)(-3.32114,-1.58)(-3.30508,-1.57)(-3.28903,-1.56)(-3.27298,-1.55)(-3.25694,-1.54)(-3.24091,-1.53)(-3.22488,-1.52)(-3.20885,-1.51)(-3.19284,-1.5)(-3.17683,-1.49)(-3.16082,-1.48)(-3.14483,-1.47)(-3.12883,-1.46)(-3.11285,-1.45)(-3.09687,-1.44)(-3.0809,-1.43)(-3.06494,-1.42)(-3.04898,-1.41)(-3.03303,-1.4)(-3.01709,-1.39)(-3.00115,-1.38)(-2.98522,-1.37)(-2.9693,-1.36)(-2.95339,-1.35)(-2.93748,-1.34)(-2.92159,-1.33)(-2.9057,-1.32)(-2.88982,-1.31)(-2.87394,-1.3)(-2.85808,-1.29)(-2.84222,-1.28)(-2.82637,-1.27)(-2.81053,-1.26)(-2.7947,-1.25)(-2.77888,-1.24)(-2.76306,-1.23)(-2.74726,-1.22)(-2.73146,-1.21)(-2.71568,-1.2)(-2.6999,-1.19)(-2.68413,-1.18)(-2.66838,-1.17)(-2.65263,-1.16)(-2.63689,-1.15)(-2.62117,-1.14)(-2.60545,-1.13)(-2.58974,-1.12)(-2.57405,-1.11)(-2.55836,-1.1)(-2.54269,-1.09)(-2.52703,-1.08)(-2.51138,-1.07)(-2.49574,-1.06)(-2.48011,-1.05)(-2.4645,-1.04)(-2.44889,-1.03)(-2.4333,-1.02)(-2.41772,-1.01)(-2.40216,-1.)(-2.38661,-0.99)(-2.37107,-0.98)(-2.35554,-0.97)(-2.34003,-0.96)(-2.32453,-0.95)(-2.30904,-0.94)(-2.29357,-0.93)(-2.27812,-0.92)(-2.26268,-0.91)(-2.24725,-0.9)(-2.23184,-0.89)(-2.21644,-0.88)(-2.20106,-0.87)(-2.1857,-0.86)(-2.17035,-0.85)(-2.15502,-0.84)(-2.13971,-0.83)(-2.12442,-0.82)(-2.10914,-0.81)(-2.09388,-0.8)(-2.07864,-0.79)(-2.06341,-0.78)(-2.04821,-0.77)(-2.03303,-0.76)(-2.01786,-0.75)(-2.00272,-0.74)(-1.98759,-0.73)(-1.97249,-0.72)(-1.95741,-0.71)(-1.94235,-0.7)(-1.92731,-0.69)(-1.9123,-0.68)(-1.8973,-0.67)(-1.88234,-0.66)(-1.86739,-0.65)(-1.85247,-0.64)(-1.83758,-0.63)(-1.82271,-0.62)(-1.80787,-0.61)(-1.79306,-0.6)(-1.77827,-0.59)(-1.76351,-0.58)(-1.74878,-0.57)(-1.73408,-0.56)(-1.71941,-0.55)(-1.70477,-0.54)(-1.69016,-0.53)(-1.67558,-0.52)(-1.66104,-0.51)(-1.64653,-0.5)(-1.63206,-0.49)(-1.61762,-0.48)(-1.60322,-0.47)(-1.58885,-0.46)(-1.57453,-0.45)(-1.56024,-0.44)(-1.54599,-0.43)(-1.53179,-0.42)(-1.51763,-0.41)(-1.50351,-0.4)(-1.48944,-0.39)(-1.47541,-0.38)(-1.46143,-0.37)(-1.4475,-0.36)(-1.43362,-0.35)(-1.4198,-0.34)(-1.40603,-0.33)(-1.39231,-0.32)(-1.37865,-0.31)(-1.36505,-0.3)(-1.35151,-0.29)(-1.33804,-0.28)(-1.32463,-0.27)(-1.31129,-0.26)(-1.29802,-0.25)(-1.28482,-0.24)(-1.2717,-0.23)(-1.25865,-0.22)(-1.24569,-0.21)(-1.23281,-0.2)(-1.22001,-0.19)(-1.20731,-0.18)(-1.19471,-0.17)(-1.1822,-0.16)(-1.1698,-0.15)(-1.1575,-0.14)(-1.14532,-0.13)(-1.13325,-0.12)(-1.12131,-0.11)(-1.10949,-0.1)(-1.09781,-0.09)(-1.08627,-0.08)(-1.07488,-0.07)(-1.06365,-0.06)(-1.05257,-0.05)(-1.04168,-0.04)(-1.03096,-0.03)(-1.02043,-0.02)(-1.01011,-0.01)(-1.,0.)(-0.990115,0.01)(-0.980467,0.02)(-0.971071,0.03)(-0.961941,0.04)(-0.953093,0.05)(-0.944542,0.06)(-0.936307,0.07)(-0.928406,0.08)(-0.920858,0.09)(-0.913684,0.1)(-0.906905,0.11)(-0.900542,0.12)(-0.89462,0.13)(-0.88916,0.14)(-0.884188,0.15)(-0.879726,0.16)(-0.875798,0.17)(-0.872428,0.18)(-0.869638,0.19)(-0.867449,0.2)(-0.865881,0.21)(-0.864951,0.22)(-0.864674,0.23)(-0.865063,0.24)(-0.866126,0.25)(-0.867869,0.26)(-0.870294,0.27)(-0.873401,0.28)(-0.877182,0.29)(-0.881631,0.3)(-0.886736,0.31)(-0.892481,0.32)(-0.898849,0.33)(-0.905821,0.34)(-0.913376,0.35)(-0.92149,0.36)(-0.93014,0.37)(-0.939302,0.38)(-0.948952,0.39)(-0.959065,0.4)(-0.969616,0.41)(-0.980583,0.42)(-0.991942,0.43)(-1.00367,0.44)(-1.01575,0.45)(-1.02816,0.46)(-1.04087,0.47)(-1.05388,0.48)(-1.06716,0.49)(-1.0807,0.5)(-1.09449,0.51)(-1.1085,0.52)(-1.12272,0.53)(-1.13715,0.54)(-1.15177,0.55)(-1.16657,0.56)(-1.18154,0.57)(-1.19667,0.58)(-1.21195,0.59)(-1.22738,0.6)(-1.24293,0.61)(-1.25862,0.62)(-1.27443,0.63)(-1.29035,0.64)(-1.30639,0.65)(-1.32252,0.66)(-1.33876,0.67)(-1.35508,0.68)(-1.3715,0.69)(-1.388,0.7)(-1.40458,0.71)(-1.42124,0.72)(-1.43797,0.73)(-1.45477,0.74)(-1.47164,0.75)(-1.48857,0.76)(-1.50556,0.77)(-1.52261,0.78)(-1.53972,0.79)(-1.55688,0.8)(-1.57409,0.81)(-1.59135,0.82)(-1.60866,0.83)(-1.62602,0.84)(-1.64341,0.85)(-1.66085,0.86)(-1.67833,0.87)(-1.69585,0.88)(-1.71341,0.89)(-1.731,0.9)(-1.74863,0.91)(-1.76629,0.92)(-1.78398,0.93)(-1.80171,0.94)(-1.81946,0.95)(-1.83724,0.96)(-1.85506,0.97)(-1.87289,0.98)(-1.89076,0.99)(-1.90865,1.)(-1.92657,1.01)(-1.94451,1.02)(-1.96247,1.03)(-1.98045,1.04)(-1.99846,1.05)(-2.01649,1.06)(-2.03454,1.07)(-2.0526,1.08)(-2.07069,1.09)(-2.0888,1.1)(-2.10692,1.11)(-2.12506,1.12)(-2.14322,1.13)(-2.1614,1.14)(-2.17959,1.15)(-2.1978,1.16)(-2.21602,1.17)(-2.23426,1.18)(-2.25251,1.19)(-2.27078,1.2)(-2.28906,1.21)(-2.30735,1.22)(-2.32566,1.23)(-2.34398,1.24)(-2.36231,1.25)(-2.38065,1.26)(-2.39901,1.27)(-2.41738,1.28)(-2.43576,1.29)(-2.45415,1.3)(-2.47255,1.31)(-2.49096,1.32)(-2.50938,1.33)(-2.52781,1.34)(-2.54625,1.35)(-2.5647,1.36)(-2.58316,1.37)(-2.60163,1.38)(-2.62011,1.39)(-2.6386,1.4)(-2.65709,1.41)(-2.6756,1.42)(-2.69411,1.43)(-2.71263,1.44)(-2.73116,1.45)(-2.74969,1.46)(-2.76823,1.47)(-2.78679,1.48)(-2.80534,1.49)(-2.82391,1.5)(-2.84248,1.51)(-2.86106,1.52)(-2.87964,1.53)(-2.89823,1.54)(-2.91683,1.55)(-2.93544,1.56)(-2.95405,1.57)(-2.97266,1.58)(-2.99128,1.59)(-3.00991,1.6)(-3.02855,1.61)(-3.04719,1.62)(-3.06583,1.63)(-3.08448,1.64)(-3.10314,1.65)(-3.1218,1.66)(-3.14046,1.67)(-3.15913,1.68)(-3.17781,1.69)(-3.19649,1.7)(-3.21517,1.71)(-3.23386,1.72)(-3.25256,1.73)(-3.27126,1.74)(-3.28996,1.75)(-3.30867,1.76)(-3.32738,1.77)(-3.34609,1.78)(-3.36481,1.79)(-3.38354,1.8)(-3.40226,1.81)(-3.421,1.82)(-3.43973,1.83)(-3.45847,1.84)(-3.47721,1.85)(-3.49596,1.86)(-3.51471,1.87)(-3.53347,1.88)(-3.55222,1.89)(-3.57098,1.9)(-3.58975,1.91)(-3.60852,1.92)(-3.62729,1.93)(-3.64606,1.94)(-3.66484,1.95)(-3.68362,1.96)(-3.7024,1.97)(-3.72119,1.98)(-3.73998,1.99)(-3.75877,2.)(-3.77757,2.01)(-3.79636,2.02)(-3.81516,2.03)(-3.83397,2.04)(-3.85278,2.05)(-3.87158,2.06)(-3.8904,2.07)(-3.90921,2.08)(-3.92803,2.09)(-3.94685,2.1)(-3.96567,2.11)(-3.98449,2.12)(-4.00332,2.13)(-4.02215,2.14)(-4.04098,2.15)(-4.05982,2.16)(-4.07865,2.17)(-4.09749,2.18)(-4.11633,2.19)(-4.13517,2.2)(-4.15402,2.21)(-4.17287,2.22)(-4.19172,2.23)(-4.21057,2.24)(-4.22942,2.25)(-4.24828,2.26)(-4.26713,2.27)(-4.28599,2.28)(-4.30486,2.29)(-4.32372,2.3)(-4.34258,2.31)(-4.36145,2.32)(-4.38032,2.33)(-4.39919,2.34)(-4.41806,2.35)(-4.43694,2.36)(-4.45582,2.37)(-4.47469,2.38)(-4.49357,2.39)(-4.51245,2.4)(-4.53134,2.41)(-4.55022,2.42)(-4.56911,2.43)(-4.588,2.44)(-4.60689,2.45)(-4.62578,2.46)(-4.64467,2.47)(-4.66357,2.48)(-4.68246,2.49)(-4.70136,2.5)(-4.72026,2.51)(-4.73916,2.52)(-4.75806,2.53)(-4.77696,2.54)(-4.79587,2.55)(-4.81477,2.56)(-4.83368,2.57)(-4.85259,2.58)(-4.8715,2.59)(-4.89041,2.6)(-4.90932,2.61)(-4.92824,2.62)(-4.94715,2.63)(-4.96607,2.64)(-4.98499,2.65)(-5.00391,2.66)(-5.02283,2.67)(-5.04175,2.68)(-5.06067,2.69)(-5.0796,2.7)(-5.09852,2.71)(-5.11745,2.72)(-5.13638,2.73)(-5.1553,2.74)(-5.17423,2.75)(-5.19317,2.76)(-5.2121,2.77)(-5.23103,2.78)(-5.24996,2.79)(-5.2689,2.8)(-5.28784,2.81)(-5.30677,2.82)(-5.32571,2.83)(-5.34465,2.84)(-5.36359,2.85)(-5.38253,2.86)(-5.40148,2.87)(-5.42042,2.88)(-5.43936,2.89)(-5.45831,2.9)(-5.47726,2.91)(-5.4962,2.92)(-5.51515,2.93)(-5.5341,2.94)(-5.55305,2.95)(-5.572,2.96)(-5.59095,2.97)(-5.60991,2.98)(-5.62886,2.99)(-5.64782,3.)};
\end{axis}
\end{tikzpicture}
\caption{Space-time diagram of particle scattering in the rational Calogero-Moser model with coupling $g=1$, and initial state $q(0)=(1,0,-1)$, $p(0)=(1,-1,1)$.}
\label{fig:8}
\end{figure}
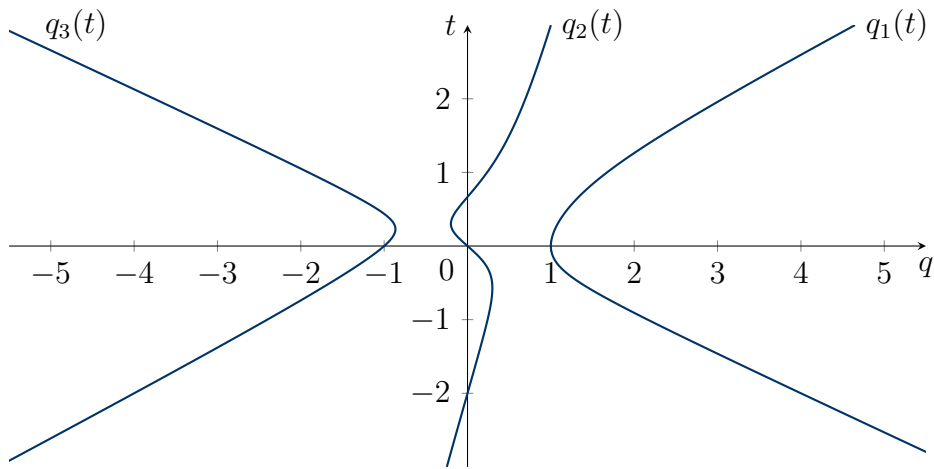

\chapter{Appendices to Chapter \ref{chap:2}}
\label{chap:B}

\section{Application of Jacobi's theorem on complementary minors}
\label{sec:B.1}

In this appendix, we complete the proof of Lemma \ref{lem:2.7} by a calculation based
on Jacobi's theorem on complementary minors (e.g.~\cite{Pr94}), which will be recalled
shortly. Our reasoning below is adapted from Pusztai \cite{Pu11}. A significant difference
is that in our case we need the strong regularity conditions \eqref{2.87} and \eqref{2.100} to avoid dividing by zero during
the calculation. In fact, this appendix is presented mainly
to explain the origin of the strong regularity conditions.

For an $m\times m$ matrix $M$ let
$M\big(\begin{smallmatrix}r_1&\dots&r_k\\c_1&\dots&c_k\end{smallmatrix}\big)$
denote the determinant formed from the entries lying on the intersection of
the rows $r_1,\dots,r_k$ with the columns $c_1,\dots,c_k$ of $M$ $(k\leq m)$,
$$M\begin{pmatrix}r_1&\dots&r_k\\c_1&\dots&c_k\end{pmatrix}
=\det(M_{r_i,c_j})_{i,j=1}^k.$$

\begin{theorem}[Jacobi]
\label{thm:B.1}
Let $A$ be an invertible $N\times N$ matrix with $\det(A)=1$ and $B=(A^{-1})^\top$.
For a fixed permutation
$\big(\begin{smallmatrix}j_1&\dots&j_N\\k_1&\dots&k_N\end{smallmatrix}\big)$
of the pairwise distinct indices $j_1,\dots,j_N\in\{1,\dots,N\}$ and
any $1\leq p<N$ we have
\begin{equation}
B\begin{pmatrix}j_1&\dots&j_p\\k_1&\dots&k_p\end{pmatrix}
=\sgn\begin{pmatrix}j_1&\dots&j_N\\k_1&\dots&k_N\end{pmatrix}
A\begin{pmatrix}j_{p+1}&\dots&j_N\\k_{p+1}&\dots&k_N\end{pmatrix}.
\label{B.1}
\end{equation}
\end{theorem}

Applying Jacobi's theorem to $\check A$ \eqref{2.101} we now derive
the two equations \eqref{2.103} and \eqref{2.104}
for the pair of functions $(W_a,W_{n+a})$ for each $a=1,\dots, n$, which are defined
by $W_k=w_k\cF_k$ with $\cF_k=\vert F_k\vert^2$ \eqref{2.93} and $w_k$ in \eqref{2.102}.

\begin{lemma}
\label{lem:B.2}
Fix any strongly regular $\lambda$, i.e. $\lambda\in\R^n$
for which \eqref{2.87} and \eqref{2.100} hold,
and use the above notations for $(W_a,W_{n+a})$.
If $\check A$ given by \eqref{2.101} is a unitary matrix, then
$(W_a,W_{n+a})$ satisfies the two equations \eqref{2.103} and
\eqref{2.104} for each $a=1,\dots,n$.
\end{lemma}

\begin{proof}
Let $\check B=(\check{A}^{-1})^\top$, i.e. $\check{B}_{j,k}=\overline{\check{A}}_{j,k}$, $j,k\in\{1,\dots,N\}$
and $a\in\{1,\dots,n\}$ be a fixed index. Since $\det(\check A)=1$, by
Jacobi's theorem with $j_b=b$, ($b\in\N_N$) and $k_c=c$, ($c\in\N_N\setminus\{a,n+a\}$), $k_a=n+a$, $k_{n+a}=a$ and $p=n$
we have
\begin{equation}
\check{B}\begin{pmatrix}
1&\dots&a&\dots&n\\
1&\dots&n+a&\dots&n
\end{pmatrix}=
-\check{A}\begin{pmatrix}
n+1&\dots&n+a&\dots&N\\
n+1&\dots&a&\dots&N
\end{pmatrix}.
\label{B.2}
\end{equation}
Denote the corresponding $n\times n$ submatrices of $\check B$ and $\check A$
 by $\xi$ and $\eta$, respectively.
One can check that
\begin{equation}
\xi=\Psi-\frac{\mu-\nu}{\mu-\lambda_a}E_{a,a},\quad
\eta=\Xi-\frac{\mu-\nu}{\mu+\lambda_a}E_{a,a},
\label{B.3}
\end{equation}
where $E_{j,k}$ stands for the $n\times n$ elementary matrix
$(E_{j,k})_{j',k'}=\delta_{j,j'}\delta_{k,k'}$
and $\Psi$ and $\Xi$ are the Cauchy-like matrices
\begin{equation}
\Psi_{j,k}=\begin{cases}
\dfrac{2\mu\overline{F}_jF_{n+k}}{2\mu-\lambda_j+\lambda_k},&\text{if}\ k\neq a,\\[.5cm]
\dfrac{2\mu\overline{F}_jF_a}{2\mu-\lambda_j-\lambda_a},&\text{if}\ k=a,
\end{cases}
\qquad\text{and}\qquad
\Xi_{j,k}=\begin{cases}
\dfrac{2\mu F_{n+j}\overline{F}_k}{2\mu+\lambda_j-\lambda_k},&\mbox{if}\ k\neq a,\\[.5cm]
\dfrac{2\mu F_{n+j}\overline{F}_{n+a}}{2\mu+\lambda_j+\lambda_a},&\mbox{if}\ k=a,
\end{cases}
\label{B.4}
\end{equation}
$j,k\in\{1,\dots,n\}$.
Expanding $\det(\xi)$ and $\det(\eta)$ along the $a$-th column we obtain the formulae
\begin{equation}
\det(\xi)=\det(\Psi)-\frac{\mu-\nu}{\mu-\lambda_a}\cC_{a,a},\quad
\det(\eta)=\det(\Xi)-\frac{\mu-\nu}{\mu+\lambda_a}\cC_{a,a},
\label{B.5}
\end{equation}
where $\cC_{a,a}$ is the cofactor of $\Psi$ associated with entry $\Psi_{a,a}$.
Since $\Psi$ and $\Xi$ are both Cauchy-like matrices we have
\begin{equation}
\det(\Psi)=\frac{1}{\mu-\lambda_a}D_aW_a,\quad
\det(\Xi) =\frac{1}{\mu+\lambda_a}D_aW_{n+a},
\label{B.6}
\end{equation}
where
\begin{equation}
D_a=\prod_{\substack{b=1\\(b\neq a)}}^n\overline{F}_bF_{n+b}
\prod_{\substack{c,d=1\\(a\neq c\neq d\neq a)}}^n
\frac{\lambda_c-\lambda_d}{2\mu+\lambda_c-\lambda_d}.
\label{B.7}
\end{equation}
It can be easily seen that $\cC_{a,a}=D_a$, therefore formulae
\eqref{B.2}, \eqref{B.5}, \eqref{B.6} lead to the equation
\begin{equation}
(\mu+\lambda_a)W_a+(\mu-\lambda_a)W_{n+a}-2(\mu-\nu)=0.
\label{B.8}
\end{equation}

It should be noticed that in the last step we divided by $D_a$,
which is legitimate since $D_a$ is non-vanishing due to the strong-regularity
condition given by \eqref{2.87} and \eqref{2.100}. To see this, assume momentarily that
$F_i=0$ for some $i=1,\dots,n$ at some strongly regular $\lambda$.
The denominator in \eqref{2.101} does not vanish, and the unitarity of
$\check A$ implies that we must have $\check A_{i, i+n}=1$ or $\check A_{i, i+n}=-1$.
These in turn are equivalent to
\begin{equation}
\lambda_i = 2\mu -\nu
\quad\text{or}\quad \lambda_i =\nu,
\label{B.9}
\end{equation}
which are excluded by \eqref{2.100}. One can similarly check that the vanishing of
$F_{n+i}$ would require
\begin{equation}
\lambda_i =\nu - 2\mu
\quad\text{or}\quad \lambda_i =-\nu,
\label{B.10}
\end{equation}
which are also excluded. These remarks pinpoint the origin of the second half of
the conditions imposed in \eqref{2.100}.

Next, we apply Jacobi's theorem by setting
$j_b=k_b=b$, ($b\in\N_n$), $j_{n+1}=k_{n+1}=n+a$, $j_{n+c}=k_{n+c}=n+c-1$,
($c\in\N_{n-1}$) and $p=n+1$. Thus
\begin{equation}
\check{B}\begin{pmatrix}
1&\dots&n&n+a\\
1&\dots&n&n+a
\end{pmatrix}=
\check{A}\begin{pmatrix}
n+1&\dots&\widehat{n+a}&\dots&N\\
n+1&\dots&\widehat{n+a}&\dots&N
\end{pmatrix},
\label{B.11}
\end{equation}
where $\widehat{n+a}$ indicates that the $(n+a)$-th row and column are omitted.
Now denote the submatrices of size $(n+1)$ and $(n-1)$ corresponding to the determinants in \eqref{B.11}
 by $X$ and $Y$, respectively.
From \eqref{B.11} and \eqref{2.101} it follows that $\det(X)=\det(Y)=D_a$ \eqref{B.7}.
The submatrix $X$ can be written in the form
\begin{equation}
X=\Phi-\frac{\mu-\nu}{\mu-\lambda_a}E_{a,n+1}-\frac{\mu-\nu}{\mu+\lambda_a}E_{n+1,a},
\label{B.12}
\end{equation}
i.e. $X$ is a rank two perturbation of the Cauchy-like matrix $\Phi$
having the entries
\begin{equation}
\begin{gathered}
\Phi_{j,k}=\frac{2\mu\overline{F}_jF_{n+k}}{2\mu-\lambda_j+\lambda_k},\quad
\Phi_{j,n+1}=\frac{2\mu\overline{F}_jF_a}{2\mu-\lambda_j-\lambda_a},\\
\Phi_{n+1,k}=\frac{2\mu\overline{F}_{n+a}F_{n+k}}{2\mu+\lambda_a+\lambda_k},\quad
\Phi_{n+1,n+1}=\overline{F}_{n+a}F_{a},
\end{gathered}
\label{B.13}
\end{equation}
where $j,k\in\{1,\dots,n\}$.
The determinant of $\Phi$ is
\begin{equation}
\det(\Phi)=-\frac{\lambda_a^2}{\mu^2-\lambda_a^2}D_a W_a W_{n+a},
\label{B.14}
\end{equation}
which cannot vanish because $\lambda$ is strongly regular.
Since $X$ is a rank two perturbation of $\Phi$ we obtain
\begin{equation}
\det(X)=\det(\Phi)-(\mu-\nu)\bigg(\frac{\cC_{a,n+1}}{\mu-\lambda_a}+\frac{\cC_{n+1,a}}{\mu+\lambda_a}\bigg)
+(\mu-\nu)^2\frac{\cC_{a,n+1}\cC_{n+1,a}-\cC_{a,a}\cC_{n+1,n+1}}{(\mu-\lambda_a)(\mu+\lambda_a)\det(\Phi)},
\label{B.15}
\end{equation}
where $\cC$ now is used to denote the cofactors of $\Phi$.
By calculating the necessary cofactors we derive
\begin{equation}
\begin{gathered}
\cC_{a,a}\cC_{n+1,n+1}=D_a^2W_aW_{n+a},\\
\cC_{a,n+1}=-\frac{1}{\mu+\lambda_a}D_aW_{n+a},\quad
\cC_{n+1,a}=-\frac{1}{\mu-\lambda_a}D_aW_a.
\end{gathered}
\label{B.16}
\end{equation}
Equations \eqref{B.14}-\eqref{B.16} together with $\det(X)=D_a$ imply
\begin{equation}
\lambda_a^2(W_aW_{n+a}-1)-\mu(\mu-\nu)(W_a+W_{n+a}-2)+\nu^2=0.
\label{B.17}
\end{equation}
Equations \eqref{B.8} and \eqref{B.17} coincide with \eqref{2.103} and \eqref{2.104},
respectively.
\end{proof}

\section{$\cH_l^\vD$ as elementary symmetric function}
\label{sec:B.2}

Fix an arbitrary $n\in\N$ and $l\in\{0,1,\dots,n\}$ and let $e_l$ stand for the
$l$-th elementary symmetric polynomial in $n$ variables $x_1,\dots,x_n$, i.e. $e_0(x_1,\dots,x_n)=1$ and for $l\geq 1$
\begin{equation}
e_l(x_1,\dots,x_n)=\sum_{1\leq j_1\dots<j_l\leq n}x_{j_1}\dots x_{j_l}.
\label{B.18}
\end{equation}
At the end of Chapter \ref{chap:2}, we referred to the following useful result due to van Diejen
\cite[Proposition 2.3]{vD95}. For convenience, we present it together
with a direct proof.

\begin{proposition}
\label{prop:B.1}
By using \eqref{2.221} it can be shown that
\begin{equation}
\cH_l^\vD(q)=4^le_l(\sinh^2\frac{q_1}{2},\dots,\sinh^2\frac{q_n}{2}).
\label{B.19}
\end{equation}
\end{proposition}

\begin{proof}
First, $e_l$ has the equivalent form
\begin{equation}
e_l(\sinh^2\frac{q_1}{2},\dots,\sinh^2\frac{q_n}{2})
=\sum_{J\subset\{1,\dots,n\},\ |J|=l}\;\prod_{j\in J}\sinh^2\frac{q_j}{2}.
\label{B.20}
\end{equation}
Utilizing the identity $\sinh^2(\alpha/2)=[\cosh(\alpha)-1]/2$
casts the right-hand side into
\begin{equation}
\sum_{J\subset\{1,\dots,n\},\ |J|=l}2^{-l}
\prod_{j\in J}[\cosh(q_j)-1]=
\sum_{J\subset\{1,\dots,n\},\ |J|=l}2^{-l}
\sum_{K\subset J}(-1)^{l-|K|}\prod_{k\in K}\cosh(q_k).
\label{B.21}
\end{equation}
The two sums on the right-hand side can be merged into one,
but the multiplicity of subsets must remain the same.
This results in the appearance of a binomial coefficient
\begin{multline}
\sum_{J\subset\{1,\dots,n\},\ |J|\leq l}\frac{(-1)^{l-|J|}}{2^l}{n-|J|\choose l-|J|}
\prod_{j\in J}\cosh(q_j)=\\
=\sum_{\substack{J\subset\{1,\dots,n\},\ |J|\leq l\\\varepsilon_j=\pm 1,\ j\in J}}
\frac{(-1)^{l-|J|}}{2^{l+|J|}}{n-|J|\choose l-|J|}
\prod_{j\in J}\cosh(\varepsilon_j q_j),
\label{B.22}
\end{multline}
where we also used that $\cosh$ is an even function and compensated the
`over-counting' of terms. Now, let us simply pull a $4^{-l}$ factor out of the sum to get
\begin{equation}
4^{-l}
\sum_{\substack{J\subset\{1,\dots,n\},\ |J|\leq l\\\varepsilon_j=\pm 1,\ j\in J}}
(-2)^{l-|J|}{n-|J|\choose l-|J|}\prod_{j\in J}\cosh(\varepsilon_j q_j).
\label{B.23}
\end{equation}
Recall the following identity for the hyperbolic cosine of the sum of a finite number, say $N$, real arguments (see \cite[Art. 132]{H1841} and apply $\cos(\ri\alpha)=\cosh(\alpha)$)
\begin{equation}
\cosh\bigg(\sum_{k=1}^N\alpha_k\bigg)
=\bigg[\prod_{k=1}^N\cosh(\alpha_k)\bigg]
\bigg[\sum_{m=0}^{\big\lfloor\tfrac{ N}{2}\big\rfloor}e_{2m}(\tanh(\alpha_1),\dots,\tanh(\alpha_N))\bigg],
\label{B.24}
\end{equation}
where $e_{2m}$ are now elementary symmetric functions with arguments
$\tanh(\alpha_k)$.
Note that for any $m>0$ and set of signs $\varepsilon$ there is another one
$\varepsilon'$, such that $e_{2m}^{J,\varepsilon'}=-e_{2m}^{J,\varepsilon}$,
therefore by using \eqref{B.24} we see that \eqref{B.23} equals to
\begin{multline}
4^{-l}
\sum_{\substack{J\subset\{1,\dots,n\},\ |J|\leq l\\\varepsilon_j=\pm 1,\ j\in J}}
(-2)^{l-|J|}{n-|J|\choose l-|J|}
\prod_{j\in J}
\cosh(\varepsilon_j q_j)\sum_{m=0}^{\big\lfloor\tfrac{|J|}{2}\big\rfloor}
s_{2m}^{J,\varepsilon}=\\
=4^{-l}
\sum_{\substack{J\subset\{1,\dots,n\},\ |J|\leq l\\\varepsilon_j=\pm 1,\ j\in J}}
(-2)^{l-|J|}{n-|J|\choose l-|J|}
\cosh(q_{\varepsilon J}).
\label{B.25}
\end{multline}
Applying \eqref{2.221} concludes the proof.
\end{proof}

\chapter{Appendices to Chapter \ref{chap:3}}
\label{chap:C}

\section{Links to systems of van Diejen and Schneider}
\label{sec:C.1}

Recall that the trigonometric $\BC_n$ van Diejen system \cite{vD94} has the Hamiltonian
\begin{equation}
H_{\text{vD}}(\lambda,\theta)=\sum_{j=1}^n\big(\cosh(\theta_j)
\sV_j(\lambda)^{1/2}\sV_{-j}(\lambda)^{1/2}-[\sV_j(\lambda)+\sV_{-j}(\lambda)]/2\big),
\label{C.1}
\end{equation}
with $\sV_{\pm j}$ ($j=1,\dots,n$) defined by
\begin{equation}
\sV_{\pm j}(\lambda)=\sw(\pm\lambda_j)\prod_{\substack{k=1\\(k\neq j)}}^n
\sv(\pm\lambda_j+\lambda_k)\sv(\pm\lambda_j-\lambda_k),
\label{C.2}
\end{equation}
and $\sv,\sw$ denoting the trigonometric potentials
\begin{equation}
\sv(z)=\frac{\sin(\mu+z)}{\sin(z)}\quad\text{and}\quad
\sw(z)=\frac{\sin(\mu_0+z)}{\sin(z)}
\frac{\cos(\mu_1+z)}{\cos(z)}
\frac{\sin(\mu'_0+z)}{\sin(z)}
\frac{\cos(\mu'_1+z)}{\cos(z)},
\label{C.3}
\end{equation}
where $\mu,\mu_0,\mu_1,\mu'_0,\mu'_1$ are arbitrary parameters. By making the substitutions
\begin{equation}
\begin{gathered}\lambda_j\to\ri(\hat p_j+R),\\\theta_j\to\ri\hat q_j,\end{gathered}
\quad\forall j\quad\text{and}\quad
\mu\to\ri g/2,\quad
\begin{gathered}\mu_0\to\ri(g_0+R),\\\mu'_0\to\ri(g'_0-R),\end{gathered}\quad
\begin{gathered}\mu_1\to\ri g_1+\pi/2,\\\mu'_1\to\ri g'_1+\pi/2\end{gathered}
\label{C.4}
\end{equation}
the potentials become hyperbolic functions and their $R\to\infty$ limit exists, namely
\begin{equation}
\lim_{R\to\infty}\sv(\pm (\lambda_j+\lambda_k))=e^{\pm g/2},\quad
\lim_{R\to\infty}\sv(\pm (\lambda_j-\lambda_k))
=\frac{\sinh(\pm g/2+\hat p_j-\hat p_k)}{\sinh(\hat p_j-\hat p_k)},\quad\forall j,k
\label{C.5}
\end{equation}
and
\begin{equation}
\lim_{R\to\infty}\sw(\pm\lambda_j)
=e^{g_0-g'_0\pm(g_1+g'_1)-2\hat p_j}-e^{\pm(g_0+g'_0+g_1+g'_1)},\quad\forall j.
\label{C.6}
\end{equation}
In the $1$-particle case we have $V_{\pm}(\lambda)=\sw(\pm\lambda)$,
thus $H_{\text{vD}}$ takes the following form
\begin{equation}
H_{\text{vD}}(\lambda,\theta)=\cosh(\theta)\sw(\lambda)^{1/2}\sw(-\lambda)^{1/2}
-[\sw(\lambda)+\sw(-\lambda)]/2.
\label{C.7}
\end{equation}
By utilizing \eqref{C.6} one obtains
\begin{equation}
\begin{gathered}
\lim_{R\to\infty}\sw(\lambda)^{1/2}\sw(-\lambda)^{1/2}
=\big[1-(e^{2g_0}+e^{-2g'_0})e^{-2\hat p}+e^{2g_0-2g'_0-4\hat p}\big]^{1/2},\\
\lim_{R\to\infty}[\sw(\lambda)+\sw(-\lambda)]/2
=\frac{e^{g_0-g'_0+g_1+g'_1}+e^{g_0-g'_0-g_1-g'_1}}{2}e^{-2\hat p}
-\cosh(g_0+g'_0+g_1+g'_1).
\end{gathered}
\label{C.8}
\end{equation}
Equating the $R\to\infty$ limit of $H_{\text{vD}}(\lambda,\theta)$ \eqref{C.7} with
the Hamiltonian $H(\hat p,\hat q;x,a,b)$ \eqref{3.1} yields a system of linear equations
involving $g_0,g_1,g'_0,g'_1$ as unknowns and $u,v$ as parameters. Actually, four
sets of linear equations can be constructed, each with infinitely many solutions
depending on one (real) parameter, but these sets are `equivalent' under the
exchanges: $g_0\leftrightarrow g'_0$ or $g_1\leftrightarrow g'_1$.
Therefore it is sufficient to give only one set of solutions, e.g.
\begin{equation}
g_0=a,\quad g_0'=0,\quad g_1=b-g'_1,\quad g'_1\in\R.
\label{C.9}
\end{equation}
Setting $g=x$ and $g'_1=0$ provides the following special choice of couplings in \eqref{C.4}
\begin{equation}
\mu=\ri x/2,\quad
\mu_0=\ri(a+R),\quad
\mu'_0=-\ri R,\quad
\mu_1=\ri b+\pi/2,\quad
\mu'_1=\pi/2,
\label{C.10}
\end{equation}
and one finds the following
\begin{equation}
\lim_{R\to\infty}H_{\text{vD}}\big(\lambda(\hat p,R),\theta(\hat q)\big)
=-H(\hat p,\hat q;x,a,b)+\cosh(b-a).
\label{C.11}
\end{equation}
In the $n$-particle case, by using \eqref{C.5} and \eqref{C.6} it can be shown
that with \eqref{C.10} one has
\begin{equation}
\lim_{R\to\infty}H_{\text{vD}}\big(\lambda(\hat p,R),\theta(\hat q)\big)
=-H(\hat p,\hat q;x,a,b)+\sum_{j=1}^n\cosh\big((j-1)x+b-a\big),
\label{C.12}
\end{equation}
i.e., the Hamiltonian $H$ \eqref{3.1} is recovered as a singular limit of $H_{\text{vD}}$
\eqref{C.1}.

Consider now the function $H(\hat p,\hat q;x,a,b)$ and introduce the real parameter
$\sigma$ through the substitutions
\begin{equation}
b\to b-2\sigma
\label{C.13}
\end{equation}
and apply the canonical transformation
\begin{equation}
\hat p_j\to -Q_j+\sigma,\quad \hat q_j\to -P_j,\quad\forall j.
\label{C.14}
\end{equation}
Then we have
\begin{equation}
\lim_{\sigma\to\infty}H(\hat p(Q,\sigma),\hat q(P),x,a,b(\sigma))
=H_{\text{Sch}}(Q,P,x,a-b),
\label{C.15}
\end{equation}
with Schneider's \cite{Sc87} Hamiltonian
\begin{equation}
H_{\text{Sch}}(Q,P,x,a-b)
=\frac{e^{a-b}}{2}\sum_{j=1}^ne^{2Q_j}-\sum_{j=1}^n\cos(P_j)
\prod_{\substack{k=1\\(k\neq j)}}^n
\bigg[1-\frac{\sinh^2\big(\frac{x}{2}\big)}
{\sinh^2(Q_j-Q_k)}\bigg]^{\tfrac{1}{2}}.
\label{C.16}
\end{equation}

\begin{remark}
\label{rem:C.1}
(\emph{i}) In \eqref{C.4} only two of the four external field couplings
$\mu_0,\mu_0',\mu_1,\mu_1'$ are scaled with $R$.
However, scaling all four of these parameters also leads to
an integrable Ruijsenaars-Schneider type system with a more
general $4$-parameter external field. For details, see \cite[Section
II.B]{vD95-2}. (\emph{ii}) The connection to Schneider's Hamiltonian
was mentioned in \cite[Remark 7.1]{Ma15} as well, where a singular limit,
similar to \eqref{C.15} was taken.
\end{remark}

\section{Proof of Proposition \ref{prop:3.2}}
\label{sec:C.2}

In this appendix, we prove Proposition \ref{prop:3.2} that states that the range of the
`position variable' $\hat p$ is contained in the closed thick-walled Weyl chamber
$\bar\cC_x$ \eqref{3.58}.

\begin{proof}[Proof of Proposition \ref{prop:3.2}]
According to \eqref{3.57} the matrices $e^{2\hat p}$ and
$e^{2\hat p-x\1_n}+\sgn(x)e^{\hat p}ww^\dag e^{\hat p}$
are similar and therefore have the same characteristic polynomial.
This gives the identity
\begin{equation}
\prod_{j=1}^n(e^{2\hat p_j}-\lambda)=\prod_{j=1}^n(e^{2\hat p_j-x}-\lambda)
+\sgn(x)\sum_{j=1}^n\bigg[e^{2\hat p_j}|w_j|^2\prod_{\substack{k=1\\(k\neq j)}}^n
(e^{2\hat p_k-x}-\lambda)\bigg],
\label{C.17}
\end{equation}
where $\lambda$ is an arbitrary complex parameter. The constraint on $\hat p$ arises
from the fact that $|w_m|^2$ $(m=1,\dots,n)$ must be non-negative and not all zero,
because of the definition \eqref{3.56}.

Let us assume for a moment that the components of $\hat p$ are distinct such that
$\hat p_1>\dots>\hat p_n$. This enables us to express $|w_m|^2$ for all
$m\in \{1,\dots,n\}$ from the above equation by evaluating it at $n$ different values
of $\lambda$, viz. $\lambda=e^{2\hat p_m-x}$, $m=1,\dots,n$. We obtain the following
\begin{equation}
|w_m|^2=\sgn(x)(1-e^{-x})\prod_{\substack{j=1\\(j\neq m)}}^n
\frac{e^{2\hat p_j+x}-e^{2\hat p_m}}{e^{2\hat p_j}-e^{2\hat p_m}},
\quad m=1,\dots,n.
\label{C.18}
\end{equation}
For $x>0$ and any $\hat p$ with $\hat p_1>\dots>\hat p_n$ the formula \eqref{C.18}
implies that $|w_n|^2>0$ and for $m=1,\dots,n-1$ we have $|w_m|^2\geq 0$ if and only
if $\hat p_m-\hat p_{m+1}\geq x/2$. Similarly, if $x<0$ and $\hat p\in\R^n$ with
$\hat p_1>\dots>\hat p_n$, then \eqref{C.18} implies $|w_1|^2>0$ and for $m=2,\dots,n$
we have $|w_m|^2\geq 0$ if and only if $\hat p_{m-1}-\hat p_m\geq -x/2$. In summary,
if $\hat p_1>\dots>\hat p_n$, then $|w_m|^2\geq 0$ $\forall m$ implies that
$\hat p\in\bar\cC_x$.

Now, let us prove our assumption, that all components of $\hat p$ must be different.
Indirectly, suppose that some (or maybe all) of the $\hat p_j$'s coincide.
This can be captured by a partition of the positive integer
\begin{equation}
n=k_1+\dots+k_r,
\label{C.19}
\end{equation}
where $r<n$ (or equivalently, at least one integer $k_1,\dots,k_r$ must be
greater than $1$) and the indirect assumption can be written as
\begin{equation}
\hat p_1=\dots=\hat p_{k_1},\quad
\hat p_{k_1+1}=\dots=\hat p_{k_1+k_2},\quad\dots,\quad
\hat p_{k_1+\dots+k_{r-1}+1}=\dots=\hat p_{k_1+\dots+k_r}\equiv\hat p_n.
\label{C.20}
\end{equation}
Then \eqref{C.17} can be reformulated as
\begin{equation}
\prod_{j=1}^r(\Delta_j-\lambda)^{k_j}
=\prod_{j=1}^r(\Delta_je^{-x}-\lambda)^{k_j}
+\sgn(x)\sum_{m=1}^rZ_m\Delta_m(\Delta_me^{-x}-\lambda)^{k_m-1}
\prod_{\substack{j=1\\(j\neq m)}}^r(\Delta_je^{-x}-\lambda)^{k_j},
\label{C.21}
\end{equation}
where we introduced $r$ distinct variables
\begin{equation}
\Delta_1=e^{2\hat p_{k_1}},\quad
\Delta_2=e^{2\hat p_{k_1+k_2}},\quad\dots,\quad
\Delta_r=e^{2\hat p_{k_1+\dots+k_r}}\equiv e^{2\hat p_n},
\label{C.22}
\end{equation}
and $r$ non-negative real variables
\begin{equation}
\begin{gathered}
Z_1=|w_1|^2+\dots+|w_{k_1}|^2,\quad
Z_2=|w_{k_1+1}|^2+\dots+|w_{k_1+k_2}|^2,\\
\dots,\quad Z_r=|w_{k_1+\dots+k_{r-1}+1}|^2+\dots+|w_n|^2.
\end{gathered}
\label{C.23}
\end{equation}
Notice that $Z_1+\dots+Z_r=|w|^2=\sgn(x)e^{-x}(e^{nx}-1)>0$, therefore at least one
of the $Z_j$'s must be positive. Next, we define the rational function of $\lambda$
\begin{equation}
Q(\Delta,x,\lambda)=\prod_{j=1}^r\frac{(\Delta_j-\lambda)^{k_j}}
{(\Delta_je^{-x}-\lambda)^{k_j-1}},
\label{C.24}
\end{equation}
and use it to rewrite \eqref{C.21} as
\begin{equation}
Q(\Delta,x,\lambda)=\prod_{j=1}^r(\Delta_je^{-x}-\lambda)
+\sgn(x)\sum_{m=1}^rZ_m\Delta_m
\prod_{\substack{j=1\\(j\neq m)}}^r(\Delta_je^{-x}-\lambda).
\label{C.25}
\end{equation}
The above equation implies that all poles of $Q$ are apparent, i.e., there must
be cancelling factors in its numerator. This observation has a straightforward
implication on the $\Delta$'s.
\begin{center}
($\ast$)\quad For every index $m\in\{1,\dots,r\}$ with $k_m>1$,
there exists an index $s\in\{1,\dots,r\}$ s.t. $\Delta_s=\Delta_me^{-x}$
and $k_s\geq k_m-1$.
\end{center}
The quantities $Z_m=Z_m(\Delta,x)$ can be uniquely determined by evaluating
\eqref{C.25} at $r$ different values of the parameter $\lambda$, namely
$\lambda_m=\Delta_me^{-x}$ ($m=1,\dots,r$). However, there are $3$ disjoint
cases which are to be handled separately.

\noindent
\underline{Case 1:} $k_m=1$ and $\nexists s\in\{1,\dots,r\}$:
$\Delta_s=\Delta_me^{-x}$. Then we find
\begin{equation}
Z_m=\sgn(x)(1-e^{-x})e^{(n-1)x}\prod_{\substack{j=1\\(j\neq m)}}^r
\bigg(\frac{\Delta_j-\Delta_me^{-x}}{\Delta_j-\Delta_m}\bigg)^{k_j}>0.
\label{C.26}
\end{equation}
\underline{Case 2:} $k_m>1$ and $k_s=k_m-1$. Then we find
\begin{equation}
Z_m=(-1)^{k_m+1}\sgn(x)(1-e^{-x})e^{(n-k_m)x}\prod_{\substack{j=1\\(j\neq m,s)}}^r
\bigg(\frac{\Delta_j-\Delta_me^{-x}}{\Delta_j-\Delta_m}\bigg)^{k_j}>0.
\label{C.27}
\end{equation}
\underline{Case 3:} $k_m=1$ and $\exists s\in\{1,\dots,r\}$:
$\Delta_s=\Delta_me^{-x}$ or $k_m>1$ and $k_s>k_m-1$. Then we get
\begin{equation}
Z_m=0.
\label{C.28}
\end{equation}
Since there is at least one $Z_m$ which is positive, the set of indices belonging to
Case 1 or Case 2 must be non-empty. Introduce a real positive parameter $\varepsilon$
and associate to every degenerate configuration \eqref{C.20} a continuous
family of configurations, denoted by $\hat p(\varepsilon)$, with components
$\hat p(\varepsilon)_1,\dots,\hat p(\varepsilon)_n$ defined by the formulae
\begin{equation}
\begin{gathered}
\exp(2\hat p(\varepsilon)_a+a\varepsilon)=\Delta_1,\quad a=1,\dots,k_1,\\
\exp(2\hat p(\varepsilon)_{\sum_{m=1}^{j-1}k_m+a}+a\varepsilon)=\Delta_j,
\quad a=1,\dots,k_j,\quad j=2,\dots,r.
\end{gathered}
\label{C.29}
\end{equation}
This way coinciding components of $\hat p$ \eqref{C.20} are `pulled apart' to points
successively separated by $\varepsilon/2$. It is clear that with sufficiently small
separation the configuration $\hat p(\varepsilon)$ sits in the chamber
$\{\hat x\in\R^n\mid 0>\hat x_1>\dots>\hat x_n\}$. For such non-degenerate
configurations $\hat p(\varepsilon)$, let us consider the expressions
\begin{equation}
|w_\ell(\hat p(\varepsilon),x)|^2=\sgn(x)(1-e^{-x})
\prod_{\substack{j=1\\(j\neq\ell)}}^n
\frac{e^{2\hat p(\varepsilon)_j+x}-e^{2\hat p(\varepsilon)_\ell}}
{e^{2\hat p(\varepsilon)_j}-e^{2\hat p(\varepsilon)_\ell}},\quad\ell=1,\dots,n,
\label{C.30}
\end{equation}
which give the unique solution of equation \eqref{C.17} at $\hat p(\varepsilon)$.
The limits $\lim_{\varepsilon\to 0}|w_\ell(\hat p(\varepsilon),x)|^2$ exist,
and do not vanish for $\ell=k_1+\dots+k_m$ if $k_m$ belongs to Case 1 or Case 2.
For such $\ell=k_1+\dots+k_m$ we must have
\begin{equation}
\lim_{\varepsilon\to 0}|w_{k_1+\dots+k_m}(\hat p(\varepsilon),x)|^2=Z_m(\Delta,x)>0,
\label{C.31}
\end{equation}
where $Z_m$ is given by \eqref{C.26} in Case 1 and by \eqref{C.27} in Case 2.
It can be also seen that
\begin{equation}
|w_\ell(\hat p(\varepsilon),x)|^2\equiv 0
\quad\Longleftrightarrow\quad
\begin{cases}
\ell\notin\{k_1,k_1+k_2,\dots,k_1+\dots+k_r\}\\
\text{or}\\
\ell=k_1+\dots+k_m\ \text{with}\ k_m\ \text{from Case 3},
\end{cases}
\label{C.32}
\end{equation}
i.e. $|w_\ell(\hat p(\varepsilon),x)|^2$ vanishes identically except for the
components in \eqref{C.31}. Notice that for a small enough $\varepsilon$ some
coordinates of $\hat p(\varepsilon)$ are separated by less than $|x|/2$. Thus,
as it was shown at beginning the proof, we have $|w_\ell(\hat p(\varepsilon),x)|^2<0$
for some index $\ell$, which might depend on $\varepsilon$. Moreover, \eqref{C.32}
implies that the index in question must have the form
$\ell=k_1+\dots+k_{m^\ast}$ for some $m^\ast$ appearing in \eqref{C.31}.
But since the number of indices is finite, a monotonically decreasing sequence
$\{\varepsilon_N\}_{N=1}^\infty$ tending to zero can be chosen such that
$|w_{k_1+\dots+k_{m^\ast}}(\hat p(\varepsilon_N),x)|^2<0$ for all $N$.
This together with \eqref{C.32} gives the contradiction
\begin{equation}
0\geq \lim_{N\to\infty}|w_{k_1+\dots+k_{m^\ast}}(\hat p(\varepsilon_N),x)|^2
=Z_{m^\ast}(\Delta,x)>0
\label{C.33}
\end{equation}
proving that all components of $\hat p$ must be distinct. This concludes the proof.
\end{proof}

The above proof is a straightforward adaptation of the proofs of \cite[Lemma 5.2]{FK11}
and \cite[Theorem 2]{FK12}. We presented it since it could be awkward to extract the
arguments from those lengthy papers, and also our notations and the ranges of our
variables are different.

\section{Proof of Lemma \ref{lem:3.6}}
\label{sec:C.3}

We here prove the following equivalent formulation of Lemma \ref{lem:3.6}.

\begin{lemma}
\label{lem:C.2}
Suppose that $\frac{\pi}{2}\geq q_1>\dots>q_n>0$ and
\begin{equation}
\begin{bmatrix}\eta_L(1)&\0_n\\\0_n&\eta_L(2)\end{bmatrix}
\begin{bmatrix}\cos q&\ri\sin q\\\ri\sin q&\cos q\end{bmatrix}
\begin{bmatrix}\eta_R(1)^{-1} &\0_n\\\0_n&\eta_R(2)^{-1}\end{bmatrix}
=\begin{bmatrix}\cos q&\ri\sin q\\\ri\sin q&\cos q\end{bmatrix}
\label{C.34}
\end{equation}
for $\eta_L,\eta_R\in G_+$. Then
\begin{equation}
\eta_L(1)=\eta_R(2)=m_1,\quad\eta_L(2)=\eta_R(1)=m_2
\label{C.35}
\end{equation}
with some diagonal matrices $m_1,m_2\in\T_n$ having the form
\begin{equation}
m_1=\diag(a,\xi),\quad m_2=\diag(b,\xi),\quad\xi\in\T_{n-1},\
a,b\in\T_1,\quad\det(m_1 m_2)=1.
\label{C.36}
\end{equation}
If in addition $\frac{\pi}{2}>q_1$, then $m_1=m_2$.
\end{lemma}

\begin{proof}
The block off-diagonal components of the equality \eqref{C.34} give
\begin{equation}
\eta_L(1)=(\sin q)\eta_R(2)(\sin q)^{-1},\quad
\eta_L(2)=(\sin q)\eta_R(1)(\sin q)^{-1}.
\label{C.37}
\end{equation}
Since $\eta_L(1)^{-1}=\eta_L(1)^\dag$, the first of these relations implies
$\eta_R(2)=(\sin q)^2\eta_R(2)(\sin q)^{-2}$. As the entries of $(\sin q)$ are all
different, this entails that $\eta_R(2)$ is diagonal, and consequently we obtain the
relations in \eqref{C.35} with some diagonal matrices $m_1$ and $m_2$. On the other
hand, the block-diagonal components of \eqref{C.34} require that
\begin{equation}
\cos q =\eta_L(1)(\cos q)\eta_R(1)^{-1},\quad
\cos q =\eta_L(2)(\cos q)\eta_R(2)^{-1}.
\label{C.38}
\end{equation}
Since $\cos q_k\neq 0$ for $k=2,\dots,n$, the formula \eqref{C.36} follows.
If an addition $\cos q_1\neq 0$, then we also obtain from \eqref{C.38} that $a=b$,
i.e., $m_1=m_2=m$ with some $m\in\T_n$.
\end{proof}

\section{Auxiliary material on Poisson-Lie symmetry}
\label{sec:C.4}

The statements presented here are direct analogues of well-known results \cite{AMM81,GS82}
about Hamiltonian group actions with zero Poisson bracket on the symmetry group.
They are surely familiar to experts, although we could not find them in a reference.

Let us consider a Poisson-Lie group $G$ with dual group $G^\ast$ and a symplectic
manifold $P$ equipped with a left Poisson action of $G$. Essentially following Lu
\cite{Lu91} (cf.~Remark \ref{rem:C.6}), we say that the $G$-action admits the momentum map
$\psi\colon P\to G^\ast$ if for any $X\in\cG$, the Lie algebra of $G$, and any
$f\in C^\infty(P)$ we have
\begin{equation}
(\cL_{X_P}f)(p)=\langle X,\{f,\psi\}(p)\psi(p)^{-1}\rangle,\quad\forall p\in P,
\label{C.39}
\end{equation}
where $X_P$ is the vector field on $P$ corresponding to $X$, $\langle.,.\rangle$
stands for the canonical pairing between the Lie algebras of $G$ and $G^\ast$, and
the notation pretends that $G^\ast$ is a matrix group. Using the Hamiltonian vector
field $V_f$ defined by $\cL_{V_f}h=-\{f,h\}$ ($\forall h\in C^\infty(P)$), we can
spell out equation \eqref{C.39} equivalently as
\begin{equation}
(\cL_{X_P}f)(p)=-\langle X,\big(D_{\psi(p)}R_{\psi(p)^{-1}}\big)
\big((D_p\psi)(V_f(p))\big)\rangle,\quad\forall p\in P,
\label{C.40}
\end{equation}
where $D_p\psi\colon T_p P\to T_{\psi(p)}G^\ast$ is the derivative, and
$R_{\psi(p)^{-1}}$ denotes the right-translation on $G^\ast$ by $\psi(p)^{-1}$.
Since the vectors of the form $V_f(p)$ span $T_p P$, we obtain the following
characterization of the Lie algebra of the isotropy subgroup $G_p<G$ of $p\in P$.

\begin{lemma}
\label{lem:C.3}
With the above notations, we have
\begin{equation}
\mathrm{Lie}(G_p)=\big[\big(D_{\psi(p)}R_{\psi(p)^{-1}}\big)
\big(\mathrm{Im}(D_p\psi)\big)\big]^\perp.
\label{C.41}
\end{equation}
\end{lemma}

This directly leads to the next statement.

\begin{corollary}
\label{cor:C.4}
An element $\mu\in G^\ast$ is a regular value of the momentum
map $\psi$ if and only if $\mathrm{Lie}(G_p)=\{0\}$ for every
$p\in\psi^{-1}(\mu)=\{p\in P\mid\psi(p)=\mu\}$.
\end{corollary}

Let us further suppose that $\psi\colon P\to G^\ast$ is $G$-equivariant, with respect
to the appropriate dressing action of $G$ on $G^\ast$. Then we have
\begin{equation}
G_p<G_\mu,\quad\forall p\in\psi^{-1}(\mu).
\label{C.42}
\end{equation}
Here $G_p$ and $G_\mu$ refer to the respective actions of $G$ on $P$ and on $G^\ast$.
Corollary \ref{cor:C.4} and equation \eqref{C.42} together imply the following useful result.

\begin{corollary}
\label{cor:C.5}
If $G_\mu$ acts locally freely on $\psi^{-1}(\mu)$, then $\mu$
is a regular value of the equivariant momentum map $\psi$. Consequently,
$\psi^{-1}(\mu)$ is an embedded submanifold of $P$.
\end{corollary}

We finish by a clarifying remark concerning the momentum map.

\begin{remark}
\label{rem:C.6}
Let $B$ be the Poisson tensor on $P$, for which $\{f,h\}=B(df,dh)=\cL_{V_h}f$.
We can write $V_h=B^\sharp(dh)$ with the corresponding bundle map
$B^\sharp\colon T^\ast P\to TP$. Any $X\in\cG=T_eG=(T_{e'}G^\ast)^\ast$ extends to a
unique right-invariant 1-form $\vartheta_X$ on $G^\ast$ ($e\in G$ and $e'\in G^\ast$ are
the unit elements). With this at hand, equation \eqref{C.39} can be reformulated as
\begin{equation}
X_P=B^\sharp(\psi^\ast(\vartheta_X)),
\label{C.43}
\end{equation}
which is a slight variation of the defining equation of the momentum map found
in \cite{Lu91}.
\end{remark}

\section{On the reduced Hamiltonians}
\label{sec:C.5}

In this appendix we prove the claim, made in Section \ref{sec:3.4}, that on
the momentum surface $\Phi_+^{-1}(\mu)$ the Hamiltonians $\cH_j$, $j\in\Z^\ast$
\eqref{3.190} are linear combinations of $h_k$, $k=1,\dots,n$ \eqref{3.194}.
This will be achieved by establishing the form of the integer powers of the matrix
displayed in \eqref{3.193}, which we denote here by $\cL$, i.e.
\begin{equation}
\cL
=\begin{bmatrix}
e^{-2v}\1_n&-e^{-v}\alpha\\
e^{-v}\alpha^\dag&e^{2v}\1_n-\alpha^\dag\alpha
\end{bmatrix}.
\label{C.44}
\end{equation}

\begin{lemma}
\label{lem:C.7}
For any positive integer $j$, the $j$-th power of the $2n\times 2n$ matrix $\cL$
\eqref{C.44} reads
\begin{equation}
\cL^j=\begin{bmatrix}
\cL^j_{11}&\cL^j_{12}\\
\cL^j_{21}&\cL^j_{22}
\end{bmatrix},
\label{C.45}
\end{equation}
where $\cL^j_{11},\cL^j_{12},\cL^j_{21},\cL^j_{22}$ are $n\times n$ blocks of
the form
\begin{equation}
\begin{split}
&\cL^j_{11}=\sum_{m=1}^ja_m^{(j)}(\alpha\alpha^\dag)^{j-m},\quad
\cL^j_{12}=\alpha\sum_{m=1}^jb_m^{(j)}(\alpha^\dag\alpha)^{j-m},\\
&\cL^j_{21}=\alpha^\dag\sum_{m=1}^jc_m^{(j)}(\alpha\alpha^\dag)^{j-m},\quad
\cL^j_{22}=(-1)^j(\alpha^\dag\alpha)^j+\sum_{m=1}^jd_m^{(j)}(\alpha^\dag\alpha)^{j-m},
\end{split}
\label{C.46}
\end{equation}
with the $4j$ coefficients $a_m^{(j)},b_m^{(j)},c_m^{(j)},d_m^{(j)}$, $m=1,\dots,j$
depending only on the parameter $v$.
\end{lemma}

\begin{proof}
We proceed by induction on $j$. For $j=1$ the statement clearly holds,
and supposing that \eqref{C.45}-\eqref{C.46} is valid for some
fixed integer $j>0$ we simply calculate the $(j+1)$-th power
$\cL^{j+1}=\cL\cL^j$. This proves the statement.
\end{proof}

Our claim of linear expressibility follows at once, that is for any positive integer
$j$ we have
\begin{equation}
\cH_j=(-1)^jh_j
+\sum_{k=1}^{j-1}\frac{k}{j}(a_{j-k}^{(j)}+d_{j-k}^{(j)})h_k
+\frac{n}{2j}(a_j^{(j)}+d_j^{(j)}).
\label{C.47}
\end{equation}
Incidentally, one also obtains a recursion for the
coefficients $a_m^{(j)},b_m^{(j)},c_m^{(j)},d_m^{(j)}$
from the proof of Lemma \ref{lem:C.7}.
If they are required,
this should enable one to establish the values of the constants that occur in \eqref{C.47}.

As for the negative powers of $\cL$, one readily checks that the inverse of
$\cL$ is
\begin{equation}
\cL^{-1}
=\begin{bmatrix}
e^{2v}\1_n-\alpha\alpha^\dag&e^{-v}\alpha\\
-e^{-v}\alpha^\dag&e^{-2v}\1_n
\end{bmatrix},
\label{C.48}
\end{equation}
which has essentially the same form as $\cL$ does, thus the blocks of $\cL^{-j}$
$(j>0)$ can be expressed similarly as in Lemma \ref{lem:C.7}.
In fact, conjugating $\cL^{-1}$ with the $2n\times 2n$ involutory
block-matrix
\begin{equation}
\boldsymbol{C}=\begin{bmatrix}
\0_n&\1_n\\
\1_n&\0_n
\end{bmatrix},
\label{C.49}
\end{equation}
leads to the following formula
\begin{equation}
\boldsymbol{C}\cL^{-1}\boldsymbol{C}
=\begin{bmatrix}
e^{-2v}\1_n&-e^{-v}\alpha^\dag\\
e^{-v}\alpha&e^{2v}\1_n-\alpha\alpha^\dag
\end{bmatrix},
\label{C.50}
\end{equation}
which implies that the blocks of $\cL^{-j}$ are obtained from those of $\cL^j$
by reversing their order and interchanging the role of $\alpha$ and $\alpha^\dag$.
Furthermore, since $\tr((\alpha\alpha^\dag)^k)=\tr((\alpha^\dag\alpha)^k)$ we get
\begin{equation}
\cH_{-j}=-\cH_j\qquad \forall j\in\Z^\ast.
\label{C.51}
\end{equation}

\chapter{Appendices to Chapter \ref{chap:4}}
\label{chap:D}

It is clear that our results on the $2$-parameter family of hyperbolic systems \eqref{H}
open up a plethora of interesting  problems. Besides, based on our numerical calculations,
below we also wish to  discuss some possible generalizations in two further directions. 

\section{Lax matrix with spectral parameter}
\label{sec:D.1}

First, it is  a time-honoured principle that the inclusion of a spectral
parameter into the  Lax matrix of an integrable system can greatly enrich
the analysis by  borrowing techniques from complex geometry. Bearing this
fact in mind, with  the aid of the function
\be
    \Phi(x \, | \, \eta) 
    = e^{x \coth(\eta)} \left( \coth(x) - \coth(\eta) \right)
\label{Phi}
\ee
depending on the complex variables $x$ and $\eta$, over the phase space 
$P$ \eqref{P} we define the matrix valued smooth function 
$\cL = \cL(\lambda, \theta; \mu, \nu \, | \, \eta)$ with entries
\be
    \cL_{k, l}
    = \left(
        \ri \sin(\mu) F_k \bar{F}_l + \ri \sin(\mu - \nu) C_{k, l})
    \right)
    \Phi(\ri \mu + \Lambda_j - \Lambda_k \, | \, \eta)
    \qquad
    (k, l \in \bN_N).
\label{cL_matrix}
\ee
One of the outcomes of our numerical investigations is that for any values 
of $\eta$ the eigenvalues of $\cL$ provide a family of first integrals in 
involution for the van Diejen system \eqref{H}. Thinking of $\eta$ as a 
spectral parameter, let us also observe that, in the limit 
$\R \ni \eta \to \infty$, from $\cL$ we can recover our Lax matrix $L$ 
\eqref{L}; that is, $\cL \to L$. Although the spectral parameter dependent 
matrix $\cL$ does not take values in the Lie group $\UN(n, n)$ \eqref{G}, we 
find it interesting that the constituent function $\Phi$ \eqref{Phi} can 
be seen as a hyperbolic limit of the elliptic Lam\'e function, that plays 
a prominent role in the theory of the elliptic CMS and RS systems (see e.g. 
the papers \cite{Kr80,Ru87} and the monograph 
\cite{BBT03}). Therefore, it is tempting to think that an 
appropriate elliptic deformation of $\cL$ \eqref{L} may lead to a spectral 
parameter dependent Lax matrix of the elliptic van Diejen system with 
coupling parameters $\mu$ and $\nu$.

\section{Lax matrix with three couplings}
\label{sec:D.2}

In Chapter \ref{chap:4} we have studied the van Diejen system \eqref{H} with only two 
independent coupling parameters. Though a construction of a Lax matrix 
for the most general hyperbolic van Diejen system with five independent 
coupling parameters still seems to be out of reach, we can offer a plausible 
conjecture for a Lax matrix with \emph{three} independent coupling constants. 
Simply by generalizing the formulae appearing in the theory of the rational 
$\BC_n$ RSvD systems \cite{Pu13}, with the aid of an additional 
real parameter $\kappa$ let us define the real valued functions $\alpha$ 
and $\beta$ for any $x > 0$ by the formulae
\be
    \alpha(x)
    = \frac{1}{\sqrt{2}} 
        \left(
            1 + \left( 1 + \frac{\sin(\kappa)^2}{\sinh(2 x)^2} \right)^\half 
        \right)^\half
    \quad \text{and} \quad
    \beta(x) 
    = \frac{\ri}{\sqrt{2}}
        \left(
            -1 + \left( 1 + \frac{\sin(\kappa)^2}{\sinh(2 x)^2} \right)^\half
        \right)^\half.
\label{alpha_beta}
\ee
Built upon these functions, let us also introduce the Hermitian $N \times N$ 
matrix
\be
    h(\lambda) 
        = \begin{bmatrix}
        \diag(\alpha(\lambda_1), \ldots, \alpha(\lambda_n)) 
            & \diag(\beta(\lambda_1), \ldots, \beta(\lambda_n)) 
        \\
        -\diag(\beta(\lambda_1), \ldots, \beta(\lambda_n)) 
            & \diag(\alpha(\lambda_1), \ldots, \alpha(\lambda_n))
        \end{bmatrix}.
\label{h}
\ee
One can easily show that $h C h = C$, whence the matrix valued function
\be
    \tilde{L} = h^{-1} L h^{-1}
\label{L_tilde}
\ee
also takes values in the Lie group $\UN(n, n)$ \eqref{G}. Notice that the 
rational limit of matrix $\tilde{L}$ gives back the Lax matrix of the 
rational $\BC_n$ RSvD system, that first appeared in equation (4.51) of 
paper \cite{Pu12}. Moreover, upon setting
\be
    g = \mu,
    \quad
    g_0 = g_1 = \frac{\nu}{2},
    \quad
    g'_0 = g'_1 = \frac{\kappa}{2},
\label{3parameters}
\ee
for van Diejen's main Hamiltonian $H_1$ \eqref{H_vD} we get that
\be
    H_1 
    = 2 \sum_{a=1}^n
        \cosh(\theta_a) u_a 
        \left(
            1 + \frac{\sin(\kappa)^2}{\sinh(2 \lambda_a)^2}
        \right)^\half
    + 2 \sum_{a=1}^n
        \Real 
        \left( 
            z_a \frac{\sinh(\ri \kappa + 2 \lambda_a)}{\sinh(2 \lambda_a)}
        \right),
\label{H_1_mu_nu_kappa}
\ee
with the functions $z_a$ and $u_a$ defined in the equations \eqref{z_a} and 
\eqref{u_a}, respectively. The point is that, in complete analogy with 
\eqref{H1_vs_H}, one can establish the relationship
\be
    H_1 
    + 2 \cos 
        \left( 
            \nu + \kappa + (n - 1) \mu 
        \right) 
        \frac{\sin(n \mu)}{\sin(\mu)}
    = \tr(\tilde{L}).
\label{H1_vs_H_kappa}
\ee
Furthermore, based on numerical calculations for small values of $n$, 
it appears that the eigenvalues of $\tilde{L}$ \eqref{L_tilde} provide 
a commuting family of first integrals for the van Diejen system 
\eqref{H_1_mu_nu_kappa}. To sum up, we have numerous evidences that matrix 
$\tilde{L}$ \eqref{L_tilde} is a Lax matrix for the $3$-parameter family 
of van Diejen systems \eqref{H_1_mu_nu_kappa}, if the pertinent parameters 
are connected by the relationships displayed in \eqref{3parameters}.
As can be seen in \cite{Pu12}, the new parameter $\kappa$ causes
many non-trivial technical difficulties even at the level of the rational
van Diejen system. Part of the difficulties can be traced back to the fact 
that for $\sin(\kappa) \neq 0$ the matrix $\tilde{L}$ \eqref{L_tilde} does 
not belong to the symmetric space $\exp(\mfp)$ \eqref{symm_space}, whence 
the diagonalization of $\tilde{L}$ requires a less direct approach than 
that provided by the canonical form \eqref{mfp_and_mfa_and_K}. We wish to 
come back to these problems in later publications.

\chapter{Appendix to Chapter \ref{chap:5}}
\label{chap:E}

\section{Explicit form of the functions $\Lambda_{j,\ell}^y$}
\label{sec:E.1}

In this appendix, we display the building blocks \eqref{L^y+IV} of the global
elliptic Lax matrix explicitly. Below, $\xi$ varies in the closed simplex $\cA_y$
associated with a type (i) coupling $y$ \eqref{typeI-y} for fixed $p$ and $M$.
The function $\cJ$ was defined in \eqref{J2}. The trigonometric case is obtained
by simply replacing the $\ws$-function \eqref{sigma} everywhere by the sine function.

\smallskip
\noindent\underline{Special components:}
For $1\leq j\leq n-p$
$$
\Lambda_{j,j+p}^y(\xi)=-\sgn(M)\ws(y)
\frac{\big[\prod_{\substack{m=1\\(m\neq p)}}^{n-1}
\ws(\sum_{k=j}^{j+m-1}\xi_k-y)
\ws(\sum_{k=j+p}^{j+p+n-m-1}\xi_k+y)\big]^{\tfrac{1}{2}}}
{\prod_{m=1}^{n-1}\big[\ws(\sum_{k=j}^{j+m-1}\xi_k)
\ws(\sum_{k=j+p}^{j+p+m-1}\xi_k)\big]^{\tfrac{1}{2}}}.
$$
For $n-p<j\leq n$
$$
\Lambda_{j,j+p-n}^{y}(\xi)=\sgn(M)\ws(y)
\frac{\big[\prod_{\substack{m=1\\(m\neq p)}}^{n-1}
\ws(\sum_{k=j}^{j+m-1}\xi_k-y)
\ws(\sum_{k=j+p-n}^{j+p-m-1}\xi_k+y)\big]^{\tfrac{1}{2}}}
{\prod_{m=1}^{n-1}\big[\ws(\sum_{k=j}^{j+m-1}\xi_k)
\ws(\sum_{k=j+p-n}^{j+p-m-1}\xi_k)\big]^{\tfrac{1}{2}}}.
$$
\underline{Diagonal components:}
For $1\leq j=\ell\leq p$
$$
\Lambda_{j,j}^y(\xi)=\big[\cJ(|u_j|^2)\cJ(|u_{j+n-p}|^2)\big]^{\tfrac{1}{2}}
\frac{\big[\prod_{\substack{m=1\\(m\neq p)}}^{n-1}
\ws(\sum_{k=j}^{j+m-1}\xi_k-y)
\ws(\sum_{k=j}^{j+n-m-1}\xi_k+y)\big]^{\tfrac{1}{2}}}
{\prod_{m=1}^{n-1}\ws(\sum_{k=j}^{j+m-1}\xi_k)}.
$$
For $p<j=\ell\leq n$
$$
\Lambda_{j,j}^y(\xi)=\big[\cJ(|u_j|^2)\cJ(|u_{j-p}|^2)\big]^{\tfrac{1}{2}}
\frac{\big[\prod_{\substack{m=1\\(m\neq p)}}^{n-1}
\ws(\sum_{k=j}^{j+m-1}\xi_k-y)
\ws(\sum_{k=j}^{j+n-m-1}\xi_k+y)\big]^{\tfrac{1}{2}}}
{\prod_{m=1}^{n-1}\ws(\sum_{k=j}^{j+m-1}\xi_k)}.
$$
\underline{Components above the diagonal:}
For $1\leq j<\ell\leq p$
$$
\Lambda_{j,\ell}^y(\xi)
=\ws(y)\big[\cJ(|u_j|^2)\cJ(|u_{\ell+n-p}|^2)\big]^{\tfrac{1}{2}}
\frac{\big[\prod_{\substack{m=1\\(m\neq\ell-j,p)}}^{n-1}
\ws(\sum_{k=j}^{j+m-1}\xi_k-y)
\ws(\sum_{k=\ell}^{\ell+n-m-1}\xi_k+y)\big]^{\tfrac{1}{2}}}
{\prod_{m=1}^{n-1}\big[\ws(\sum_{k=j}^{j+m-1}\xi_k)
\ws(\sum_{k=\ell}^{\ell+m-1}\xi_k)\big]^{\tfrac{1}{2}}}.
$$
For $1\leq j<\ell\leq n$ with $p<\ell$ and $\ell\neq j+p$
$$
\Lambda_{j,\ell}^y(\xi)=\frac{\ws(y)\big[\cJ(|u_j|^2)
\cJ(|u_{\ell-p}|^2)\big]^{\tfrac{1}{2}}}{\sgn(j+p-\ell)}
\frac{\big[\prod_{\substack{m=1\\(m\neq\ell-j,p)}}^{n-1}
\ws(\sum_{k=j}^{j+m-1}\xi_k-y)
\ws(\sum_{k=\ell}^{\ell+n-m-1}\xi_k+y)\big]^{\tfrac{1}{2}}}
{\prod_{m=1}^{n-1}\big[\ws(\sum_{k=j}^{j+m-1}\xi_k)
\ws(\sum_{k=\ell}^{\ell+m-1}\xi_k)\big]^{\tfrac{1}{2}}}.
$$
\underline{Components below the diagonal:}
For $1\leq\ell<j\leq n$ with $\ell\leq p$ and $\ell\neq j+p-n$
$$
\Lambda_{j,\ell}^y(\xi)
=\frac{\ws(y)\big[\cJ(|u_j|^2)
\cJ(|u_{\ell+n-p}|^2)\big]^{\tfrac{1}{2}}}{\sgn(\ell+n -j- p )}
\frac{\big[\prod_{\substack{m=1\\(m\neq j-\ell,p)}}^{n-1}
\ws(\sum_{k=j}^{j+n-m-1}\xi_k-y)
\ws(\sum_{k=\ell}^{\ell+m-1}\xi_k+y)\big]^{\tfrac{1}{2}}}
{\prod_{m=1}^{n-1}\big[\ws(\sum_{k=j}^{j+m-1}\xi_k)
\ws(\sum_{k=\ell}^{\ell+m-1}\xi_k)\big]^{\tfrac{1}{2}}}.
$$
For $p<\ell<j\leq n$
$$
\Lambda_{j,\ell}^y(\xi)=\ws(y)\big[\cJ(|u_j|^2)
\cJ(|u_{\ell-p}|^2)\big]^{\tfrac{1}{2}}
\frac{\big[\prod_{\substack{m=1\\(m\neq j-\ell,p)}}^{n-1}
\ws(\sum_{k=j}^{j+n-m-1}\xi_k-y)
\ws(\sum_{k=\ell}^{\ell+m-1}\xi_k+y)\big]^{\tfrac{1}{2}}}
{\prod_{m=1}^{n-1}\big[\ws(\sum_{k=j}^{j+m-1}\xi_k)
\ws(\sum_{k=\ell}^{\ell+m-1}\xi_k)\big]^{\tfrac{1}{2}}}.
$$

\backmatter
\bookmarksetup{startatroot}

\refstepcounter{summ}
\chapter*{Summary}
\label{chap:sum}
\markboth{Summary}{}
\addcontentsline{toc}{chapter}{Summary}

Integrable many-body systems in one spatial dimension form an important class of exactly solvable Hamiltonian systems with their diverse mathematical structure and widespread applicability in physics \cite{Pe90,Po06,Su04}. Among these models, the systems of Calogero-Ruijsenaars type occupy a central position, due to their intimate relation with soliton theory \cite{Ru90}, and since many other interesting models (e.g. Toda lattice) can be obtained from them by taking various limits and analytic continuations \cite{Ru99}. Calogero-Ruijsenaars systems describe interacting particles moving on a line or circle. They come in different types called rational (I), hyperbolic (II), trigonometric (III), and elliptic (IV) depending on the functional form of their Hamiltonian. They exist in nonrelativistic and relativistic form, and both at the classical and quantum level. This already means sixteen different models (captured by Figure \ref{fig:3}). But there are other interesting extensions maintaining integrability, as is exemplified by versions attached to (non-A type) root systems \cite{OP76} or allowing internal degrees of freedom (spin) \cite{GH84}.

A fascinating aspect of Calogero-Ruijsenaars type systems is their duality relations, which first appeared in the famous paper \cite{KKS78} by Kazhdan, Kostant, and Sternberg, and was extensively explored by Ruijsenaars \cite{Ru88}. Here the word \emph{duality} refers to action-angle duality, which involves two Liouville integrable many-body Hamiltonian systems $(M,\omega,H)$ and $(\tilde{M},\tilde{\omega},\tilde{H})$ with Darboux coordinates $q,p$ and $\tilde{q},\tilde{p}$. These are said to be duals of each other if there is a global symplectomorphism $\cR\colon M\to\tilde M$ of the phase spaces, which exchanges the canonical coordinates with the action-angle variables for the Hamiltonians. Practically, this means that $H\circ\cR^{-1}$ depends only on $\tilde{q}$, while $\tilde H\circ\cR$ only on $q$. In more detail, $q$ are the particle positions for $H$ and action variables for $\tilde H$, and similarly, $\tilde{q}$ are the positions of particles modelled by the Hamiltonian $\tilde H$ and action variables for $H$. The idea that dualities can be interpreted in terms of Hamiltonian reduction can be distilled from \cite{KKS78} and was put forward explicitly in several papers in the 1990s, e.g.~\cite{FGNR00,GN95}.

During the past ten years, Feh\'{e}r and collaborators undertook the study of these dualities within the framework of reduction  \cite{FK09,FA10,AF10,FK11,FK12,Fe13,FKl14}. The list of action-angle dualities was enlarged by Pusztai \cite{Pu11,Pu11-2,Pu12,Pu13,Pu15} to dual pairs associated with non-A type root systems. The primary aim of the work presented in this thesis was to further develop these earlier developments. In this effort, our main tool was the method of Hamiltonian reduction, which belongs to the set of standard toolkits applicable to study a great variety of problems ranging from geometric mechanics to field theory and harmonic analysis \cite{MR02,Et07}. It is especially useful in the theory of integrable Hamiltonian systems \cite{BBT03}, where one of the maxims is that one should view the systems of interests as reductions of obviously solvable `free' systems \cite{OP81}. This is often advantageous, for example since the reduction produces global phase spaces on which the reduced free flows are automatically complete, which is an indispensable property of any integrable system.

The thesis is divided into two parts. Part \ref{part:1} of the thesis takes this reduction approach to Calogero-Ruijsenaars type systems. Part \ref{part:2} collects our work that were initiated by reduction results, but were obtained using direct methods.

\section*{Results}

Here we collect the main results of the thesis, going chapter by chapter.
In each title, we cite our related contribution (\ref{chap:publ})

\subsection*{Spectral coordinates of the rational Calogero-Moser system \citepalias{Go16}}

Chapter \ref{chap:1} starts by recalling the pivotal work of Kazhdan, Kostant, Sternberg \cite{KKS78}. Subsection \ref{subsec:1.1.2} is a recap of their result about the complete integrability and action-angle duality for the rational Calogero-Moser system in the context of Hamiltonian reduction. In Section \ref{sec:1.2}, we put these ideas into use, when we identify the canonical variables of \cite{FM16} in terms of the reduction picture (Lemma \ref{lem:1.1}), and prove the relation conjectured in that paper (Theorem \ref{thm:1.2}). We attain Sklyanin's formula as a corollary (Corollary \ref{cor:1.3}).

\subsection*{Action-angle duality for the trigonometric BC${}_{\textit{n}}$ Calogero-Moser-Sutherland system \citepalias{FG14,Go14,GF15}}

Chapter \ref{chap:2} is a study on the trigonometric Sutherland system attached to the $\BC_n$ root system. We start by providing a physical interpretation of the model in Section \ref{sec:2.1}. This is followed by the preparatory Section \ref{sec:2.2}, where the group-theoretic ingredients of reduction are introduced together with the unreduced Abelian Poisson algebras and the symplectic reduction to be performed. In Section \ref{sec:2.3}, we solve the momentum equations, hence obtaining the first model of the reduced phase space (Theorem \ref{thm:2.1}). In a nutshell, this first model of the reduced phase space carries the Sutherland Hamiltonian as the reduction of the free Hamiltonian governing geodesic motion on the Lie group $\UN(2n)$ (Corollary \ref{cor:2.2}). The content of this section, and even its quantum analogue, is fairly standard \cite{FP10}. The heart of the second chapter is Subsection \ref{subsec:2.3.2}, in which we first describe the reduced phase space locally (Theorem \ref{thm:2.3}), then extend this construction to a global model (Theorem \ref{thm:2.16}). This gives rise to the action-angle dual of the trigonometric $\BC_n$ Sutherland system. The main Hamiltonian of the dual system turns out to be a real form of the rational Ruijsenaars-Schneider Hamiltonian with three independent couplings. In Section \ref{sec:2.4}, we apply our duality map to various problems, such as finding the equilibrium configuration of the Sutherland model, proving the maximal superintegrability of the dual model, and connecting the Hamiltonians of the hyperbolic analogue with a family of Hamiltonians found by van Diejen.

\subsection*{A Poisson-Lie deformation of the trigonometric BC${}_{\textit{n}}$ Calogero-Moser-Sutherland system \citepalias{FG15,FG16}}

Chapter \ref{chap:3} generalises certain parts of the previous chapter as it derives a $1$-parameter deformation of the trigonometric $\BC_n$ Sutherland system by applying Hamiltonian reduction to the Heisenberg double of the Poisson-Lie group $\SU(2n)$. Here, we were also motivated by the recent work of Marshall \cite{Ma15}.
We start in Section \ref{sec:3.1} by defining the reduction of interest. In Section \ref{sec:3.2}, we observe that several technical results of \cite{FK11} can be applied for analysing the reduction at hand, and solve the momentum map constraints by taking advantage of this observation (Proposition \ref{prop:3.2}, \ref{prop:3.3}). The main result of this chapter is contained in Section \ref{sec:3.3}, where we characterize the reduced system. In Subsection \ref{subsec:3.3.1}, we prove that the reduced phase space is smooth (Theorem \ref{thm:3.9}). Then in Subsection \ref{subsec:3.3.2} we focus on a dense open submanifold on which the Hamiltonian `lives' (Theorem \ref{thm:3.11}). The demonstration of the Liouville integrability of the reduced free flows is given in Subsection \ref{subsec:3.3.3}. In particular, we prove the integrability of the completion of the pertinent system carried by the full reduced phase space. Our main result in this chapter is Theorem \ref{thm:3.14} (proved in Subsection \ref{subsec:3.3.4}), which establishes a globally valid model of the reduced phase space. In Section \ref{sec:3.4}, we complete (Theorem \ref{thm:3.22}) the recent derivation of the hyperbolic analogue by Marshall \cite{Ma15}.

\subsection*{Lax representation of the hyperbolic BC${}_{\textit{n}}$ Ruijsenaars-Schneider-van Diejen system \citepalias{PG16}}

Chapter \ref{chap:4} contains our construction of a Lax pair for the classical hyperbolic
van Diejen system with two independent coupling parameters. In Section \ref{sec:4.1}, we start with a short overview on some relevant group-theoretic facts and fix notation. In
Section \ref{sec:4.2}, we define our Lax matrix for the van Diejen system, and also investigate its main algebraic properties (Proposition \ref{PROPOSITION_L_in_G}, Lemma \ref{LEMMA_L_in_exp_p}).
These results can be used to show that the dynamics can be solved by a projection
method, which in turn allows us to initiate the study of the scattering properties. This was done by B.G.~Pusztai and is included in Subsections \ref{subsec:4.2.2}, \ref{subsec:4.3.1},
\ref{subsec:4.3.3}, \ref{subsec:4.3.4}, \ref{subsec:4.3.5}. In Subsection \ref{subsec:4.4.1}, we elaborate the link between our special $2$-parameter family of Hamiltonians and the most general $5$-parameter family of hyperbolic van Diejen systems. We affirm the equivalence between van Diejen's
commuting family of Hamiltonians and the coefficients of the characteristic 
polynomial of our Lax matrix (Lemma \ref{LEMMA_linear_relation}). Based on this technical result, we can infer that the eigenvalues of the proposed Lax matrix provide a commuting family of first integrals for the Hamiltonian system (Theorem \ref{THEOREM_commuting_eigenvalues}).

\subsection*{Trigonometric and elliptic Ruijsenaars-Schneider models on the complex projective space \citepalias{FG16-2}}

Chapter \ref{chap:5} is concerned with the compactified Ruijsenaars-Schneider systems with so-called type (i) couplings \cite{FKl14}. We reconstruct the corresponding compactification on $\CP^{n-1}$ using only direct, elementary methods (Proposition \ref{prop:5.1}, Theorem \ref{thm:5.5}). Such construction was not known previously except for special type (i) cases. By doing so, we gain a better understanding of the structure of these trigonometric systems. This part of the chapter fills Sections \ref{sec:5.1} and \ref{sec:5.2}. In Section \ref{sec:5.3}, we explain that the direct method is applicable to obtain type (i) compactifications of the elliptic Ruijsenaars-Schneider system as well (Theorem \ref{thm:5.7}). This new result extends the remarks of Ruijsenaars \cite{Ru90,Ru99}.

\selectlanguage{magyar}

\chapter*{Összefoglal\'o}
\markboth{\"Osszefoglal\'o}{}
\addcontentsline{toc}{chapter}{\"Osszefoglal\'o}

\section*{Bevezetés}

Az egydimenziós integrálható sokrészecske modellek széleskörû fizikai alkalmazásaik és gazdag matematikai hátterük okán az egzaktul megoldható hamiltoni rendszerek fontos osztályát képezik. A Calogero-Ruijsenaars típusú rendszerek központi helyet foglalnak el ezek között. Ez egyrészt a szolitonok elméletével való kapcsolatuknak, másrészt annak köszönhetõ, hogy számos más érdekes modell (pl. a Toda-molekula) származtatható belõlük, határesetekként és analitikus kiterjesztéssel. A Calogero-Ruijsenaars típusú modellek egyenesen vagy körön mozgó kölcsönható részecskéket írnak le. A kölcsönhatás jellege szerint négy típust különböztetünk meg. Ezek a racionális (I), a hiperbolikus (II), a trigonometrikus (III) és az elliptikus (IV) rendszerek. A modelleknek létezik nemrelativisztikus és relativisztikus, valamint klasszikus- és kvantummechanikai változata is. Integrálható általánosításaik közül kiemelendõk a gyökrendszereken alapuló és a belsõ szabadsági fokot is megengedõ (spin) modellek.

\section*{Tudományos el\H{o}zmények}

Tekintsük az $(M,\omega,H)$, $(\tilde{M},\tilde{\omega},\tilde{H})$ Liouville integrálható rendszereket. A két rendszer \emph{hatás-szög dualitásáról} akkor beszélünk, ha létezik a fázisterek között egy $\mathcal{R}\colon M\to\tilde{M}$ szimplektomorf leképezés, amely az $\tilde{M}$ tér valamely $(\tilde{q},\tilde{p})$ kanonikus koordinátáit a $H$ Hamilton-függvényhez tartozó rendszer hatás-szög változóiba viszi át, és fordítva, az $M$ térnek léteznek $(q,p)$ kanonikus koordinátái, amelyek a $\tilde{H}$ Hamilton-függvény rendszerének hatás-szög változói lesznek. Ekkor $\mathcal{R}$ az ún. \emph{hatás-szög leképezés}. Ezáltal $H\circ\mathcal{R}^{-1}$ kizárólag $\tilde{q}$-tól, és $\tilde H \circ \mathcal{R}$ csakis $q$-tól függ. Mindemellett az általunk vizsgált rendszerek esetén a $H$ Hamilton-függvény $(q,p)$ koordinátás alakja kölcsönható részecskék egy olyan modelljét adja, amelyben $q$ a részecske-koordináták szerepét játssza, és hasonlóan, a $\tilde{H}$ függvény a $(\tilde{q},\tilde{p})$ változókkal kifejezve $\tilde{q}$ pozíciókba elhelyezett részecskék kölcsönhatását írja le. Ezen különleges kapcsolat jelentõségét mutatja, hogy a kvantummechanikai tárgyalásban is megjelenik mint a fontos speciális függvényekkel kifejezett hullámfüggvények bispektrális tulajdonsága \cite{DG86,Ru90}.

Dualitásban álló sokrészecske rendszereket vizsgált Ruijsenaars \cite{Ru88,Ru90-2,Ru95,Ru99}, aki közvetlen úton konstruált hatás-szög leképezéseket az A$_{n-1}$ gyökrendszerhez asszociált Calogero-Ruijsenaars és Toda típusú integrálható rendszerekhez. Ezen dualitási relációk redukciós ér\-tel\-me\-zésérõl számos cikk született az 1990-es években, pl. \cite{FGNR00,GN95}. Az ezekben felmerült ötleteket továbbfejlesztve és szisztematikusan kidolgozva az elmúlt évtizedben Fehér és munkatársai ilyen kapcsolatokat vezettek le redukciós keretek között \cite{FK09,FA10,AF10,FK11,FK12,Fe13,FKl14} (ld. \emph{Alkalmazott módszerek}). Az a természetes várakozás, hogy hasonló dualitások fennállnak másfajta gyökrendszerek esetén is Pusztai munkájában nyert igazolást \cite{Pu11,Pu11-2,Pu12,Pu13,Pu15}.

Kutatásunkban során olyan új eredmények elérését t\H{u}ztük ki célul (ld. \ref{chap:publ}), amelyek ezen korábbi fejleményekhez kapcsolódnak.

\section*{Célkit\H{u}zések}

A disszertációban bemutatott doktori munka céljai az alábbi pontokba foglalhatók össze:

\begin{enumerate}
\item[I.] A racionális A$_{n-1}$ Calogero-Moser modell hatás-szög változóira vonatkozó Sklyanin-formula bizonyítása redukciós módszerrel.
\item[II.] A trigonometrikus BC$_n$ Sutherland rendszer hatás-szög duálisának részletes kidolgozása hamiltoni redukciós keretek között.
\item[III.] Az elõzõ pont eredményeit és Marshall egy korábbi munkáját általánosítva a trigonometrikus BC$_n$ Sutherland modell egy $1$-paraméteres integrálható deformációjának megalkotása.
\item[IV.] A Lax formalizmus kiterjesztése az egynél több csatolási állandót tartalmazó általánosított hiperbolikus Ruijsenaars-Schneider rendszerekre.
\item[V.] Új elliptikus A$_{n-1}$ Ruijsenaars-Schneider modellek konstruálása az $n$-dimenziós komplex projektív téren.
\end{enumerate}

\noindent
A fenti kutatási elképzeléseket sikeresen valósítottuk meg, sõt további, a kezdeti várakozásokon túlmutató elõrelépéseket is tettünk.

\section*{Alkalmazott módszerek}

A fenti célok eléréséhez az úgynevezett \emph{hamiltoni redukció} módszerét,
valamint standard matematikai eszközöket alkalmaztunk.

Dióhéjban összefoglalva, a redukciós eljárás során a levezetendõ rendszerek
részecskéinek bonyolult mozgását egy magasabb dimenziós térben mozgó, nagyfokú
szimmetriával bíró szabad részecske `ügyesen' választott vetületeként nyerjük.

Pontosabban fogalmazva, a redukció kiindulásaként egy csoportelméleti eredetû
fázisteret választunk. Ez lehet egy $X$ mátrix Lie-csoport vagy Lie-algebra
$P=T^\ast X$ koérintõnyalábja. A $P$ nyalábon természetes módon megadható $\Omega$
szimplektikus forma és egy $\mathcal{H}\colon P\to\mathbb{R}$ Hamilton-függvény megválasztása
egy $(P,\Omega,\mathcal{H})$ hamiltoni rendszert eredményez. Ha a $\mathcal{H}$ Hamilton-függvény kellõen egyszerû alakot ölt, akkor a mozgásegyenletek explicit módon megoldhatók, sõt akár $\{\mathcal{H}_j\}$ Poisson kommutáló elsõ integrálok egy egész serege felírható. Ekkor egy megfelelõen választott $G$ csoport hatása az $X$ (és ezáltal a $P$) téren, amelyre nézve a $\mathcal{H}_j$ függvények invariánsak, lehetõvé teszi az asszociált $\Phi\colon P\to\mathfrak{g}^\ast$ momentum leképezés felírását. A $\Phi$ momentum leképezés értékének $\mu\in\mathfrak{g}^\ast$ elemre történõ rögzítése egy $\Phi^{-1}(\mu)$ szintfelületet jelöl ki a $P$ fázistérben. Ez a kényszerfelület a momentum érték $G_\mu\subset G$ izotrópia-részcsoportjának pályáiból áll. Ezen pályák alkotják a $(P_\mathrm{red},\omega_\mathrm{red},H)$ redukált fázistér pontjait. A fenti konstrukciónak köszönhetõen az involúcióban álló $\{\mathcal{H}_j\}$ mozgásállandók hamiltoni folyamai  invariánsan hagyják a momentum szintfelületet és a $\{\mathcal{H}_j\}$ függvények állandók $G_\mu$ pályái mentén. Következésképpen értelmezhetõk a függvények $H_j\colon P_\mathrm{red}\to\mathbb{R}$ redukciói, amelyek Poisson zárójele továbbra eltûnik, és ily módon a származtatott $(P_\mathrm{red},\omega_\mathrm{red},H)$ hamiltoni rendszer Liouville értelemben integrálható. A gyakorlatban jellemzõen a redukált fázisteret a $G_\mu$ csoport pályáinak egy $S$ sima szelésével azonosítjuk. Ilyen szelést a $\Phi=\mu$ egyenlet megoldásával nyerünk. Két így kapott $S,\tilde{S}$ modell lehet egymás hatás-szög duálisa.

\section*{Új tudományos eredmények}

Az alábbiakban röviden ismertetem a disszertációban elért tudományos eredményeket.
A kapcsolódó publikációkat a disszertáció végén található \ref{chap:publ} lista gy\H{u}jti össze.
A közleményekre az azoknak megfelel\H{o} tézispontok címében, illetve szükség esetén a szövegben hivatkozok.

\subsection*{I. A racionális Calogero-Moser rendszer spektrális koordinátái \citepalias{Go16}}

\begin{itemize}
\item[+] A hamiltoni redukció módszerének alkalmazásával azonosítottam a racionális Calogero-Moser rendszer Falqui és Mencattini \cite{FM16} által felírt kanonikus koordinátáit.
\item[+] Bizonyítottam egy Falqui és Mencattini \cite{FM16} által megsejtett össze\-függést.
\item[+] Igazoltam Sklyanin \cite{Sk13} formuláját, amely spektrális kanonikus koordinátákat szolgáltat a racionális Calogero-Moser rendszerhez.
\end{itemize}

\subsection*{II. A trigonometrikus BC${}_{\textit{n}}$ Sutherland rendszer hatás-szög duálisa \citepalias{FG14,Go14,GF15}}

\begin{itemize}
\item[+] Hamiltoni redukció útján származtattam a trigonometrikus BC$_n$ Sutherland modell hatás-szög duálisát, amelyben a racionális BC$_n$ Ruijsenaars-Schneider rendszer egy valós formáját ismertem fel.
\item[+] Bizonyítottam, hogy a duális modell lokális leírásában használt változók kanonikus koordináta-rendszert alkotnak \citepalias{Go14}.
\item[+] Felírtam ezen duális rendszer Lax-mátrixát explicit alakban.
\item[+] Megadtam a duális modell fázisterének, valamint Lax-mátrixának globális leírását \citepalias{FG14}.
\item[+] Jellemeztem a trigonometrikus BC$_n$ Sutherland modell egyensúlyi konfigurációit a hatás-szög dualitás segítségével.
\item[+] További alkalmazásként igazoltam, hogy a duális rendszer $(n-1)$ extra mozgás\-állandóval rendelkezik, következésképp maximálisan szuperintegrálható.
\item[+] Végül bizonyítottam, hogy a hiperbolikus BC$_n$ Sutherland modell Pusztai \cite{Pu12} által konstruált involúcióban álló mozgásállandói és a van Diejen \cite{vD94} által talált Poisson kommutáló elsõ integrálok ekvivalensek, azaz ugyanazt az abeli algebrát generálják \citepalias{GF15}. A két említett függvénycsalád közötti lineáris kapcsolatot explicit formában felírtam és igazoltam.
\end{itemize}

\subsection*{III. A trigonometrikus BC${}_{\textit{n}}$ Sutherland rendszer\\Poisson-Lie deformációja \citepalias{FG15,FG16}}

\begin{itemize}
\item[+] Marshall korábbi, hiperbolikus esettel foglalkozó munkáját \cite{Ma15} általánosítva levezettem a trigonometrikus BC$_n$ Sutherland rendszer egy $1$-paraméteres integrálható deformációját a $2n\times 2n$-es egységnyi determinánsú unitér mátrixok alkotta Poisson-Lie csoport Heisenberg duplájának általánosított Marsden-Weinstein redukciójából.
\item[+] Megoldottam a momentum kényszer-egyenletet, visszavezetve azt egy Fehér és Klim\v c\'\i k \cite{FK11} által korábban már részletesen vizsgált egyenletre.
\item[+] A fejezet fõ eredményeként globálisan jellemeztem a redukált rendszert \citepalias{FG15}. Igazoltam, hogy a levezetett rendszer Liouville integrálható.
\item[+] Továbbá megmutattam, hogy a modell miként kapható meg van Diejen \cite{vD94} öt csatolási állandót tartalmazó modelljébõl. Ezáltal a levezetett modellt sikerült beilleszteni a Calogero-Ruijsenaars típusú integrálható rendszerek közé.
\item[+] Végül teljessé tettem a hiperbolikus verzió Marshall \cite{Ma15} által adott származtatását \citepalias{FG16}.
\end{itemize}

\subsection*{IV. A hiperbolikus BC${}_{\textit{n}}$ Ruijsenaars-Schneider-van Diejen rendszer Lax reprezentációja \citepalias{PG16}}

\begin{itemize}
\item[+] Igazoltam, hogy a Lax mátrix eleme az $(n,n)$-szignatúrájú `belsõ szorzással' definiált pszeudounitér mátrixok Lie-csoportjának.
\item[+] Pusztai korábbi eredményét \cite{Pu11-2} felhasználva bizonyítottam, hogy a Lax mátrix pozitív definit.
\item[+] Megmutattam a Pusztai által levezetett szóráselméleti eredmények segítségével, hogy a Lax mátrixból származó spektrális invariánsok és van Diejen \cite{vD94} öt paramétert tartalmazó Poisson kommutáló függvénycsaládjának megfelelõ specializációja ekvivalensek.
\item[+] Ennek segítségével bebizonyítottam, hogy a Lax mátrix független sajátértékei Poisson kommutáló mozgásállandók teljes rendszerét alkotják.
\end{itemize}

\subsection*{V. Trigonometrikus és elliptikus Ruijsenaars-Schneider modellek a komplex projektív téren \citepalias{FG16-2}}

\begin{itemize}
\item[+] Megvizsgáltam a Fehér és Kluck \cite{FKl14} által korábban felfedezett ún. \emph{egyes típusú} csatolási állandóval jellemzett kompaktifikált Ruijsenaars-Schneider modelleket, és közvetlen, elemi úton megmutattam, hogy a trigonometrikus esetben ezen rendszerek miként ágyazhatók be a megfelelõ komplex projektív térbe.
\item[+] A trigonometrikus esetben alkalmazott eljárást általánosítottam az elliptikus potenciálok esetére is, ezáltal új elliptikus Ruijsenaars-Schneider modelleket konstruáltam a komplex projektív téren. Ezzel kiterjesztettem Ruijsenaars korábbi eredményeit \cite{Ru90,Ru99}.
\end{itemize}

\selectlanguage{english}

\renewcommand\bibname{Publications}
\refstepcounter{publ}

\renewcommand\bibname{Bibliography}
\refstepcounter{bibl}

\end{document}